\documentclass[11pt]{article}

\usepackage[sectionbib]{natbib}
\usepackage{algorithm}
\usepackage{algorithmic}
\usepackage{booktabs}
\usepackage{lineno}
\usepackage{graphicx}
\usepackage{color,xcolor}
\usepackage{subfigure} 
\usepackage{geometry}
\usepackage{appendix}
\usepackage{multirow}

\usepackage[colorlinks=true,linkcolor=blue,citecolor=blue]{hyperref}%

\usepackage{amsmath}
\usepackage{amssymb}
\usepackage{amsthm}
\newtheorem{cond}{Condition}

\usepackage[utf8]{inputenc} 
\usepackage{hyperref}       
\usepackage{url}            
\usepackage{booktabs}       
\usepackage{amsfonts}       
\usepackage{nicefrac}       
\usepackage{microtype}      

\usepackage{natbib}

\usepackage{array,epsfig,rotating}
\usepackage{tabularx}
\usepackage{subfigure}
\usepackage{graphicx}
\usepackage{verbatim}
\usepackage{arydshln}
\usepackage{authblk}

\newtheorem{theorem}{Theorem}
\newtheorem{lemma}{Lemma}

\newtheorem{corollary}{Corollary}

\newtheorem{assumption}{Assumption}

\allowdisplaybreaks

\def\AA{\mathbf A}

\def\hat{\widehat}
\def\tilde{\widetilde}

\newcommand{\ttt}{\boldsymbol \theta}

\newcommand{\ee}{\boldsymbol e}
\newcommand{\pp}{{\boldsymbol p}}

\newcommand{\RR}{\mathbf R}
\newcommand{\XX}{\mathbf X}
\newcommand{\ZZ}{\mathbf Z}
\newcommand{\DD}{\mathbf D}
\newcommand{\zz}{\mathbf z}
\newcommand{\BT}{\mathbf T}
\newcommand{\UU}{\mathbf U}
\newcommand{\VV}{\mathbf V}
\newcommand{\MM}{\mathbf M}
\newcommand{\WW}{\mathbf W}
\newcommand{\II}{\mathbf I}

\newcommand{\TT}{\boldsymbol \theta}

\newcommand{\mmu}{\boldsymbol \mu}
\newcommand{\uu}{\boldsymbol u}
\newcommand{\vv}{\boldsymbol v}

\newcommand{\ME}{\mathbb E}
\newcommand{\MP}{\mathbb P}

\DeclareMathOperator*{\argmax}{arg\,max}

\def\spacingset#1{\renewcommand{\baselinestretch}%
{#1}\small\normalsize} \spacingset{1}
\spacingset{1.5}

\usepackage{float}

\begin{document}



\title{A Tensor-EM Method for Large-Scale Latent Class Analysis with Binary Responses}
\author{Zhenghao Zeng}
\affil{Carnegie Mellon University}
\author{Yuqi Gu}
\affil{Columbia University}
\author{Gongjun Xu}
\affil{University of Michigan, Ann Arbor}
\date{}

\maketitle

\begin{abstract}
Latent class models are powerful statistical modeling tools widely used in psychological, behavioral, and social sciences. In the modern era of data science, researchers often have access to response data collected from large-scale surveys or assessments, featuring many items (large $J$) and many subjects (large $N$).
This is in contrary to the traditional regime with fixed $J$ and large $N$. 
To analyze such large-scale data, it is important to develop methods that are both   computationally efficient and theoretically valid.
In terms of computation, the conventional EM algorithm for latent class models tends to have a slow algorithmic convergence rate for large-scale data and may converge to some local optima instead of the maximum likelihood estimator (MLE). 
Motivated by this, we introduce the tensor decomposition perspective into latent class analysis with binary responses. Methodologically, we propose to use a moment-based tensor power method in the first step, and then use the obtained estimates as initialization for the EM algorithm in the second step. 
Theoretically, we establish the clustering consistency of the MLE in assigning subjects into latent classes when $N$ and $J$ both go to infinity. 
Simulation studies suggest that the proposed tensor-EM pipeline enjoys both good accuracy and computational efficiency for large-scale data with binary responses. 
We also apply the proposed method to an educational assessment dataset as an illustration.
\end{abstract}


\newpage
\def\spacingset#1{\renewcommand{\baselinestretch}%
{#1}\small\normalsize} \spacingset{1}
\spacingset{1.7}

\section{Introduction} \label{Intro}

Latent class models (LCMs) \citep{lazarsfeld1968latent, goodman1974} are powerful statistical modeling tools widely used in psychological, behavioral, and social sciences. LCMs use a categorical latent variable to model the unobserved heterogeneity of multivariate categorical data and identify meaningful latent subgroups of subjects. 
LCMs have seen broad applications in a variety of scientific fields, including psychology and psychiatry \citep{bucholz2000latent, keel2004application},
sociology and organizational research \citep{vermunt2003applications, wang2011latent}, and
biomedical and epidemiological studies \citep{bandeen1997, dean2010, kongsted2017latent}.
For instance, \cite{bucholz2000latent} explored the existence of potential subtypes of Antisocial Personality Disorder via latent class analysis. \cite{keel2004application} applied LCMs to empirically define four eating disorder phenotypes and identified features that differentiate between phenotypes. \cite{wang2011latent} summarized several areas in organizational research where LCMs are particularly useful, such as  
identifying unobserved subpopulations  and recognizing the unobserved heterogeneity in measurement functioning. \cite{vermunt2003applications} provided an overview on applications of LCM and its extensions in social science research. \cite{dean2010} considered variable selection in LCMs and identified meaningful group structures in single nucleotide polymorphism data. \cite{kongsted2017latent} introduced and illustrated the applications of LCMs in health research.
There are also various extensions and generalizations building upon LCMs, including LCMs with covariates or distal outcomes \citep{vermunt2010lcmcov, lanza2013subgroup,ouyang2022identifiability}, longitudinal LCMs and latent transition analysis \citep{dunn2006characterizing, collins2009latent},  factor mixture models \citep{lubke2005investigating,muthen1999finite},  and also restricted LCMs known as diagnostic classification models that involve multiple categorical latent variables \citep{rupp2008dcm, xu2017identifiability, von2019handbook}.
LCMs may also serve as initial modeling step before fitting a more delicate cognitive diagnostic model \citep{ma2022learning}. For general introductions and applications of LCMs, see \cite{hagenaars2002applied} and \cite{collins2009latent}.

In this paper, we focus on LCMs with binary responses for large-scale data, which are typically collected in modern educational assessments (correct/wrong responses) and psychological or social science surveys (yes/no responses).
Such data are characterized by many test takers with large $N$ and also many items with large $J$. 
This is in contrary to the traditional regime with large $N$ and fixed $J$, a relatively well understood setting.
Such a large scope of data poses challenges to classical statistical analysis methods and calls for new developments for LCMs.
In the following, we summarize the two main questions motivating our study.

The first question of large-scale latent class analysis is how to perform computations efficiently. A conventional estimation method is the Expectation-Maximization \citep[EM;][]{dempster1977em} algorithm to maximize the marginal likelihood. EM algorithms have two potential drawbacks, the slow algorithmic convergence rate in high-dimensional problems and a tendency to converge to some local optima when the initial values are poorly chosen \citep{balakrishnan2017em}. In practice, researchers often run EM with many random initializations and select the one that gives the largest log-likelihood value. This procedure can be very time-consuming, especially for large-scale data. It is thus desirable to develop more efficient computational tools for large-scale latent class analysis.
Motivated by this, we introduce the tensor decomposition perspective into latent class analysis.
%


There has been active research on tensor decompositions since they were introduced in \cite{hitchcock1927expression}. 
The concept of tensors appeared in the literature of psychometrics dating back to 1960s-1970s \citep{tucker1964extension, tucker1966some, harshman1970foundations, kruskal1976psych}.
The interest of tensor decompositions has  expanded to other areas, including chemometrics \citep{smilde2005multi}, signal processing \citep{de1998matrix}, and data mining \citep{mccullagh2018tensor}. In particular, tensor methods have also been used in learning latent variable models. 
{\cite{anandkumar2012spectral} derived tensor structures for low-order moments of latent Dirichlet allocation and applied tensor power method to learning the parameters. 
\cite{anandkumar2012mixture} used the method of moments for mixture models and hidden Markov models as a viable alternative to EM algorithms.
\cite{hsu2013learning} derived similar tensor structure in mixtures of spherical Gaussian models.}
\cite{anandkumar2014tensor} summarized the common structures in several different latent variable models and used tensor power method to learn the parameters under a unified framework. 
Generally speaking, these tensor methods are all based on lower-order moments of observed variables rather than the entire likelihood function. 
As a result, an advantage of using moment-based tensor decomposition algorithms for learning latent variable models is the provable consistency guarantee; see \cite{anandkumar2014tensor} for more details.


Besides the computational challenge, the second question of large-scale latent class analysis is how to ensure the estimators in the large-$N$ and large-$J$ regime are theoretically valid and meaningful.
Traditionally, the subjects' latent class indicators in an LCM are often treated as random variables and marginalized out to obtain the marginal likelihood; we call the resulting model a random-effect LCM.
On the other hand, an alternative approach is to treat the subjects' latent class indicators as fixed unknown parameters and directly incorporate them into the likelihood;  we call the resulting model a fixed-effect LCM.
In the classical scenario with sample size $N$ going to infinity and the number of items $J$ held fixed, the fixed-effect LCMs are known to be inconsistent for estimating the subject-level latent class indicators \citep[e.g., see][]{neyman1948consistent}.
However, for data featuring large $N$ and large $J$, with an increasing amount of information collected per subject, an interesting theoretical question is whether we can obtain consistency in clustering the subjects into their corresponding true latent class in fixed-effect LCMs?

In this paper, in the regime where both $N$ and $J$ go to infinity, we propose an efficient computational pipeline and develop the theory of clustering consistency for LCMs with binary responses. 
It is known that the method of moments can be viewed as good complementary to the maximum likelihood approach \citep{chaganty2013spectral,zhang2014spectral,balakrishnan2017em}. {\cite{balakrishnan2017em} theoretically examined the properties of two-stage estimators where a suitable initial estimator is refined with the EM algorithm. \cite{chaganty2013spectral} and \cite{zhang2014spectral} considered mixed linear models and multi-class crowd labeling problem, respectively. They showed that in both problems the tensor estimator based on moments serves as an effective initialization for EM algorithm. Inspired by their work, we introduce the two-step estimator for LCM. }
On the computational side, we propose an efficient two-step estimation pipeline integrating the moment-based tensor decomposition method and the EM algorithm.
{In the first step, we apply the tensor power method in \cite{anandkumar2014tensor} for LCMs to quickly and reliably find roughly accurate parameter estimates.
In the second step, we propose to use the tensor estimates as initialization for the EM algorithm to refine the parameter estimation.
With good initialization, EM algorithms typically converge in very few iterations. 
Therefore, such an estimation pipeline combines the advantages of both the tensor decomposition algorithm and the EM algorithm for latent class analysis.
Our extensive simulation studies empirically show that such an estimation pipeline enjoys both computational efficiency and estimation accuracy.
Further, on the theoretical side, we prove the clustering consistency of the joint maximum likelihood estimator (joint MLE) for fixed-effect LCMs. 
That is, we prove that the joint MLE is consistent in estimating the subjects' latent class memberships under certain mild assumptions when $N$ and $J$ both go to infinity. 
We also derive a bound on the rate of convergence of the joint MLE's clustering performance. The consistency of item parameters is established as a corollary of clustering consistency.}

The rest of this paper is organized as follows. The setups of random-effect and fixed-effect LCMs are introduced in Section \ref{LCM-model}. The proposed estimation procedures of large-scale LCM are presented in Section \ref{estimation}. Some preliminaries about tensor are also provided in Section \ref{estimation} to make this section self-contained.
Section \ref{consistency} presents our theoretical results on clustering consistency of joint MLE. Section \ref{simulation} presents simulation studies that evaluate the proposed  estimation procedures and assess the empirical behavior of clustering consistency. A real data example is shown in Section \ref{real-data}, and we conclude this paper with some discussion in Section \ref{discussion}.
All the proofs and additional simulation results are presented in Supplementary Material. 
	



\section{Latent Class Models with Binary Responses}\label{LCM-model}

In this section we introduce two perspectives of LCM. In random-effect LCMs, the latent class indicators are random variables; 
while in fixed-effect LCM, the latent class indicators are fixed and treated as unknown parameters. 
These two models share common assumptions on how the observed variables depend on the latent ones. 

\subsection{Latent Variables as Random Effects}

We first introduce random-effect LCM. Consider a binary-outcome latent class model with $J$ items and $L$ classes. {Throughout this paper we will use boldface type to denote vectors, matrices and tensors while standard type is used to denote scalars.} There are two types of individual-specific variables in the model, that is a binary response vector $\RR_i \in \{0,1\}^J$ and a latent variable $z_i \in [L]$. Here $[L] = \{1,2,\dots,L\}$ is the set of positive integers smaller than or equal to $L$. The  response vector $\RR_i=(R_{i,1},\ldots,R_{i,J})$ contains the observed responses to the $J$ items of $i$-th subject. The $j$-th component of $\RR_i$ will be 1 if this subject gives a positive response to the $j$-th item and will be 0 otherwise. For instance, in a test with $J$ items, if a student answers the $j$-th item correctly, then $R_{i,j}$, the $j$-th component of $\RR_i$, will be 1. If the student fails to give a right answer then $R_{i,j} = 0$. The latent variable $z_i$ is introduced to categorize different observations and explain the dependence among items. 

The generative process for a response vector $\RR_i$ of an observation is as follows: first the class of this observation $z_i$ is drawn from a discrete distribution specified by the probability vector $\mathbf{p} = (p_1,p_2,\dots,p_L)$, where $p_k\geq 0$ and $\sum_{k=1}^L p_k=1$. So we have
\[
P(z_i=\ell) = p_{\ell}, \; \ell \in [L],
\]
where $p_\ell$ is the proportion of subjects belonging to $\ell$-th class in the population. Then given the latent class $z_i=\ell$, the responses to $J$ items are drawn conditionally independently from a Bernoulli distribution with parameter $\theta_{j,\ell}$ for each item $j$. That is
\[
P(R_{i,j} = 1 | z_i=\ell) = \theta_{j,\ell}.
\]
So $\theta_{j,\ell}$ measures the ability of subjects from $\ell$-th class to give a positive response on item $j$ and is also known as item parameters. 
Like many other latent variables, local independence is assumed here, implying the dependence of item responses is fully explained by the latent classes. We collect all the item parameters for the $L$ classes in the matrix $\boldsymbol{\TT} = (\theta_{j,\ell}) \in [0,1]^{J\times L}$ whose rows are indexed by the $J$ items and columns indexed by the $L$ classes. All the response vectors are collected in a $N \times J$ matrix $\mathbf{R}$, and the corresponding log-likelihood function under the random-effect LCM is 
\[
\ell(\mathbf{R}; \, \mathbf{p}, \, \TT) = \log\left\{ \prod_{i=1}^{N}\left[\sum_{\ell=1}^{L}p_{\ell} \prod_{j=1}^{J}\theta_{j,\ell}^{R_{i,j}}(1-\theta_{j,\ell})^{1-R_{i,j}}\right] \right\},
\]
with  $(\mathbf{p},\, \TT)$ the parameters to be estimated.

\subsection{Latent Variables as Fixed Effects}

Another way to model latent classes is to view latent class assignment as fixed unknown parameters. For a fixed-effect LCM, denote the $i$-th subject's latent class membership by a vector of binary entries $\ZZ_{i,\cdot} = (Z_{i,1},\ldots,Z_{i,L})$, with $Z_{i,\ell} = 1$ if subject $i$ belongs to the latent class $\ell$. We also introduce another notation for the latent class membership $z_i\in\{1,2,\ldots,L\}$ and $z_i = \ell$ corresponds to $Z_{i,l} = 1$. Given a sample of size $N$, collect all the $\ZZ_{1,\cdot},\ldots, \ZZ_{N,\cdot}$ in a $N\times L$ matrix $\mathbf Z$, then each row of $\mathbf Z$ contains only one entry of ``1" and the remaining entries are zeros. We will use the two equivalent notations $\ZZ$ and $\zz=(z_1,\ldots,z_N)$ interchangeably. The components of response vector $\RR_i$ are independent Bernoulli variables with parameters specified by $\TT$. So we have $P(R_{i,j} = 1) = \theta_{j, z_i}$.
The log-likelihood for $(\mathbf Z, \, \TT)$ takes the following form
\begin{align}\label{eq-loglike}
	\ell(\RR;\,\mathbf Z, \, \TT)
	=&~\log \left\{\prod_{i=1}^N \prod_{l=1}^L \left[\prod_{j=1}^J 
	(\theta_{j,\ell}) ^{R_{i,j}} (1-\theta_{j,\ell})^{1-R_{i,j}}
	\right]^{Z_{i,\ell}}
	 \right\}\\ \notag
	=&~\sum_i \sum_j \Big\{ R_{i,j}\Big[\sum_{\ell=1}^L Z_{i,\ell}\log(\theta_{j,\ell})\Big] + (1-R_{i,j})\Big[\sum_{\ell=1}^L Z_{i,\ell}\log(1-\theta_{j,\ell})\Big] \Big\} \\ \notag
	=&~\sum_i \sum_j  \Big\{ R_{i,j}\log(\theta_{j,z_i}) + (1-R_{i,j})\log(1-\theta_{j,z_i}) \Big\}.
\end{align}
The parameters to estimate are $(\mathbf Z, \, \TT)$. The above display is also called the complete data likelihood in the literature.
In the next section we will discuss how to apply tensor method to efficiently estimate the parameters in these two types of latent class models. 


\section{Estimation Procedures }\label{estimation}

The EM algorithm is a popular method to maximize likelihood and estimate parameters in LCM by iterating between E-step and M-step. In E-step the probability of each subject belonging to each class is updated by current estimates of item parameters, and in M-step item parameters are updated given the probabilities of each subject's latent class membership. However, the likelihood function under LCM is nonconcave due to the mixture model formulation. Hence, EM algorithm may suffer from convergence to local optima and slow convergence rate under poor initializations.
Good initial values are critical to the success of the EM algorithm. In this section we introduce tensor method in \cite{anandkumar2014tensor} to find good initializations and hence improve the performance of EM algorithm. {We first introduce some basics about tensor in Section \ref{est-preliminaries} and show the tensor structure in random-effect LCM in Section \ref{est-structure}. In Section \ref{sec-sub-tpm}, we introduce the tensor power method, which is central to recovering the parameters $(\TT,\pp)$ from the tensor structure. The tensor-EM method, which uses the tensor estimates of $(\TT, \pp)$ as initializations for EM algorithm, is given in Section \ref{sec-est-tem}. In Section \ref{SelectL} we discuss how to select the number of latent classes $L$. }

\subsection{Preliminaries about Tensor}\label{est-preliminaries}

We will follow the discussions of \cite{anandkumar2014tensor} and be succinct and self-contained. First we introduce some notations borrowed from \cite{anandkumar2014tensor}. A real $p$-th order tensor  $\BT \in \otimes_{i=1}^p \mathbb{R}^{n_i}$ is a $p$-way array of real numbers where $[\BT]_{i_1,\dots,i_p}$ is the $(i_1,\dots,i_p)$-th entry in the array. We will mostly consider the case where $n_i = n$ for all $i \in [p]$.  Vectors and matrices are special cases of tensors where $p=1$ and $p=2$, respectively. 
Another view of tensor is that it is a multilinear map from a set of matrices $\{\VV_i \in \mathbb{R}^{n\times m_i} : i \in [p]\}$ to a $p$-th order tensor $\BT(\VV_1,\dots,\VV_p) \in \mathbb{R}^{m_1\times \dots \times m_p}$, {where $m_1,\ldots,m_p$ are positive integers,} defined as
\begin{equation}\label{multilinear}
	[\BT(\VV_1,\dots,\VV_p)]_{i_1,\dots,i_p}  :=\sum_{j_1,j_2,\dots,j_p} [\BT]_{j_1,j_2,\dots,j_p}[\VV_1]_{j_1,i_1}\dots [\VV_p]_{j_p,i_p}.
\end{equation}

In this paper we will mainly consider three-way tensor and third-order case of this multilinear map. For a third-order tensor $\BT \in \otimes^3 \mathbb{R}^{d}$ and a vector $\uu \in \mathbb{R}^d$, we will make use of the following vector-valued map in the iteration of tensor power methods
\begin{equation}\label{multi2}
	\BT(\II,\uu,\uu) = \sum_{i=1}^{d} \sum_{1\leq j,l\leq d} [\BT]_{i,j,l} (\ee_j^\mathrm{T}\uu)(\ee_l^\mathrm{T}\uu) \ee_i,
\end{equation}
{where $\II$ is the $d$-dimensional identity matrix} and $\ee_1,\dots,\ee_d$ are the canonical basis vectors of $\mathbb{R}^d$; that is, each $\ee_k$ is a $d$-dimensional vector with only the $k$-th entry being one and the other entries being zero. To obtain \eqref{multi2} from \eqref{multilinear}, we note $\BT(\II,\uu,\uu)$ is a $d$-dimensional vector and
\[
[\BT(\II,\uu,\uu)]_k = \sum_{i=1}^{d} \sum_{1\leq j,l\leq d} [\BT]_{i,j,l} \II_{i,k}u_j u_l = \sum_{i=1}^{d} \sum_{1\leq j,l\leq d} [\BT]_{i,j,l} (\ee_j^\mathrm{T}\uu)(\ee_l^\mathrm{T}\uu) \ee_{i,k}.
\]
We will also use the following map in the iteration
\begin{equation}\label{multi3}
	\BT(\uu,\uu,\uu) = \sum_{i,j,k} [\BT]_{i,j,k}(\ee_i^\mathrm{T}\uu)(\ee_j^\mathrm{T}\uu)(\ee_k^\mathrm{T}\uu).
\end{equation}
These maps are all special cases of \eqref{multilinear}.

Most tensors we consider in this paper are symmetric tensors, which means that an element of a tensor is invariant to permutations of its coordinates. If $\BT \in \otimes^p \mathbb{R}^d$ is a symmetric tensor, then we have $[\BT]_{i_1,\dots,i_p} = [\BT]_{i_{\pi(1)},\dots,i_{\pi(p)}}$ for all permutations $\pi$ on $[p]$. This concept is a generalization of symmetric matrices.

A simple case of a tensor is called rank-one tensor. A rank-one tensor $\BT \in \otimes^p \mathbb{R}^d$ can be expressed as tensor product of $p$ vectors: $\BT = \vv_1 \otimes \vv_2 \otimes \dots \otimes \vv_p$ for some vectors $\vv_1,\ldots,\vv_p \in \mathbb{R}^d$, where $[\BT]_{i_1,\dots,i_p} = \prod_{k=1}^{p} [\vv_{k}]_{i_k}$ and $[\vv_{k}]_{i_k}$ is the $i_k$th component of $\vv_k$. When $\vv_k=\vv$ for all $k$, we can get a symmetric tensor. More detailed discussion and introductions about tensor can be found in \cite{kolda2009tensor}.

\subsection{Tensor Structure in Random-Effect LCM}\label{est-structure}

\cite{anandkumar2014tensor} showed that for some latent variable models, their low-order moments can be expressed as a sum of rank-one tensors. Once this structure of cross moments is obtained for a particular model, one can apply the orthogonal tensor decomposition to learn the parameters of the model. For random-effect LCM, we show that there is also useful tensor structure in low-order moments by examining Theorem 3.6 in \cite{anandkumar2014tensor}, which studies the multi-view models. {Although LCM can be viewed as a special case of multi-view models there, we believe it is still inspiring to introduce tensor method to estimate the parameters in LCM from many perspectives. 
First, the tensor method is based on lower-order moments of responses (2nd and 3rd orders) and has consistency guarantees. It is also computationally efficient based on our simulations. Moreover, in psychometrics, researchers usually use likelihood to study the identification and estimation of parameters. And we will show that with appropriate manipulations on second and third order cross-moments, we can also uniquely recover the parameters in a random-effect LCM.
Hence, tensor method provides a new insight into identifying and estimating parameters in latent variable models widely used in psychometrics.} 
On the other hand, we would also like to clarify that the EM algorithm employed in our estimation method is not considered in \cite{anandkumar2014tensor}. To our best knowledge, the proposed 
tensor-EM method of combining tensor decomposition and the EM algorithm for latent class analysis is new, and moreover, the  established consistency theory of our final estimators in Section \ref{consistency} is not previously studied.

Recall that we use $\RR_i$ to generally denote subject $i$'s response vector of length $J$.
We consider response vector on a population level and divide the items into three disjoint parts (so we assume $J\geq 3$) $\RR_i^1$, $\RR_i^2$, $\RR_i^3$ with each $\RR_i^t \in \mathbb{R}^{J_t}$ and $J_1 + J_2 + J_3 = J$. {The goal is to relate the cross-moments of these three parts with the parameters we want to estimate.} The item paramters of $\RR_i^t$ are denoted by $\TT_t \in \mathbb{R}^{J_t\times L}$, which is a sub-matrix of $\TT$, with rows corresponding to rows in $\RR_i^t$. We need the following assumption to derive the tensor structure.

\begin{cond} \label{cond1}
	Each $\TT_t$ has full column rank $L$ for t = 1,2,3. 
\end{cond}

Note that the partition of items can be arbitrary as long as the item parameters for each part satisfy Condition \ref{cond1}. So we can try different partitions to estimate the parameters and take average to obtain the final estimates. 

We denote the $i$-th column of $\TT_t$ to be $\TT_{t,i}$. The following theorem restates Theorem 3.6 in \cite{anandkumar2014tensor} in our setting and characterizes the tensor structure in random-effect LCM. 

\begin{theorem}\label{tensorstr}
	Assume that Condition \ref{cond1} holds and $p_1,\dots,p_L>0$. Define
	\begin{equation}
		\begin{aligned}
		\tilde{\RR}_i^{2} :=&\, \mathbb{E} [\RR_i^1\otimes \RR_i^3] \mathbb{E} [\RR_i^2 \otimes \RR_i^3]^+\RR_i^2\\
		\tilde{\RR}_i^{3} :=&\, \mathbb{E} [\RR_i^1\otimes \RR_i^2] \mathbb{E} [\RR_i^3 \otimes \RR_i^2]^+\RR_i^3\\
		\MM_2 :=& \, \mathbb{E} [\RR_i^1\otimes \tilde{\RR}_i^{2}]\\
		\MM_3 : = & \, \mathbb{E} [\RR_i^1 \otimes \tilde{\RR}_i^{2} \otimes \tilde{\RR}_i^{3 }]
		\end{aligned}
	\end{equation}
	where $\AA^+$ denotes the Moore-Penrose pseudoinverse of matrix $\AA$. Then we have
	\begin{equation} \label{moment}
	\begin{aligned}
		\MM_2 = & \, \sum_{k=1}^{L} p_k\, \TT_{1,k} \otimes \TT_{1,k},\\
		\MM_3 = & \, \sum_{k=1}^{L} p_k \,\TT_{1,k} \otimes \TT_{1,k} \otimes \TT_{1,k}.
		\end{aligned}
	\end{equation}
\end{theorem}

\begin{proof}
	First we compute the cross moment. For $t \neq t'$, $\RR_i^t$ and $\RR_i^{t'}$ are conditionally independent, we have
	\[
	\mathbb{E} [\RR_i^t \otimes \RR_i^{t'}] = \sum_{k=1}^{L}p_k \mathbb{E}[\RR_k^t \otimes \RR_k^{t'}|z_i=k] = \sum_{k=1}^{L} p_k \mathbb{E}[\RR_i^t|z_i=k] \otimes \mathbb{E} [\RR_i^{t'}|z_i=k]=\sum_{k=1}^{L} p_k\TT_{t,k} \otimes \TT_{t',k}.
	\]
	If we denote $\DD=\mathrm{diag}\{p_1,\dots,p_L\}$, then we have $\mathbb{E} [\RR_i^t \otimes \RR_i^{t'}]=\TT_t \DD \TT_{t'}^\mathrm{T}$.
	{In the following calculations we need to use the Moore–Penrose inverse of $\mathbb{E} [\RR_i^t \otimes \RR_i^{t'}]$ and we first compute it. The following fact is useful: To compute the Moore-Penrose inverse of $AB$, if  $A$ has linearly independent columns and $B$ has linearly independent rows, then $(AB)^+ = B^+ A^+$. Now by condition \ref{cond1}, $\TT_t \DD$ has linearly independent columns and $\TT_{t'}^\mathrm{T}$ has linearly independent rows. So we can write $(\TT_t \DD \TT_{t'}^\mathrm{T})^+ = (\TT_{t'}^\mathrm{T})^+ (\TT_t \DD)^+ $. Apply the fact again on $(\TT_t \DD)^+$ we have $(\TT_t \DD \TT_{t'}^\mathrm{T})^+ = (\TT_{t'}^\mathrm{T})^+ \DD^{-1} \TT_t^+$.}
	
	Then we calculate the conditional mean
	\[
	\mathbb{E}[\tilde{\RR}_i^{2}|z_i=k] = \mathbb{E} [\RR_i^1\otimes \RR_i^3] \mathbb{E} [\RR_i^2 \otimes \RR_i^3]^+\mathbb{E}[\RR_i^2|z_i=k].
	\]
	According to the model setting $\mathbb{E}[\RR_i^2|z_i=k]=\TT_2 \ee_k$, then we have
	\[
		\mathbb{E}[\tilde{\RR}_i^{2}|z_i=k] = \TT_1 \DD \TT_{3}^\mathrm{T} (\TT_2 \DD \TT_{3}^\mathrm{T})^+\TT_2 \ee_k=\TT_1 \DD (\TT_{3}^+\TT_{3})^\mathrm{T} \DD^{-1} \TT_2^+\TT_2\ee_k.
	\]
	By condition \ref{cond1}, $\TT_t^+\TT_t = I_L$ for all $t$, thus $\mathbb{E}[\tilde{\RR}_i^{2}|z_i=k]  = \TT_{1,k}$. Similarly, $\mathbb{E}[\tilde{\RR}_i^{3}|z_i=k]  = \TT_{1,k}$. So we have
	\[
	\MM_2 = \sum_{k=1}^{L}p_k \mathbb{E}[\RR_i^1 \otimes \tilde{\RR}_i^{2}|z_i=k] =  \sum_{k=1}^{L}p_k \mathbb{E}[\RR_i^1 |z_i=k]  \otimes\mathbb{E}[ \tilde{\RR}_i^{2}|z_i=k]  = \sum_{k=1}^{L} p_k\, \TT_{1,k} \otimes \TT_{1,k}.
	\] 
	Similarly one can get the decomposition for $\MM_3$ in \eqref{moment}.
\end{proof}

In applications we only have finite samples and the moments in Theorem \ref{tensorstr} should be approximated by empirical moments. In particular, once we have samples $\RR_1, \dots, \RR_N \in \mathbb{R}^J$, we partition each sample and obtain $\RR_i^{t} \in \mathbb{R}^{J_t}$ corresponding to the partition on population level. Then the transformed response and estimated moments can be computed by
\begin{equation}\label{moments-est}
    \begin{array}{l}
    \hat{\mathbb{E}} [\RR_i^t \otimes \RR_i^{t^{\prime}}] := \frac{1}{N} \sum_{j=1}^{N} \RR_{j}^{t} \otimes \RR_{j}^{t^{\prime}} \\
    {\tilde{\RR}_{i,e}^{2}}:=\hat{\mathbb{E}} [\RR_i^1 \otimes \RR_i^{3}] \left(\hat{\mathbb{E}} [\RR_i^2 \otimes \RR_i^{3}]\right)^{+} \RR_{i}^2 \\
    \tilde{\RR}_{i,e}^{3}:=\hat{\mathbb{E}} [\RR_i^1 \otimes \RR_i^{2}]\left(\hat{\mathbb{E}} [\RR_i^3 \otimes \RR_i^{2}]\right)^{+} \RR_i^{3} \\
    \widehat{\MM}_{2}:=\frac{1}{N} \sum_{j=1}^{N} {\RR}_{j}^{1} \otimes \tilde{\RR}_{j,e}^{2}, \\
    \widehat{\MM}_{3}:=\frac{1}{N} \sum_{j=1}^{N} \RR_{j}^{1} \otimes \tilde{\RR}_{j,e}^{2} \otimes \tilde{\RR}_{j,e}^{3}.
    \end{array}
\end{equation}

Due to the randomness of sample, it is possible that $r=:$ rank$(\hat{\mathbb{E}} [\RR_i^2 \otimes \RR_i^3]) > L$ and $\hat{\mathbb{E}} [\RR_i^2 \otimes \RR_i^3]$ has $(r-L)$ extra non-zero singular values. These singular values will be small since they equal to 0 in $\hat{\mathbb{E}} [\RR_i^2 \otimes \RR_i^3]$'s population counterpart ${\mathbb{E}} [\RR_i^2 \otimes \RR_i^3]$. In this case we should discard these extra singular values and only use first $
L$ sigular values when calculating $\hat{\mathbb{E}} [\RR_i^2 \otimes \RR_i^3]^+$, otherwise one has to compute the inverse of these small singular values, which will incur large error.

After learning $\TT_1$ from data by the tensor power method to be introduced in Section \ref{sec-sub-tpm}, we can obtain $\TT_2$ and $\TT_3$ by setting $\TT_2 = \mathbb{E} [\RR_i^2\otimes \RR_i^3] \mathbb{E} [\RR_i^1 \otimes \RR_i^3]^+\TT_1$  and $\TT_3 = \mathbb{E} [\RR_i^3\otimes \RR_i^2] \mathbb{E} [\RR_i^1 \otimes \RR_i^2]^+\TT_1$. This can be derived in a same way as Theorem \ref{tensorstr}.  So the main problem is to estimate $\TT_1$ from moments $\MM_2$ and $\MM_3$. 

Although this structure only holds for random-effect LCMs,  in fixed-effect LCMs we can view $z_1,\dots,z_N$ as random with some prior distribution. For instance, they are sampled from some discrete distributions on $[L]$ independently from $\RR$. Then the data generation process of random-effect and fixed-effect LCMs are the same and the estimation procedures for random-effect LCMs also apply to fixed-effect LCMs. 

\subsection{Tensor Method to Learn the Parameters}\label{sec-sub-tpm}

In this section we briefly describe the procedures in \cite{anandkumar2014tensor} to recover the parameters in \eqref{moment}. That is, given
\begin{equation}\label{moments}
	\begin{aligned}
	\MM_2 = & \, \sum_{i=1}^{L} w_i\, \mmu_i \otimes \mmu_i,\\
	\MM_3 = & \, \sum_{i=1}^{L} w_i \,\mmu_i \otimes \mmu_i \otimes \mmu_i\\
	\end{aligned}
\end{equation}
where $\mmu_i \in \mathbb{R}^d$, we want to obtain the elements of decomposition $(w_i,\mmu_i)$'s from $\MM_2$ and $\MM_3$. Condition \ref{cond1} now becomes $\{\mmu_1,\dots,\mmu_L\}$ are linearly independent. 
First we introduce the orthogonal decomposition of a tensor. Then we see how we can use tensor power method to recover the orthogonal decomposition of a tensor and estimate the parameters.

\subsubsection{Orthogonal Decomposition}

Since the moments structures in \eqref{moments} are about at most a third-order tensor, we only consider the case $p=3$ (third-order tensor).

A symmetric tensor $\BT \in \otimes^3 \mathbb{R}^d$ has an orthogonal decomposition if there exists a collection of orthonormal unit vectors $\{\vv_1,\dots \vv_L \}$ and positive scalers $\lambda_i > 0$ such that 
\begin{equation} \label{orthogonal}
	\BT = \sum_{i=1}^{L} \lambda_i \vv_i \otimes \vv_i \otimes \vv_i.
\end{equation}

Without loss of generality we assume $\lambda_i > 0$ because for third-order tensor we have $-\lambda_i \vv_i^{\otimes 3} = \lambda_i (-\vv_i)^{\otimes 3}$. However we do not assume $\lambda_i$'s are ordered. {In fact, according to Theorem \ref{tp-converge} in Section \ref{tensorpower}, the eigenvector that tensor power method converges to depends on the magnitude of elements in \{$|\lambda_i\vv_i^\mathrm{T}\uu_0|$,$1\leq i\leq L$\} instead of the magnitude of $\lambda_i$'s. Here $\uu_0$ is the initial point for tensor power method.} This definition is a generalization of spectral decomposition for a symmetric matrix. We can also generalize the concept of eigenvalue and eigenvectors.

Recall the definition of the multilinear map induced by a tensor in \eqref{multilinear}.
A unit vector $\uu \in \mathbb{R}^d$ is an eigenvector of $\BT$ with corresponding eigenvalue $\lambda \in \mathbb{R}$ if $\BT(\II,\uu,\uu) = \lambda \uu$. For an orthogonally decomposable tensor $\BT = \sum_{i=1}^{L} \lambda_i \vv_i \otimes \vv_i \otimes \vv_i$, one can check operation \eqref{multi2} is
\[
\BT(\II,\uu,\uu) = \sum_{i=1}^{L} \lambda_i (\uu^\mathrm{T}\vv_i)^2 \vv_i
\]
and operation \eqref{multi3} reduces to
\[
\BT(\uu,\uu,\uu) = \sum_{i=1}^{L} \lambda_i (\uu^\mathrm{T}\vv_i)^3.
\]
By the orthogonality of $\vv_i$'s, $\BT(\II,\vv_i,\vv_i) = \lambda_i \vv_i$ and $\BT(\vv_i,\vv_i,\vv_i) = \lambda_i$ for all $i \in [L]$. Thus $(\lambda_i,\vv_i)$ is an eigenvector/eigenvalue pair of $\BT$.

The eigenvalues and eigenvectors of a tensor are more complicated than those of a matrix, and there are some subtle points. {For example, unlike matrices, orthogonal
decompositions do not necessarily exist for each symmetric tensor. Moreover, if a tensor $\BT$ admits an orthogonal decomposition in \eqref{orthogonal}, then this is the unique orthogonal decomposition of $\BT$. This is very different from the spectral decomposition of a matrix. In fact, $\{\vv_1, \dots, \vv_L\}$ are the set of robust eigenvectors for $\BT$. See \cite{anandkumar2014tensor} for more discussion.}  

\subsubsection{Whitening Process}\label{sec-whitening}

Comparing the third-order moment structure $\MM_3$ in \eqref{moments} with the orthogonal decomposition form \eqref{orthogonal}, we find they have almost the same form except that the vectors $\mmu_i$'s in  \eqref{moments} may not necessarily be orthogonal to each other. So we need to whiten the tensor $\MM_3$ to $\tilde{\MM}_3$, which has an orthogonal decomposition. In the whitening process we will make use of $\MM_2$ in \eqref{moments}.

Let $\WW \in \mathbb{R}^{d \times L}$ satisfy $\MM_2(\WW,\WW) = \WW^\mathrm{T}\MM_2\WW = \II_L$. We can take $\WW = \UU\DD^{-1/2}$, where $\DD$ is the diagnoal matrix containing all positive eigenvalues of $\MM_2$ and $\UU \in \mathbb{R}^{d\times L}$ is the matrix of corresponding orthogonal eigenvectors of $\MM_2$. $\DD^{-1/2}$ is well-defined since we assume $\mmu_i$'s are linearly independent and thus $\MM_2$ is of rank $L$.

{Suppose $\MM_2$ and $\MM_3$ admit the decomposition as in \eqref{moments},} define $\tilde{\mmu}_i :=\sqrt{\omega_i} \,\WW^\mathrm{T}\mmu_i$ and observe that
\[
\II_L = \MM_2(\WW,\WW) = \sum_{i=1}^{L}\WW^\mathrm{T}(\sqrt{\omega_i}\mmu_i)(\sqrt{\omega_i}\mmu_i)^\mathrm{T}\WW = \sum_{i=1}^{L} \tilde{\mmu}_i \tilde{\mmu}_i^\mathrm{T} .
\]
So $\tilde{\mmu}_i\, 's$ are orthonormal vectors.

Define 
\[
\tilde{\MM}_3 := \MM_3(\WW,\WW,\WW) = \sum_{i=1}^{L}\omega_i (\WW^\mathrm{T}\mmu_i)^{\otimes 3} =  \sum_{i=1}^{L}\frac{1}{\sqrt{\omega_i}} \tilde{\mmu}_i^{\otimes 3}.
\]
Since $\tilde{\mmu}_i\, 's$ are orthonormal vectors, this is the orthogonal decomposition of $\tilde{\MM}_3$. We can use tensor power method described in Section \ref{tensorpower} to obtain the eigenvalue/eigenvector pairs $(\lambda_i,\vv_i) = (1/\sqrt{\omega_i},\tilde{\mmu}_i)$. Then we can recover the parameters  $\omega_i$'s and $\mmu_i$'s as $(\omega_i,\mmu_i) = (\frac{1}{\lambda_i^2},\lambda_i(\WW^\mathrm{T})^+ \tilde{\mmu}_i)$, where $(\WW^\mathrm{T})^+$ is the Moore-Penrose pseudoinverse of $\WW^\mathrm{T}$.

\subsubsection{Tensor Power Method}\label{tensorpower}

Now we show how to recover the parameters $(\lambda_i,\vv_i)$'s in \eqref{orthogonal} from a tensor $\BT$. In analogy to matrix power method, here we use the tensor power method of \cite{de2000best} to obtain the eigenvalue/eigenvector pairs $(\lambda_i,\vv_i)$ in \eqref{orthogonal}. First suppose a third-order tensor has an exact orthogonal decomposition. We have the following result on the algorithmic convergence of tensor power method (Lemma 5.1 in \cite{anandkumar2014tensor}).

\begin{theorem}\label{tp-converge}
	Let $\BT \in \otimes^3 \mathbb{R}^d$ have an orthogonal decomposition as given in \eqref{orthogonal}. For a vector $\uu_0 \in \mathbb{R}^d$, suppose that the set of numbers \{$|\lambda_i\vv_i^\mathrm{T}\uu_0|$,$1\leq i\leq L$\} has a unique largest value. Without loss of generality, say $|\lambda_1\vv_1^\mathrm{T}\uu_0|$ is this largest value and $|\lambda_2\vv_2^\mathrm{T}\uu_0|$ is the
second largest value. For $t=1,2\dots$, let
	\[
	\uu_t := \frac{\BT(\II,\uu_{t-1},\uu_{t-1})}{||\BT(\II,\uu_{t-1},\uu_{t-1})||}.
	\]
	Then 
	\[
	||\vv_1-\uu_t||^2 \leq (2\lambda_1^2 \sum_{i=2}^{K} \lambda_i^{-2})\left|\frac{\lambda_2 \vv_2^\mathrm{T}\uu_0}{\lambda_1 \vv_1^\mathrm{T}\uu_0}\right|^{2^{t+1}},
	\]
    where $\|\cdot \|$ is the $\ell_2$ norm.
\end{theorem}

The result shows that the repeated iteration starting from $\uu_0$ converges to $\vv_1$ at a quadratic rate. The reason why the tensor power method enjoys a quadratic convergence rate in Theorem \ref{tp-converge} while the usual matrix power method has a (relatively) slower linear convergence rate is that the iteration step in the tensor case is quadratic while that step is linear in the matrix case. Specifically, the (unnormalized) iteration in the tensor power method is
\[
\bar{\uu}_{t+1}=\BT\left(\II, \bar{\uu}_{t}, \bar{\uu}_{t}\right)=\sum_{i=1}^{k} \lambda_{i}\left(\vv_{i}^{\top} \bar{\uu}_{t}\right)^{2} \vv_{i},
\]
and the (unnormalized) iteration in the matrix power method is
\[
\bar{\uu}_{t+1}=\BT \bar{\uu}_t=\sum_{i=1}^{k} \lambda_{i}\left(\vv_{i}^{\top} \bar{\uu}_{t}\right) \vv_{i},
\]
where $(\lambda_i, \vv_i)$'s are the eigenvalues and eigenvectors of the tensor/matrix $\BT$. In the tensor case, $\bar{\uu}_{t+1}$ depends on $\vv_{i}^{\top} \bar{\uu}_{t}$ via $\left(\vv_{i}^{\top} \bar{\uu}_{t}\right)^2$, and by induction, one can show $\bar{\uu}_{t}=\sum_{i=1}^{k} \lambda_{i}^{2^{t}-1} c_{i}^{2^{t}} \vv_{i}$ for $c_i = \vv_{i}^{\top} \bar{\uu}_{0}$. Then the quadratic rate $O\left(\rho^{2^{t}}\right)$ for $\rho = \left|\frac{\lambda_2 c_2}{\lambda_1 c_1}\right|$ arises from the normalization of $\bar{\uu}_t$. In the matrix case, $\bar{\uu}_{t}=\sum_{i=1}^{k} \lambda_{i}^{t} c_{i} \vv_{i}$, which is different from the form in the tensor case, and after normalization, the convergence rate is $O(\rho^t)$ for $\rho = \left|\frac{\lambda_2}{\lambda_1} \right|$. Therefore, the tensor power method enjoys a faster convergence rate than matrix power method. To obtain all the eigenvalue/eigenvector pairs, we use ``deflation'' after getting an eigenvalue/eigenvector pair $(\lambda_i,\vv_i)$. That is, to obtain the $j$-th eigenvalue/eigenvector pair, we subtract the previous $j-1$ rank-one structures from $\BT$ and then execute the power method on $\BT-\sum_{i=1}^{j-1} \lambda_i \vv_i\otimes \vv_i\otimes \vv_i$.

However, in practice we plug-in the empirical estimate of $\MM_2$ and $\MM_3$ in \eqref{moment} and the estimate of whitened tensor $\tilde{\MM}_3$ may not have an exact orthogonal decomposition. So we should consider the case where we only have an approximation $\hat{\BT}$ of $\BT$ and need a more robust algorithm to use an orthogonal decomposition to approximate $\hat{\BT}$. Following \cite{anandkumar2014tensor}, we present Algorithm \ref{robust-tp} as a more robust method.
 Multiple starting points are used in Algorithm \ref{robust-tp} to ensure approximate convergence at first stage. Intuitively, by restarting from different points we can start from a point from which the initial $n$ iterations from it dominates the error $\hat{\BT}-\BT$.  {To be concrete, define operator norm for a tensor $\BT$ as follows
 $$
\|\BT\|_{\text{op}}:=\sup _{\|\uu \|=1}|\BT(\uu, \uu, \uu)| .
$$
Perturbation analysis (Theorem 5.1 in \cite{anandkumar2014tensor}) shows that under some conditions if $\|\hat{\BT} - \BT\|_{\text{op}}$ is small (in our setting, this requires the cross moments are estimated accurately), then the estimated eigenvalue/eigenvector pairs returned by Algorithm \ref{robust-tp} are close to the true eigenvalue/eigenvector pairs. Note that in contrast to Davis-Kahan's theorem (which holds for all symmetric matrix), Theorem 5.1 in \cite{anandkumar2014tensor} is an algorithm-dependent perturbation analysis and only applies to Algorithm \ref{robust-tp} since in general $\hat{\BT}$ may not even have an orthogonal decomposition. Furthermore, Theorem 5.1 in \cite{anandkumar2014tensor} shows we can obtain good estimates of eigenvalue/eigenvector pairs with high probability with $K = \mathrm{poly}(L)$ trials. When $L$ is large the required number of initial trials $K$ is close to linear in $L$. Hence, the robust tensor power method (Algorithm \ref{robust-tp}) can recover the eigenvalue/eigenvector pairs of $\BT$ efficiently from an estimator $\hat{\BT}$. }
 

\begin{algorithm}[h]
	\caption {Robust tensor power method}
	\label{robust-tp}
	\begin{algorithmic}[1]
		\REQUIRE  symmetric tensor $\tilde{\BT} \in \mathbb{R}^{d\times d \times d}$, number of iterations $K$, $n$.
		\ENSURE  estimates of one of eigenvalue/eigenvector pairs; the deflated tensor 
		\FOR{$\tau = 1 \;\mathrm{to}\; K$}
			\STATE Draw $\uu_0^{(\tau)}$ uniformly from unit sphere in $\mathbb{R}^d$.
			\FOR{$t = 1  \;\mathrm{to}\; n$}
				\STATE Compute power iteration and re-normalization 
				\[
				\uu_t^{(\tau)} = \frac{\tilde{\BT}(\II,\uu_{t-1}^{(\tau)},\uu_{t-1}^{(\tau)})}{||\tilde{\BT}(\II,\uu_{t-1}^{(\tau)},\uu_{t-1}^{(\tau)})||}
				\]
			\ENDFOR
		\ENDFOR
		\STATE Let $\tau^* = \mathrm{argmax}_{\tau \in [K]} \{ \tilde{\BT}(\uu_{t-1}^{(\tau)},\uu_{t-1}^{(\tau)},\uu_{t-1}^{(\tau)})\}$
		\STATE Do $n$ power iteration updates further starting from $\uu_n^{(\tau^*)}$ to obtain $\hat{\uu}$, and set $\hat{\lambda} = \tilde{\BT}(\hat{\uu},\hat{\uu},\hat{\uu})$.
		\RETURN the estimated eigenvalue/eigenvector pair $(\hat{\lambda},\hat{\uu}) $; the deflated tensor $\BT - \hat{\lambda}\hat{\uu}\otimes \hat{\uu}\otimes \hat{\uu}.$
		
\end{algorithmic}
\end{algorithm}

 In conclusion, the procedures to learn $(\omega_i,\mmu_i)$ from $\MM_2$ and $\MM_3$ in \eqref{moments} are: (1) Use the information of $\MM_2$ to whiten $\MM_3$ and get $\tilde{\MM}_3$; (2) Apply tensor power method to $\tilde{\MM}_3$ and learn the orthogonal decomposition of it; (3) Obtain the original parameters with an inverse transformation of whitening. {We then elaborate on how to use the tensor method to estimate paprameters in LCM.}
 
\subsection{Tensor-EM Method}\label{sec-est-tem}

{In Section \ref{sec-sub-tpm} we introduce how to recover parameters $(w_i, \mmu_i)$'s from $\MM_2$ and $\MM_3$ in \eqref{moments}. Since the tensor structure in random-effect LCM in Theorem \ref{tensorstr} is in the exact form of \eqref{moments}, we can apply the methods in Section \ref{sec-sub-tpm} to the empirical estimates of $\MM_2$ and $\MM_3$ shown in \eqref{moments-est} and obtain the tensor estimator for $\TT_1$ and $\pp$. The relations $\TT_2 = \mathbb{E} [\RR_i^2\otimes \RR_i^3] \mathbb{E} [\RR_i^1 \otimes \RR_i^3]^+\TT_1$  and $\TT_3 = \mathbb{E} [\RR_i^3\otimes \RR_i^2] \mathbb{E} [\RR_i^1 \otimes \RR_i^2]^+\TT_1$ then give us the tensor estimates of $\TT_2$ and $\TT_3$. We denote the tensor estimates as $\hat{\TT}_{T}$ and $\hat{\pp}_T$. When the sample size is large enough, this method alone can yield an estimator close to the true parameters. However if we do not have so many samples, estimates based on tensor are not so accurate since in tensor method we only take advantage of low-order moments and ignore other information of sampling distributions. So after obtaining the tensor estimates, we use them as initial values for EM algorithms to improve the accuracy. We call this two-step estimation procedure the \textit{tensor-EM method}. We further derive the EM algorithm for random-effect LCM here. Consider the complete log-likelihood}

\begin{equation*}
    \ell_{comp} (\TT, \pp| \RR, \ZZ) = \sum_{i=1}^N \sum_{l=1}^L Z_{i,l} \log p_l + \sum_{i=1}^N \sum_{l=1}^L\sum_{j=1}^J Z_{i,l} [R_{i,j}\theta_{j,l}+(1-R_{i,j})(1-\theta_{j,l})],
\end{equation*}
{where we use the same notation $Z_{i,l}$ as in the fixed-effect LCM to denote the indicator $I$(subject $i$ is in class $l$). Given $(\TT^{(t)}, \pp^{(t)}),$ in the E-step we compute $\varphi_{i,l}^{(t+1)} = \ME[Z_{i,l} \mid \TT^{(t)}, \pp^{(t)}, \RR]$ by the posterior probability,} 
 
\begin{equation}\label{post-prob}
\varphi_{i,l}^{(t+1)} = \frac{p_l^{(t)} \prod_{j} {\theta_{j,l}^{(t)}}^{R_{i,j}} (1-\theta_{j,l}^{(t)})^{1-R_{i,j}}}{\sum_{l}p_l^{(t)} \prod_{j} {\theta_{j,l}^{(t)}}^{R_{i,j}} (1-\theta_{j,l}^{(t)})^{1-R_{i,j}}}
\end{equation}
{In the M-step, we replace $Z_{i,l}$ with $\varphi_{i,l}^{(t+1)}$ and obtain}
\[
Q(\TT, \pp | \TT^{(t)}, \pp^{(t)}) = \sum_{i=1}^N \sum_{l=1}^L \varphi_{i,l}^{(t+1)} \log p_l + \sum_{i=1}^N \sum_{l=1}^L\sum_{j=1}^J \varphi_{i,l}^{(t+1)} [R_{i,j}\theta_{j,l}+(1-R_{i,j})(1-\theta_{j,l})].
\]
{After maximizing over $\TT, \pp$ in $Q(\TT, \pp | \TT^{(t)}, \pp^{(t)})$ we arrive at the updated parameters}
\begin{equation*}
    \theta_{j,l}^{(t+1)} = \frac{\sum_{i} \varphi_{i,l}^{(t+1)}R_{i,j}}{\sum_{i}\varphi_{i,l}^{(t+1)}}, \quad p_{l}^{(t+1)} = \frac{\sum_{i=1} \varphi_{i,l}^{(t+1)}}{\sum_{l}\sum_{i=1} \varphi_{i,l}^{t+1}}.
\end{equation*}
{We keep iterating until some convergence criterion is met (e.g. the log-likelihood improves little after one iteration). When no prior knowledge is available, people may try many random initial values $(\TT^0, \pp^0)$ and take the one that has the maximum log-likelihood after the algorithm converges. In tensor-EM method we set the initial values $(\TT^0, \pp^0)$ to be $(\hat{\TT}_{T},\hat{\pp}_T)$. }
 
For fixed-effect LCM, we first learn $p_i$'s and $\TT$ from \eqref{moment} as a initialization step and apply Classification-EM (CEM) algorithm proposed in \cite{celeux1992cem} to obtain the final estimator. The main difference between CEM and EM algorithm summarized above is that in CEM algorithm we want to estimate latent class membership $z_i^{(t+1)}$, hence we need to find the index that maximizes posterior probability $\{\varphi_{i,l}^{(t+1)}, l\in [L]\}$ for each $i$. Formally, given $(\TT^{(t)}, \pp^{(t)}),$ in the E-step we compute $\varphi_{i,l}^{(t+1)}$ in \eqref{post-prob}. Then we let the estimated latent class membership for subject $i$ at step $t+1$ be
\[
z_i^{(t+1)} = \argmax_{l\in\{1,\ldots,L\}} \varphi_{i,l}^{(t+1)}
\]
{and correspondingly set $\ZZ^{(t+1)}$. In the M-step the parameters $(\TT, \pp)$ are updated using $\ZZ^{(t+1)}$ instead of $\varphi_{i,l}^{(t+1)}$ as follows}

\begin{equation*}
    \theta_{j,l}^{(t+1)} = \frac{\sum_{i} Z_{i,l}^{(t+1)}R_{i,j}}{\sum_{i}Z_{i,l}^{(t+1)}}, \quad p_{l}^{(t+1)} = \frac{\sum_{i=1} Z_{i,l}^{(t+1)}}{N}.
\end{equation*}
When the algorithm converges, we output $(\hat{\TT}, \hat{\ZZ})$ (recall the parameters in fixed-effect LCM are $\TT$ and $\ZZ$). Since the CEM algorithm only requires $(\TT,\pp)$ to obtain the next-step updates, we can set the initial values to be $(\hat{\TT}_{T},\hat{\pp}_T)$. When there are some components of $\TT$ outside the range $[0,1]$, we can set the negative
components to be a small positive number (e.g. 0.001) and those over one to be a number close to 1
(e.g. 0.999).

Empirically, we find that the accuracy of tensor-EM method is comparable to estimates obtained with EM algorithm starting from true parameters, indicating that the tensor-EM method can give the MLE of latent class model when the model is correctly specified. 
Moreover, it is more computationally efficient than EM algorithm with several random initial values, especially in large-scale data. See our simulation study for more details.


\subsection{Selecting the Number of Classes} \label{SelectL} 
In the discussion above we assume the number of classes $L$ is known. Next we discuss the selection of $L$.
There exists a rich literature in selecting number of classes.  \cite{nylund2007deciding} performed a Monte Carlo simulation study on several commonly used methods and found BIC proposed in \cite{schwarz1978estimating} and likelihood ratio test based on bootstrap in \cite{mclachlan2004finite} have a better performance. They recommend BIC and likelihood ratio test based on bootstrap to select the number of classes. 
Since we focus on large-scale datasets containing many items and samples, we  follow the discussion of \cite{chen2017regularized} and apply generalized information criterion proposed in \cite{nishii1984asymptotic} to selecting the number of classes.

  Specifically, for a candidate set $\mathcal{L}$ and any $L \in \mathcal{L}$ , we apply the tensor-EM algorithm to learning the parameters and compute the generalized information criterion for random-effect and fixed-effect LCM as follows:
\[
\mathrm{GIC}_R(L) = -2\,\ell(\mathbf{R}; \, \hat{\mathbf{p}}^L, \, \hat{\TT}^L) + a_N \text{dim}(\pp, \TT),
\]
\[
\mathrm{GIC}_F(L) = -2\,\ell(\mathbf{R}; \, \hat{\mathbf{\ZZ}}^L, \, \hat{\TT}^L) + a_N \text{dim}(\ZZ, \TT)
\]
where $\text{dim}(\pp, \TT)$ is the dimension of parameters to estimate and measures the model complexity. We have $\text{dim}(\pp, \TT) = JL + L-1$ in random-effect LCM and $\text{dim}(\ZZ, \TT) = JL +N$ in fixed-effect LCM. Sample size dependent quantity $a_N$ measures the level of penalty on model complexity. Here we consider two choice of $a_N$.
\begin{itemize}
	\item[{\textbullet}] $\mathrm{GIC}_1$: $a_N$ = $\mathrm{log}(N)$. This case corresponds to BIC and it enjoys some consistent results shown in \cite{nishii1984asymptotic} when the model has low complexity (i.e. the dimension of parameters is fixed). 
	\item[{\textbullet}] $\mathrm{GIC}_2$: $a_N$ = $\mathrm{log}[\mathrm{log}(N)]\mathrm{log}N$. This choice is considered in \cite{fan2013tuning} in generalized linear model to address the case where the dimension of parameter space $d$ increases at a polynomial order of $N$, that is, $d = O(N^c)\; \mathrm{for\; some }\; c > 0$. The large-scale latent class model we consider tends to have many items. The dimension of parameters $\text{dim}(\pp, \TT) = J \times L + L-1$ can be large and should not be treated as fixed. For instance, in the simulation study a random-effect LCM we consider has ten classes and one hundred items. This model has dimension $d=1009$ while the sample size is $N=1000$. So it is more appropriate to adopt this choice in this setting. See \cite{fan2013tuning} for discussion about theoretical results of this choice of $a_N$.
\end{itemize} 

After calculating the GIC for different models, we select the number of classes to minimize GIC($L$):
\[
L^* = \mathrm{arg} \min_{L \in \mathcal{L}} \mathrm{GIC}(L).
\]

Although our simulation results show that the tensor-EM method gives the MLE of the model when the number of classes $L$ is correctly selected, it may converge to some local optima when $L$ is incorrect. It is also likely that the tensor method yields inaccurate results when $L$ is misspecified as $\tilde{L}$. For instance, some item parameters may be negative or over one, which happens when the sample size is small. In this case, we propose to revise the tensor estimates as follows: we set the negative components to be a small positive number (e.g. 0.001) and those over one to be a number close to 1 (e.g. 0.999). 
Such a procedure will help us select the  number of classes because we know as long as the true number of classes is specified, the tensor-EM method can yield an ideal estimate (MLE) with a small GIC value. 
On the other hand, a poor estimate based on a misspecified model will give a larger GIC value, indicating misfit of the model. 
Hence, GIC values computed by tensor-EM method can provide useful information to select the number of classes. According to our simulation study, the proposed method can select the right model most of the time.

{After proposing the computational methods to select number of classes and to find the MLE, we next examine the theoretical properties of MLE in the large-$N$ and large-$J$ scenario. For random-effect LCM with fixed number of items $J$, the MLE is known to be consistent. However, 
the joint MLE for fixed-effect LCM may not be consistent \citep{neyman1948consistent} when $J$ is fixed. Intuitively, one cannot hope to recover each subject's latent class membership accurately with only a finite number of items observed for each subject. So in the next section, we will consider the consistency of joint MLE when $J$ also goes to infinity in fixed-effect LCMs.}

\section{Clustering Consistency of the Joint MLE}\label{consistency}

In this section we consider large-scale fixed-effect LCMs and characterize the behavior of latent class assignment estimator $\hat{\ZZ}$ under suitable conditions, where $(\hat\TT,\hat\ZZ)$ is the joint maximum likelihood estimator (MLE).
We use a similar proof technique as in \cite{gu2020joint} to establish the clustering consistency of the joint MLE for fixed-effect LCMs.

First we need to define some notations. 
Denote the true parameters by $(\TT^0,\ZZ^0)$.
Define
\begin{align*}
	P_{i,j}=&~\mathbb P(R_{i,j}=1)=\theta^0_{j,z_i^0},\\
	M =&~ \frac{1}{NJ}\sum_{i=1}^N\sum_{j=1}^J \MP(R_{i,j}=1).
\end{align*}
The $M\in[0,1]$ above measures the average positive response rate over all subjects and items.
Denote the expectation of log-likelihood $\ell(\RR;\,\mathbf Z, \, \TT)$ in \eqref{eq-loglike}  by 
\begin{align*}
	\bar\ell(\mathbf Z, \TT) 
	= \ME [\ell(\RR;\,\mathbf Z, \, \TT)]
	=&~ \sum_{i=1}^N \sum_{j=1}^N \Big\{ P_{i,j}\log(\theta_{j,z_i}) + (1-P_{i,j})\log(1-\theta_{j,z_i}) \Big\},
\end{align*}
where the expectation is taken with respect to the distribution of $\RR$.

Given arbitrary $\mathbf \ZZ$, denote 
\begin{align*}
\ell(\RR;\, \ZZ) =&~ \sup_{\TT} \ell(\RR;\,\mathbf Z, \, \TT) = \ell(\RR;\, \ZZ, \hat\TT^{(\ZZ)}),\\ \notag
\quad
\bar\ell(\ZZ) =&~ \sup_{\TT} \bar\ell(\mathbf Z, \, \TT) = \bar\ell(\ZZ, \bar\TT^{(\ZZ)}),
\end{align*}
where $\hat\TT^{(\ZZ)} = \arg\max_{\TT} \ell(\RR;\,\mathbf Z, \, \TT)$ and $\bar \TT^{(\ZZ)} = \arg\max_{\TT}\bar\ell(\mathbf Z,\TT)$. Then under any realization of $\mathbf Z$, the following holds for any latent class $a\in[L]$,
\begin{equation}\label{eq-zmle}
	\hat\theta^{(z)}_{j,a} = \frac{\sum_{i}Z_{i,a} R_{i,j}}{\sum_{i}Z_{i,a}},\quad
	\bar\theta^{(z)}_{j,a} = \frac{\sum_{i}Z_{i,a} P_{i,j}}{\sum_{i}Z_{i,a}}.
\end{equation}

We consider the joint maximum likelihood estimator $(\hat\TT,\hat\ZZ)$ subject to fitting a $L$-class fixed-effect LCM with true parameters $(\TT^0,\ZZ^0)$,
\[
(\hat\TT,\hat\ZZ) = \argmax_{(\TT,\ZZ)} \ell(\RR;\,\mathbf Z, \, \TT) .
\]
Note that $\hat\ZZ = \argmax_{\ZZ} \ell(\RR;\,\ZZ,\hat\TT^{(\ZZ)})= \argmax_{\ZZ} \ell(\RR;\,\ZZ)$, where $\hat\TT^{(\ZZ)}$ maximizes the profile likelihood $\ell(\RR;\,\mathbf Z, \, \TT)$ given a particular realization $\ZZ$. One can apply the procedures in Section \ref{estimation} to compute the joint MLE efficiently.

We impose the following assumptions on the true parameters.

\begin{assumption}
	\label{assume-pij}
There exists a constant $\gamma>0$ such that
	\begin{equation}\label{eq-pijb}
		\frac{1}{J^\gamma} \leq 
	\min_{1\leq j\leq J,\atop 1\leq a\leq L} \theta^0_{j,a}
	\leq 
	\max_{1\leq j\leq J,\atop 1\leq a\leq L} \theta^0_{j,a}
	\leq 1-\frac{1}{J^\gamma}.	
	\end{equation}
\end{assumption}

\begin{assumption}
	\label{cond-gap}
	There exists a positive sequence $\{\beta_{J}\}$ such that
	\begin{align}
\label{eq-as-gen2}
      \frac{1}{J}\min_{1\leq a\neq b\leq L} \|\TT^0_{\cdot,a} - \TT^0_{\cdot,b}\|^2 \geq &~\beta_J,
	\end{align}
	where $\|\cdot\|$ denotes the $\ell_2$ norm. 
\end{assumption}

Assumption \ref{assume-pij} guarantees that the components of $\TT$ are bounded away from 0 and 1 but allowed to become very close to 0 or 1 as $J$ becomes larger. It is a quite mild technical assumption. 
Assumption \ref{cond-gap} is an identification condition for latent classes and guarantees that the item parameters of different classes are different enough. 
Note that we allow different classes to have same probability to answer a single item correctly (i.e. $\theta_{j,a} = \theta_{j,b}$ for some $a\neq b$). But their average performance on the $J$ items should be different.


In fitting the latent class model, we are interested in controlling the number of incorrect latent class assignments. 
 {Formally, after obtaining some estimator $\hat{\zz}$, let $\hat{C}_{1}, \dots, \hat{C}_{L}$ be clusters from our estimator $\hat{\zz}$ such that subjects sharing same estimated membership are in one cluster. For instance, suppose there are eight subjects whose latent class memberships are $\zz^0 = (2,2,2,2,1,2,1,1)$. Here ``1" and ``2" represents the true class index for each subject. The estimates are $\hat{\zz} = (1,1,1,1,1,2,2,2)$ so we have $\hat{C}_1=\{\text{subject } 1,2,3,4,5\}$ and $\hat{C}_2=\{\text{subject } 6,7,8\}$.  Let $m_l = \argmax_{a \in [L]} \, \sum_{i\in \hat{C}_l}Z_{i,a}^{0} $ be the majority of true class membership among subjects in $\hat{C}_l$. In our example, $\hat{C}_1$ has four subjects whose true class indices under $\zz^0$ are ``2" and hence $m_1 = 2$ and similarly $m_2 = 1$. Since the latent classes are identified up to permutations of class index, $m_l$ should be viewed as the class index of $\hat{C}_l$ corresponding to true class assignments $\zz^0$. In our example, although the estimates indicate that subjects in $\hat{C}_1$ are in the latent class ``1", we should obtain the latent class index of $\hat{C}_1$ under $\zz^0$, which is $m_1 = 2$. The number of correct latent class assignments under $\hat{\zz}$ in $\hat{C}_l$ is then $ \left|\{i\in \hat{C}_l: z_i^0=m_l \}\right|$. In the above example we have $m_1 = 2, \left|\{i\in \hat{C}_1: z_i^0=2 \}\right|=4 $, indicating that in $\hat{C}_1$ four subjects are correctly assigned to their true class membership. Similarly $m_2 = 1, \left|\{i\in \hat{C}_1: z_i^0=2 \}\right|=2$. The total number of correctly assigned subjects is then $\sum_{l}\left|\{i\in \hat{C}_l: z_i^0=m_l \}\right|$. In our example this number is $6$. The number of incorrect class assignments under estimated latent class assignments $N_e(\hat{\zz})$ is defined as: }
\begin{equation}
    \label{error}
    N_e(\hat{\zz}) = N - \sum_{l=1}^{L} \left|\{i\in \hat{C}_l: z_i^0=m_l \}\right|
\end{equation}
So every subject $i\in\{1,\ldots,N\}$
whose true class under $\zz^0$ is not in the majority within its estimated class under $\hat{\zz}$ is counted. In the example above $N_e(\hat{\zz}) = 2$. 

We have the following main theorem on the clustering consistency of joint MLE for fixed-effect LCM, which characterizes the asymptotic behavior of error rate $N(\hat{\zz})/N$.

\begin{theorem}
    \label{thm-joint} 
Under assumptions \ref{assume-pij} and \ref{cond-gap}, assume the following when $N,J\to\infty$,
\begin{equation}\label{thm-joint-scaling}
\frac{MJ}{\log L} \to \infty \;, \frac{N}{L} \to \infty,
\end{equation}
\[
\sqrt{\frac{M}{J}} \left( \frac{N}{L}\right)^{1-\xi} \rightarrow \infty \; \text{for some small $\xi>0$},
\]
for the joint maximum likelihood estimator $\hat{\zz}$, we have 
\begin{align}
\frac{N_e(\hat{\zz})}{N} = o_P\left(\frac{(\log J)^{1+\eta}\cdot\sqrt{M \log L} }{\sqrt{J}\beta_J}\right)	
\end{align}
for any $\eta>0$.
	
\end{theorem}
Assigning each subject to latent class resembles the process of clustering and hence we name our results as ``clustering consistency". In Theorem \ref{thm-joint}, we allow $L \rightarrow \infty$ or $L = O(1)$ and as long as the scaling conditions hold, our results hold. In particular, if $J$ remains bounded as $N,J \rightarrow \infty$, we have the following corollary. 

\begin{corollary}
    \label{cor-joint} 
Under assumptions \ref{assume-pij} and \ref{cond-gap}, assume $N,J\to\infty$ and $L = O(1)$,
\begin{equation}\label{cor-joint-scaling}
MJ \to \infty \;, 
\end{equation}
\[
\sqrt{\frac{M}{J}} N^{1-\xi}  \rightarrow \infty \; \text{for some small $\xi>0$},
\]
for the joint maximum likelihood estimator $\hat{\zz}$, we have 
\begin{align}
\frac{N_e(\hat{\zz})}{N} = o_P\left(\frac{(\log J)^{1+\eta}\cdot\sqrt{M } }{\sqrt{J}\beta_J}\right)	
\end{align}
for any $\eta>0$.
\end{corollary}

In particular if $M=\frac{1}{NJ}\sum_{i=1}^N\sum_{j=1}^J \MP(R_{i,j}=1)$ is of the constant order (denoted by $M =\Theta(1)$), the only scaling condition would become $J = o(N^{2(1-\xi)})$ for some small $\xi > 0$ and this is very mild. The rate depends on $\beta_J$ specified in \eqref{eq-as-gen2}. If the item parameters between different classes differ by a constant then we have $\beta_J = \Theta(1)$ and the final error rate $N_e(\hat{\zz})/{N} = o_P\left((\log J)^{1+\eta}/{\sqrt{J}}\right)$ decays towards zero as $N, J$ increases.

The following corollary shows the item parameters can be consistently estimated via joint MLE under some conditions as $N, J \rightarrow \infty.$

\begin{corollary}\label{item-consistency}

    Under assumptions \ref{assume-pij} and \ref{cond-gap} and the scaling conditions in Theorem \ref{thm-joint}, if we further assume clustering consistency holds:
    \[
    \frac{N_e(\hat{\zz})}{N} \stackrel{P}{\longrightarrow} 0, \text{  as } N, J \rightarrow \infty
    \]
    and there exists some positive constant $\tau$ such that $n_{l}^0/N \geq \tau $ for all $l \in [L]$ where $n_{l}^0 = \sum_{i}Z_{i,l}^0$ is the number of samples in latent class $l$. Then as $N,J \rightarrow \infty$, with probability approaching 1, for any $l \in [L]$ there exists a unique $a \in [L]$ such that $m_a = l$, i.e. the $a$-th cluster represents $l$-th class. Furthermore we have for any $l \in [L]$
    \[
    \max_{j} |\hat{\theta}_{j,a} - \theta_{j,l}^0| \stackrel{P}{\longrightarrow} 0 .
    \]
\end{corollary}

 The proof of Corollary \ref{item-consistency} can be found in the appendix. The parameter estimation consistency relies on clustering consistency established in Theorem \ref{thm-joint}. The condition $n_l^0/N \geq \tau$ for all $l$ guarantees that there are enough samples to estimate the item parameters for each class. 
 According to the theory presented above, both latent class membership and item parameters can be consistently estimated under mild conditions, which provides theoretical guarantees for real-world applications of large-scale latent class analysis.

 It is interesting to mention some pioneer work in high dimensional item factor analysis model \citep{chen2019joint, chen2019structured}. The fixed-effect LCM can be viewed as a special case of multidimensional IRT model in \cite{chen2019joint} where the person parameters correspond to latent class membership in this work and factor loadings correspond to item parameters. Then one may use constrained joint MLE approach in \cite{chen2019joint} to obtain estimates for person parameters and factor loadings with consistency guarantees. The main difference in our work and \cite{chen2019joint} is that to obtain the joint MLE, we maximize $\ZZ$ over matrices with one-hot rows (exactly one component in each row of $\ZZ$ can be 1) while the person parameters are optimized over any real number (satisfying certain constraints on the norm) in \cite{chen2019joint}. Since the joint MLEs are obtained in a different way, it is unclear how the joint MLE in LCM and constrained joint MLE in \cite{chen2019joint} are correlated. The continuous estimates for person parameters cannot translate to discrete latent class membership directly, making it hard to examine the relations between two estimators. Hence different techniques to establish the consistency of joint MLE are applied in our work. We will leave the connections between IRT models and LCM models for future explorations. In applications, both IRT models and latent class models can be fit to have different interpretations on the data.  

We further discuss the connection and difference between our results and those in \cite{gu2020joint}. \cite{gu2020joint} considered large-scale structured latent attribute models and established consistency of the joint MLE in their models. They also treated the latent part in their model (latent attribute profile) as fixed and derived the consistency of estimating the latent attribute profiles. Our work differs from \cite{gu2020joint} in the following respects: First, the assumption 2 in \cite{gu2020joint} requires each component of item parameters to be quite different for respondents with different latent attribute profiles. However the assumption \ref{cond-gap} in our work only requires the $L_2$ distance between item parameters for different latent classes to be quite different. The model structures considered in \cite{gu2020joint} are more delicate and may require stronger assumptions on the item parameters. The other main difference lies in the technical proof. After proving a bound on $\bar{\ell}\left(\mathbf{Z}^{0}\right)-\bar{\ell}(\widehat{\mathbf{Z}})$, we obtain the clustering consistency by a refined partition argument while \cite{gu2020joint} considered the structures implied by Q-matrix and identification assumptions to prove the consistency of estimating the Q-matrix vectors and latent attribute profiles. See the proof in the appendix for details.


\section{Simulation Study}\label{simulation}

In this section, we perform simulation studies to assess the performance of the tensor-EM method. Specifically, in Section \ref{simu-EM-tensor}, we examine the estimation accuracy and speed of tensor-EM method under LCMs. In Section  \ref{simu-local-dependence}, we consider the setting where local independence is violated and evaluate the robustness of the tensor-EM method. The clustering consistency is verified empirically together with comparisons to several other clustering methods in Section \ref{simu-cluster}. We also empirically evaluate the performance of GIC in selecting the number of classes in Section \ref{simu-GIC}.

\subsection{Performance of tensor-EM method under local independence}\label{simu-EM-tensor}

 We consider 24 different settings: $\{N = 1000,10000,20000\} \otimes \{J=100,200\}\otimes \{L = 5,10\}\otimes \{\text{item parameters} \in \{0.1,0.2,0.8,0.9\}$ or $\{0.2,0.4,0.6,0.8\}$\}. By item parameters $\in \{0.1,0.2,0.8,0.9\}$ we mean we generate the true $\TT$'s elements $\theta_{j,a}$ independently and uniformly from $\{ 0.1,0.2,0.8,0.9\}$. Note that under the considered large-scale LCM with many items, the generic identifiability conditions stated in Corollary 5 in \cite{allman2009identifiability} is guaranteed. 

We compare the performance of the  proposed tensor-EM method in Section \ref{sec-est-tem} with three other methods:\begin{enumerate}
    \item[(1)] EM-true, which is the EM algorithm starting from the true parameters as initial values;
    \item[(2)]  {EM-random, where we randomly generate the initial values for the EM algorithm. In random-effect LCM, we keep trying different initial values until we find the EM algorithm converges in 1000 iterations on five initial values and then we select the estimators corresponding to maximum log-likelihood. In fixed-effect LCM, we generate five initial values and run CEM algorithm on them until it converges or the number of iterations exceeds 1000, then we select the estimators corresponding to maximum log-likelihood. The first mechanism can hopefully find better solutions but the second one can save more time. We use EM-random algorithm with these two mechanisms and compare the results with tensor-EM to show the good performance of tensor-EM;}
    \item[(3)] the tensor method alone. 
    In this third competitor, we permute the items randomly and obtain a tensor estimate for each permutation.  Let $\pi$ be a permutation of $[J]$ and $\RR_i^{\pi}$ be subject response vector corresponding to permutation $\pi$ (i.e. $[\RR_i^{\pi}]_k = R_{i,\pi(k)}$). We obtain tensor estimates $\hat{\TT}^{\pi}$ from cross moments of $\RR_i^{\pi}$ and set the tensor estimates of original item parameters as $\hat{\theta}_{j,a} = \hat{\theta}_{\pi^{-1}(j),a}^{\pi}$. We then repeat this procedure five times and finally take average of them. This repetition can reduce the MSE of item parameters a little but will remain the same magnitude; 
    \item[(4)] {EM-tensor. This is the tensor-EM algorithm we detail in Section \ref{sec-est-tem}. We first apply the tensor method and obtain the tensor estimates. We then use tensor estimates as initializations for EM and CEM algorithm in random-effect LCM and fixed-effect LCM, respectively. }
\end{enumerate}   
We emphasize that in the proposed tensor-EM method, we do not repeat the tensor power method or take any average. 
Empirically, just one implementation of the tensor power method gives good initial values for the EM algorithm.

The running time and MSE are reported, where MSE = $\sum_{j=1}^J \sum_{l=1}^L (\theta_{j,l} - \hat \theta_{j,l})^2/(JL)$. 
In some settings, the EM-random estimates have too large MSEs, so we also present the plots excluding the EM-random estimates to better visualize the MSEs of the other three methods. 
The results are based on 100 replications in each simulation setting. For the random-effect LCM, the population proportion vector $\pp$ and item parameters $\TT$ are first generated and the process of generating samples and estimation is repeated. The proportion vector $\pp$ is generated randomly to guarantee each class has enough samples (in settings with five classes we have $p_l \geq 0.1$ and in settings with ten classes we have $p_l \geq 0.08$ for all $l$). For the fixed-effect LCM, the latent class membership $\ZZ$ and $\TT$ are generated and the process of sampling and estimation is repeated (so unlike in random-effect LCM, $\ZZ$ is the same for each replication). The convergence criterion for EM(CEM) algorithm is set as when the improvement in likelihood is less than 0.1 (which has a relative tolerance smaller than $10^{-5}$ in the considered simulation settings).

Due to the space limitation, we present here two representative figures for each type of LCM (i.e., random-effect LCM and fixed-effect LCM)  in Figures \ref{fig-rand-small}--\ref{fig-fix-big}, and provide the rest simulation results in Supplementary Material. 

\begin{figure}[H]
	\centering
	\subfigure[MSE of item parameters]{
		\begin{minipage}[t]{0.45\linewidth}
			\centering
			\includegraphics[width=2in]{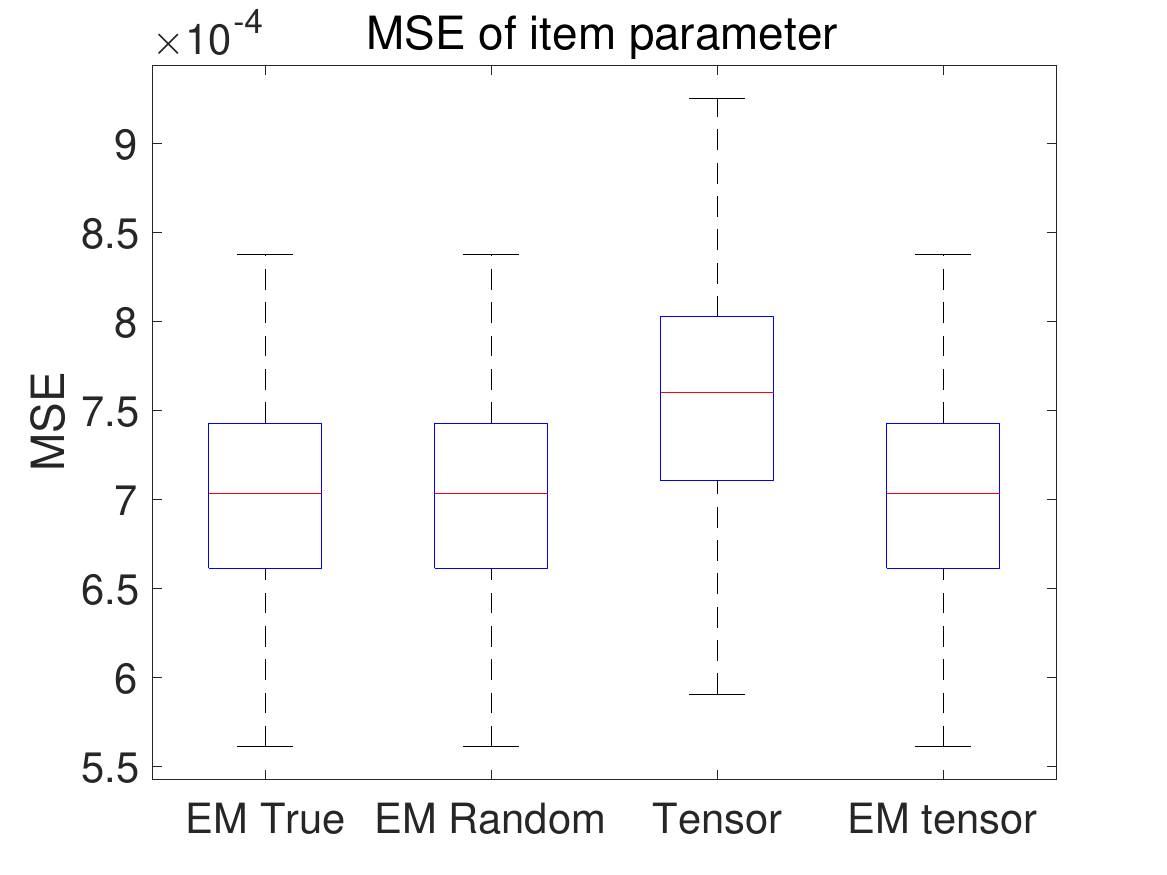}
		\end{minipage}%
	}%
	\subfigure[Running time of the algorithms]{
		\begin{minipage}[t]{0.45\linewidth}
			\centering
			\includegraphics[width=2in]{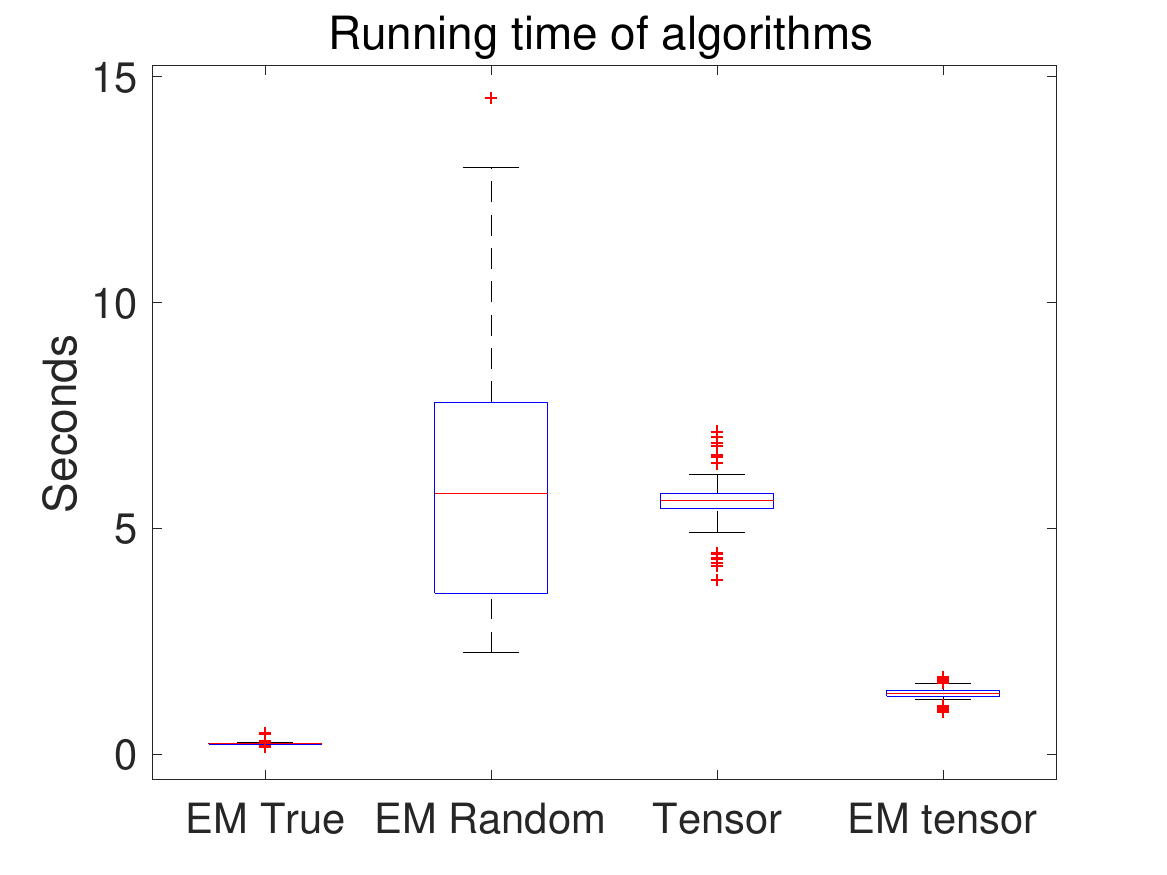}
		\end{minipage}%
	}%
	\centering
	\caption{Random-effect LCM, $N = 1000, J= 100, L=5, \theta_{j,a}\in \{0.1,0.2,0.8,0.9\}$}
	\label{fig-rand-small}
\end{figure}

\begin{figure}[H]
	\centering
	\subfigure[MSE of item parameters]{
		\begin{minipage}[t]{0.33\linewidth}
			\centering
			\includegraphics[width=2in]{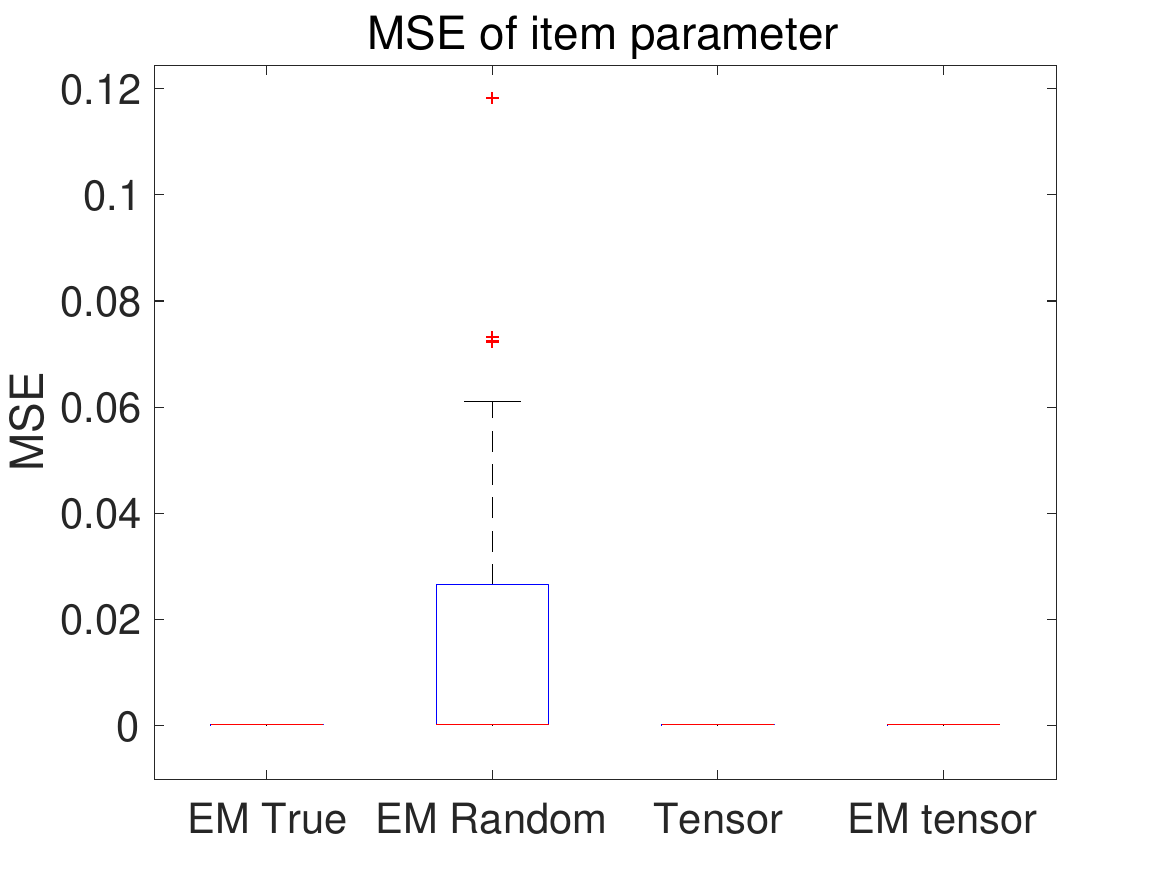}
		\end{minipage}%
	}%
		\subfigure[MSE without EM-random]{
		\begin{minipage}[t]{0.33\linewidth}
			\centering
			\includegraphics[width=2in]{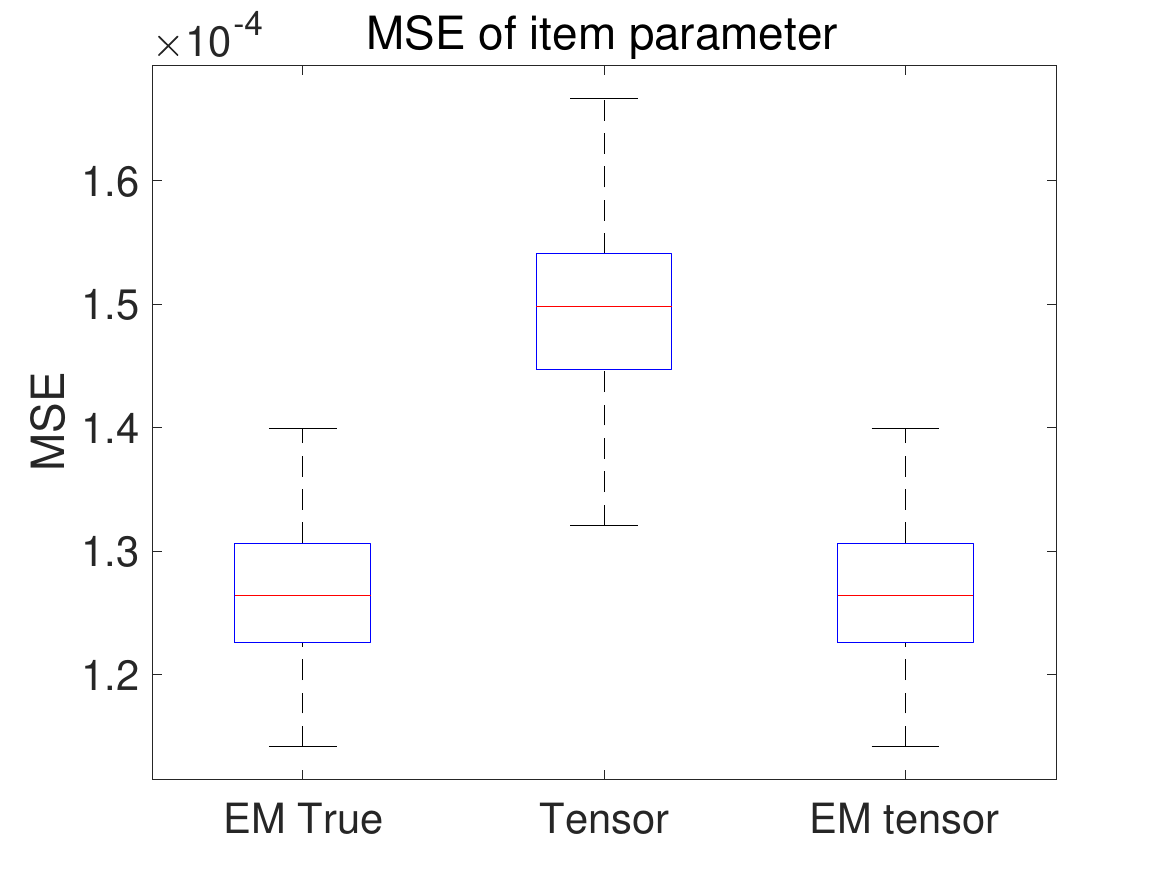}
		\end{minipage}%
	}%
	\subfigure[Running time of the algorithms]{
		\begin{minipage}[t]{0.33\linewidth}
			\centering
			\includegraphics[width=2in]{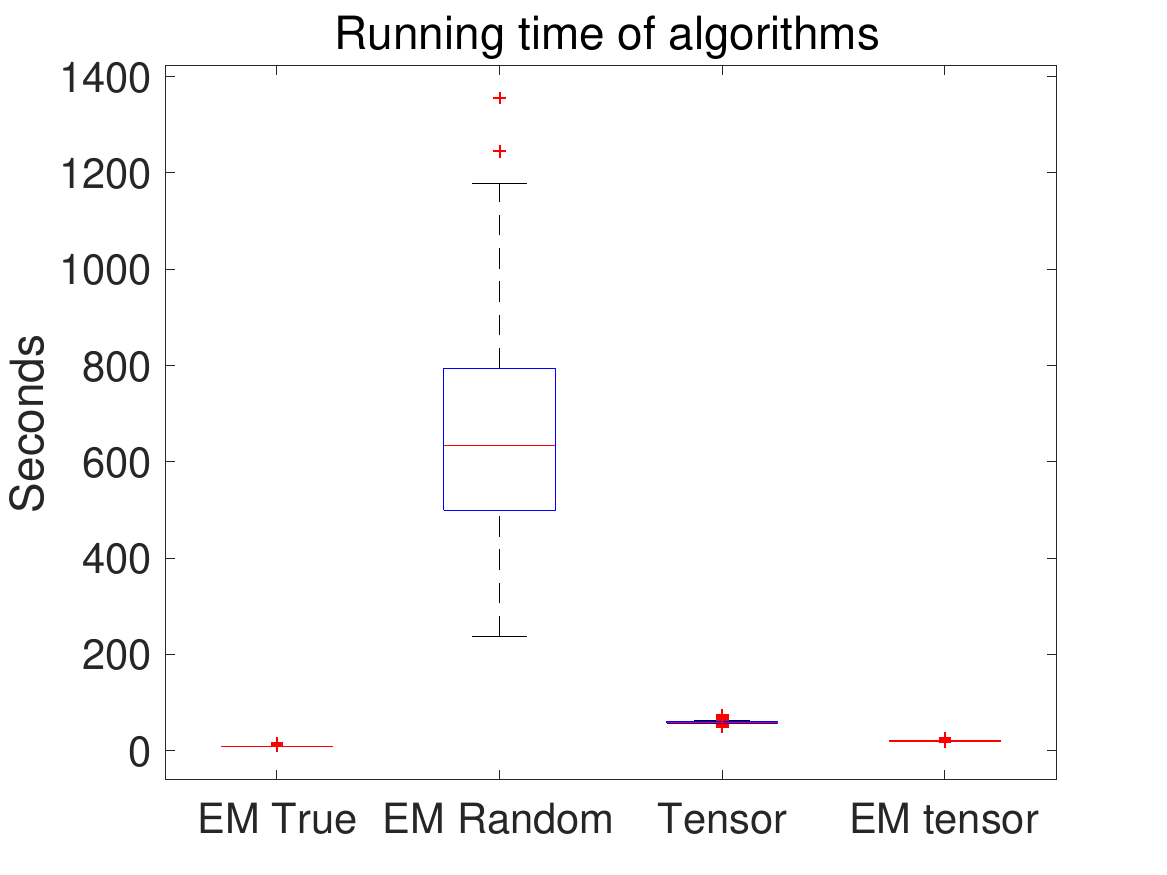}
		\end{minipage}%
	}%
	\centering
	\caption{Random-effect LCM, $N = 10000, J= 100, L=10,\theta_{j,a}\in \{0.1,0.2,0.8,0.9\}$}
	\label{fig-rand-big}
\end{figure}

\begin{figure}[H]
	\centering
	\subfigure[MSE of item parameters]{
		\begin{minipage}[t]{0.33\linewidth}
			\centering
			\includegraphics[width=2in]{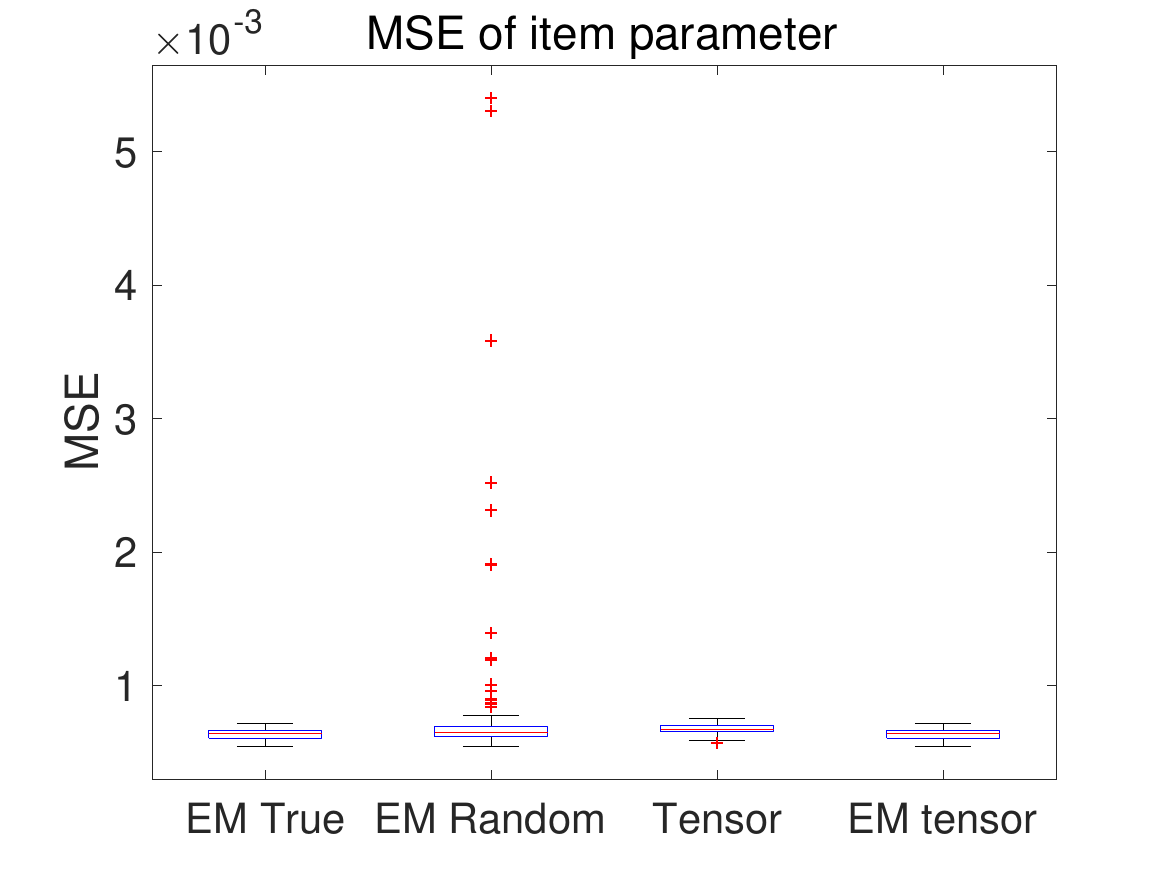}
		\end{minipage}%
	}%
	\subfigure[MSE without EM-random]{
		\begin{minipage}[t]{0.33\linewidth}
			\centering
			\includegraphics[width=2in]{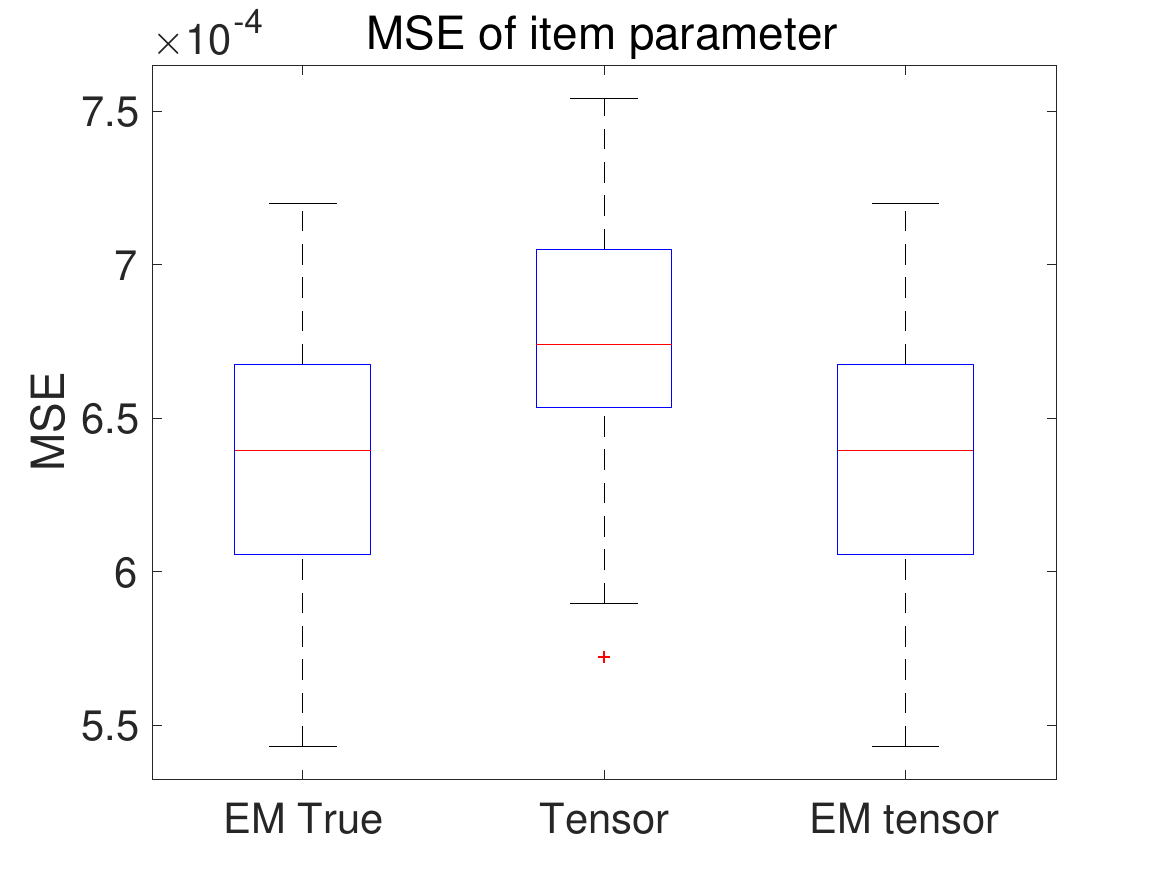}
		\end{minipage}%
	}%
	\subfigure[Running time of the algorithms]{
		\begin{minipage}[t]{0.33\linewidth}
			\centering
			\includegraphics[width=2in]{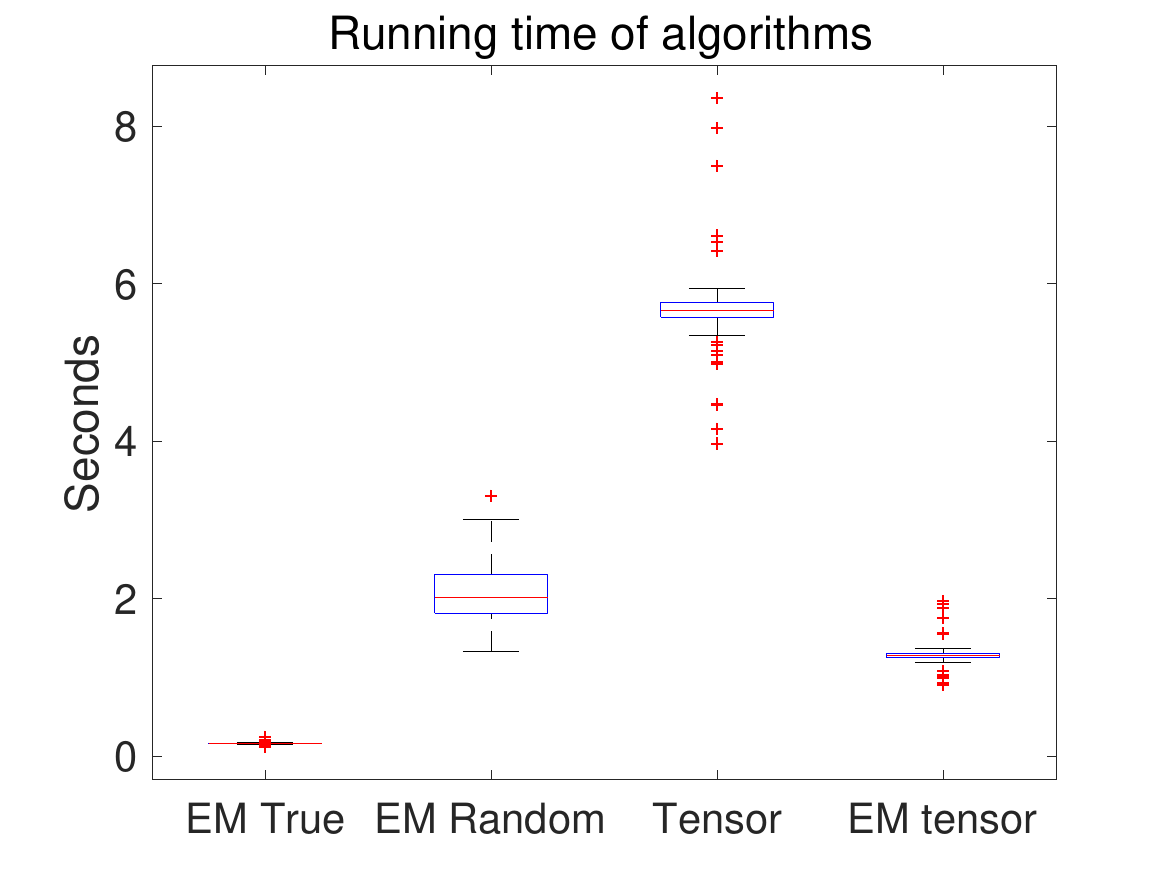}
		\end{minipage}%
	}%
	\centering
	\caption{Fixed-effect LCM, $N = 1000, J= 100, L=5,\theta_{j,a}\in \{0.1,0.2,0.8,0.9\}$}
	\label{fig-fix-small}
\end{figure}

\begin{figure}[H]
	\centering
	\subfigure[MSE of item parameters]{
		\begin{minipage}[t]{0.33\linewidth}
			\centering
			\includegraphics[width=2in]{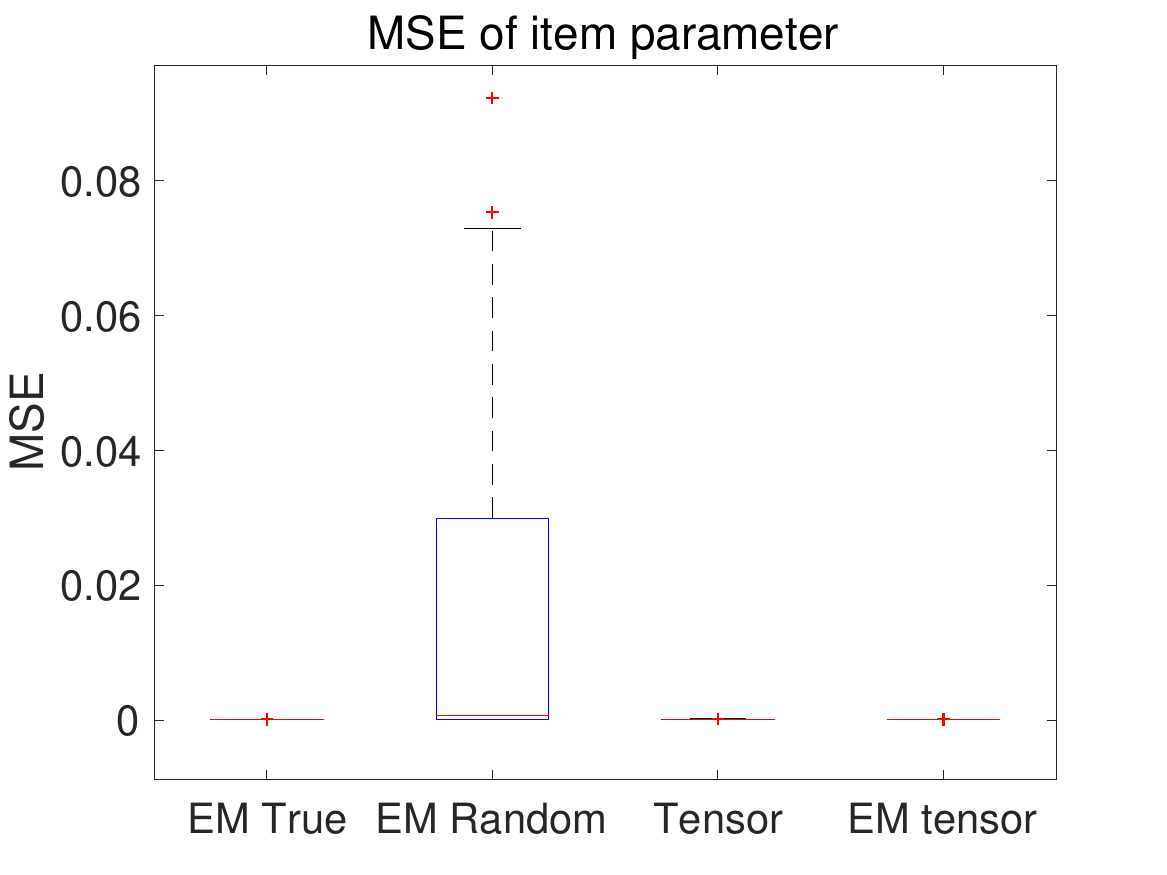}
		\end{minipage}%
	}%
	\subfigure[MSE without EM-random]{
		\begin{minipage}[t]{0.33\linewidth}
			\centering
			\includegraphics[width=2in]{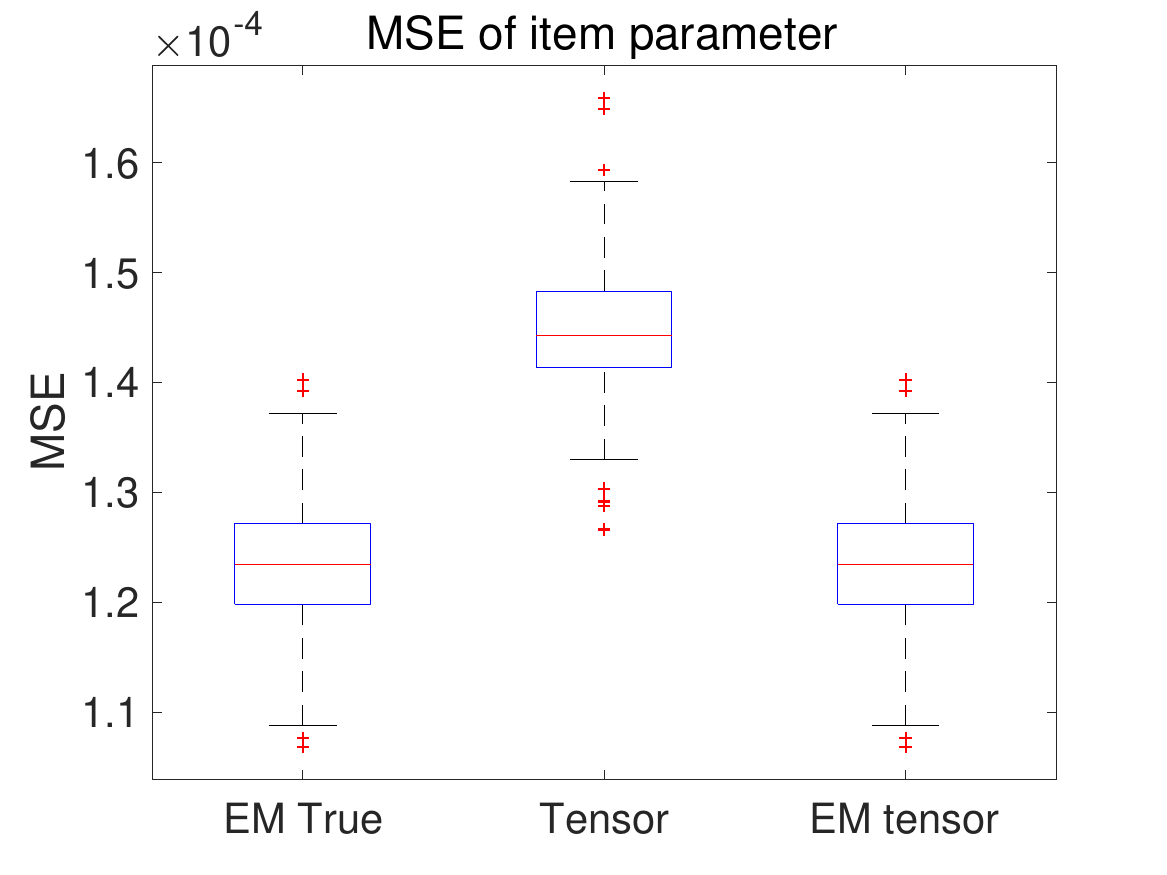}
		\end{minipage}%
	}%
	\subfigure[Running time of the algorithms]{
		\begin{minipage}[t]{0.33\linewidth}
			\centering
			\includegraphics[width=2in]{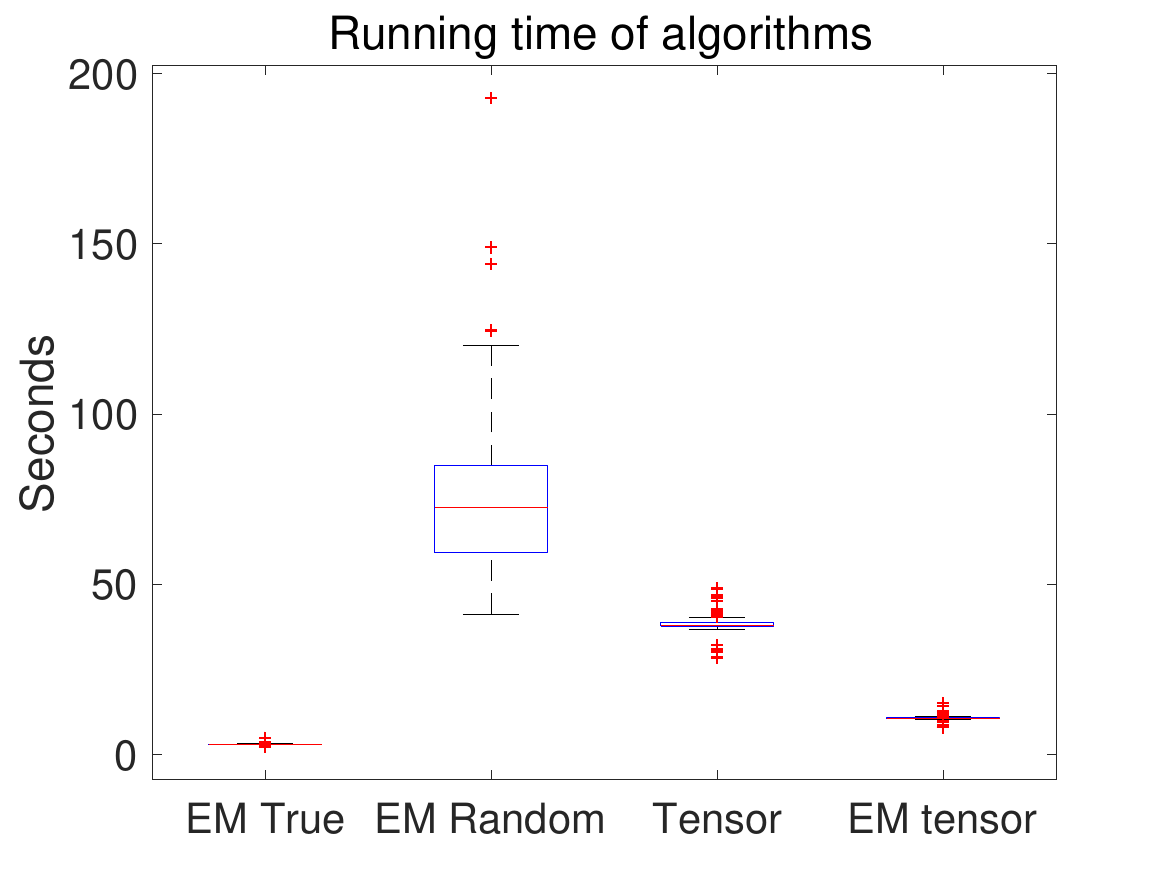}
		\end{minipage}%
	}%
	\centering
	\caption{Fixed-effect LCM, $N = 10000, J= 100, L=10,\theta_{j,a}\in \{0.1,0.2,0.8,0.9\}$}
	\label{fig-fix-big}
\end{figure}

From the boxplots in Figures \ref{fig-rand-small}--\ref{fig-fix-big} and those in Supplementary Material, we can see that for each setting, the MSE of the tensor-EM method is  almost the same as that of the EM-true method. The  EM-random method   sometimes yields local maximizer of log-likelihood function and thus its estimates have a large MSE. The tensor estimates alone have a larger MSE compared with the tensor-EM estimates but are more stable than the EM-random estimates.
Comparing the running time of different algorithms, we can find that the tensor-EM method is computationally efficient, only second to the performance of EM-true with true parameters as initials values. On the other hand, the EM-random method can be computationally intensive because it needs more steps to converge.

The sample size $N$, number of classes $L$ and number of items $J$ all affect the accuracy and running time of the methods examined. As the sample size $N$ increases, the accuracy of the methods is improved while they all need more time. As the number of classes $L$ increases, the MSE of EM-tensor estimates also becomes larger because we have more parameters to estimate. As the number of items $J$ increases, the accuracy of EM-tensor remains comparable to EM-true while the accuracy of tensor method alone is improved. The running time increases as $J, L$ becomes large. When the item parameters are generated in $ \{ 0.1,0.2,0.8,0.9\}$, the signal strength is strong and the estimates have smaller MSE compared with cases where item parameters are generated from $\{0.2,0.4,0.6,0.8\}$ where the signal strength is weaker. We also note that the random-effect and fixed-effect LCMs with the same $N, J, L$ and item parameters share similar orders of MSEs.

To further show the advantage of tensor-EM method over EM-random, we perform more simulations. First, we let them start from the same initial points (under some transformations) and evaluate their estimation accuracy and running time. Recall we need to use second-order moments $\MM_2$ to whiten $\MM_3$ and get orthogonal decomposable tensor $\tilde{\MM}_3$. To ensure they start from the same initial points, we mimic the whitening process and take the following strategy.

We first divide the item parameters $\TT$ into three parts as described in Section \ref{estimation}. Then we randomly generate initial values $\TT_1^0$ for $\TT_1$ from $U(0,1)$ . For EM-random we also need initial values for $\TT_2$ and $\TT_3$. From the relations between $\TT_i$'s in Section 3.2, we set $\TT_2^0 = \hat{\mathbb{E}} [\RR_i^2\otimes \RR_i^3] \hat{\mathbb{E}} [\RR_i^1 \otimes \RR_i^3]^+\TT_1^0$  and $\TT_3^0 = \hat{\mathbb{E}} [\RR_i^3\otimes \RR_i^2] \hat{\mathbb{E}} [\RR_i^1 \otimes \RR_i^2]^+\TT_1^0$ and concatenate them to form $\TT^0$ and feed to EM-random algorithm. For Tensor-EM method, we need to transform $\TT_1^0$. Recall the columns of $\TT_1$ play the role of $\mmu_i$'s in \eqref{moments} and after we define $\tilde{\mmu}_i = \sqrt{w_i}\WW^{\top} \mmu_i$, then $\tilde{\mmu}_i$'s are sets of eigenvectors for an orthogonal decomposable tensor $\tilde{\MM}_3$, which we perform tensor power method on. Now let $\tilde{\mmu}_i^0 = \sqrt{p_i} \WW^{\top} \TT_{1,i}^0$ for $i \in [L]$ and we use $\tilde{\mmu}_i^0$ as the initial value in Algorithm \ref{robust-tp} to obtain the $i$-th eigenvalue/eigenvector pair of the estimated $\MM_3$ from the samples (i.e. we do not perform $K$ random initializations as shown in Algorithm \ref{robust-tp}). After we obtain estimates of $\tilde{\mmu}_i$, we use the relation shown in the last paragraph of Section \ref{sec-whitening} to get the tensor estimates of ($\mmu_i,\TT_{1,i}$) and hence obtain the tensor estimates for $\TT$. We then implement EM algorithms starting from it and evaluate its performance. 

 {We also implement a ``smarter" version of EM-random. The idea is to use a large tolerance to run EM with multiple random starting points and then run EM with a refined tolerance with the solution that gives largest likelihood in the first stage. We first run EM with 10 random initial values for 20 iterations and then find the solution that yields the largest likelihood. Then we run EM starting from that solution until convergence. The above three algorithms together with EM-true are run with 100 replications under settings: $\{N = 1000,10000\} \otimes \{J=100,200\}\otimes \{L = 5,10\}\otimes \{\text{item parameters} \in \{0.1,0.2,0.8,0.9\}$ or $\{0.2,0.4,0.6,0.8\}$\}.  Some results are shown in Figures \ref{fig:moresimu_smallN} and \ref{fig:moresimu_largeN}, where EM-random and EM-tensor use the same starting values as detailed above and EM-random(refined) uses the ``smarter" refined-tolerance version of EM-random. More simulation results can be found in the appendix.}

\begin{figure}[H]
	\centering
	\subfigure[MSE of item parameters]{
		\begin{minipage}[t]{0.45\linewidth}
			\centering
			\includegraphics[width=2in]{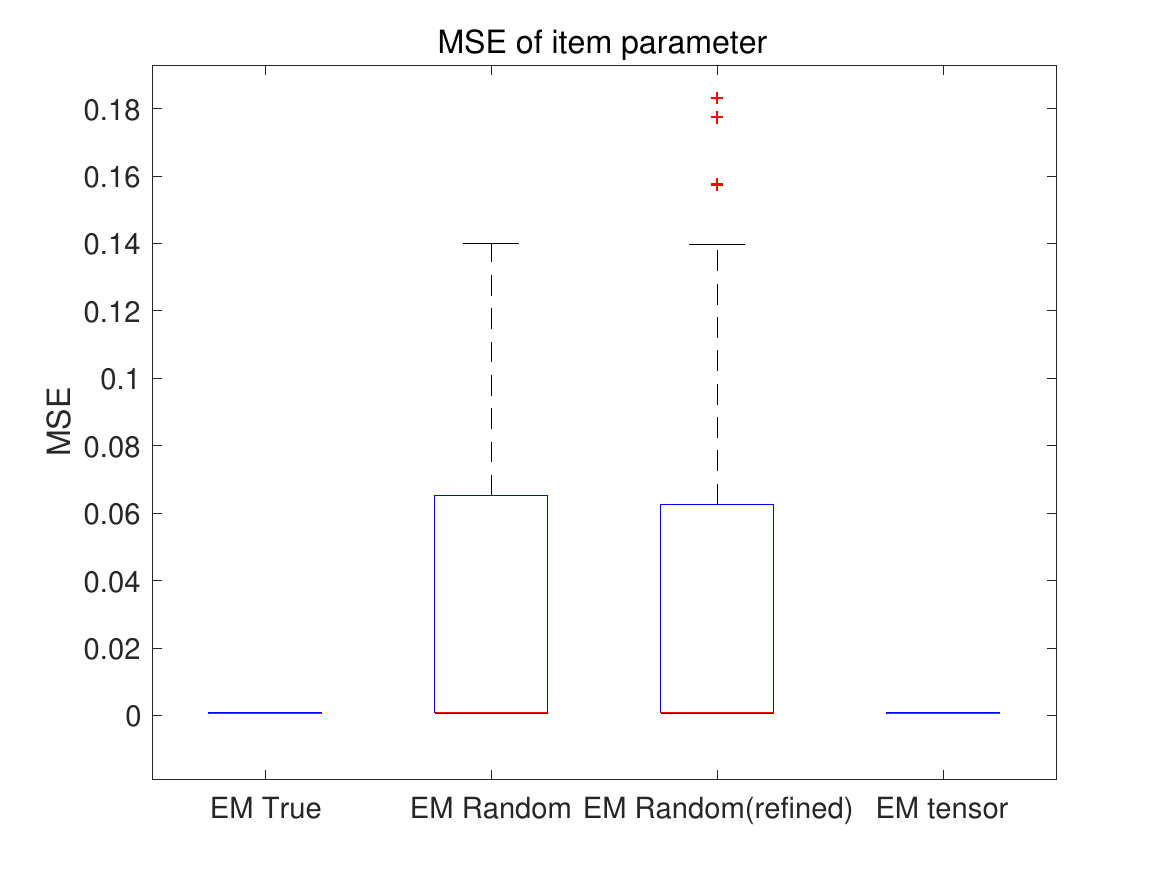}
		\end{minipage}%
	}%
	\subfigure[Running time of the algorithms]{
		\begin{minipage}[t]{0.45\linewidth}
			\centering
			\includegraphics[width=2in]{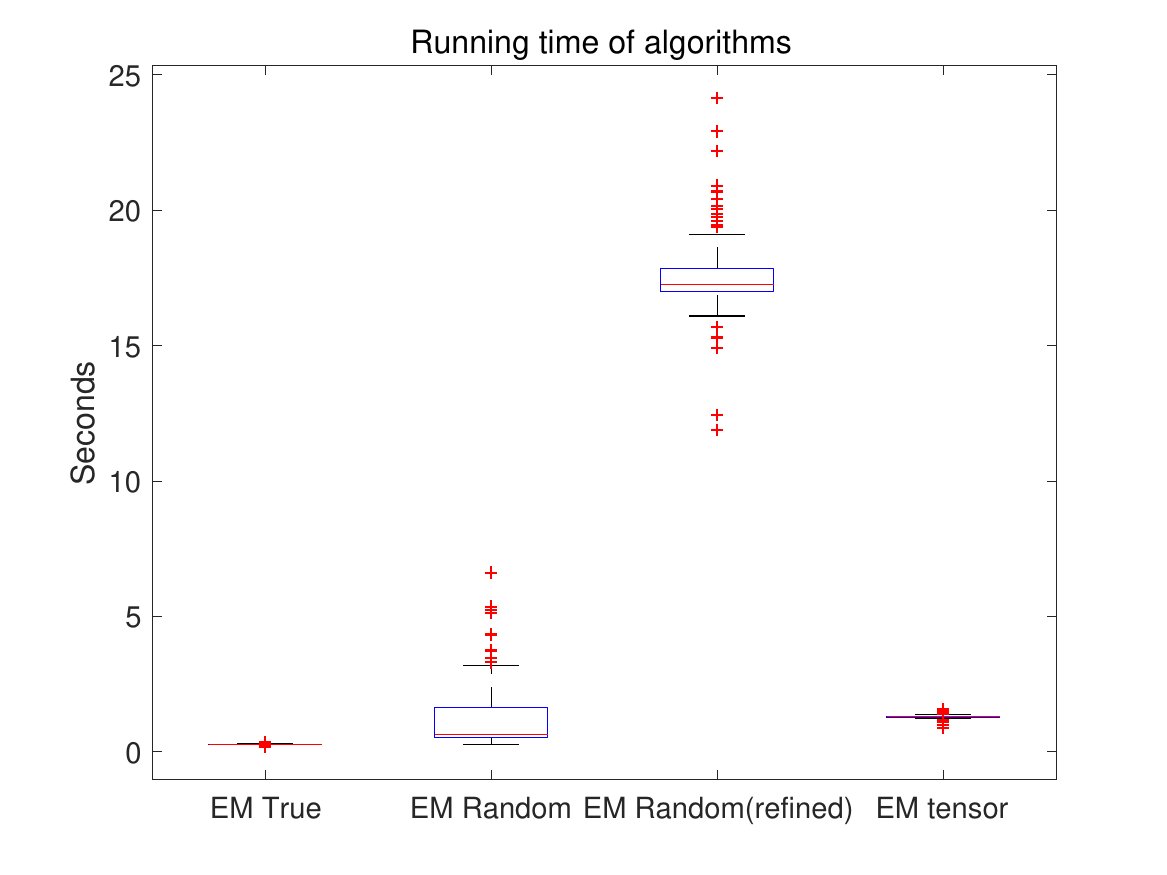}
		\end{minipage}%
	}%
	\centering
	\caption{Random-effect LCM, $N = 1000, J= 100, L=5, \theta_{j,a}\in \{0.1,0.2,0.8,0.9\}$}
	\label{fig:moresimu_smallN}
\end{figure}

\begin{figure}[H]
		\centering
		\subfigure[MSE of item parameters]{
			\begin{minipage}[t]{0.45\linewidth}
				\centering
				\includegraphics[width=2in]{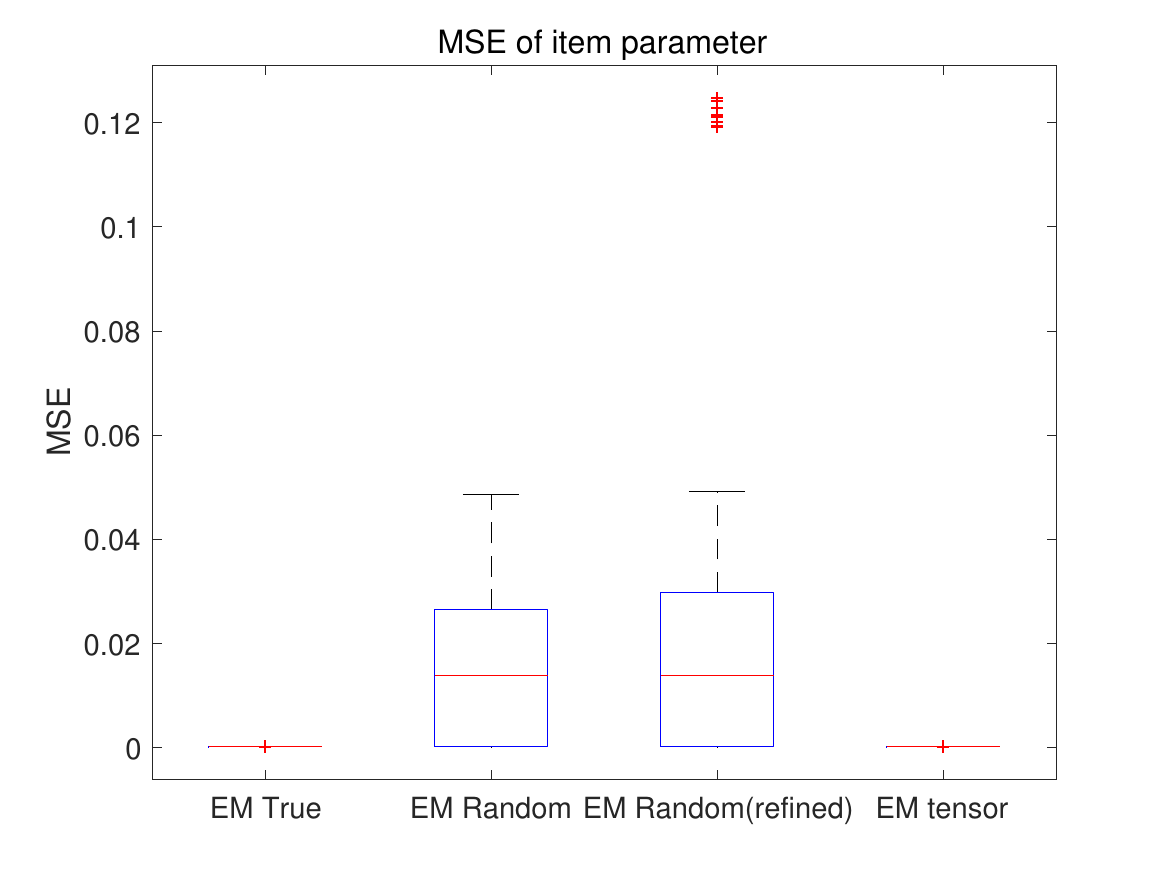}
			\end{minipage}%
		}%
		\subfigure[Running time of the algorithms]{
			\begin{minipage}[t]{0.45\linewidth}
				\centering
				\includegraphics[width=2in]{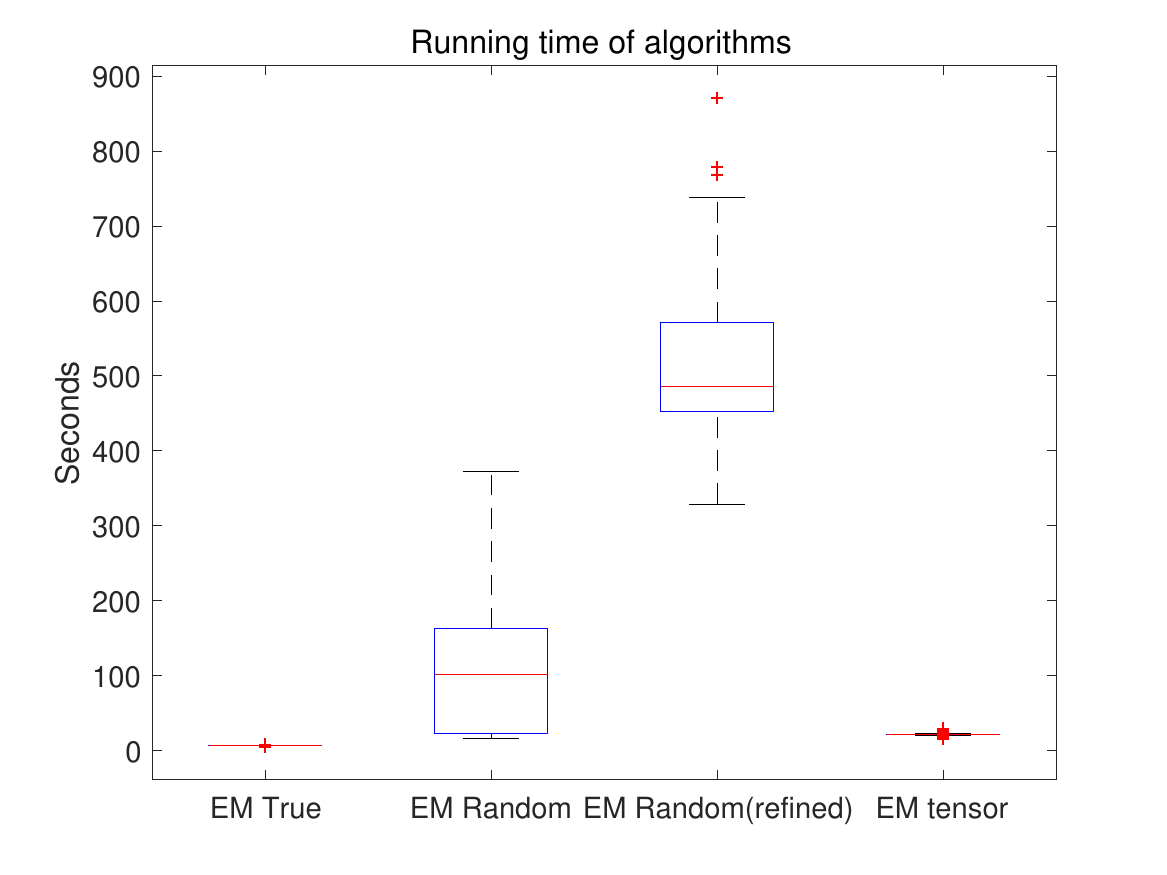}
			\end{minipage}%
		}%
		\centering
		\caption{Random-effect LCM, $N = 10000, J= 200, L=10,$ $\theta_{j,a} \in \{0.2,0.4,0.6,0.8\}$}
		\label{fig:moresimu_largeN}
\end{figure}

 {As shown in Figures \ref{fig:moresimu_smallN} and \ref{fig:moresimu_largeN}, the tensor-EM method has similar MSE with EM-true and outperforms the two EM-random algorithms. Comparing EM-random with tensor-EM, we see the tensor-EM has smaller MSE when they use the same initialization, indicating the better performance of tensor-EM. Comparing EM-random using refined tolerance with tensor-EM, the tensor-EM can yield more accurate results efficiently. One possible explanation is that the EM-random method using refined tolerance is a greedy method since it chooses the best solution only based on the first few iterations, which may not be very reliable.  }

\subsection{Performance of tensor-EM method under local dependence}\label{simu-local-dependence}

 {We further investigate the performance of EM algorithms under local dependence. The data generating process is as follows: After we generate $z_i$ from the proportion vector $\pp$, given $z_i = l$, we first obtain one sample $\XX_i$ from $J$-dimensional multivariate normal distribution $\mathcal{N}(\mathbf{0}, \mathbf{\Sigma}_\rho)$, where $\mathbf{\Sigma}_\rho$ is the covariance matrix of an auto-regression model, i.e. $[\mathbf{\Sigma}_\rho]_{i,j} = \rho^{|i-j|}$ for some $0 < \rho < 1$. Then we let $R_{i,j} = I(X_{i,j} < q(\theta_{j,l}))$ where $q(\alpha)$ is the $\alpha$-quantile of standard normal distribution. This guarantees marginally we have $\MP(R_{i,j}=1 | z_i=l) = \theta_{j,l}$ but conditioning on $z_i$, the components of $\RR_i$ are now correlated. The value of $\rho$ controls the extent to which they are correlated. $\rho=0$ corresponds to the conditional independent case. We run EM-random, tensor-EM and EM-true under the following settings: $\{N = 1000\} \otimes \{\rho=0.3, 0.7 \}\otimes \{J=100,200\}\otimes \{L = 5,10\}\otimes \{\text{item parameters} \in \{0.1,0.2,0.8,0.9\}$ or $\{0.2,0.4,0.6,0.8\}$\}. These three algorithms are run in the same way as what we did in Figures \ref{fig-rand-small}--\ref{fig-fix-big}. Some results are presented in Figures \ref{fig:localdepend1} and \ref{fig:localdepend2} and more results can be found in appendix.}

\begin{figure}[H]
	\centering
	\subfigure[MSE of item parameters $\rho=0.3$]{
		\begin{minipage}[t]{0.45\linewidth}
			\centering
			\includegraphics[width=2in]{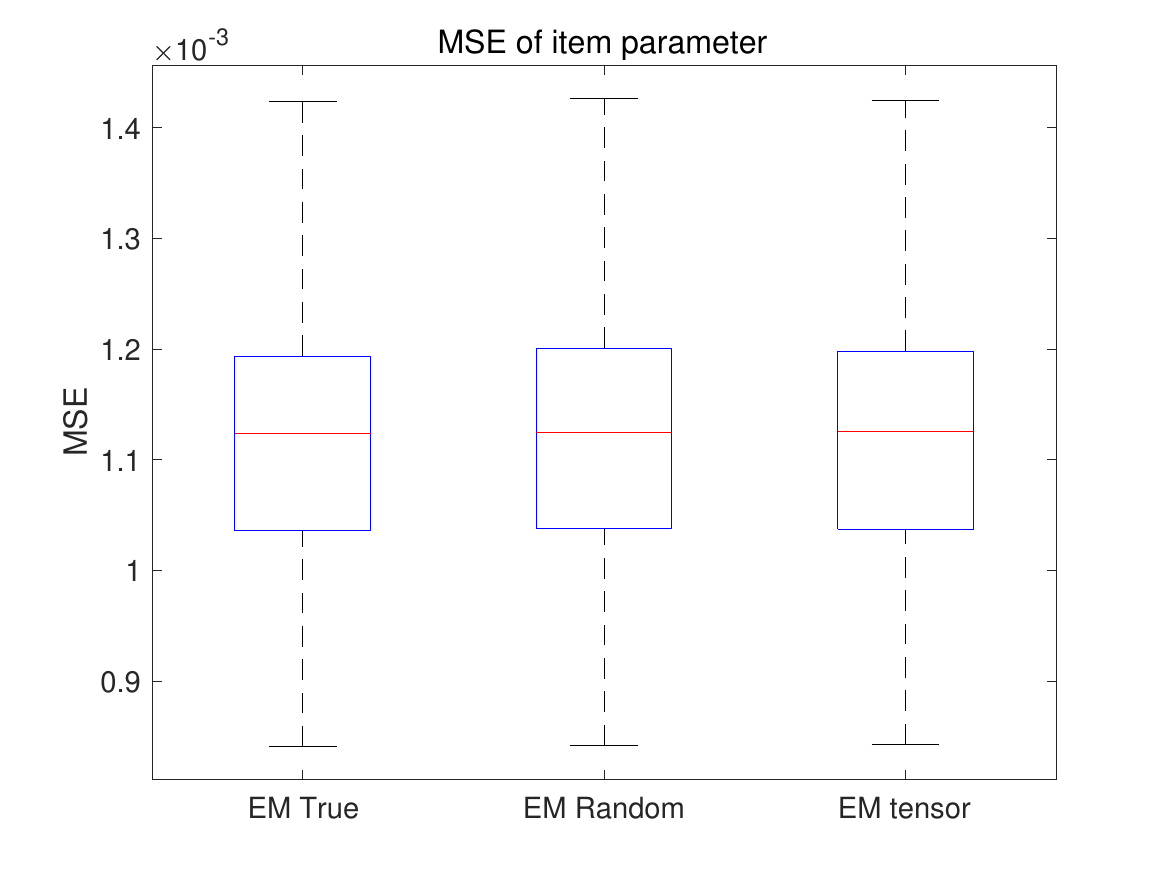}
		\end{minipage}%
	}%
	\subfigure[MSE of item parameters $\rho=0.7$]{
		\begin{minipage}[t]{0.45\linewidth}
			\centering
			\includegraphics[width=2in]{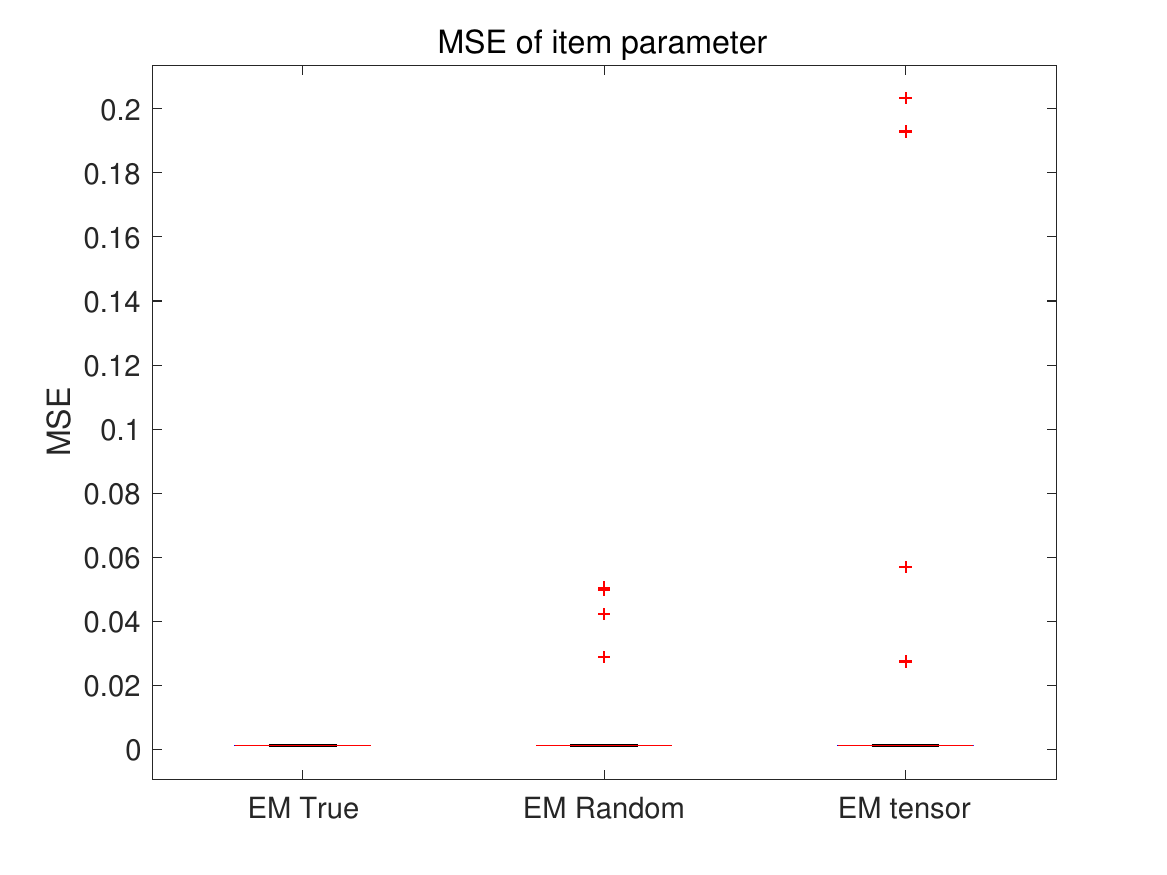}
		\end{minipage}%
	}%
	\centering
	\caption{Random-effect LCM, $N = 1000, J= 100, L=5, \theta_{j,a}\in \{0.2,0.4,0.6,0.8\}$}
	\label{fig:localdepend1}
\end{figure}

\begin{figure}[H]
	\centering
	\subfigure[MSE of item parameters $\rho=0.3$]{
		\begin{minipage}[t]{0.45\linewidth}
			\centering
			\includegraphics[width=2in]{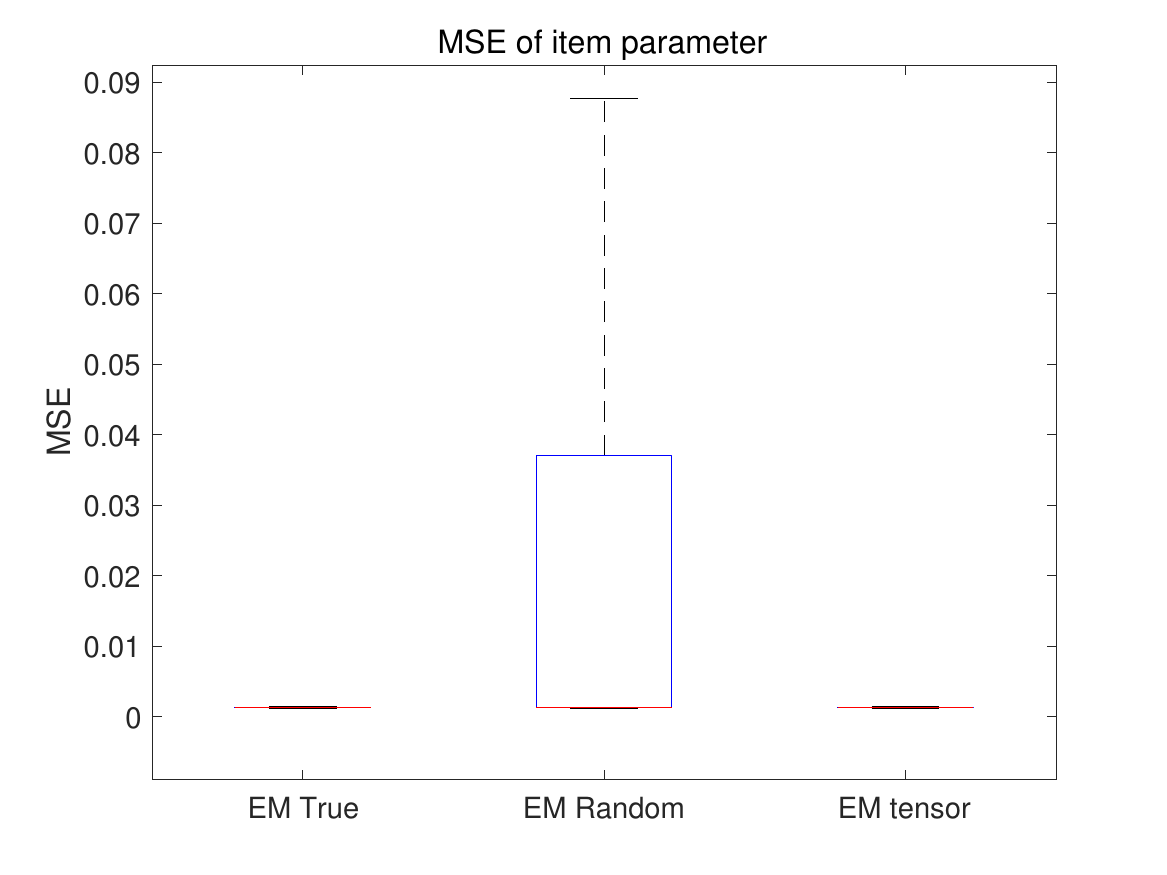}
		\end{minipage}%
	}%
	\subfigure[MSE of item parameters $\rho=0.7$]{
		\begin{minipage}[t]{0.45\linewidth}
			\centering
			\includegraphics[width=2in]{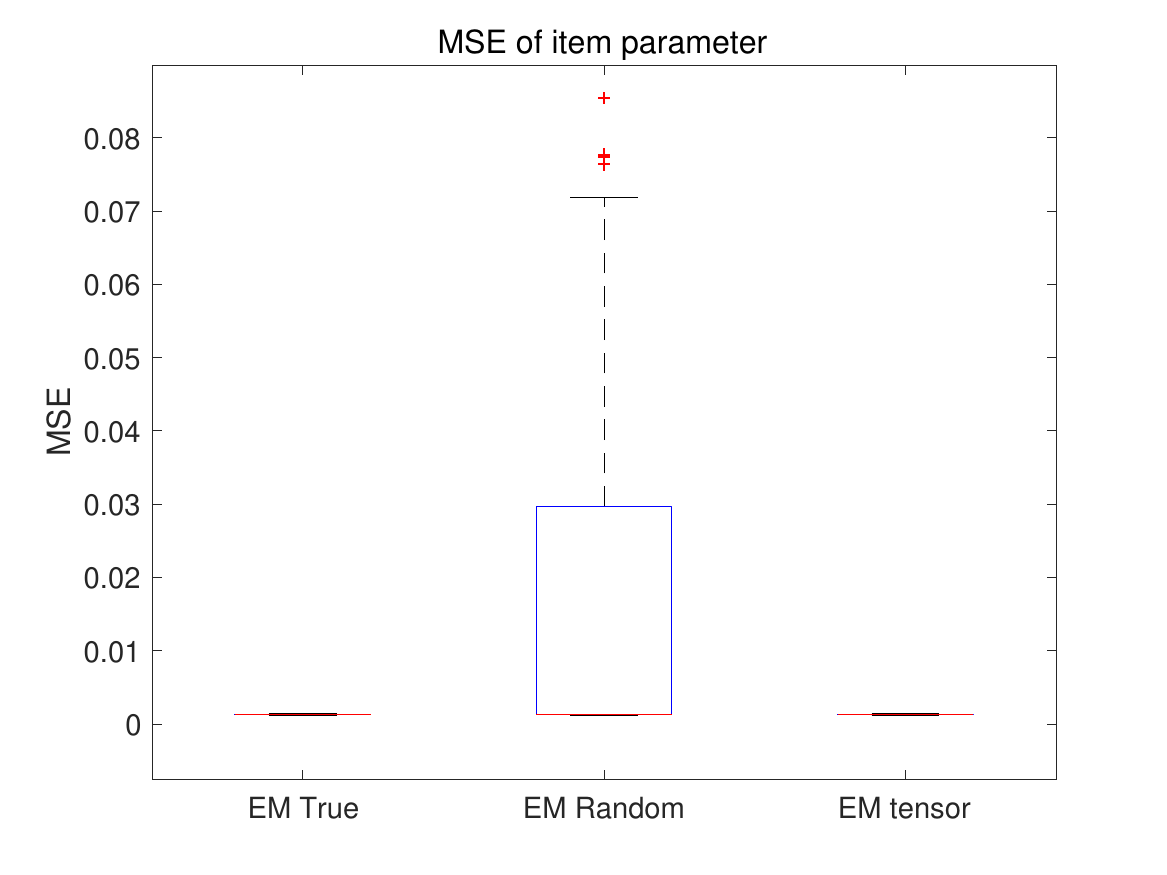}
		\end{minipage}%
	}%
	\centering
	\caption{Random-effect LCM, $N = 1000, J= 100, L=10, \theta_{j,a}\in \{0.1,0.2,0.8,0.9\}$}
	\label{fig:localdepend2}
\end{figure}

 {We note that in the local dependent setting, the log-likelihood is no longer valid and hence there is no guarantee on the EM-true method, which is also based on the log-likelihood. However, EM-true can still yield accurate estimates. The proposed tensor-EM method has similar performance with EM-true and is also robust against violations of local independence. In contrast, the EM-random method may not work well under local dependence settings.}

\subsection{Verification of Clustering Consistency}\label{simu-cluster}

In this subsection, we empirically verify Theorem \ref{thm-joint} in fixed-effect LCMs with diverging $N$ and $J$. Specifically, we consider fixed-effect LCM with $L = 5$ classes. We let $J$ increase from 30 to 100 by 10 and set $N = 10J$ in all the simulations. {The only purpose of setting $N=10J$ is that we can visualize the error rate $N_e(\hat{\zz})/N$ as a function of $J$ in a plot to see the trend.} Item parameters are generated uniformly from either $\{0.1,0.2,0.8,0.9\}$ or $\{0.2,0.4,0.6,0.8\}$ and latent class assignments $\zz$ are sampled uniformly over $[L]$. After item parameters and latent class assignments are generated, we generate response $\RR$ accordingly. 
Since we have shown that tensor-EM method has good performance in Section \ref{simu-EM-tensor}, we then apply tensor-EM method to obtain the joint MLE $(\hat{\TT}, \hat{\ZZ})$. 
The error of the estimated latent class assignments $\hat{\ZZ}$ is evaluated and the number of incorrect assignments $N_e(\hat{\zz})$ defined in \eqref{error} is computed for each replication. This process of generating $\RR$, estimating $\ZZ$ and evaluating error is replicated 100 times and the boxplots of error rate $N_e(\hat{\zz})/N$ are shown in Figure \ref{error-rate}. According to these plots, the error rate of the estimated latent class membership $\zz$ decays to zero as $N, J$ increases. Again in the strong signal setting where $\theta_{j,a} \in \{0.1,0.2,0.8,0.9 \}$ the error rate converges to zero faster.

\begin{figure}[H]
	\centering
	\subfigure[Item parameters $\in \{0.1, 0.2, 0.8, 0.9\}$]{
		\begin{minipage}[t]{0.45\linewidth}
			\centering
			\includegraphics[width=\textwidth]{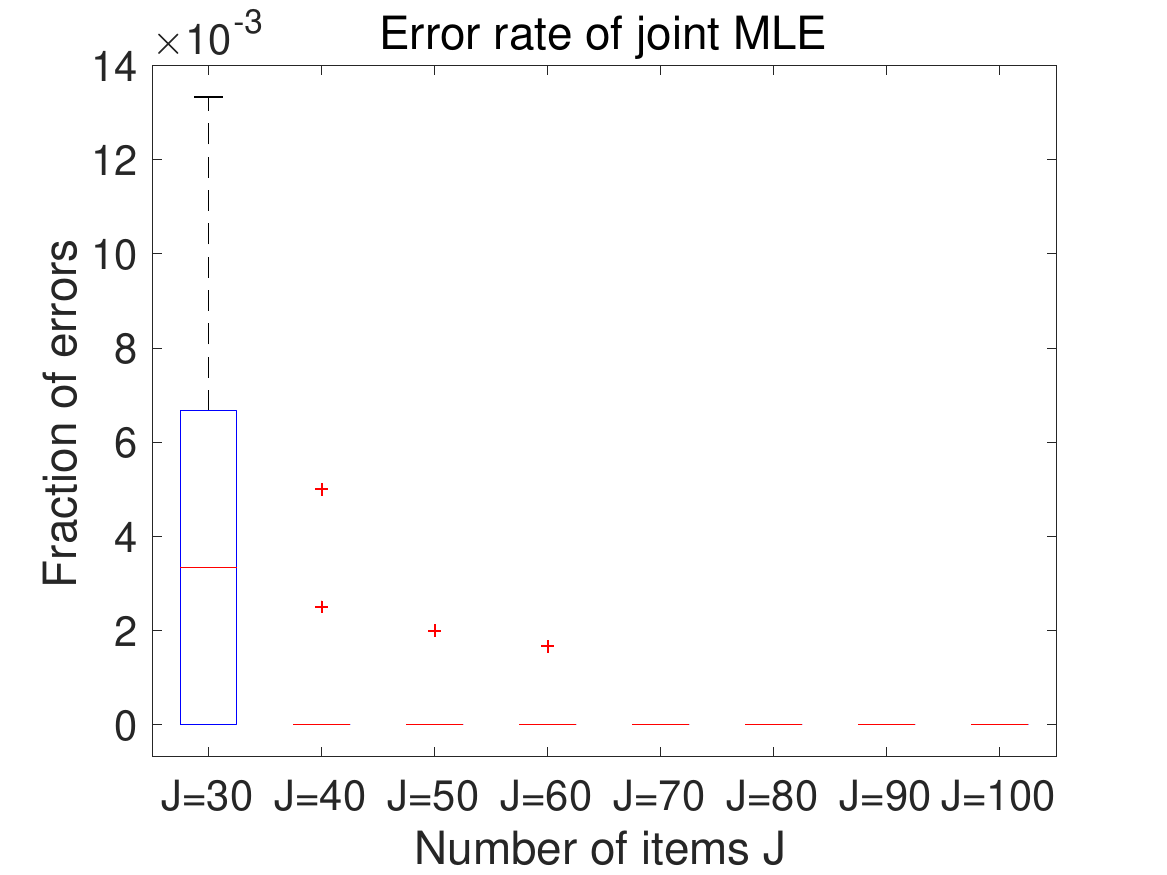}
		\end{minipage}%
	}%
	\subfigure[Item parameters $\in \{0.2, 0.4, 0.6, 0.8\}$]{
		\begin{minipage}[t]{0.45\linewidth}
			\centering
			\includegraphics[width=\textwidth]{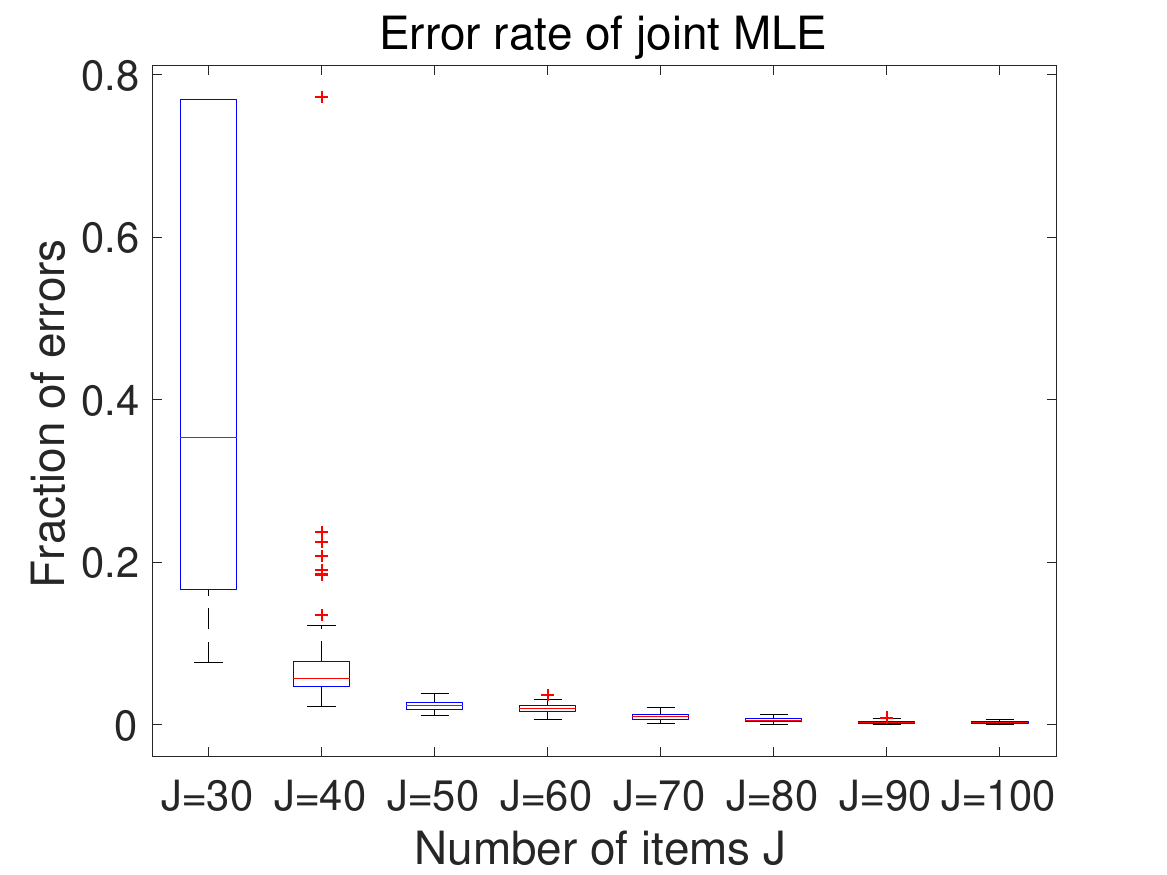}
		\end{minipage}%
	}%
	\centering
	\caption{Error rate of joint MLE in latent class assignments versus number of items $J$}
	\label{error-rate}
\end{figure}

We also compare the clustering performance of the proposed tensor-EM method with the performance of several other commonly used clustering algorithms. To be concrete, we consider the following clustering algorithms.
\begin{enumerate}
    \item Max linkage clustering. In our simulation, max linkage is found to have  a better performance compared with single linkage and average linkage, and hence here we only present the results of the max linkage clustering.
    \item K-medoids. Considering the binary response in our setting, Hamming distance is applied as a metric. For each replication, the algorithm is repeated ten times with different initial cluster centroid positions and the best is chosen as the final result.
    \item K-means. Similarly to K-medoids,  Hamming distance is used and for each replication, the algorithm is performed with ten   different initial values.
    \item Spectral clustering with normalization. Hamming distance is used to compute the similarity between data points.
\end{enumerate}
We use functions from {\it Statistics and Machine Learning Toolbox} in Matlab to implement the first three methods. Spectral clustering is implemented using the normalized random-walk Laplacian matrix (see \cite{shi2000normalized} for details). Under the same settings as in  Figure \ref{error-rate}, we generate samples from LCM and apply the above algorithms and tensor-EM to cluster the data. 
The average error rates of all the algorithms over 100 replications are computed for each $J$. The trend of error rate versus the number of items $J$ is shown in Figure \ref{compare-cluster}. We can see all the algorithms have small error rates as $J$ becomes large and the tensor-EM method has the best performance when $J$ is moderately large. One possible explanation is that the tensor-EM method is more tailored for LCMs  and may have advantage over other methods in the LCM setting.

\begin{figure}[H]
	\centering
	\subfigure[Item parameters $\in \{0.1, 0.2, 0.8, 0.9\}$]{
		\begin{minipage}[t]{0.45\linewidth}
			\centering
			\includegraphics[width=\textwidth]{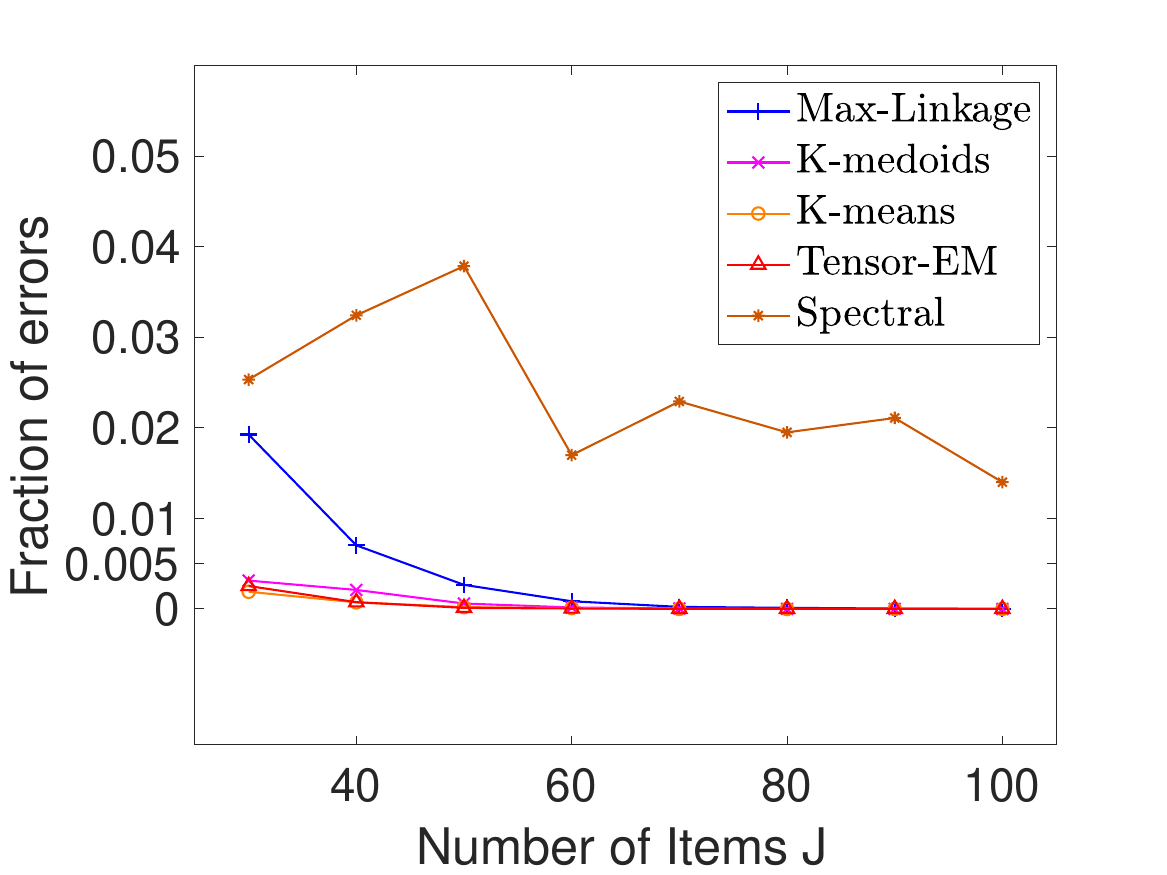}
		\end{minipage}%
	}%
	\subfigure[Item parameters $\in \{0.2, 0.4, 0.6, 0.8\}$]{
		\begin{minipage}[t]{0.45\linewidth}
			\centering
			\includegraphics[width=\textwidth]{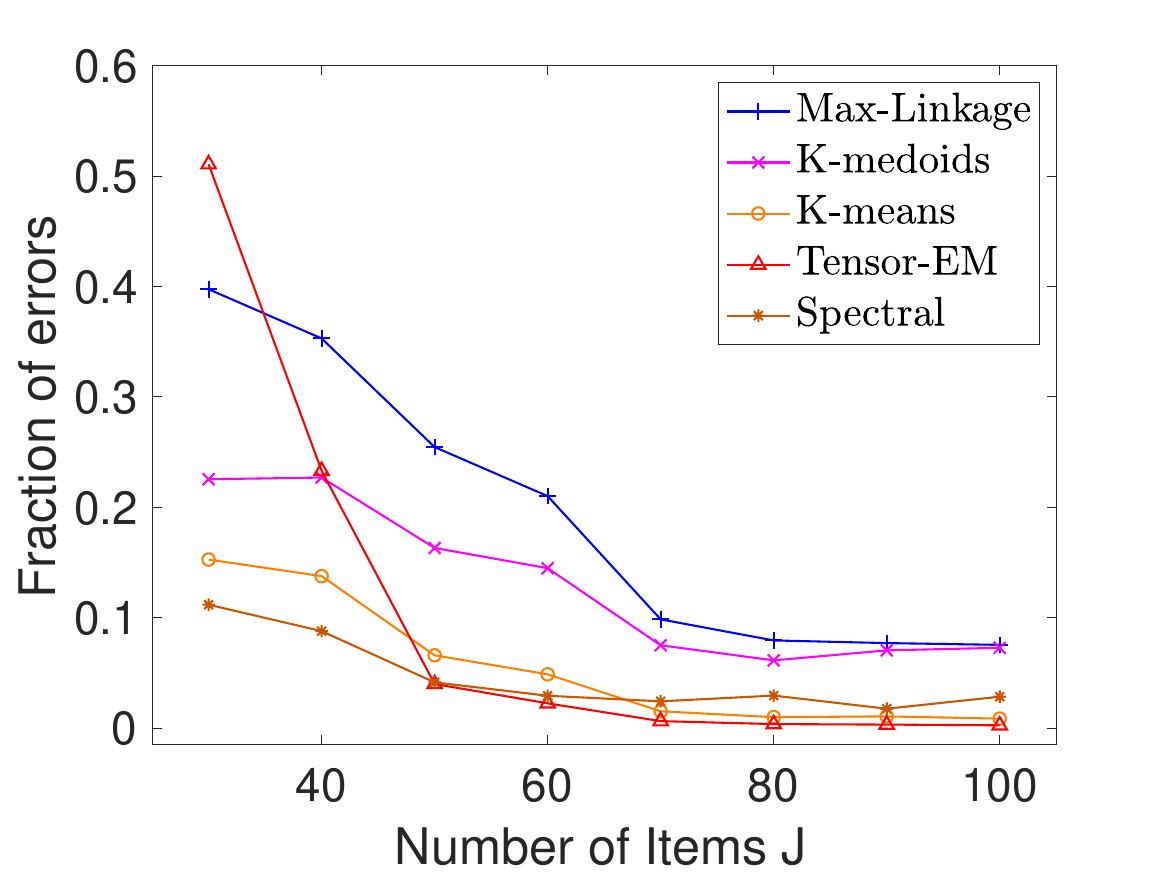}
		\end{minipage}%
	}%
	\centering
	\caption{Error rate of different clustering algorithms}
	\label{compare-cluster}
\end{figure}

\subsection{Performance of GIC in Selecting Number of Classes}\label{simu-GIC}

We consider the accuracy of GIC in selecting $L$, which needs to be estimated in practice.
The settings we consider are the same as in Section \ref{simu-EM-tensor}. The performance of GIC to select the number of classes in the random-effect LCM and the fixed-effect LCM are reported in Table \ref{tab:GIC}. For settings with the true number of classes $L$ being five, we let the candidate set of $L$ be $\{2,3,4,5,6,7\}$; while for settings with the true $L$ being ten, we let the candidate set of $L$ be $\{7,8,9,10,11,12\}$. We see that for all these settings $\mathrm{GIC}_1$ can always select the correct number of classes. And for most of the settings $\mathrm{GIC}_2$ can choose the right model. The only setting where $\mathrm{GIC}_2$ performs not so well is random-effect LCM with $N = 1000, J=100, L=10,\theta_{j,a} \in \{ 0.2,0.4,0.6,0.8\}$. In general, both $\mathrm{GIC}_1$ and $\mathrm{GIC}_2$ enjoy desirable performance in selecting the correct number of classes.

\renewcommand\arraystretch{0.5}
\begin{table}[h!]
    \centering
    \linespread{3}
    \begin{tabular}{@{\extracolsep{4pt}}cccccccc}
    \toprule
    \multirow{2}{*}{Signal strength} & \multirow{2}{*}{$N$} & \multirow{2}{*}{$J$} & \multirow{2}{*}{$L$} & \multicolumn{2}{c}{Random-effect} & \multicolumn{2}{c}{Fixed-effect} \\ 
    \cmidrule{5-6} 
    \cmidrule{7-8} 
    & & & & $\mathrm{GIC}_1$ & $\mathrm{GIC}_2$ & $\mathrm{GIC}_1$ & $\mathrm{GIC}_1$ \\
    \midrule
    \multirow{12}{*}[-20pt]{$\theta_{j,a} \in \{ 0.1,0.2,0.8,0.9\} $} & \multirow{4}{*}{1000} &  \multirow{2}{*}{100} & 5 & 1.00 & 1.00 & 1.00 & 1.00\\
 	\cmidrule{4-8}
 	&  & & 10 & 1.00 & 1.00 & 1.00 & 1.00\\
 	\cmidrule{3-8}
 	 &  & \multirow{2}{*}{200} & 5 & 1.00 & 1.00 & 1.00 & 1.00\\
 	\cmidrule{4-8}
	&  &  & 10 & 1.00 & 1.00 & 1.00 & 1.00\\
    \cmidrule{2-8}
 	& \multirow{4}{*}{10000} & \multirow{2}{*}{100} & 5 & 1.00 & 1.00 & 1.00 & 1.00\\
   \cmidrule{4-8}
 	&  &  & 10 & 1.00 & 1.00 & 1.00 & 1.00\\
    \cmidrule{3-8}
 	&  & \multirow{2}{*}{200} & 5 & 1.00 & 1.00 & 1.00 & 1.00\\
 	\cmidrule{4-8}
 	&  & & 10 & 1.00 & 1.00 & 1.00 & 1.00\\
 	\cmidrule{2-8}
 	& \multirow{4}{*}{20000} & \multirow{2}{*}{100} & 5 & 1.00 & 1.00 & 1.00 & 1.00\\
	 \cmidrule{4-8}
 	&  &  & 10 & 1.00 & 1.00 & 1.00 & 1.00\\
 	\cmidrule{3-8}
 	& & \multirow{2}{*}{200} & 5 & 1.00 & 1.00 & 1.00 & 1.00\\
 	\cmidrule{4-8}
 	&  & & 10 & 1.00 & 1.00 & 1.00 & 1.00\\
 	\midrule
 	\multirow{12}{*}[-20pt]{$\theta_{j,a} \in \{ 0.2,0.4,0.6,0.8\} $} &  \multirow{4}{*}{1000} & \multirow{2}{*}{100} & 5 & 1.00 & 1.00 & 1.00 & 1.00\\
 	\cmidrule{4-8}
  	&  &  & 10 & 1.00 & 0.77 & 1.00 & 1.00\\
 	\cmidrule{3-8}
 	&  & \multirow{2}{*}{200} & 5 & 1.00 & 1.00 & 1.00 & 1.00\\
 	\cmidrule{4-8}
 	&  &  & 10 & 1.00 & 0.99 & 1.00 & 1.00\\
 	\cmidrule{2-8}
 	& \multirow{4}{*}{10000} & \multirow{2}{*}{100} & 5 & 1.00 & 1.00 & 1.00 & 1.00\\
		 \cmidrule{4-8}
		&  &  & 10 & 1.00 & 1.00 & 1.00 & 1.00\\
		\cmidrule{3-8}
		 &  & \multirow{2}{*}{200} & 5 & 1.00 & 1.00 & 1.00 & 1.00\\
		 \cmidrule{4-8}
		 & & & 10 & 1.00 & 1.00 & 1.00 & 1.00\\
		 \cmidrule{2-8}
		 & \multirow{4}{*}{20000} & \multirow{2}{*}{100} & 5 & 1.00 & 1.00 & 1.00 & 1.00\\
		 \cmidrule{4-8}
		 &  & & 10 & 1.00 & 1.00 & 1.00 & 1.00\\
		 \cmidrule{3-8}
		 &  & \multirow{2}{*}{200} & 5 & 1.00 & 1.00 & 1.00 & 1.00\\
		 \cmidrule{4-8}
		 &  &  & 10 & 1.00 & 1.00 & 1.00 & 1.00\\
    \bottomrule
    
    \end{tabular}
    \caption{The fraction of correctly selecting the number of classes}
    \label{tab:GIC}
\end{table}

\section{Real Data Analysis}\label{real-data}

In this section we apply the proposed method to real data from Trends in International Mathematics and Science Study (TIMSS). We use a subset of TIMSS 2011 Austrian data \citep{george2015cognitive,sedat2015diagnostic} in R package \textbf{CDM}. 47 items are available to measure students’ abilities in 9 mathematical sub-competences, including  (DA) Data and Applying, (DK) Data and Knowing, (DR) Data and Reasoning, (GA) Geometry and Applying, (GK) Geometry and Knowing, (GR) Geometry and Reasoning, (NA) Numbers and Applying, (NK) Numbers and Knowing, and (NR) Numbers and Reasoning. The first Q-matrix in the R package \textbf{CDM} indicates the relations on the items and sub-competences measured, which are summarized in Table \ref{tab:comptence-item}. 

\renewcommand\arraystretch{0.6}
\begin{table}[h!]
    \centering
    \linespread{3}
    \begin{tabular}{@{\extracolsep{4pt}}c|c}
     Sub-competences    & Item index that measures the sub-competences \\
     \hline
     DA    &  46,47\\
     \hline
     DK   & 20,34 \\
     \hline
     DR & 21,35 \\
     \hline
     GA & 17,18,30,31,32,42,44 \\
     \hline
     GK & 7,8,16,19,28,29,43 \\
     \hline
     GR & 33,45\\
     \hline
     NA & 1,6,10,15,23,24,37,38,40\\
     \hline
     NK & 11,14,22,25,26,27,36 \\
     \hline
     NR & 2,3,4,5,9,12,13,39,41 
    \end{tabular}
    \caption{Relations between sub-competences and items}
    \label{tab:comptence-item}
\end{table}

One feature of  large scale education assessment data is that we only have response on a subset of items for each student. Here $48 \%$ of the components in the response matrix $\RR$ are missing. Under the missing at random (MAR) assumption, we use the multiple imputation (MI) to obtain five complete datasets. Same analysis is performed on each of the dataset and the final results (GIC and item parameters) presented are the average of results obtained from the five datasets. This average pooling strategy is often used in analyzing missing data with MI to account for the randomness of MI. We refer to \cite{gelman2006data} and \cite{van2011mice} for more details on missing data and MI.

After completing the data with MI, we fit a random-effect LCM on each of the dataset. We apply the GIC method to selecting the number of latent classes. According to Figure \ref{fig:GIC_timss}, $L=3$ and $L=6$ are plausible options and we fit two LCMs with $L=3$ and $L=6$. In the case $L=6$, one component in the estimated proportion vector $\hat{\pp}$ is fairly small (smaller than 0.005). The corresponding latent class can be neglected for better interpretability and parsimony. So we instead fit two models with $L=3$ and $L=5$. The estimated proportion vectors are $\hat{\pp}_3 = (0.25, 0.46, 0.29)$ and $\hat{\pp}_5 = (0.29, 0.37, 0.2, 0.02, 0.12)$ and the item parameters $\hat{\TT}$ are visualized in Figure \ref{fig:timss_item}.

\begin{figure}
    \centering
    \includegraphics[width=2.7in]{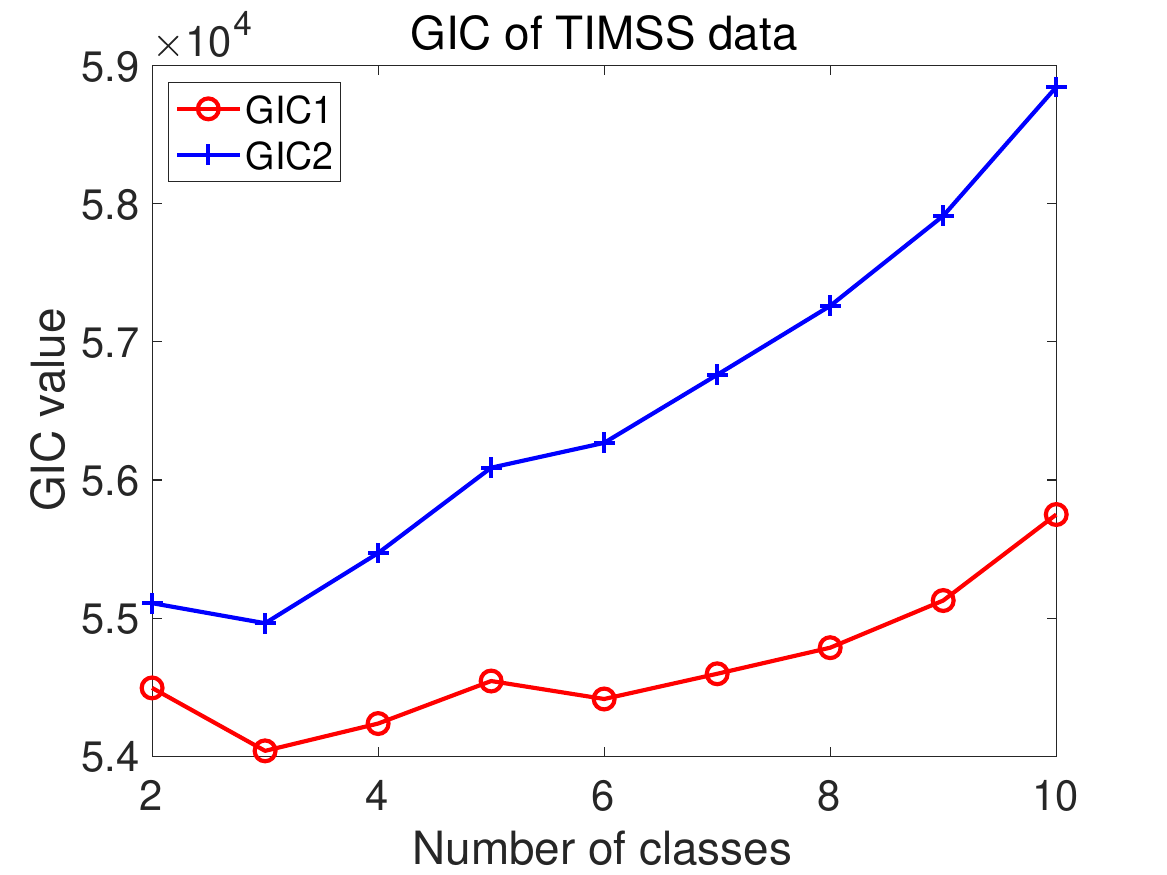}
    \caption{GIC of models with different number of classes}
    \label{fig:GIC_timss}
\end{figure}

\begin{figure}[H]
	\centering
	\subfigure[$\hat{\TT}_3$]{
		\begin{minipage}[t]{0.45\linewidth}
			\centering
			\includegraphics[width=\textwidth]{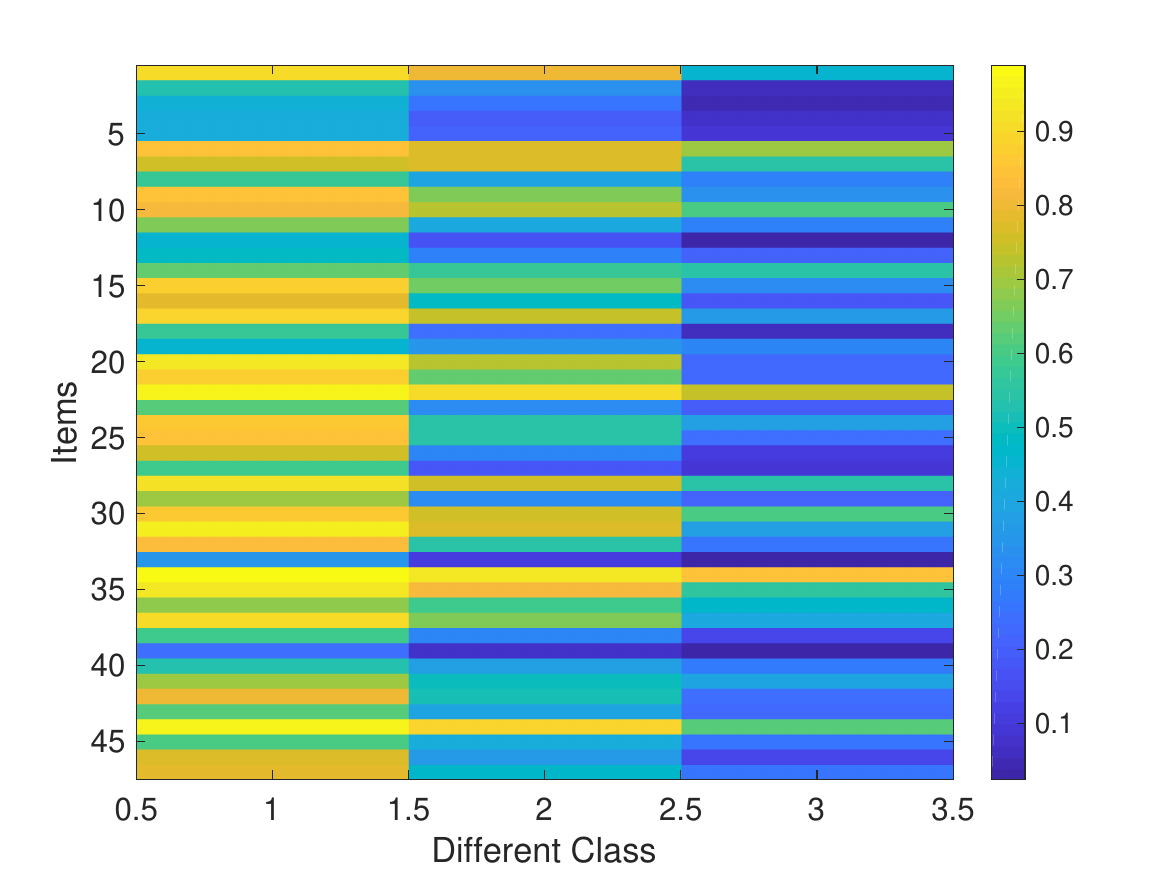}
		\end{minipage}%
	}%
	\subfigure[$\hat{\TT}_5$]{
		\begin{minipage}[t]{0.45\linewidth}
			\centering
			\includegraphics[width=\textwidth]{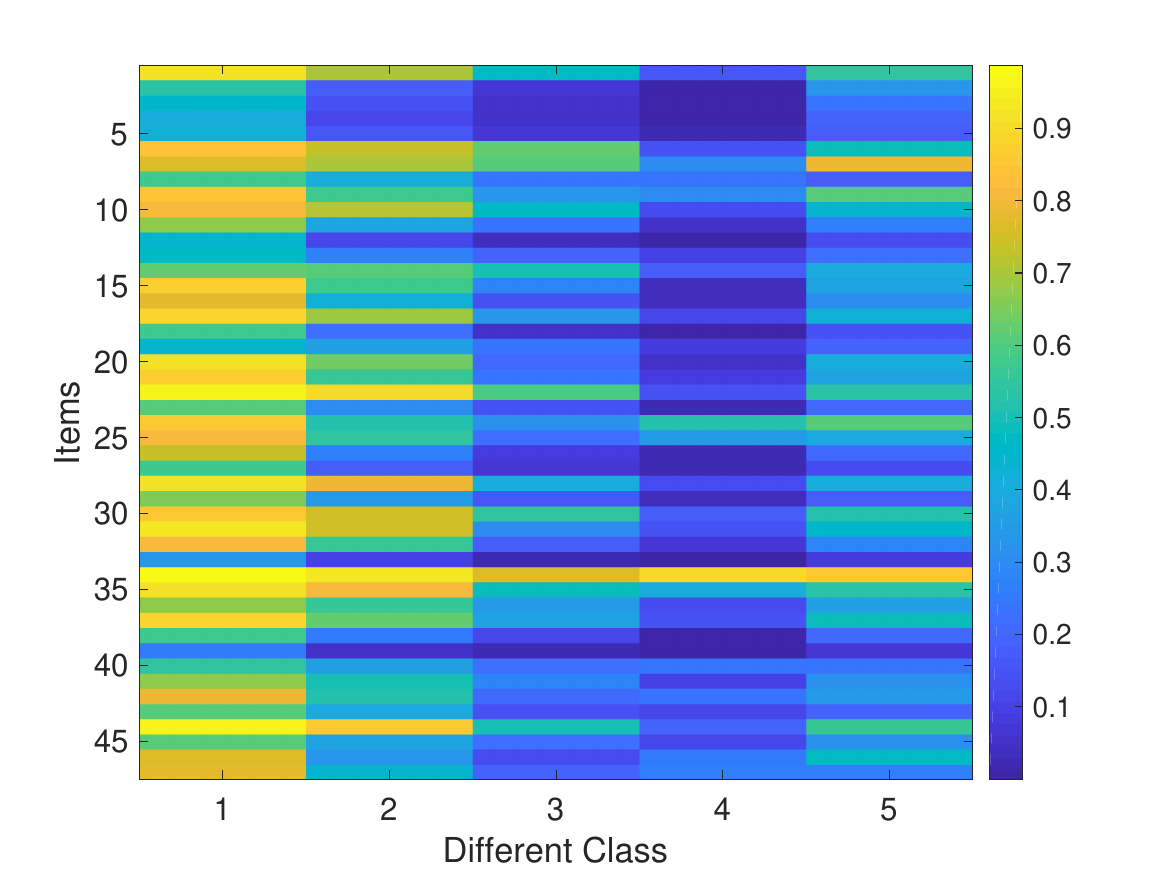}
		\end{minipage}%
	}%
	\centering
	\caption{Estimated item parameters for $L=3$ and $L=5$. A brighter color indicates a larger item parameter and hence larger probability to answer the corresponding item correctly.}
	\label{fig:timss_item}
\end{figure}

In the $L=3$ case, we see students in the three classes behave differently on the 47 items. Students in the first class top three classes and display best competence in all the sub-competences while students in the third class do not perform well on the test and need to improve their sub-competences. The competence of students in the second class is in the middle and may indicate students' average competence. Comparing two figures for $L=3$ and $L=5$, we see the first two classes roughly have the same performance on the items. Thus the $L=5$ case can be viewed as a refined analysis on $L=3$, where the third class is further divided into three parts. 
For $L=5$, according to the relations between items and sub-competences summarized in Table \ref{tab:comptence-item}, we can compare the sub-competences of different classes. For most of the sub-competences (DK, DR, GA, GK, GR, NA, NK), the first class performs best, followed by the second class, which further outperforms third and fifth classes. The fourth class has the worst performance. Note that the estimated proportion for the fourth class is 0.02 and hence only a few students have such unsatisfactory performance. For sub-competences DA and NR, the first class remains the top. The second and fifth classes follow the first class and have similar performance. The third and fourth classes have an unsatisfactory performance compared with the others.

We can also assess the difficulty of items. One goal of latent class analysis is to find items that best distinguish different classes. From the plots we see different classes have different item parameters on most of the items, indicating these items can distinguish between classes. However, there exist some items that are not ideal in this respect. For instance, all classes have relatively good performance on item 33 (M041335 in original index of the tests). This question presents four bar plots of three colors (red, green and blue) and asks students which bar plot has the smallest value for blue. Since the frequency is clearly shown in the plots, this question may be easy for fourth graders and may not be an ideal item to distinguish between different classes. Furthermore, all classes have relatively poor performance on item 39 (M051006 Cost of ice cream), indicating item 39 is of high difficulty level. 

We further explore the latent hierarchical structures of the latent classes and their interpretations. Following the idea in Section 4.2 of \cite{ma2022learning}, we first define
\[
\boldsymbol{\Gamma}=\left(\mathbb{I}\left\{ \left|\hat{\theta}_{j, l}-\max _{m \in[L]} \hat{\theta}_{j, m}\right| \leq \tau \right\}: j \in[J], l \in[L]\right) \in\{0,1\}^{J \times L}
\]
for a small $\tau>0$. Note that $\Gamma_{j,l} = 1$ indicates that $\hat{\theta}_{j, l}$ is close to $\max _{m \in[L]} \hat{\theta}_{j, m}$ and hence latent class $l$ possesses relatively high level of item parameter for item $j$. We say a latent class $l_1$ is more capable than a latent class $l_2$ and represent it as $\boldsymbol{\Gamma}_{\cdot,l_2} \rightarrow \boldsymbol{\Gamma}_{\cdot,l_1}$ if $\boldsymbol{\Gamma}_{\cdot,l_1} \succeq \boldsymbol{\Gamma}_{\cdot,l_2}$ , where for two vectors $\vv_1$ and $\vv_2$, we write $\vv_1 \succeq \vv_2$ if $v_{1,i} \geq v_{2,i}$ for each $i$. With this definition we can get partial orders among latent classes. In our real data analysis, we let $\tau = 0.25$ (roughly the standard error of all item parameters) and relax the definition of $\boldsymbol{\Gamma}_{\cdot,l_1} \succeq \boldsymbol{\Gamma}_{\cdot,l_2}$ to $\Gamma_{j, l_1} \geq \Gamma_{j, l_2}$ for $90\%$ of items when $j$ varies in $\{1,\ldots,J\}$. The resulting partial orders based on item parameters can be represented as directed acyclic graphs (DAGs) in Figure \ref{fig:partial_orders}. 

\begin{figure}[h!]
	\centering
	\subfigure[$L=3$]{
		\begin{minipage}[t]{0.12\linewidth}
			\centering
			\includegraphics[width=\textwidth]{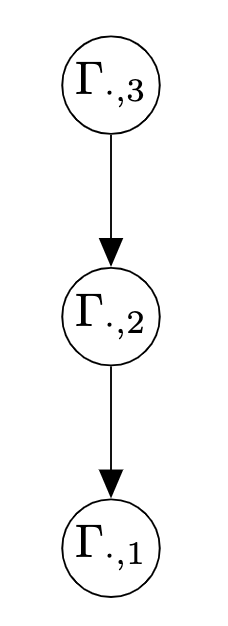}
		\end{minipage}%
	}%
	\subfigure[$L=5$]{
		\begin{minipage}[t]{0.3\linewidth}
			\centering
			\includegraphics[width=\textwidth]{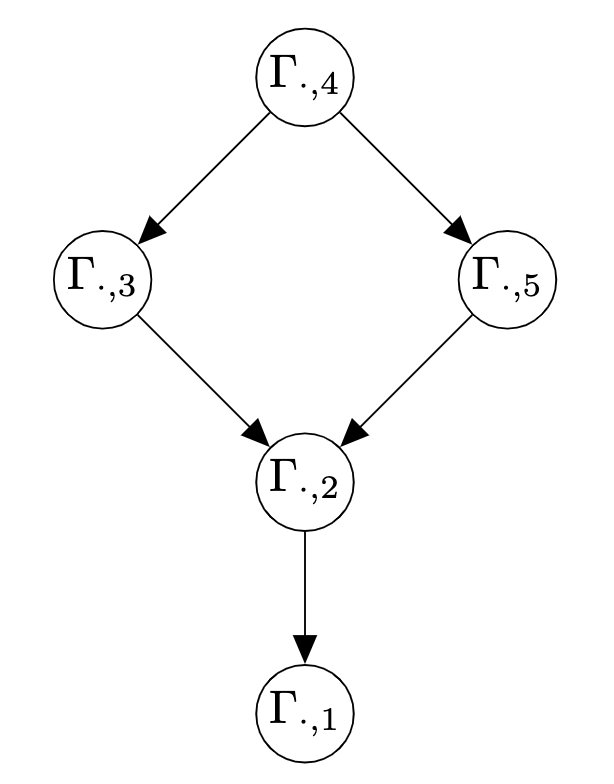}
		\end{minipage}%
	}%
	\centering
	\caption{Partial orders among latent classes}
	\label{fig:partial_orders}
\end{figure}

With the obtained partial orders, we can apply the latent hierarchy recovering algorithm in \cite{ma2022learning} to obtain the latent attribute representations of the latent classes under the cognitive diagnostic modeling framework. In particular, for $L=3$, the three latent classes can be represented as $\{ (1,1), (1,0), (0,0)\}$, and for $L=5$, the five latent classes can be represented as $\{ (1,1,1), (1,1,0), (1,0,0), (0,0,0), (0,1,0)\}$. The hierarchies among the learned latent attributes can be represented as in Figure \ref{fig:hierarchies}. Specifically,  in the case $L=3$, $\alpha_1$ is a more basic prerequisite for $\alpha_2$;
similarly in the case $L=5$, $\alpha_1$ and $\alpha_2$ may be more basic prerequisites for attribute $\alpha_3$. 

\begin{figure}[H]
	\centering
	\subfigure[$L=3$]{
		\begin{minipage}[t]{0.1\linewidth}
			\centering
			\includegraphics[width=\textwidth]{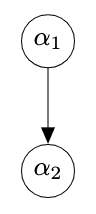}
		\end{minipage}%
	}%
	\subfigure[$L=5$]{
		\begin{minipage}[t]{0.32\linewidth}
			\centering
			\includegraphics[width=\textwidth]{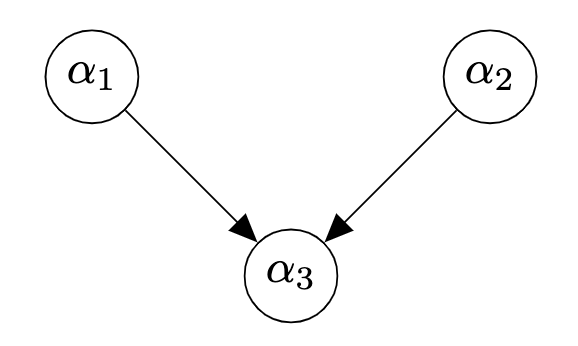}
		\end{minipage}%
	}%
	\centering
	\caption{Latent hierarchical structures of attributes}
	\label{fig:hierarchies}
\end{figure}

The estimated item parameters may help us identify which sub-competences each learned attribute $\alpha_i$ corresponds to. In the case $L=3$, we find $\hat{\theta}_{j, 1} - \hat{\theta}_{j, 2}$ is fairly large (greater than 0.2) for most items $j$ corresponding to DA, GR, NA, NK, which may be the sub-competences represented by attribute $\alpha_2$. Similarly students in class 2 greatly improve their competences in DK, DR and GA compared with students in class 3. Hence the more basic attribute $\alpha_1$ may represent DK, DR and GA. Similarly in the case $L=5$, we find $\alpha_1$ corresponds to sub-competences NA, GK, NK and GA, $\alpha_2$ corresponds to NA, GA, DK and DR, $\alpha_3$ corresponds to DA, NA, NK and NR. According to the hierarchical structure in Figure \ref{fig:hierarchies}, DA and number-related skills NA, NK and NR may be more advanced sub-competences than others.


We conclude this section by discussing the connection and difference between latent class analysis and more delicate cognitive diagnostic models in real data analysis. CDMs can be viewed as special cases of LCMs with more delicate structures on the item parameters; indeed, CDMs belong to a family of \emph{restricted} latent class models \citep{xu2017identifiability,xu2018identifying}. With the structure specified by Q-matrix, CDMs can provide more fine-grained analysis than LCMs (e.g. we can estimate the latent profile of each individual). On the other hand, LCMs do not require the prior knowledge on the Q-matrix and may serve the purpose of exploratory analysis before modeling the data with more delicate CDMs, as shown in \cite{ma2022learning} and explored in our data analysis. Hence LCMs and CDMs are closely related and both can be useful in various applications.

\section{Discussion}\label{discussion}

This paper investigates the computation and theory for large-scale latent class models.  In terms of computation, commonly used likelihood-based methods (e.g. EM algorithm)  suffer from slow convergence rate and potential convergence to local optima under poor initializations. Recent developments in tensor decomposition and its applications provide a computationally efficient moment-based method to estimate the parameters. However, such tensor method is based on low-order moments of the observed variables rather than the entire likelihood function. Hence it is not statistically efficient and generally requires a large number of samples to ensure the accuracy of estimates. 
In this work, we propose a two-stage tensor-EM estimation pipeline which combines these two methods. 
Simulation studies empirically show the proposed procedure is both computationally efficient and statistically accurate. {Moreover, based on our simulations in tensor methods the moments in \eqref{moments-est} need to be estimated accurately. When the sample size is not very large (e.g. $N=100$), the tensor method may not yield good estimates since the estimation of moments is not accurate. In applications ideally we need at least $N=500$ samples to ensure the good performance. Note that Condition \ref{cond1} implicitly constrains the number of classes $L$. If we want to fit a LCM with $L$ classes, then condition \ref{cond1} requires $L \leq \min \{J_1,J_2,J_3\} \leq J/3$. In applications we should always fit LCM with $L \leq J/3$ classes so that the Condition \ref{cond1} is satisfied and tensor method can be applied.}
Furthermore, we theoretically establish the clustering consistency (consistency of latent class membership) of large-scale fixed-effect latent class analysis where sample size $N$ and number of items $J$ both go to infinity. {Consistency of item parameters is proved as a corollary of clustering consistency. In terms of consistency under random-effect LCM, we empirically verify it in simulations. However, in the likelihood function of random-effect LCM, the latent variable is marginalized out, making the log-likelihood more challenging to analyze the consistency of parameters.  We will leave it as future work.}
 

{We also note that the response variables do not have to be binary for applying the tensor method; the tensor power method can also apply to models with polytomous or continuous responses. This is because the low-order moments constructed in and exploited by the tensor power method here \citep[also see][]{anandkumar2014tensor} essentially result from the local independence in LCMs; that is, the observed variables are conditionally independent given the latent variable.} Regarding the theoretical results for polytomous response, we believe similar proof techniques can be applied to establish the clustering consistency. The details are left as future work. 
It is also possible to find similar tensor structures in more complicated latent structure models, for example, the diagnostic classification models or cognitive diagnostic models (CDMs) \citep{rupp2008dcm, von2019handbook}. Since CDMs can be viewed as extensions of LCMs where the local independence assumption also holds, the same tensor power method used here may also be used to find rough estimates of parameters in a CDM. However, CDMs also have a unique feature, the Q-matrix \citep{tatsuoka1983rule}, which induces complicated parameter constraints on top of a latent class model.
Thus a plain tensor power method for unrestricted LCMs would not give estimates that satisfy the equality constraints under a Q-matrix.
How to develop an efficient estimation procedure for CDMs by integrating the tensor method is an interesting question left for future investigation. 
 




\bibliographystyle{apalike}
\bibliography{LCMrefer}

\newpage
\appendixpage

In this supplementary material, the proof of Theorem 3 (clustering consistency) is presented in Appendix 1. Then Corollary 2 (consistency of item parameters) is proved in Appendix 2. In Appendix 3, more simulation results are provided to show the good performance of the proposed EM-tensor method.

\section*{Appendix 1: Proof of Theorem 3}

\textbf{Outline of proof idea.}
The proof follows the following 8 steps. 

\textbf{Step 1}: Express $\ell(\RR;\, \ZZ) - \bar \ell(\ZZ)$ in terms of $\sum_{a}n_{a}\sum_{j}D(\hat \theta_{j,a}\| \bar \theta_{j,a}) + X - \mathbb E(X)$, where $X$ is a random variable depending on $\RR$ and $\bar\TT^{(\ZZ)}$ under $\ZZ$, and $n_a = \sum_{i=1}^N Z_{i,a}$.

\textbf{Step 2}: Bound the first term  $\sum_{j}\sum_{a}n_{j,a}D(\hat \theta_{j,a}\| \bar \theta_{j,a})$ in the above display uniformly over all possible $\ZZ$.

\textbf{Step 3}: Bound the second term $X-\ME(X)$ using Bernstein type inequality. Combine this and Step 2  to obtain a bound for $\sup_{\ZZ} |\ell(\RR;\, \ZZ) - \bar \ell(\ZZ)|$.

\textbf{Step 4}: (Denote the true latent class memberships by $\ZZ^0$ and joint MLE by $\hat \ZZ$.) Establish $\bar \ell(\ZZ^0) \geq \bar\ell(\ZZ)$ for all $\ZZ$. Use triangle inequality to upper-bound the non-negative quantity $\bar \ell(\ZZ^0) - \bar \ell(\hat \ZZ)$.
\begin{align*}
	0\leq \bar \ell(\ZZ^0) - \bar \ell(\hat \ZZ)
     \leq [\bar \ell(\ZZ^0) - \ell(\RR;\,\ZZ^0)] + 
     [\ell(\RR;\,\ZZ^0) - \ell(\RR;\,\hat\ZZ)]
     + [\ell(\RR;\,\hat\ZZ) - \bar \ell(\hat \ZZ)]
\end{align*}
Since in the above display the middle group of terms $[\ell(R;\,\ZZ^0) - \ell(R;\,\hat\ZZ)]\leq 0$, we have $0\leq \bar \ell(\ZZ^0) - \bar \ell(\hat \ZZ) \leq 2 \sup_{\ZZ} |\ell(\RR;\, \ZZ) - \bar \ell(\ZZ)|$.

\textbf{Step 5}: Introduce the notion of partitions and generalize $\bar\ell(\ZZ)$ to $\bar\ell(\Pi)$.

\textbf{Step 6}: Show that a refined partition increases $\bar\ell(\cdot)$. To be concrete, let $\Pi^*$ be a refined partition of $\Pi$, then we have $\bar\ell(\Pi^*) \geq \bar\ell(\Pi)$.

\textbf{Step 7}: Show that for any latent class assignment $\ZZ$, we can find a partition $\Pi^*$ that refines $\Pi^{\ZZ}$ and $\bar\ell(\ZZ^0) - \bar\ell(\Pi^*) \geq \frac{1}{2} J \beta_J N_e(\zz)$.

\textbf{Step 8}: Apply results in step 6 and step 7 to MLE $\hat{\zz}$, we have
\[
\text{bound in step 4} \geq \bar\ell(\ZZ^0) - \bar\ell(\Pi^{\hat{\ZZ}}) \geq \bar\ell(\ZZ^0) - \bar\ell(\Pi^*) \geq \frac{1}{2} J \beta_J N_e(\hat{\zz}).
\]


\bigskip
Now we formally begin the proof of Theorem 3 in the above several steps. In derivations below we abbreviate $\bar\TT^{(\ZZ)}$ as $\bar\TT$ and $\hat\TT^{(\ZZ)}$ as $\hat\TT$ to simplify notations.

\textbf{Step 1.} Define $D(p \| q) = p\log(p/q) + (1-p)\log((1-p)/(1-q))$, the Kullback-Leibler divergence of a Bernoulli distribution with parameter $p$ from that with parameter $q$. In this step we prove a lemma as follows.
\begin{lemma}
	\label{lem-express}
	Let $(R_{i,j}; \, 1\leq N, 1\leq J)$ denote independent Bernoulli trials with parameters $(P_{i,j}; \, 1\leq N, 1\leq J)$. Under a general latent class model, given an arbitrary $\ZZ$, there is 
	\begin{align}\label{eq-lemma}
		&~\sup_{\TT} \ell(\RR;\,\mathbf Z, \, \TT) - \sup_{\TT} \mathbb E[\ell(\RR;\,\mathbf Z, \, \TT)] \\ \notag
	=&~ \sum_{a=1}^L n_a \sum_j D(\hat\theta_{j,a} \| \bar \theta_{j,a})  
    + \sum_{i}\sum_j (R_{i,j} - P_{i,j})\log\Big( \frac{\bar\theta_{j,z_i}}{1-\bar\theta_{j,z_i}} \Big) \\ \notag
    = &~ \sum_{a=1}^L n_a \sum_j D(\hat\theta_{j,a} \| \bar \theta_{j,a})  
    + X-\ME X,
	\end{align}
where $X= \sum_{i}\sum_j R_{i,j}\log\Big( \frac{\bar\theta_{j,z_i}}{1-\bar\theta_{j,z_i}} \Big)$ is random variable depending on $\ZZ$ and 
\begin{equation}
\label{mle-theta}
\widehat{\theta}_{j, a}=\frac{\sum_{i} Z_{i, a} R_{i, j}}{\sum_{i} Z_{i, a}}, \quad \bar{\theta}_{j, a}=\frac{\sum_{i} Z_{i, a} P_{i, j}}{\sum_{i} Z_{i, a}}
\end{equation}
\end{lemma}

 Given a fixed $\ZZ$, denote $n^{(\ZZ)}_a = \sum_{i=1}^N Z_{i,a}$.The maximizing properties of $\hat \theta_{j,a}$ and $\bar \theta_{j,a}$ in \ref{mle-theta} imply that
\begin{equation}\label{eq-mprop}
	n_a\hat\theta_{j,a} = \sum_i Z_{i,a}R_{i,j},\quad
    n_a\bar\theta_{j,a} = \sum_i Z_{i,a}P_{i,j}.
\end{equation}
Using \eqref{eq-mprop}, we have the following,
\begin{align*}
 &~\ell(\RR;\, \ZZ) - \bar\ell(\ZZ)\\
=&~ \sum_{i}\sum_{j} \sum_{a=1}^L Z_{i,a} [R_{i,j} \log\hat\theta_{j,a} + (1-R_{i,j})\log(1-\hat\theta_{j,a})]
 \\ 
&~ \qquad\qquad 
- \sum_{i}\sum_{j} \sum_{a=1}^L Z_{i,a} [P_{i,j} \log\bar\theta_{j,a} + (1-P_{i,j})\log(1-\bar\theta_{j,a})]\\
=&~ \sum_j \sum_{a=1}^L n_a[\hat\theta_{j,a}\log \hat\theta_{j,a} + (1-\hat\theta_{j,a})\log(1-\hat\theta_{j,a})] - 
   \sum_j \sum_{a=1}^L n_a[\bar\theta_{j,a}\log \bar\theta_{j,a} + (1-\bar\theta_{j,a})\log(1-\bar\theta_{j,a})] \\
=&~ \sum_j \sum_{a=1}^L \Big\{n_a[\hat\theta_{j,a}\log \hat\theta_{j,a} + (1-\hat\theta_{j,a})\log(1-\hat\theta_{j,a})] 
 - n_a[\hat\theta_{j,a}\log \bar\theta_{j,a} + (1-\hat\theta_{j,a})\log(1-\bar\theta_{j,a})]\Big\} \\
 &~ +  \sum_j \sum_{a=1}^L \Big\{n_a[\hat\theta_{j,a}\log \bar\theta_{j,a} + (1-\hat\theta_{j,a})\log(1-\bar\theta_{j,a})] 
 -  n_a[\bar\theta_{j,a}\log \bar\theta_{j,a} + (1-\bar\theta_{j,a})\log(1-\bar\theta_{j,a})] \Big\}\\
  =&~ \sum_{a=1}^L n_a \sum_j D(\hat\theta_{j,a} \| \bar \theta_{j,a}) 
    + \sum_i \sum_j \Big\{[R_{i,j}\log \bar\theta_{j,z_i} + (1-R_{i,j})\log(1-\bar\theta_{j,z_i})] \\
    &\qquad\qquad\qquad\qquad\qquad\qquad\qquad  - [P_{i,j}\log \bar\theta_{j,z_i} + (1-P_{i,j})\log(1-\bar\theta_{j,z_i})] \Big\} \\
 =&~ \sum_{a=1}^L n_a \sum_j D(\hat\theta_{j,a} \| \bar \theta_{j,a}) 
    + \sum_{i}\sum_j R_{i,j} \log\Big( \frac{\bar\theta_{j,z_i}}{1-\bar\theta_{j,z_i}} \Big)
    - \sum_{i}\sum_j P_{i,j} \log\Big( \frac{\bar\theta_{j,z_i}}{1-\bar\theta_{j,z_i}} \Big).
\end{align*}
Define the random variable 
\begin{equation}
\label{eq-defx}	
X=\sum_{i}\sum_j R_{i,j} \log( \bar\theta_{j,z_i}/(1-\bar\theta_{j,z_i} )),
\end{equation}
then $X$  depends on $\ZZ$ and the above display becomes the summation of $\sum_{a=1}^L n_a \sum_j D(\hat\theta_{j,a} \| \bar \theta_{j,a})$ and $X-\mathbb E[X]$.
This establishes \eqref{eq-lemma} in Lemma \ref{lem-express}.
In the following, we  bound the first term $\sum_{a=1}^L n_a \sum_j D(\hat\theta_{j,a} \| \bar \theta_{j,a})$ and the second term $X-\mathbb E[X]$ in the above display uniformly over all possible $\ZZ$, respectively in Step 2 and Step 3.

\bigskip
\textbf{Step 2.} In this step we prove the following theorem.
\begin{theorem}
	\label{thm-kl}
	The following event happens with probability at least $1-\delta$,
\begin{equation*}
\max_{\ZZ}\left\{\sum_j\sum_a n_a D(\hat\theta^{\,\ZZ}_{j,a}  \| \bar \theta^{\,\ZZ}_{j,a}) \right\}
<
 N\log L + JL\log\Big(\frac{N}{L} + 1\Big) - \log\delta.
\end{equation*}

\end{theorem}
Given any fixed latent class memberships $\ZZ$, every $\hat\theta_{j,a}$ is an average of $n_a$ independent Bernoulli random variables $R_{1,j},\ldots, R_{N,j}$ with mean $\bar\theta_{j,a}$. We apply the Chernoff-Hoeffding theorem to obtain
\begin{equation}
	\mathbb P(\hat\theta_{j,a} \geq \bar \theta_{j,a} + t) 
	\leq e^{ -n_a D(\bar\theta_{j,a} + t \| \bar \theta_{j,a}) },\quad
	\mathbb P(\hat\theta_{j,a} \leq \bar \theta_{j,a} - t) \leq e^{ -n_a D(\bar\theta_{j,a} - t \| \bar \theta_{j,a}) }.
\end{equation}
Note that given a fixed $\ZZ$, each $\hat\theta_{j,a}$ can take values only in the finite set $\{0, 1/n_a, 2/n_a,\ldots,1\}$ of cardinality $n_a+1$. 
We denote this range of $\hat\theta_{j,a}$ by $\hat\Theta^{j,a}$. Then $\mathbb P(\hat\theta_{j,a} = \vartheta)\leq \exp\{ -n_a D(\vartheta \| \bar \theta_{j,a})\}$ for any $\vartheta \in \hat\Theta^{j,a}$. 
Now consider the cardinality of the set $\hat\Theta$ given $\ZZ$. Since for each of the $J\times L$ entries in $\hat\TT$, $\hat\theta_{j,a}$ can independently take on $n_a+1$ different values, there is $|\hat\Theta| = [\prod_a(n_{a}+1)]^J$. Considering the natural constraint $\sum_{a=1}^L n_a = N$, we have 
\begin{equation}\label{eq-That}
	|\hat\Theta| = \Big[\prod_{a=1}^L(n_{a}+1)\Big]^J
	\leq \Big[\Big(\frac{N}{L}+1\Big)^L\Big]^J.
\end{equation}
Define
	$\hat\Theta_\epsilon = \{\tilde\TT\in\hat\Theta:\, \sum_j\sum_a n_a D(\tilde\theta_{j,a}  \| \bar \theta_{j,a} ) \geq \epsilon\}$, then $\hat\Theta_\epsilon\subseteq \hat\Theta$. Note that the components of $\hat\ttt$ depend on different components of $\{ R_{i,j}, i \in [N], j \in [J]\}$ and thus are independent. We have
\begin{align*}
	&~ \mathbb P\Big(\sum_{j=1}^J\sum_{a=1}^L n_a D(\hat\theta_{j,a}  \| \bar \theta_{j,a} ) \geq \epsilon\Big)\\ 
= &~ \sum_{\tilde\ttt\in\hat\Theta_\epsilon}
    \mathbb P\Big(\hat\ttt=\tilde\ttt\Big)\\
\leq &~ \sum_{\tilde\ttt\in\hat\Theta_\epsilon} \prod_{j} \prod_{a}
    \exp\{ -n_a D(\tilde\theta_{j,a} \| \bar \theta_{j,a})\}
    \\
\leq &~ \sum_{\tilde\ttt\in\hat\Theta_\epsilon} \exp\{ -n_a  \sum_{j} \sum_{a} D(\tilde\theta_{j,a} \| \bar \theta_{j,a})\}
    \\
\leq &~ 
    \sum_{\tilde\ttt\in\hat\Theta_\epsilon}
    \exp\{-\epsilon\}
    \\
\leq &~ |\hat\Theta_\epsilon | e^{-\epsilon}
\leq  |\hat\Theta | e^{-\epsilon}
\leq \Big(\frac{N}{L}+1\Big)^{JL} e^{-\epsilon}.
\end{align*}
The above result holds for fixed $\ZZ$, so applying a union bound over all the $L^N$ possible assignment $\ZZ$, there is 
$$\mathbb P\Big(\max_{\ZZ} \left\{\sum_j\sum_a n_a D(\hat\theta_{j,a}  \| \bar \theta_{j,a} ) \right\} \geq \epsilon \Big) \leq L^N \Big(\frac{N}{L}+1\Big)^{JL} e^{-\epsilon}.$$
Now take $\delta = L^N\Big(\frac{N}{L}+1\Big)^{JL} e^{-\epsilon}$, then $\epsilon = N\log L + JL\log(\frac{N}{L} + 1) - \log\delta$. Therefore the following event happens with probability at least $1-\delta$,
\begin{equation*}
\max_{\ZZ} \left\{ \sum_j\sum_a n_a D(\hat\theta_{j,a}  \| \bar \theta_{j,a}) \right\}< \epsilon = N\log L + JL\log\Big(\frac{N}{L} + 1\Big) - \log\delta.
\end{equation*}
This concludes the proof of Theorem \ref{thm-kl}.

\bigskip
\textbf{Step 3.} In this step we bound $|X-\mathbb E[X]|$, with $X$ defined in \eqref{eq-defx}. 
Denote $X_{i,j}=R_{i,j}\log (\bar\theta_{j,z_i}/(1-\bar\theta_{j,z_i}))$, then $X=\sum_{i}\sum_j X_{i,j}$. Under Assumption 1, 
we have $|X_{i,j}|\leq \gamma\log J$. 
Then we have 
{$\sum_i\sum_j \ME[X_{i,j}^2] \leq \sum_{i}\sum_{j}\MP(R_{i,j}=1) \gamma^2 (\log J)^2 = \gamma^2 \sum_{i}\sum_{j} P_{i,j} (\log J)^2 = \gamma^2 MNJ (\log J)^2$.}
Applying the Bernstein's inequality to the sum of independent bounded random variables, we have the following holds for any fixed $\ZZ$,
\begin{align*}
	\mathbb P (|X-\mathbb E[X]| \geq \epsilon)
\leq &~2\exp\left\{ -\frac{(1/2)\epsilon^2}{\sum_i\sum_j \ME[X_{i,j}^2]+ (1/3)\gamma\log (J) \epsilon } \right\}\\
\leq &~ 2\exp\left\{ -\frac{(1/2)\epsilon^2}{\gamma^2 MNJ(\log J)^2 + (1/3)\gamma\log (J) \epsilon} \right\}.
\end{align*}

We next prove the following theorem.
\begin{theorem}
	\label{thm-scale}
	Under the following scaling (as $N, J \rightarrow \infty$),
\begin{equation}\label{eq-scale1}
\frac{MJ}{\log L} \to \infty \;, \frac{N}{L} \to \infty,
\end{equation}
\[
\sqrt{\frac{M}{J}} \left( \frac{N}{L}\right)^{1-\xi} \rightarrow \infty \; \text{for some small $\xi>0$},
\]
we have $$\frac{1}{NJ}\max_{\ZZ} |\ell(\RR;\,\ZZ) -  \bar\ell(\ZZ)| 
	= o_P\left(\frac{\sqrt{M \log L}}{\sqrt{J}} (\log J)^{1+\eta}\right)
	$$
for any $\eta >0$.
\end{theorem}
We need to bound $|\ell(\RR;\,\ZZ) -  \bar\ell(\ZZ)|$ uniformly over all the $\ZZ$. Combining the results of Step 2 and Step 3, since there are $L^N$ possible assignments of $\ZZ$, we apply the union bound to obtain 
\begin{align}\label{eq-deltanj}
	&~\MP (\max_{\ZZ} |\ell(\RR;\,\ZZ) - \bar\ell(\ZZ)| \geq 2\epsilon \delta_{NJ})\\ \notag
\leq &~ L^N \MP \left[
\left\{\sum_j\sum_a n_a D(\hat\theta_{j,a} \| \bar \theta_{j,a})\geq \epsilon \delta_{NJ}\right\}
\cup
\left\{|X - \ME[X]| \geq \epsilon \delta_{NJ}\right\}
\right]\\ \notag
\leq &~  \exp\Big\{N\log L + JL\log\Big(\frac{N}{L} + 1\Big) - \epsilon \delta_{NJ}\Big\} \\ \notag
&~  + 2\exp\Big\{N\log L -\frac{(1/2)\epsilon^2\delta_{NJ}^2 }{ \gamma^2 MNJ(\log J)^2 + (1/3)\gamma\log (J) \epsilon\delta_{NJ}} \Big\}.
\end{align}
In order for the term on the right hand side of the above display to go to zero, the following of $\delta_{NJ}$ would suffice,
\begin{equation}
	\label{eq-scale}
	\delta_{NJ} = N\sqrt{MJ\log L} (\log J)^{1+\eta}.
\end{equation}
for a small positive constant $\eta$.  
Then the right hand side of \eqref{eq-deltanj} goes to zero as $N, J$ go large and 
hence the scaling of $J$ described in the theorem
 yields $\MP (\max_{\ZZ} |\ell(\RR;\,\ZZ) - \bar\ell(\ZZ)| \geq 2\epsilon \delta_{NJ}) = o(1)$, which implies
\begin{equation}
	\label{eq-llrate}
	\frac{1}{NJ}\max_{\ZZ} |\ell(\RR;\,\ZZ) - \bar\ell(\ZZ)| 
	=o_P\left(\frac{\sqrt{M\log L}}{\sqrt{J}} (\log J)^{1+\eta}\right).
\end{equation}
This proves Theorem \ref{thm-scale}.

\bigskip
\textbf{Step 4.}
Denote the true class assignments by $\ZZ^0$. We first establish 
\begin{equation}
	\label{eq-tr}
	\bar \ell(\ZZ^0) \geq \bar\ell(\ZZ), \quad\text{for all }\ZZ.
\end{equation}
First note that $\theta^0_{j,z_i^0} = P_{i,j}$, and 
$$
\bar\theta_{j,z_i^0} = \frac{\sum_{m=1}^N Z_{m, z_i^0}^0 P_{m,j}}{\sum_{m=1}^N Z^0_{m, z_i^0}} 
= \frac{\sum_{m=1}^N Z_{m, z_i^0}^0 P_{i,j}}{\sum_{m=1}^N Z^0_{m, z_i^0}} = P_{i,j}.
$$
The difference $\bar \ell(\ZZ^0) - \bar\ell(\ZZ)
$ can be written as
\begin{align*}
	\bar \ell(\ZZ^0) - \bar\ell(\ZZ)
=&~ \sum_{j}\sum_i [ P_{i,j} \log\Big(\frac{\bar\theta^0_{j,z_i^0}}{\bar\theta^{\ZZ}_{j,z_i}}\Big) + (1-P_{i,j})\log\Big(\frac{1-\bar\theta^0_{j,z_i^0} }{1- \bar\theta^{\ZZ}_{j,z_i}}\Big)] \\
=&~\sum_{j}\sum_i [ P_{i,j} \log\Big(\frac{P_{i,j}}{\bar\theta^{\ZZ}_{j,z_i}}\Big) + (1-P_{i,j})\log\Big(\frac{1-P_{i,j} }{1- \bar\theta^{\ZZ}_{j,z_i}}\Big)]=\sum_i\sum_j D(P_{i,j} \| \bar\theta^{\ZZ}_{j,z_i})\geq 0,
\end{align*}
therefore establishing \eqref{eq-tr}. Since the above holds for every $\ZZ$, it also holds for the maximum likelihood estimator $\hat\ZZ$. We further upper bound $\bar \ell(\ZZ^0) - \bar\ell(\ZZ)$ from above as follows,
\begin{align*}
	0\leq \bar \ell(\ZZ^0) - \bar \ell(\hat \ZZ)
     \leq [\bar \ell(\ZZ^0) - \ell(\RR;\,\ZZ^0)] + 
     \underbrace{[\ell(\RR;\,\ZZ^0) - \ell(\RR;\,\hat\ZZ)]}_{\leq 0}
     + [\ell(\RR;\,\hat\ZZ) - \bar \ell(\hat \ZZ)],
\end{align*}
where $[\ell(\RR;\,\ZZ^0) - \ell(\RR;\,\hat\ZZ)] \leq 0$ results from the definition of $\hat\ZZ$ as the MLE, that is $\ZZ$ maximizes the $\ell(\RR;\,\ZZ,\hat\TT^{\ZZ})$. Therefore
\begin{align*}
	0\leq \bar \ell(\ZZ^0) - \bar \ell(\hat \ZZ)
	 \leq &~ [\bar \ell(\ZZ^0) - \ell(\RR;\,\ZZ^0)] +  [\ell(\RR;\,\hat\ZZ) - \bar \ell(\hat \ZZ)]
	 \\
	 \leq &~ 2\sup_{\ZZ} |\bar \ell(\ZZ) - \ell(\RR;\,\ZZ)|
	  \\
	 =&~ o_p(\delta_{NJ}).
\end{align*}

\bigskip

\textbf{Step 5.} To establish the consistency of MLE in clustering subjects into latent classes, we need to introduce the notion of partitions. First we observe that any latent class assignment $\ZZ$ defines a partition on $[N]$ into $T$ subsets $(S_1, \dots, S_T)$ via mapping  $\Pi^{\ZZ}$ from $[N]$ to $[T]$ such that for any subject we have $\theta_{j,z_i}^0 = \theta_{j,\Pi_i^{\ZZ}}^0$ for all $j$. We now generalize this notion. For any partition on $[N]$, define 
\[
\bar\theta^{\Pi}_{j,a} = \frac{1}{|S_a|} \sum_{i=1}^N \theta_{j,z_i^0}^0 I(i \in S_a) = \frac{1}{|S_a|} \sum_{i=1}^N P_{i,j} I(i \in S_a)
\]
as the average over all $i$ in the subset $S_a$ indexed by $\Pi_i = a$. We then define generalization of $\bar\ell(\ZZ)$ as
\[
\bar\ell(\Pi) = \sum_i \sum_j [P_{i,j}\log(\bar\theta^{\Pi}_{j,\Pi_i}) + (1-P_{i,j})\log(1-\bar\theta^{\Pi}_{j,\Pi_i})].
\]
Note that $\bar\theta^{\Pi^\ZZ}_{j,a} = \bar\theta^{\ZZ}_{j,a}$ and hence $\bar\ell(\Pi^\ZZ) = \bar\ell(\ZZ)$ when the partition $\Pi^\ZZ$ is induced by latent class assignment $\ZZ$.

We will proceed as follows: in step 6 we show a refined partition increases $\bar\ell(\cdot)$. We then construct a refined partition $\Pi^{*}$ for every partition $\Pi^{\ZZ}$ induced by $\ZZ$ and prove $\bar\ell(\ZZ^0) - \bar\ell(\Pi^*) \geq \frac{1}{2} N_e(\zz) \beta_J$ in step 7. Finally we apply the results to MLE $\hat{\ZZ}$ and obtain the desired results in step 8.

\bigskip

\textbf{Step 6.} We prove the following lemma:
\begin{lemma}
    \label{refine_incr}
    Let $\Pi^*$ be a refinement of any partition $\Pi$ of $[N]$, then we have $\bar\ell(\Pi^*) \geq \bar\ell(\Pi)$.
\end{lemma}

Given $a\in [T^*]$ indexing $S_a^*$ in $\Pi^{*}$, since $S_a^* \subseteq S_b$ for some $S_b$ in $\Pi$, let $F(a)$ denote its index under $\Pi$ (i.e. $b$). We have
$$
\begin{aligned}
\bar{\ell}\left(\Pi^{*}\right) &=\sum_{a=1}^{T^{*}}\left|S_{a}^{*}\right|\sum_{j=1}^J \left\{\bar{\theta}_{j,a}^{\Pi^*} \log \bar{\theta}_{j,a}^{\Pi^*}+\left(1-\bar{\theta}_{j,a}^{\Pi^*}\right) \log \left(1-\bar{\theta}_{j,a}^{\Pi^*}\right)\right\} \\
& \geq \sum_{a=1}^{T^{*}}\left|S_{a}^{*}\right|\sum_{j=1}^J \left\{\bar{\theta}_{j,a}^{\Pi^*} \log \bar{\theta}_{j,F(a)}^{\Pi}+\left(1-\bar{\theta}_{j,a}^{\Pi^*}\right) \log \left(1-\bar{\theta}_{j,F(a)}^{\Pi}\right)\right\} \\
&=\sum_{b=1}^{T}\left|S_{b}\right|\sum_{j=1}^J \left\{\bar{\theta}_{j,b}^{\Pi} \log \bar{\theta}_{j,b}^{\Pi}+\left(1-\bar{\theta}_{j,b}^{\Pi}\right) \log \left(1-\bar{\theta}_{j,b}^{\Pi}\right)\right\} =\bar\ell(\Pi).
\end{aligned}
$$
The first equality is obtained by rewriting $\bar\ell(\Pi)$ in terms of subsets. The inequality follows from non-negativity of K-L distance. Then we combine terms in same class under $\Pi$ and obtain the second equality.

\bigskip

\textbf{Step 7.} Now we prove a result on refinement.
\begin{lemma}
    \label{refine_error}
    For any latent class assignment $\ZZ$, there exists a partition $\Pi^*$ that refines $\Pi^{\ZZ}$ and
    \[
    \bar\ell(\ZZ^0) - \bar\ell(\Pi^*) \geq \frac{1}{2} N_e(\zz) J\beta_J
    \]
\end{lemma}
For a given $\ZZ$, partition each latent class assigned by $\ZZ$ into sub-classes according to true assignments $\ZZ^0$ of each sample. For each sample $i_1$ that is incorrectly assigned by $\ZZ$ (by definition this means its true class under $\ZZ^0$ is not in the majority within its estimated class under $\ZZ$), we find another sample $i_2$ assigned to same class under $\ZZ$ but $i_1$ and $i_2$ belong to different class under $\ZZ^0$ and make these two samples $(i_1, i_2)$ a pair. We allow two misclassified samples to form a pair. Note that since incorrectly assigned samples are not in the majority of that class, we can find a pair for each of them.

Here is a simple example. Suppose in one class of $\ZZ$, we have 7 samples and $\ZZ^0$ (true latent class assignments) assigns them as three sub-classes $\{1,2,3,4\}, \{5,6\}, \{7\}$. In this example samples indexed by 5,6 and 7 are misclassified. We can find pairs $(4,5), (6,7)$. 

The refined partition $\Pi^*$ contains all such pairs and remaining correctly assigned samples in all classes assigned by $\ZZ$. So for the above example, the refined subset for that class is $\{1,2,3\}, \{4,5\}, \{6,7\}$. Let $e(\zz)$ by the set of incorrectly assigned sample. Clearly $\Pi^*$ is a refinement and we have
\[
\begin{aligned}
    \bar\ell(\ZZ^0) - \bar\ell(\Pi^*) = &~ \sum_i \sum_j D(P_{i,j} \| \bar\theta_{j, \Pi_i^*}^{\Pi^*}) \\
    \geq &~ \sum_{i \in e(\zz)} \sum_j D(P_{i,j} \| \bar\theta_{j, \Pi_i^*}^{\Pi^*}) \\
    = &~ \sum_{i \in e(\zz)} \sum_j D(P_{i,j} \| \frac{P_{i,j} + P_{i',j}}{2}) 
\end{aligned}
\]
where $i$ and $i'$ are in different classes under $\ZZ^0$ while in same subset under $\Pi^*$ by definition. Apply Pinsker's inequality we have 
\[
\begin{aligned}
    D(P_{i,j} \| \frac{P_{i,j} + P_{i',j}}{2}) \geq &~ \frac{1}{2} \left[|P_{i,j}-\frac{P_{i,j} + P_{i',j}}{2}| + |1-P_{i,j} - (1- \frac{P_{i,j} + P_{i',j}}{2})|\right]^2 \\
    = &~ \frac{1}{2} (P_{i,j}-P_{i',j})^2 \\
    = &~ \frac{1}{2} (\theta_{j,z_i^0}^0- \theta_{j,z_{i'}^0}^0)^2 
\end{aligned}
\]
Hence we have
\[
\begin{aligned}
    \bar\ell(\ZZ^0) - \bar\ell(\Pi^*) \geq &~ \sum_{i \in e(\zz)} \sum_j \frac{1}{2} (\theta_{j,z_i^0}^0- \theta_{j,z_{i'}^0}^0)^2 \\
    \geq &~ \sum_{i \in e(\zz)} \frac{1}{2} \| \TT_{\cdot, z_i^0}^0-\TT_{\cdot, z_{i'}^0}^0 \|^2 \\
    \geq &~ \frac{1}{2} N_e(\zz) J\beta_J
\end{aligned}
\]

\bigskip

\textbf{Step 8.} 
Apply Lemma \ref{refine_error} to MLE $\hat{\zz}$, there exists a refinement of $\Pi^{\hat{\ZZ}}$ denoted as $\Pi^*$ such that 
\[
\bar\ell(\ZZ^0) - \bar\ell(\Pi^*) \geq  \frac{1}{2} N_e(\zz) J\beta_J
\]
By Lemma \ref{refine_incr} we have $\bar\ell({\Pi^*}) \geq \bar\ell(\Pi^{\hat{\ZZ}})$. So we conclude that 
\[
\begin{aligned}
    o_P(\delta_{NJ}) = &~ \bar \ell(\ZZ^0) - \bar \ell(\hat \ZZ) \\
    \geq &~ \bar\ell(\ZZ^0) - \bar\ell(\Pi^*) \\
    \geq &~ \frac{1}{2} N_e(\zz) J\beta_J
\end{aligned}
\]
which completes the proof.

\section*{Appendix 2: Proof of Corollary 2}

Recall $m_a = \argmax_{l \in [L]} \, \sum_{i\in \hat{C}_a}Z_{i,l}^{0}$ is the class index under $\ZZ^0$ for cluster $\hat{C}_a$. For any $0 < \epsilon < \tau$, define the following event 
\[
A_{N}^{\epsilon} = \{ N_e(\hat{\zz}) / N \leq \epsilon\}.
\]
On the event $A_{N}^{\epsilon}$, for any $l \in [L]$, since we assume $n_l^0 / N \geq \tau >0$, we claim that there is exactly one $a \in [L]$ such that $m_a = l$, i.e. the $a-$th cluster represents the $l-$th class. To see the existence of such $a$, assume by contradiction that for some $l$ there is no $a$ such that $m_a = l$, then all subjects in class $l$ are misclassified and we have 
\[
N_e(\hat{\zz})/N \geq n_l^0/N \geq \tau > \epsilon,
\]
a contradiction.
Since for each $l-$th class we can find $a-th$ cluster to represent it and there are exactly $L$ clusters $\hat{C}_1, \dots, \hat{C}_L$, such $a$ must be unique for all the $L$ classes. Note that $\MP (A_N^{\epsilon}) \rightarrow 1,$ the first statement in the corollary is proved.

For any $\epsilon \in (0,\tau)$, from the argument above, on the event $A_{N}^{\epsilon}$ for each $l$ we can find exactly one $a \in [L]$ such that $m_a = l$, then the joint MLE for $\theta_{j,l}^0$ is 
\[
\hat{\theta}_{j,a}=\hat{\theta}_{j,a}^{(\hat{\zz})}=\frac{\sum_{i=1}^{N} \hat{Z}_{i, a} R_{i, j}}{\sum_{i=1}^{N} \hat{Z}_{i, a} }.
\]
Recall we can rewrite $\theta_{j,l}^0$ as
\[
\theta_{j,l}^0=\frac{\sum_{i=1}^{N} {Z}_{i, l}^0 \theta_{j,l}^0}{\sum_{i=1}^{N} {Z}_{i, l}^0 } = \frac{\sum_{i=1}^{N} {Z}_{i, l}^0 P_{i,j}}{\sum_{i=1}^{N} {Z}_{i, l}^0 }.
\]
By triangle inequality we have
\begin{equation*}
    \begin{aligned}
    & \, \max_{j}|\hat{\theta}_{j,a} - \theta_{j,l}^0| \\
    = & \, \max_{j} \left|\frac{\sum_{i=1}^{N} \hat{Z}_{i, a} R_{i, j}}{\sum_{i=1}^{N} \hat{Z}_{i, a} } -  \frac{\sum_{i=1}^{N} {Z}_{i, l}^0 P_{i,j}}{\sum_{i=1}^{N} {Z}_{i, l}^0 }\right| \\
    \leq & \, \max_{j} \left|\frac{\sum_{i=1}^{N} \hat{Z}_{i, a} R_{i, j}}{\sum_{i=1}^{N} \hat{Z}_{i, a} } -  \frac{\sum_{i=1}^{N} \hat{Z}_{i, a} R_{i,j}}{\sum_{i=1}^{N} {Z}_{i, l}^0 }\right| + \max_{j} \left|\frac{\sum_{i=1}^{N} \hat{Z}_{i, a} R_{i, j}}{\sum_{i=1}^{N} {Z}_{i, l}^0 } -  \frac{\sum_{i=1}^{N} {Z}_{i, l}^0 R_{i,j}}{\sum_{i=1}^{N} {Z}_{i, l}^0 }\right|  \\
    + & \, \max_{j} \left|\frac{\sum_{i=1}^{N} {Z}_{i, l}^0 R_{i, j}}{\sum_{i=1}^{N} {Z}_{i, l}^0 } -  \frac{\sum_{i=1}^{N} {Z}_{i, l}^0 P_{i,j}}{\sum_{i=1}^{N} {Z}_{i, l}^0 }\right| \\
    \equiv & \, I_1 + I_2 + I_3.
    \end{aligned}
\end{equation*}
We then analyze these three terms.
\[
I_1 \leq \max_{j} \sum_{i}  \hat{Z}_{i,a} R_{i,j}  \frac{\sum_i |\hat{Z}_{i,a} - Z_{i,l}^0|}{n_l^0 \sum_{i} \hat{Z}_{i,a}} \leq \frac{\sum_i |\hat{Z}_{i,a} - Z_{i,l}^0|}{n_l^0}.
\]
There are two cases in which $|\hat{Z}_{i,a} - Z_{i,l}^0|=1$:
\begin{itemize}
    \item $\hat{Z}_{i,a} = 1, Z_{i,l}^0 = 0$, i.e. subject i is in cluster $a$ but not in class $l$. Since $m_a=l$, subject $i$ is misclassified and counted in $N_e(\hat{\zz})$.
    \item $\hat{Z}_{i,a} = 0, Z_{i,l}^0 = 1$, i.e. subject i is in class $l$ but not in cluster $a$. Since cluster $a$ is the only cluster that represents class $l$, subject $i$ must be misclassified and counted in $N_e(\hat{\zz})$.
\end{itemize}
By clustering consistency we have 
\[
I_1 \leq \frac{N_e(\hat{\zz})}{n_l^0} \leq  \frac{N_e(\hat{\zz})}{\tau N} \stackrel{P}{\longrightarrow} 0.
\]
For the second term we have
\[
I_2 = \frac{\max_j |\sum_i R_{i,j} (\hat{Z}_{i,a} - Z_{i,l}^0)|}{n_l^0} \leq \frac{\sum_i |\hat{Z}_{i,a} - Z_{i,l}^0|}{n_l^0} \leq \frac{N_e(\hat{\zz})}{n_l^0} \stackrel{P}{\longrightarrow} 0.
\]
For the third term, we apply Hoeffding's inequality and obtain 
\[
\MP (I_3\geq \delta) = \MP \left(\max_j \frac{|\sum_i Z_{i,l}^0 (R_{i,j} - P_{i,j})|}{n_l^0} \geq \delta\right) \leq J \exp (-2n_{l}^0 \delta^2) \leq J \exp(-2\tau N \delta^2) \rightarrow 0
\]
where in the last step we use the scaling condition $\sqrt{\frac{M}{J}} \left( \frac{N}{L}\right)^{1-\xi} \rightarrow \infty$. This shows
\[
I_3 \stackrel{P}{\longrightarrow} 0.
\]
Note that on $(A_{N}^{\epsilon})^c$ we may not be able to define $\hat{\theta}_{j,a}$ since the first statement in the corollary may not hold (we may not find the $a-th$ cluster for each $l-$th class). Mathematically we can arbitrarily define any $\hat{\theta}_{j,a}$ as long as it is in $[0,1]$. Since $ \MP(A_{N}^{\epsilon}) \rightarrow 1$, we only need to focus on the situation on $A_{N}^{\epsilon}$. We then have 
\begin{equation*}
\begin{aligned}
    & \, \MP(\max_j |\hat{\theta}_{j,a} -\theta_{j,l}^0| \geq \epsilon) \\
    \leq & \, \MP(\max_j |\hat{\theta}_{j,a} -\theta_{j,l}^0| I_{(A_{N}^{\epsilon})^c} \geq \epsilon/2) + \MP(\max_j |\hat{\theta}_{j,a} -\theta_{j,l}^0| I_{A_{N}^{\epsilon}}\geq \epsilon/2) \\
    \leq & \, \MP((A_{N}^{\epsilon})^c) + \MP(I_1 I_{A_{N}^{\epsilon}} \geq \epsilon/6) +  \MP(I_2 I_{A_{N}^{\epsilon}} \geq \epsilon/6)  + \MP(I_3 I_{A_{N}^{\epsilon}} \geq \epsilon/6) \rightarrow 0.
\end{aligned}
\end{equation*}
This completes the proof.

\section*{Appendix 3: More simulation results}

\subsection*{Random-effect LCM}

We first present more simulation results for random-effect LCM.

\begin{figure}[H]
	\centering
	\subfigure[MSE of item parameters]{
		\begin{minipage}[t]{0.33\linewidth}
			\centering
			\includegraphics[width=2in]{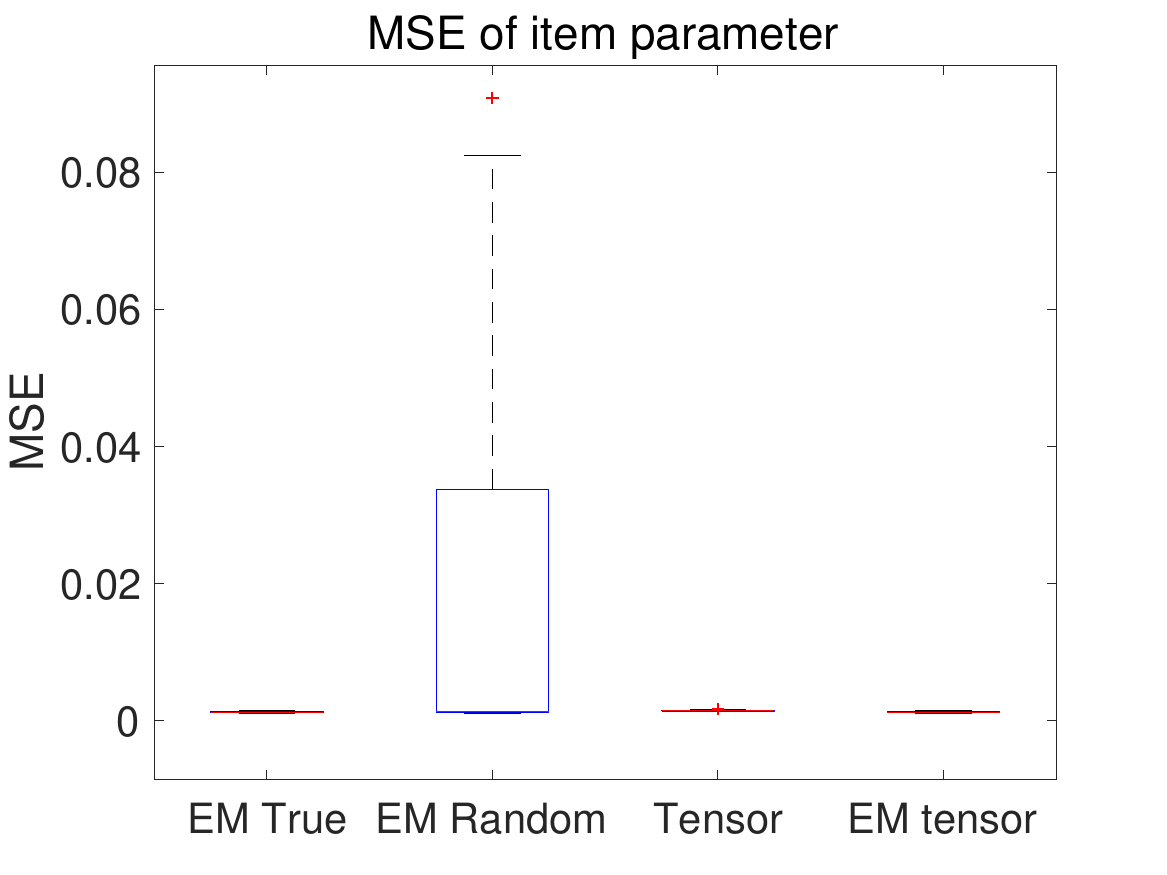}
		\end{minipage}%
	}%
		\subfigure[MSE without EM-random]{
		\begin{minipage}[t]{0.33\linewidth}
			\centering
			\includegraphics[width=2in]{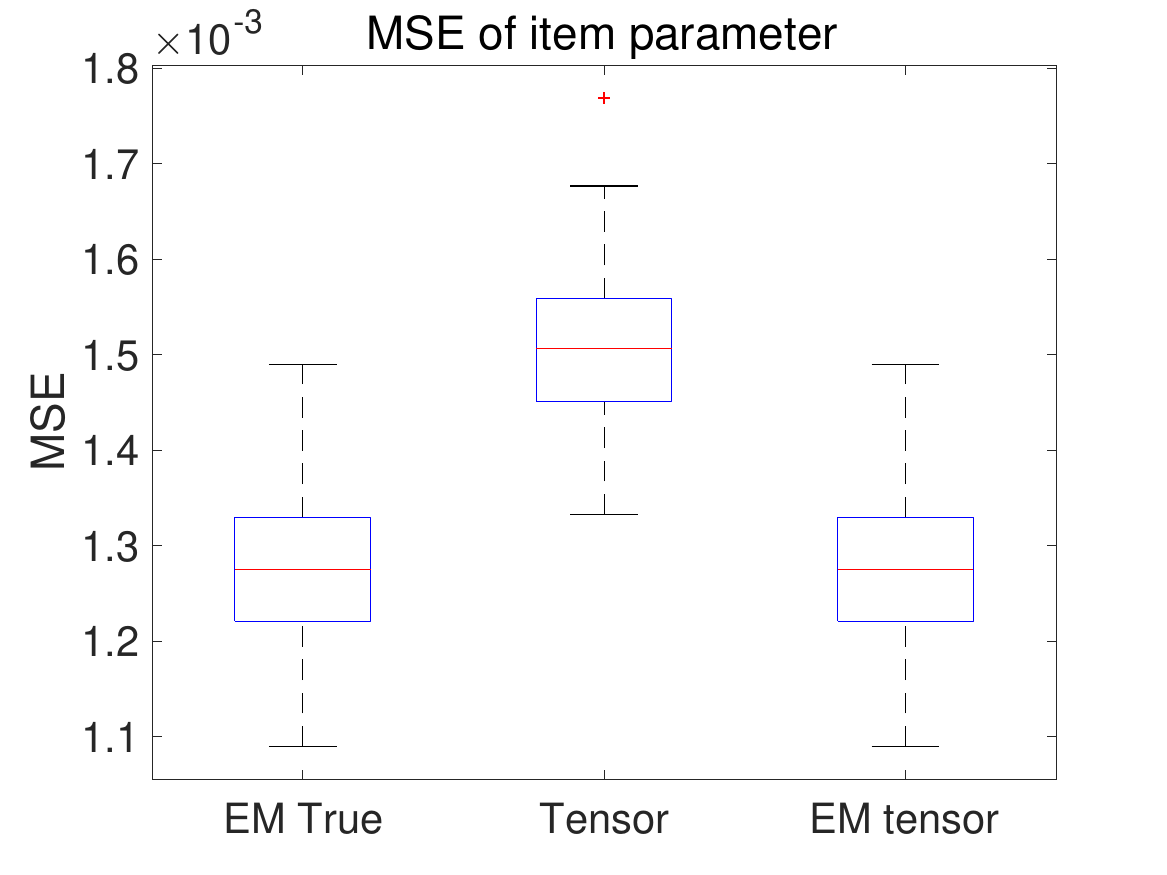}
		\end{minipage}%
	}%
	\subfigure[Running time of the algorithms]{
		\begin{minipage}[t]{0.33\linewidth}
			\centering
			\includegraphics[width=2in]{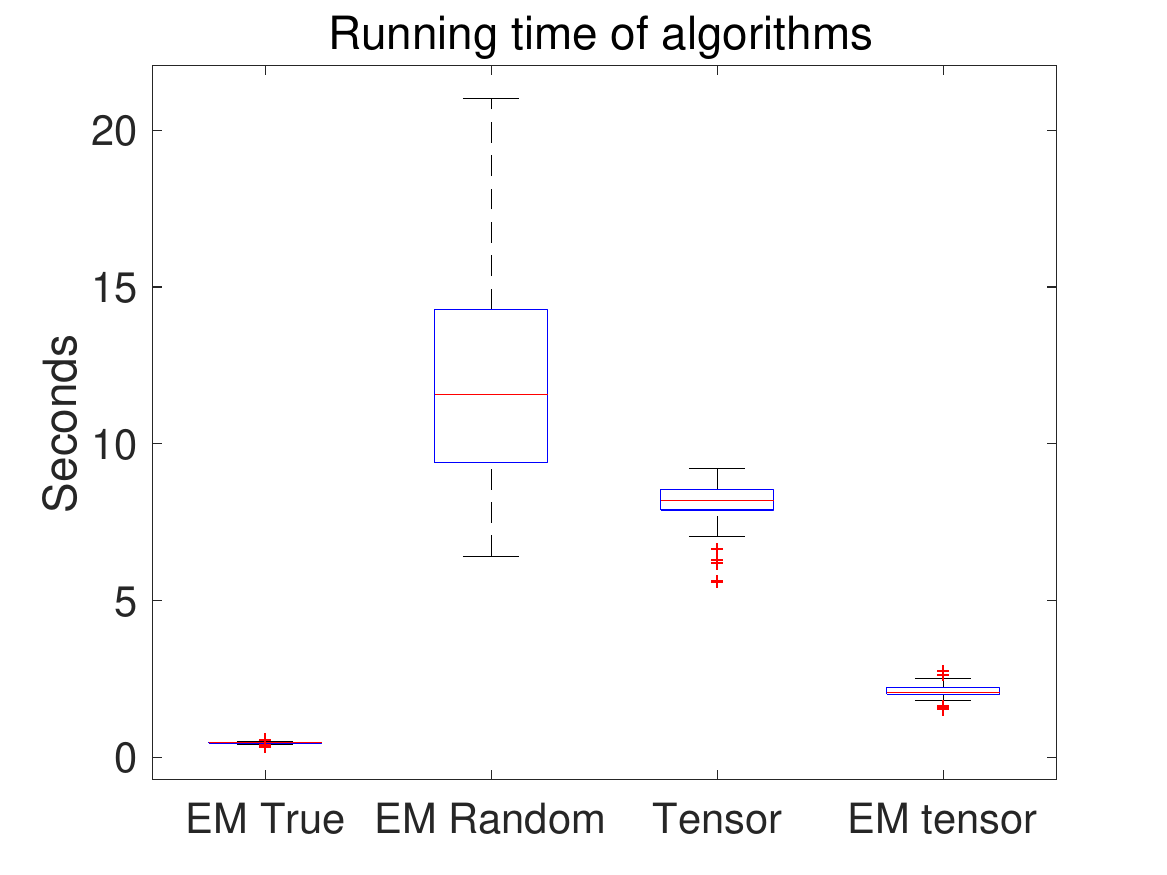}
		\end{minipage}%
	}%
	\centering
	\caption{$N = 1000, J= 100, L=10,$ item parameters $\in \{0.1,0.2,0.8,0.9\}$}
\end{figure}

\begin{figure}[H]
	\centering
	\subfigure[MSE of item parameters]{
		\begin{minipage}[t]{0.4\linewidth}
			\centering
			\includegraphics[width=2in]{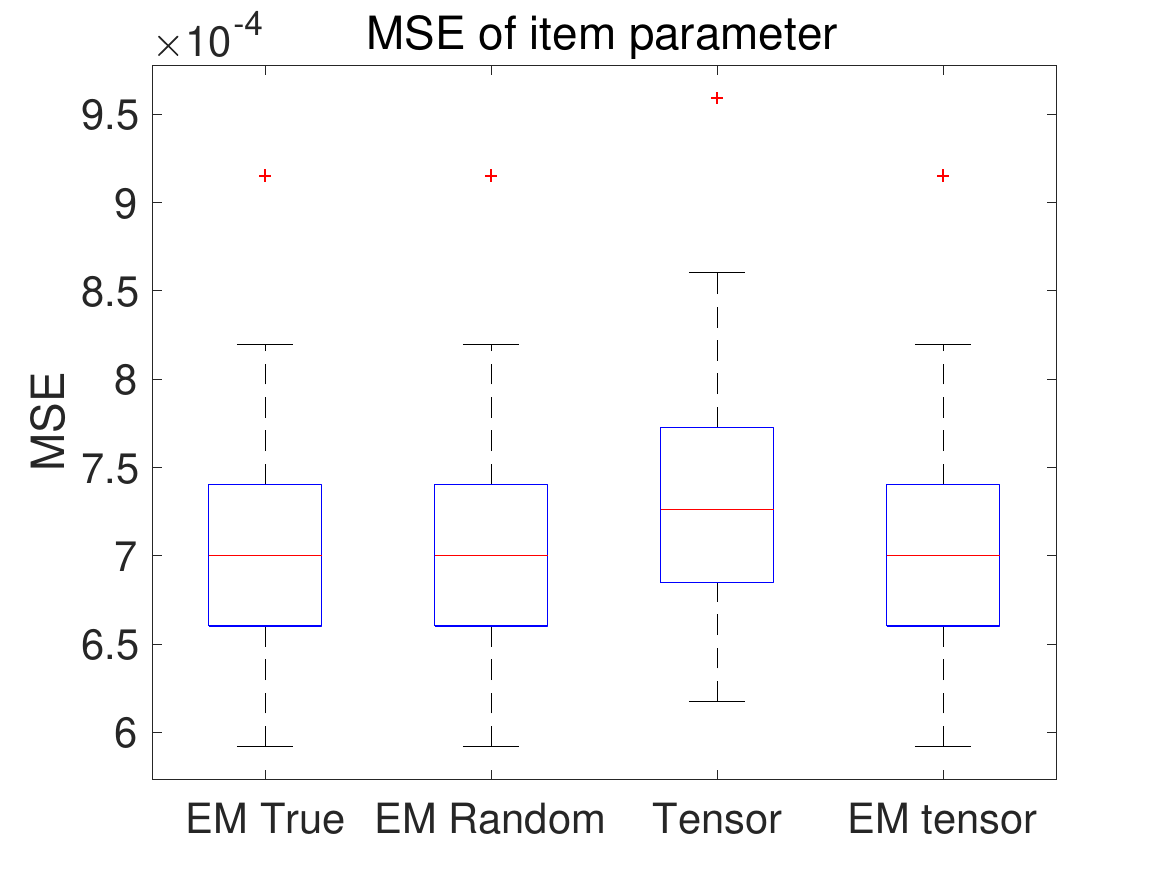}
		\end{minipage}%
	}%
	\subfigure[Running time of the algorithms]{
		\begin{minipage}[t]{0.4\linewidth}
			\centering
			\includegraphics[width=2in]{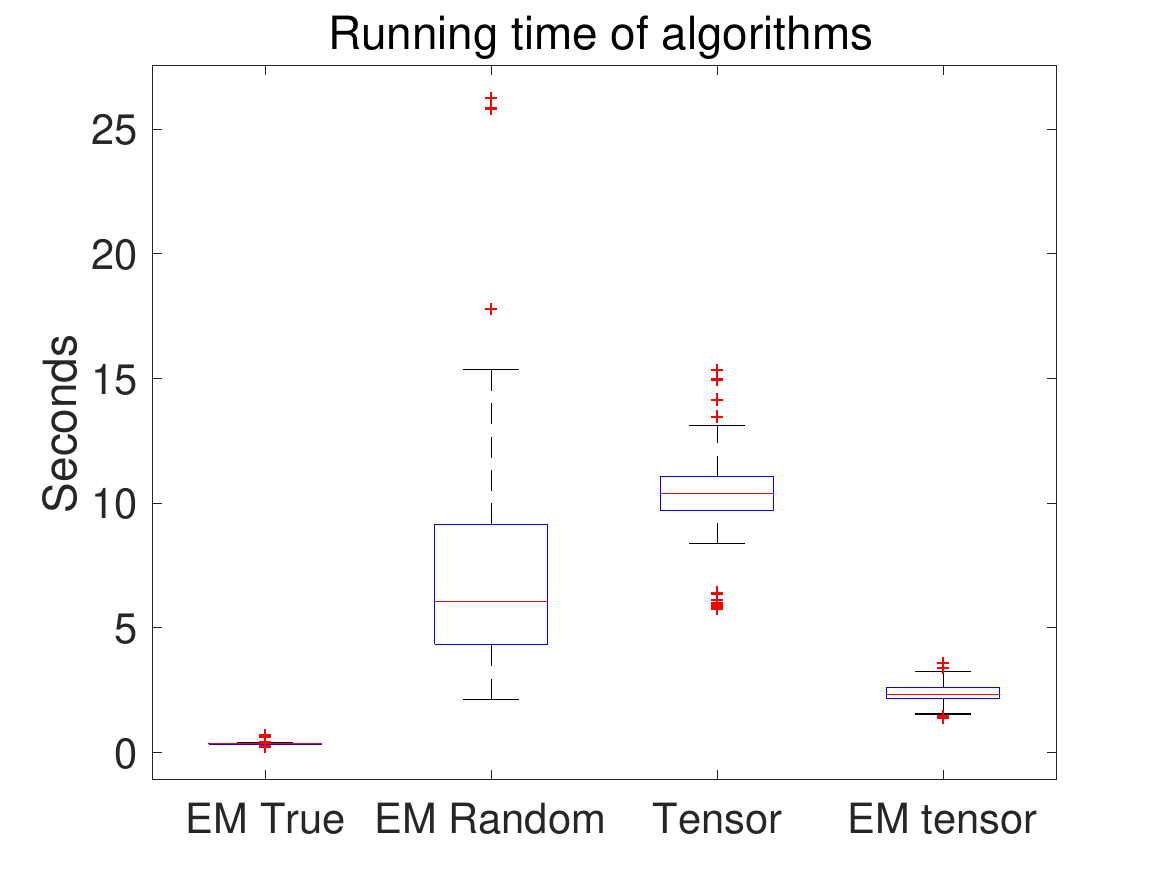}
		\end{minipage}%
	}%
	\centering
	\caption{$N = 1000, J= 200, L=5,$ item parameters $\in \{0.1,0.2,0.8,0.9\}$}
\end{figure}

\begin{figure}[H]
	\centering
	\subfigure[MSE of item parameters]{
		\begin{minipage}[t]{0.33\linewidth}
			\centering
			\includegraphics[width=2in]{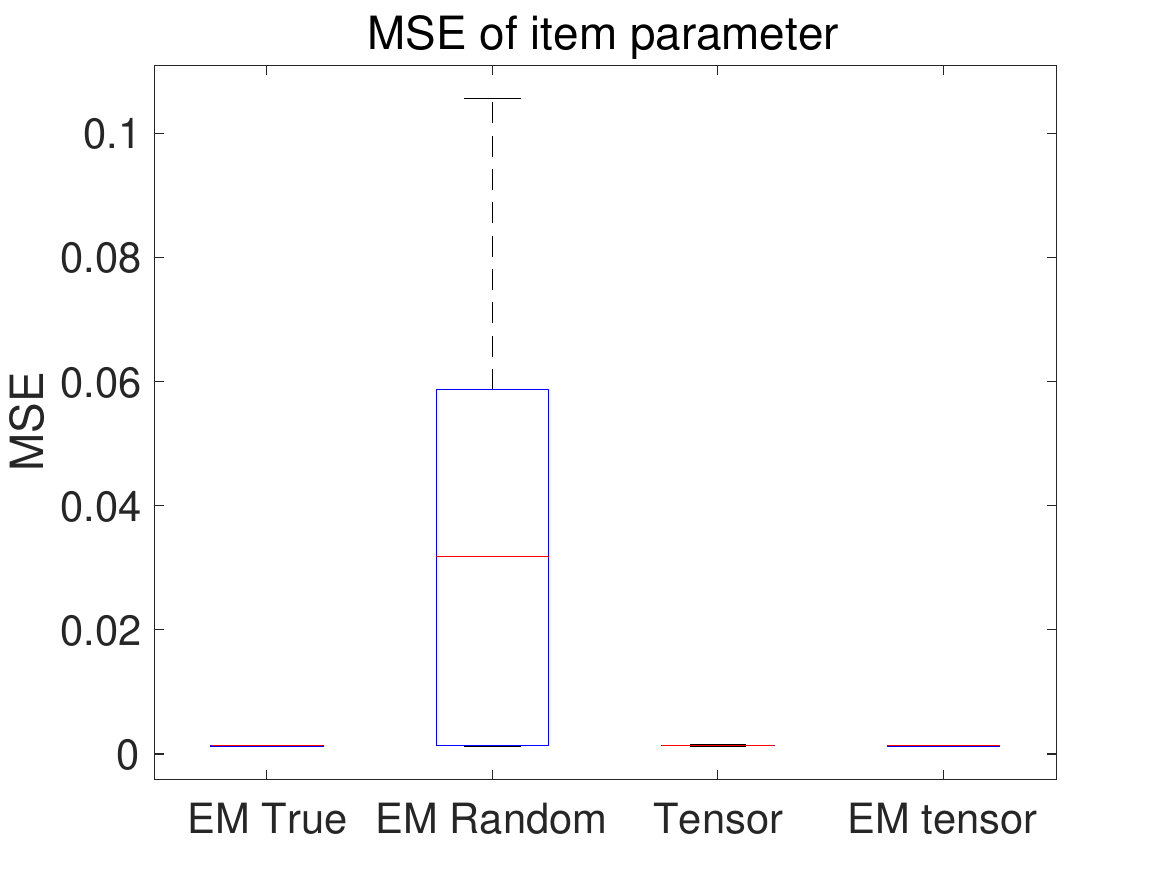}
		\end{minipage}%
	}%
    \subfigure[MSE without EM-random]{
    	\begin{minipage}[t]{0.33\linewidth}
    		\centering
    		\includegraphics[width=2in]{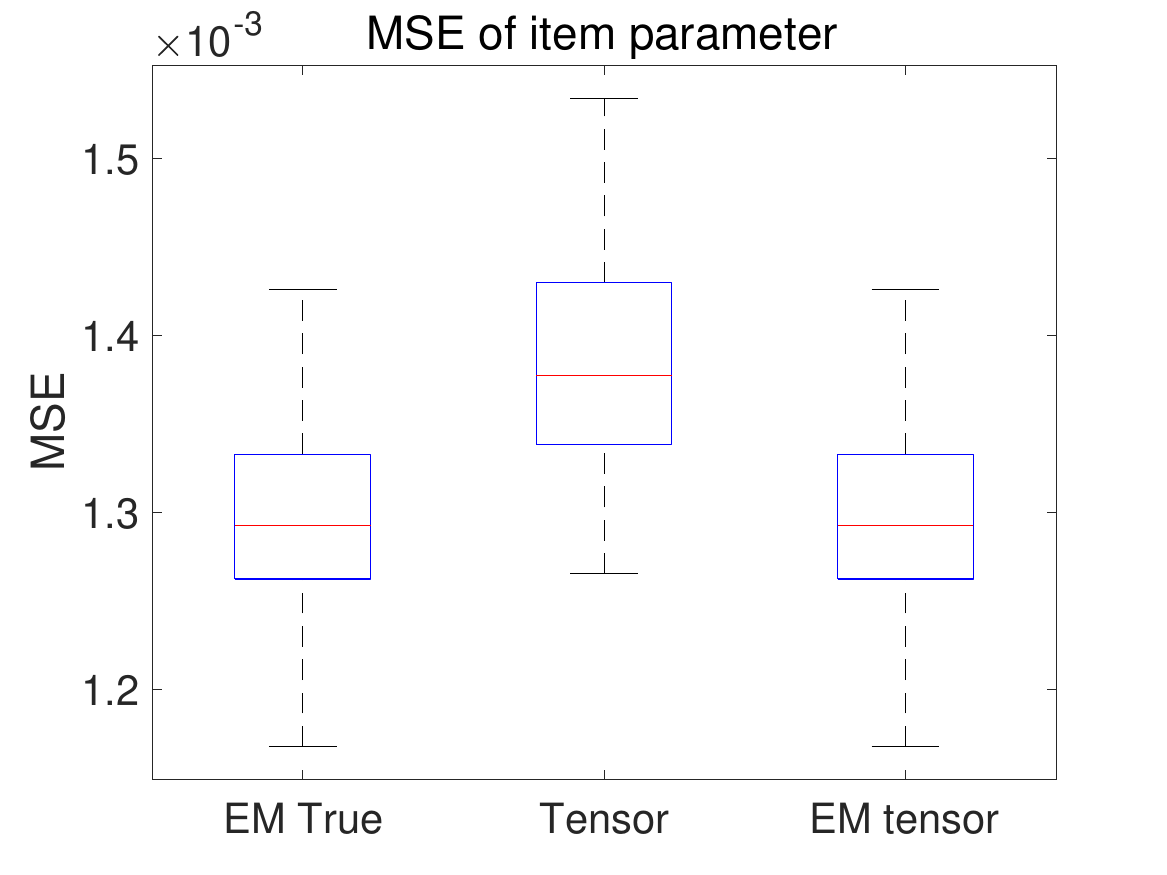}
    	\end{minipage}%
    }%
	\subfigure[Running time of the algorithms]{
		\begin{minipage}[t]{0.33\linewidth}
			\centering
			\includegraphics[width=2in]{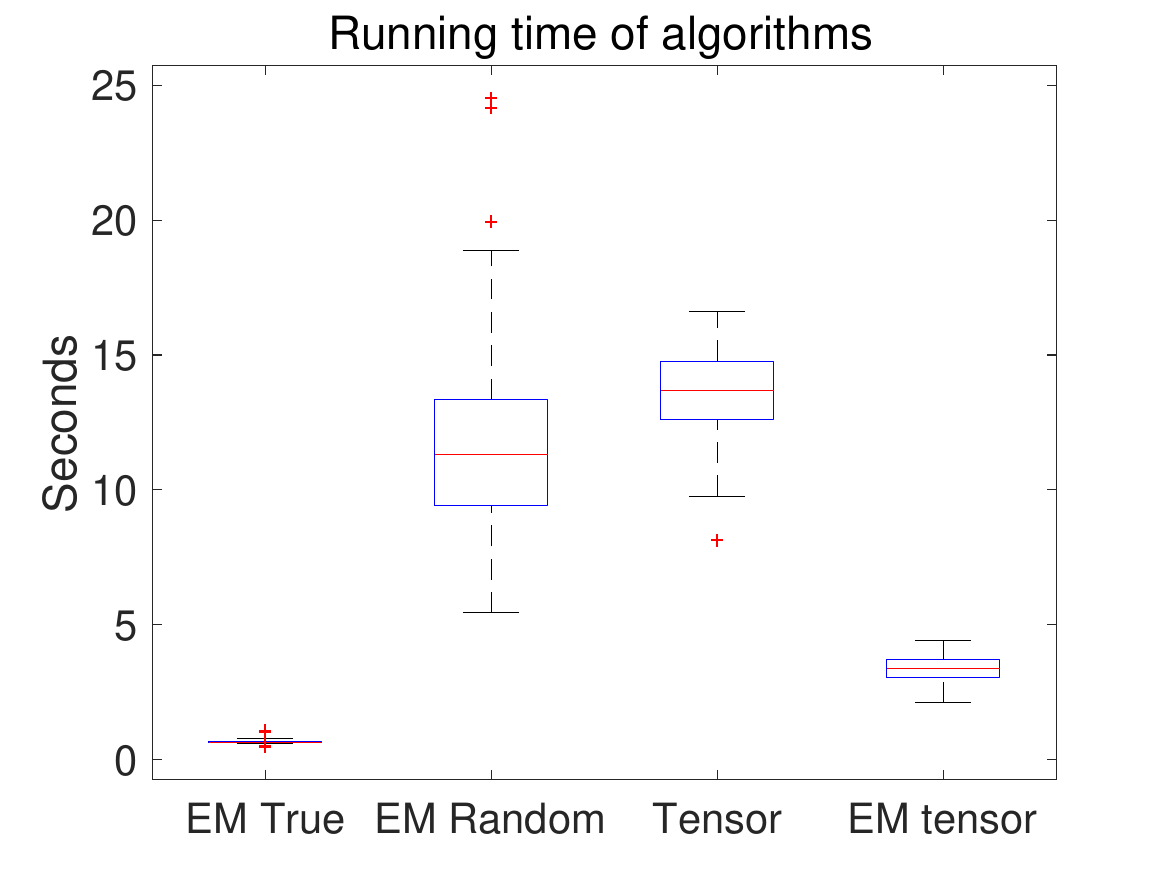}
		\end{minipage}%
	}%
	\centering
	\caption{$N = 1000, J= 200, L=10,$ item parameters $\in \{0.1,0.2,0.8,0.9\}$}
\end{figure}

\begin{figure}[H]
	\centering
	\subfigure[MSE of item parameters]{
		\begin{minipage}[t]{0.4\linewidth}
			\centering
			\includegraphics[width=2in]{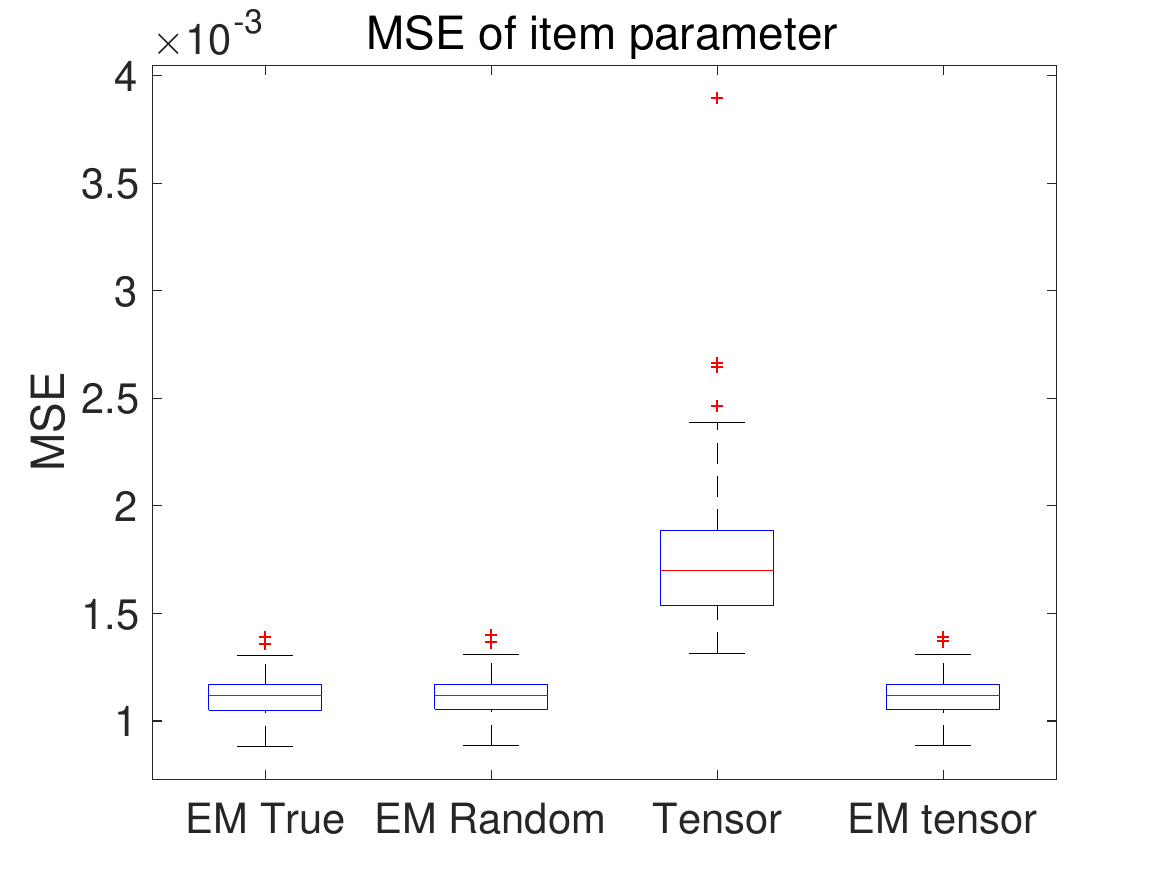}
		\end{minipage}%
	}%
	\subfigure[Running time of the algorithms]{
		\begin{minipage}[t]{0.4\linewidth}
			\centering
			\includegraphics[width=2in]{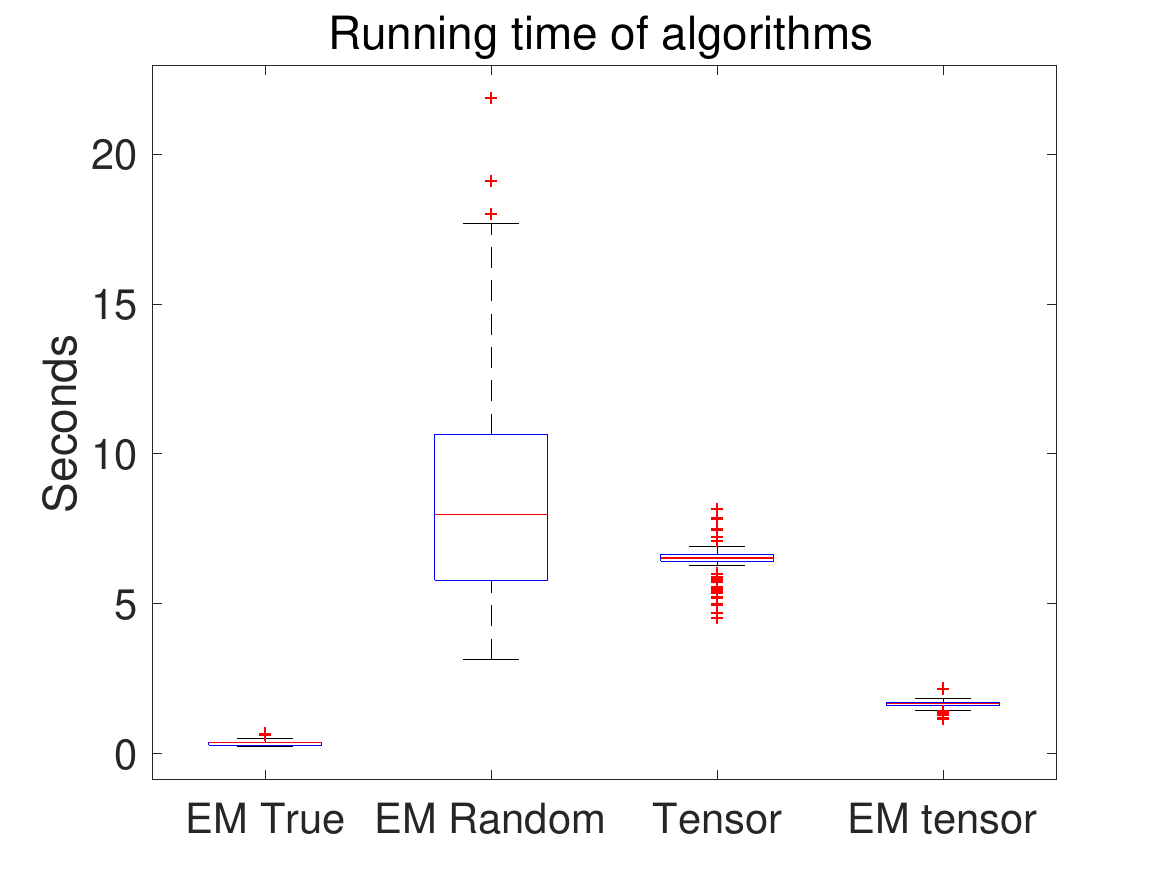}
		\end{minipage}%
	}%
	\centering
	\caption{$N = 1000, J= 100, L=5,$ item parameters $\in \{0.2,0.4,0.6,0.8\}$}
\end{figure}

\begin{figure}[H]
	\centering
	\subfigure[MSE of item parameters]{
		\begin{minipage}[t]{0.4\linewidth}
			\centering
			\includegraphics[width=2in]{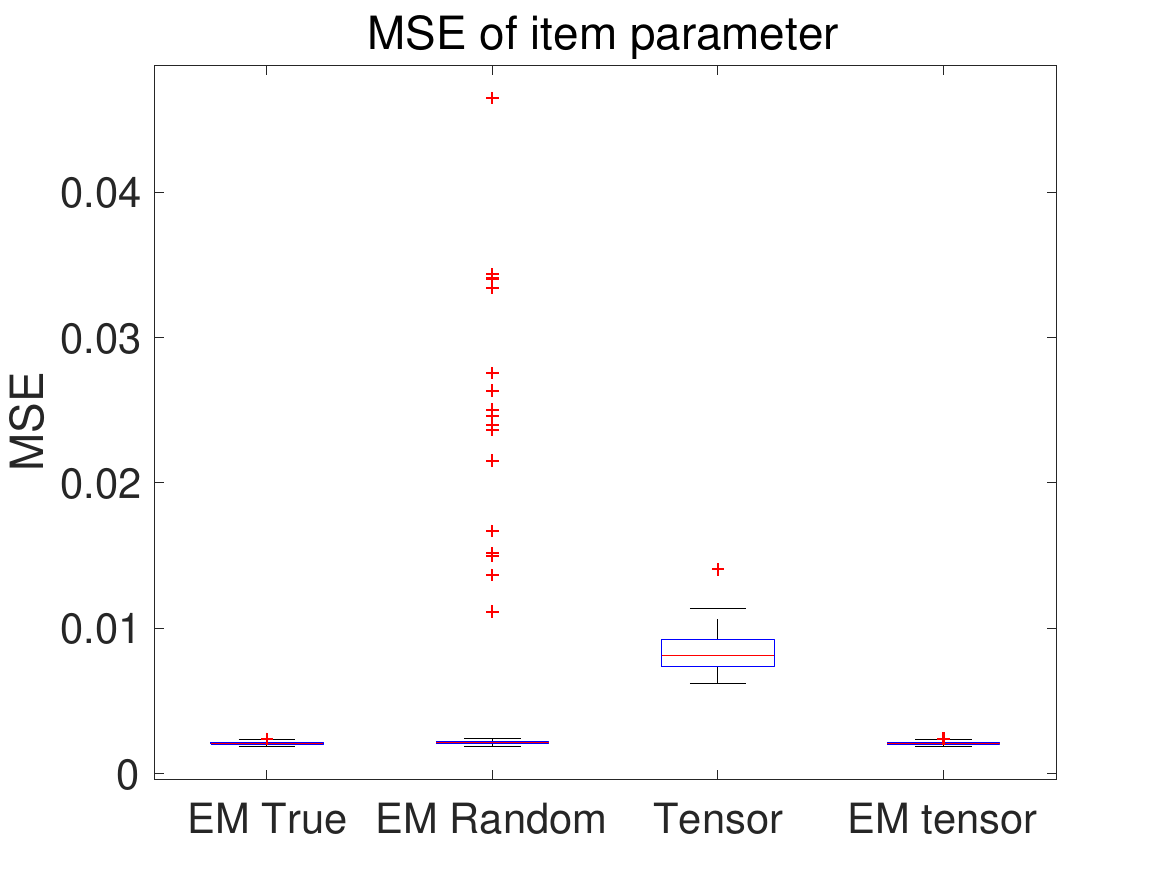}
		\end{minipage}%
	}%
	\subfigure[Running time of the algorithms]{
		\begin{minipage}[t]{0.4\linewidth}
			\centering
			\includegraphics[width=2in]{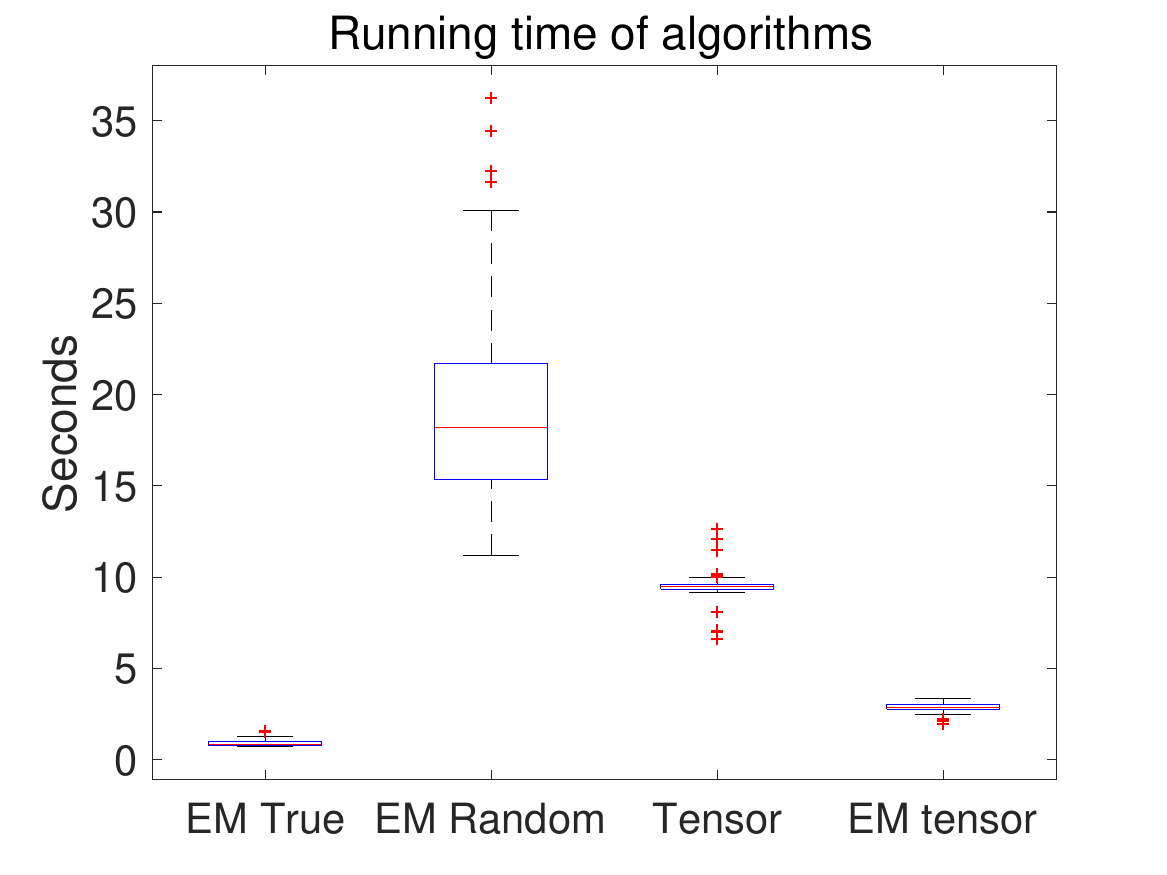}
		\end{minipage}%
	}%
	\centering
	\caption{$N = 1000, J= 100, L=10,$ item parameters $\in \{0.2,0.4,0.6,0.8\}$}
\end{figure}

\begin{figure}[H]
	\centering
	\subfigure[MSE of item parameters]{
		\begin{minipage}[t]{0.4\linewidth}
			\centering
			\includegraphics[width=2in]{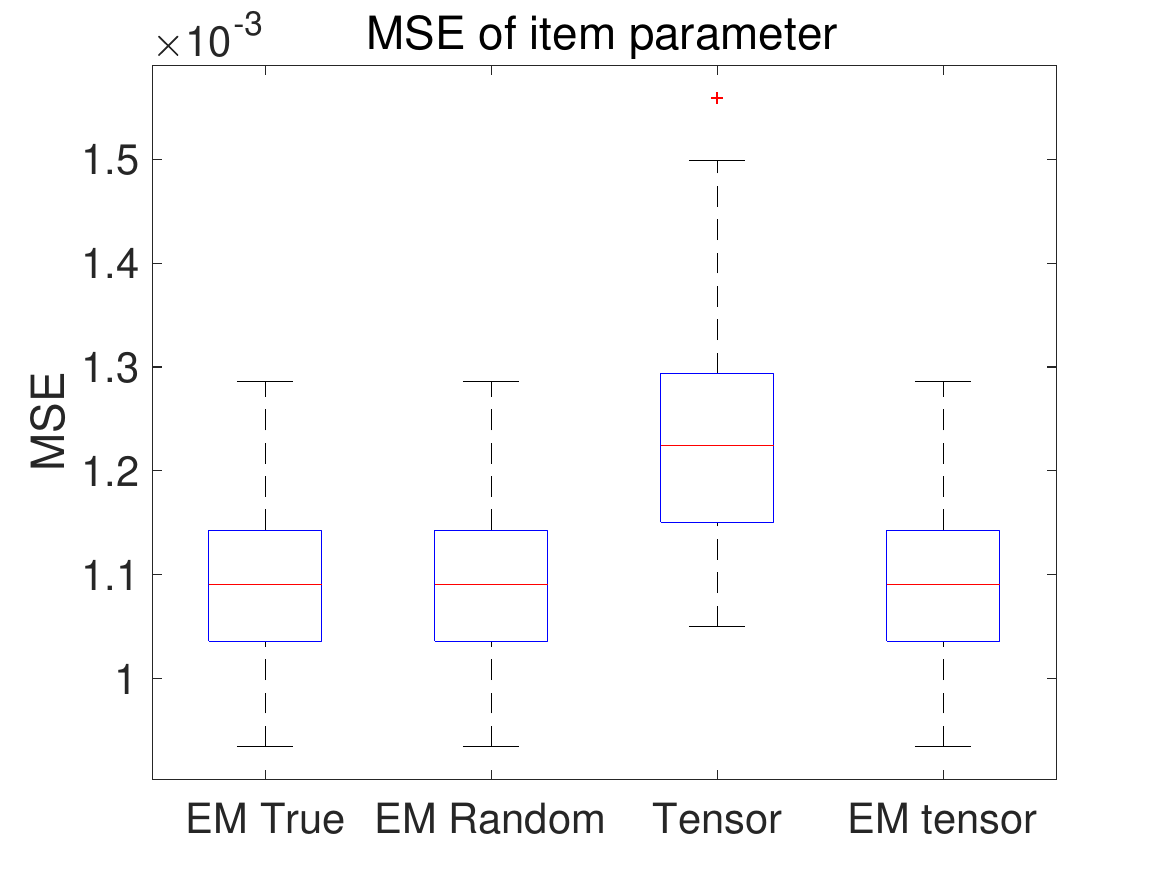}
		\end{minipage}%
	}%
	\subfigure[Running time of the algorithms]{
		\begin{minipage}[t]{0.4\linewidth}
			\centering
			\includegraphics[width=2in]{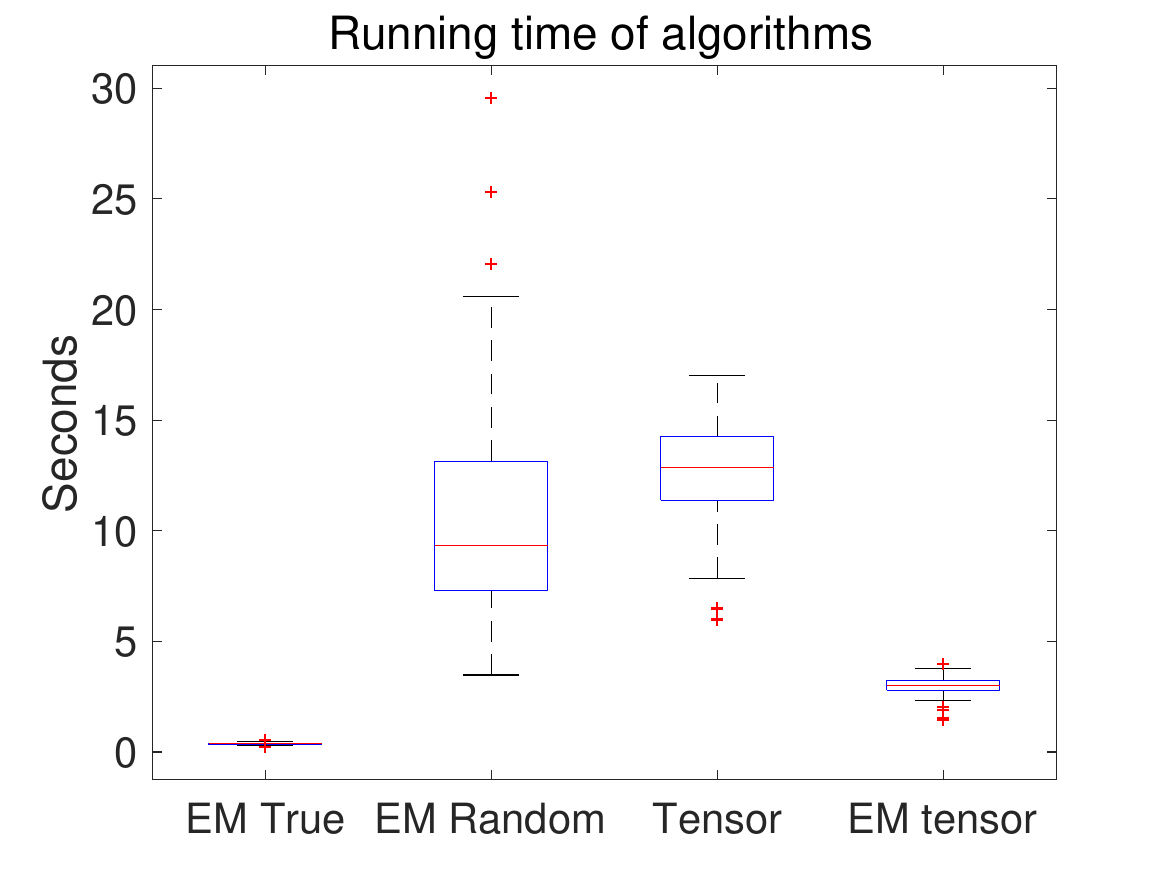}
		\end{minipage}%
	}%
	\centering
	\caption{$N = 1000, J= 200, L=5,$ item parameters $\in \{0.2,0.4,0.6,0.8\}$}
\end{figure}

\begin{figure}[H]
	\centering
	\subfigure[MSE of item parameters]{
		\begin{minipage}[t]{0.33\linewidth}
			\centering
			\includegraphics[width=2in]{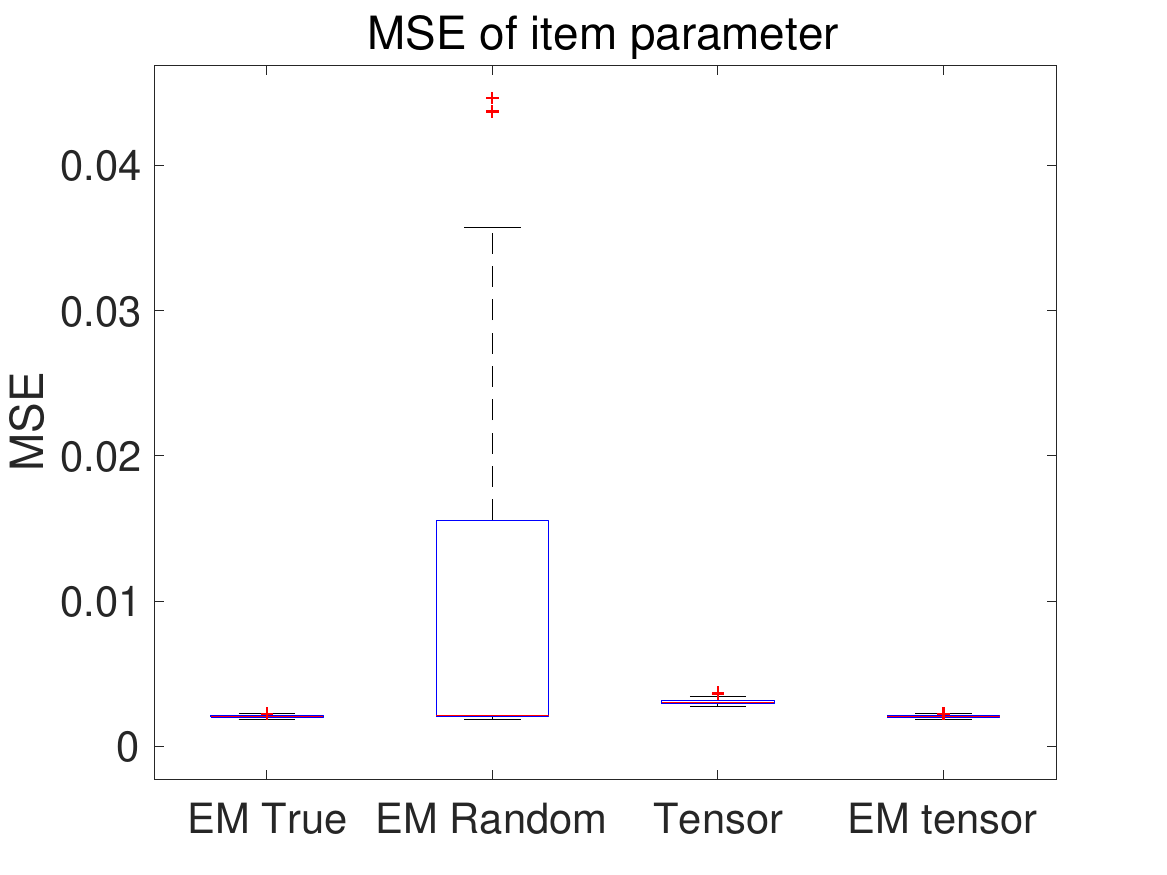}
		\end{minipage}%
	}%
	\subfigure[MSE without EM-random]{
		\begin{minipage}[t]{0.33\linewidth}
			\centering
			\includegraphics[width=2in]{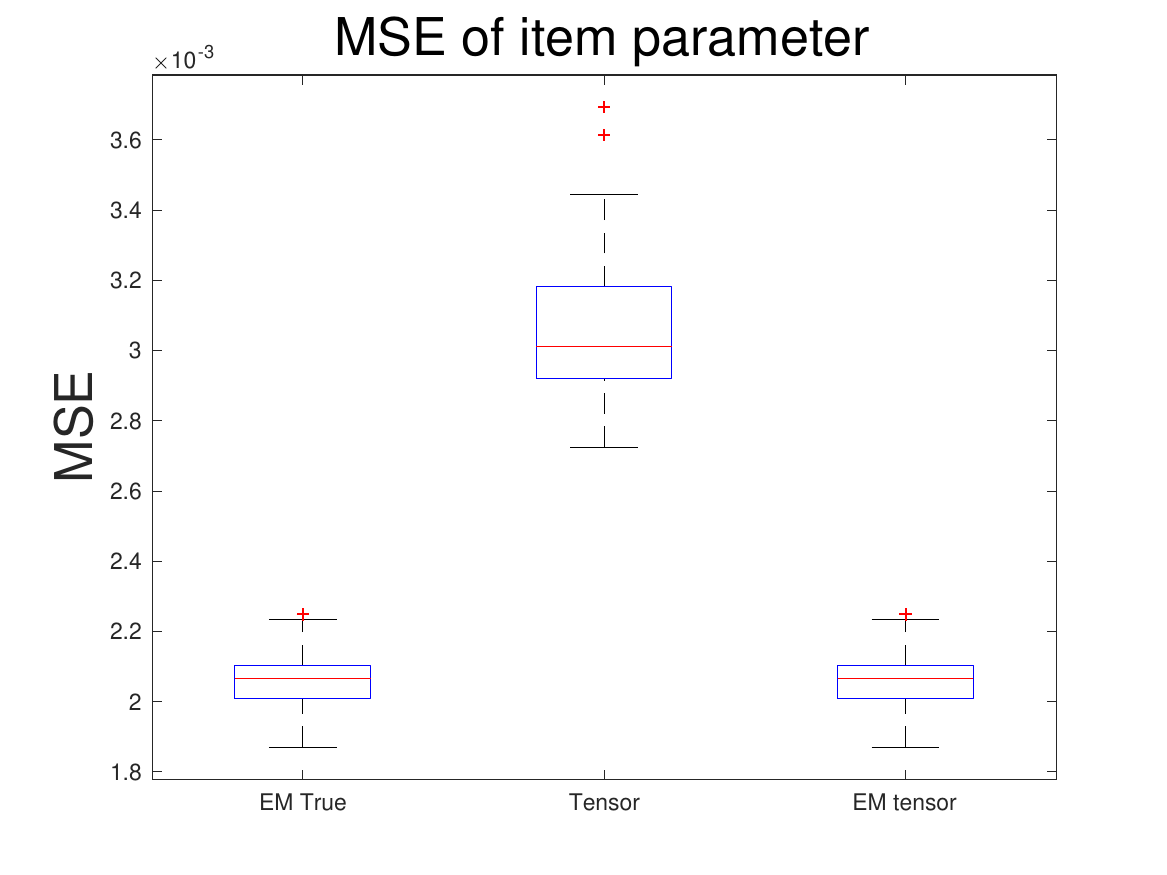}
		\end{minipage}%
	}%
	\subfigure[Running time of the algorithms]{
		\begin{minipage}[t]{0.33\linewidth}
			\centering
			\includegraphics[width=2in]{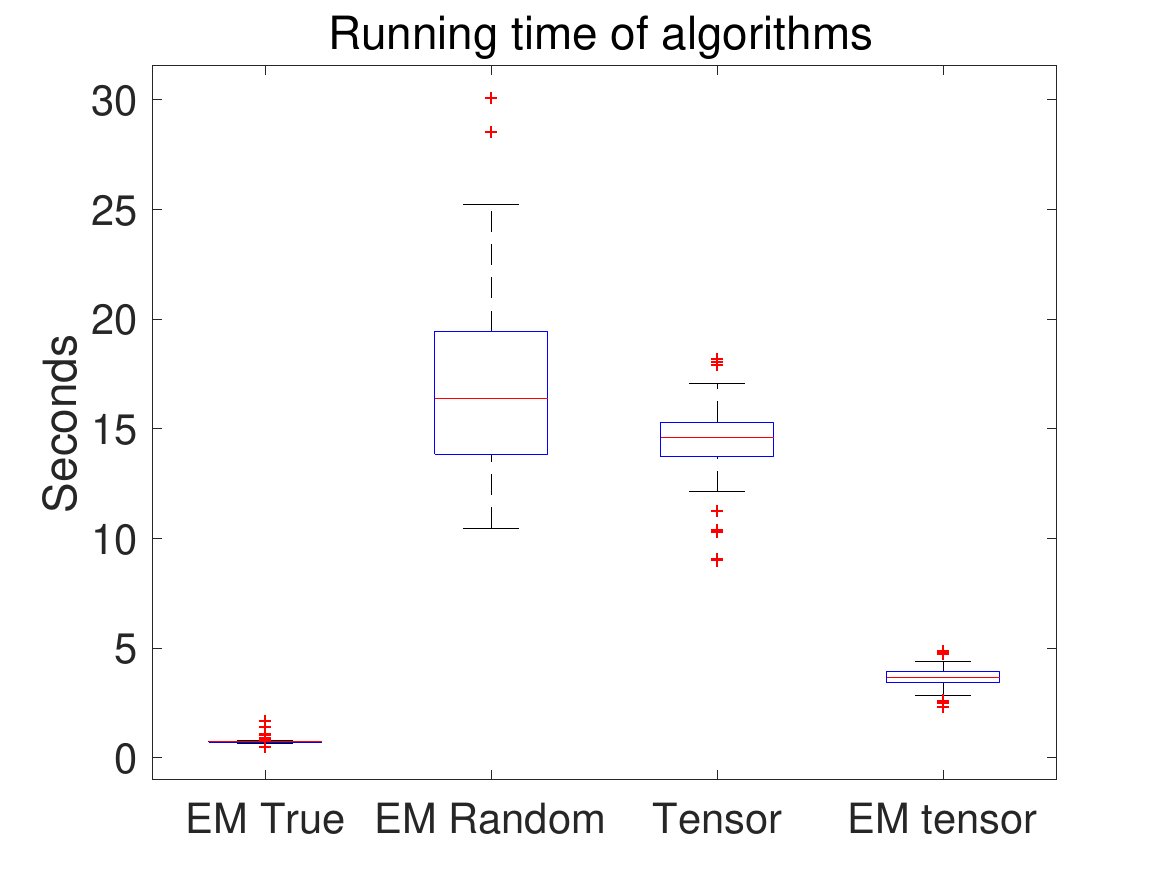}
		\end{minipage}%
	}%
	\centering
	\caption{$N = 1000, J= 200, L=10,$ item parameters $\in \{0.2,0.4,0.6,0.8\}$}
\end{figure}

\begin{figure}[H]
	\centering
	\subfigure[MSE of item parameters]{
		\begin{minipage}[t]{0.4\linewidth}
			\centering
			\includegraphics[width=2in]{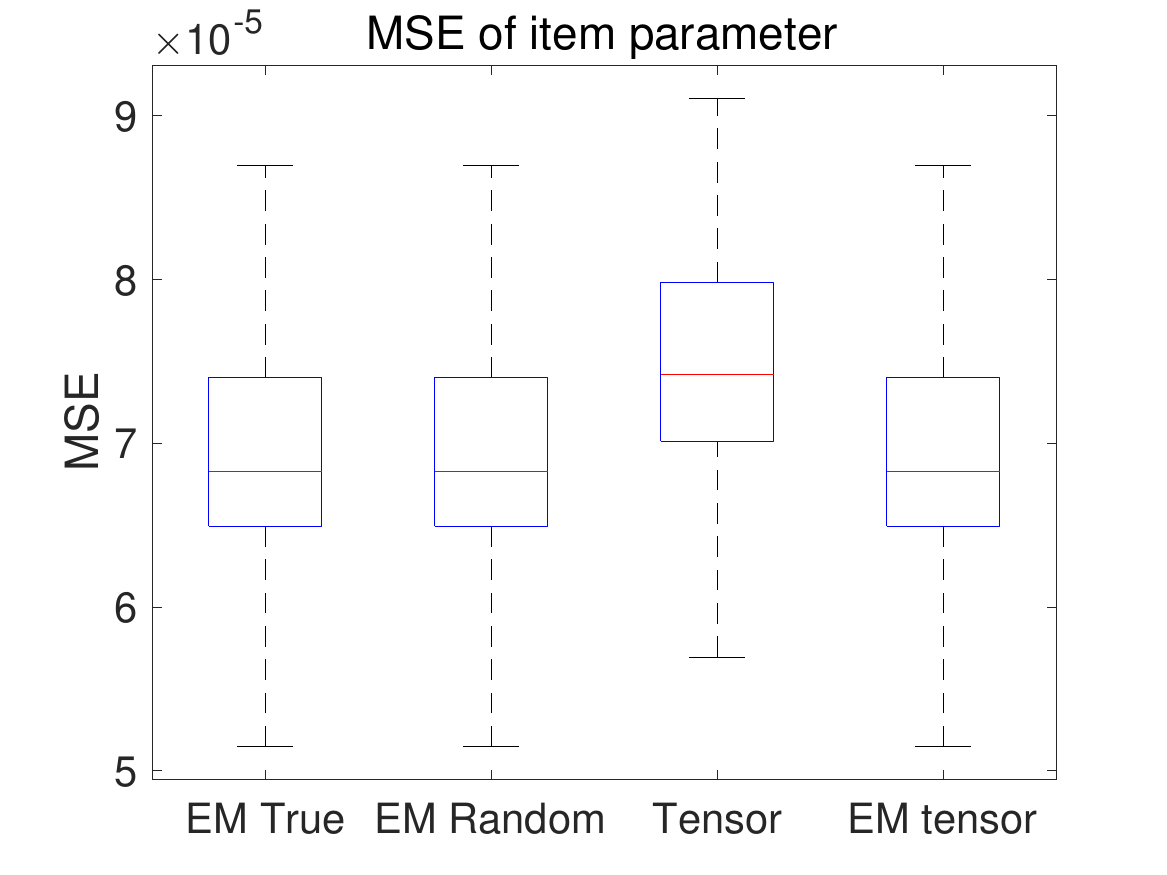}
		\end{minipage}%
	}%
	\subfigure[Running time of the algorithms]{
		\begin{minipage}[t]{0.4\linewidth}
			\centering
			\includegraphics[width=2in]{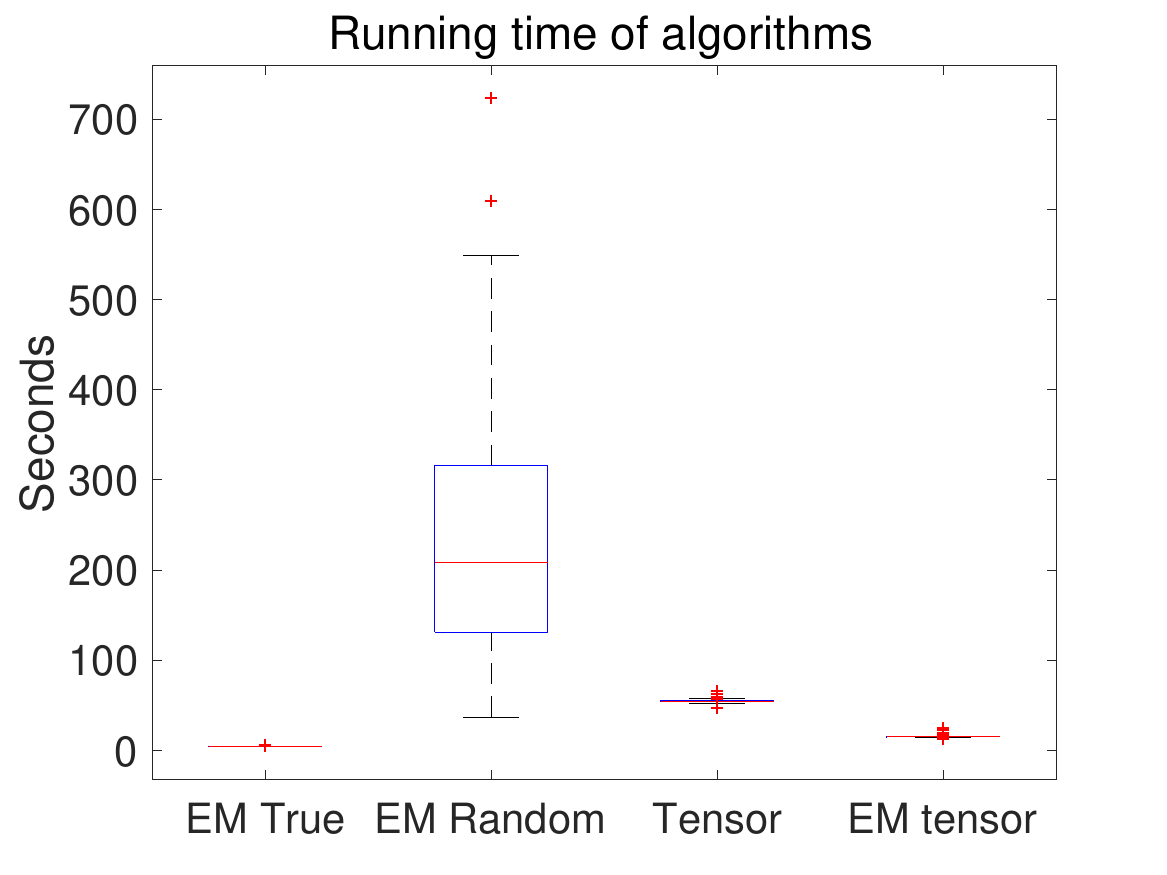}
		\end{minipage}%
	}%
	\centering
	\caption{$N = 10000, J= 100, L=5,$ item parameters $\in \{0.1,0.2,0.8,0.9\}$}
\end{figure}

\begin{figure}[H]
		\centering
		\subfigure[MSE of item parameters]{
			\begin{minipage}[t]{0.33\linewidth}
				\centering
				\includegraphics[width=2in]{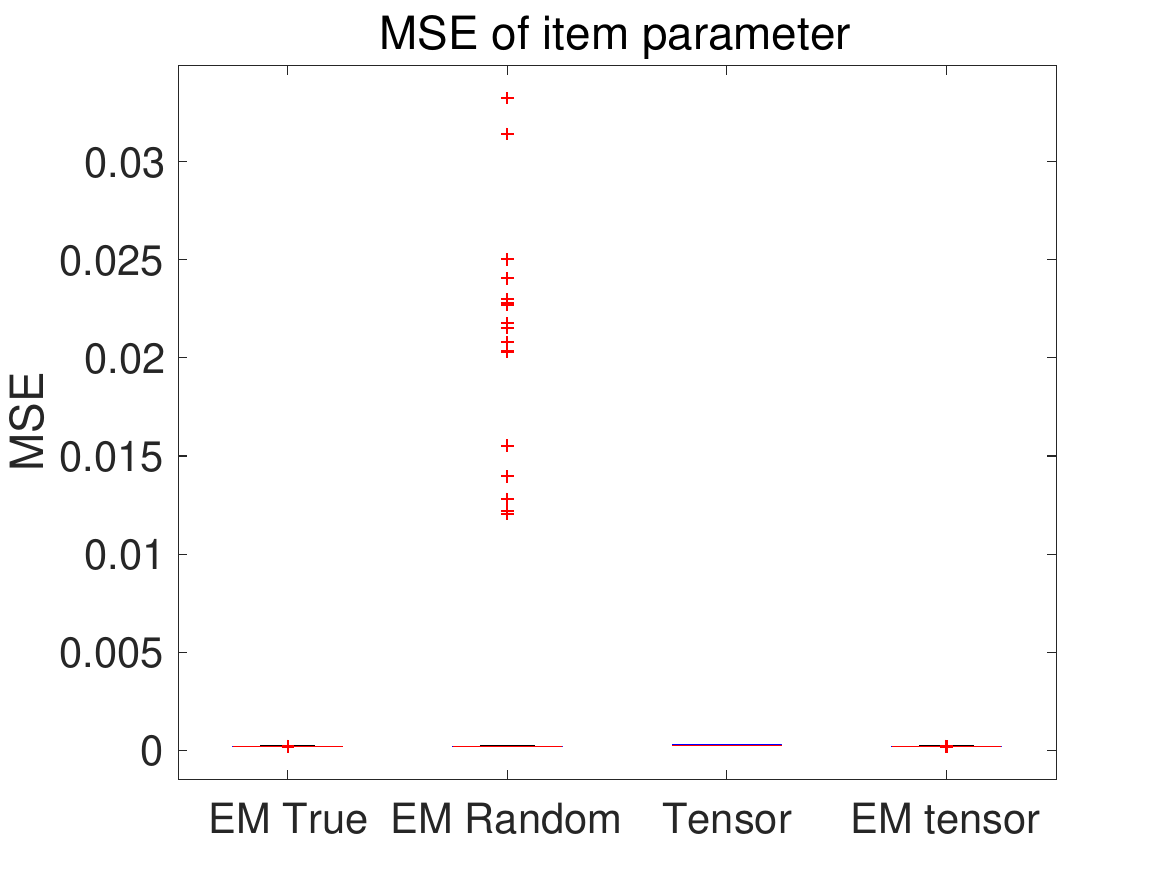}
			\end{minipage}%
		}%
		\subfigure[MSE without EM-random]{
			\begin{minipage}[t]{0.33\linewidth}
				\centering
				\includegraphics[width=2in]{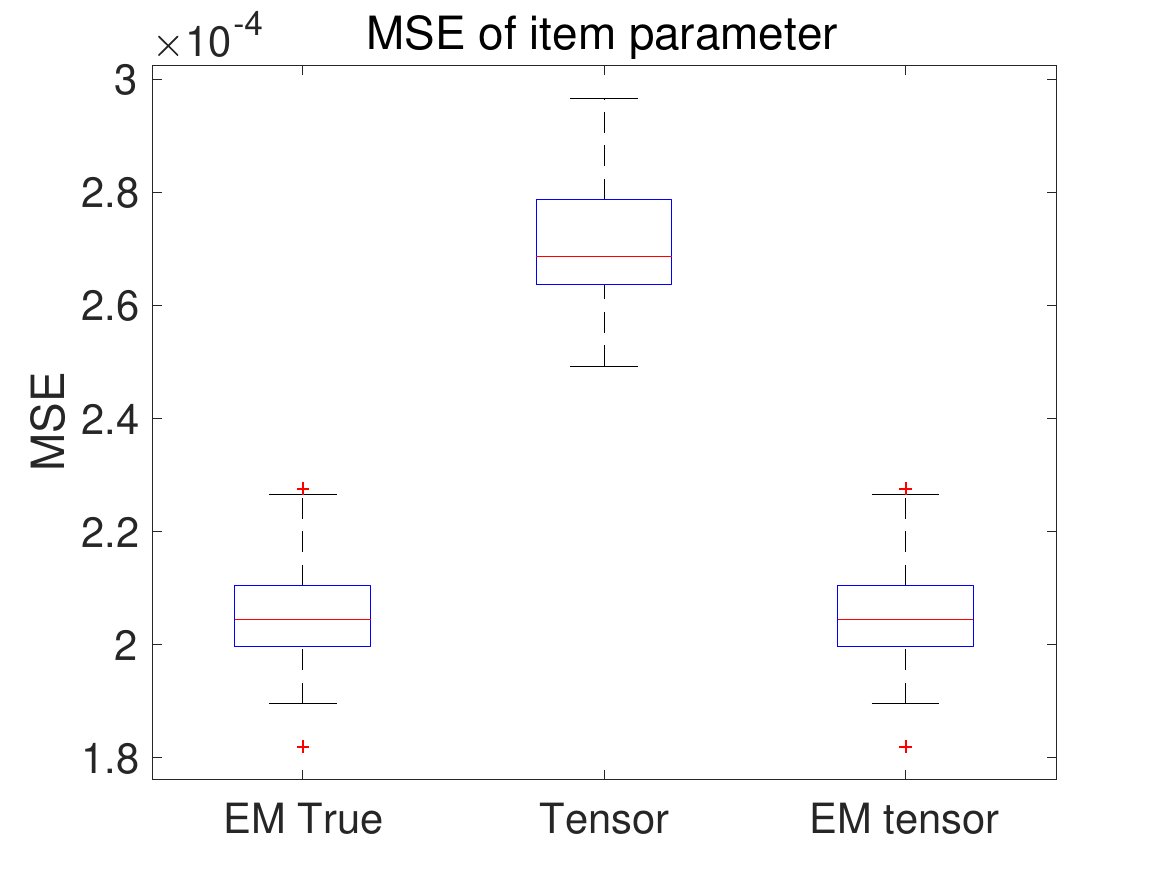}
			\end{minipage}%
		}%
		\subfigure[Running time of the algorithms]{
			\begin{minipage}[t]{0.33\linewidth}
				\centering
				\includegraphics[width=2in]{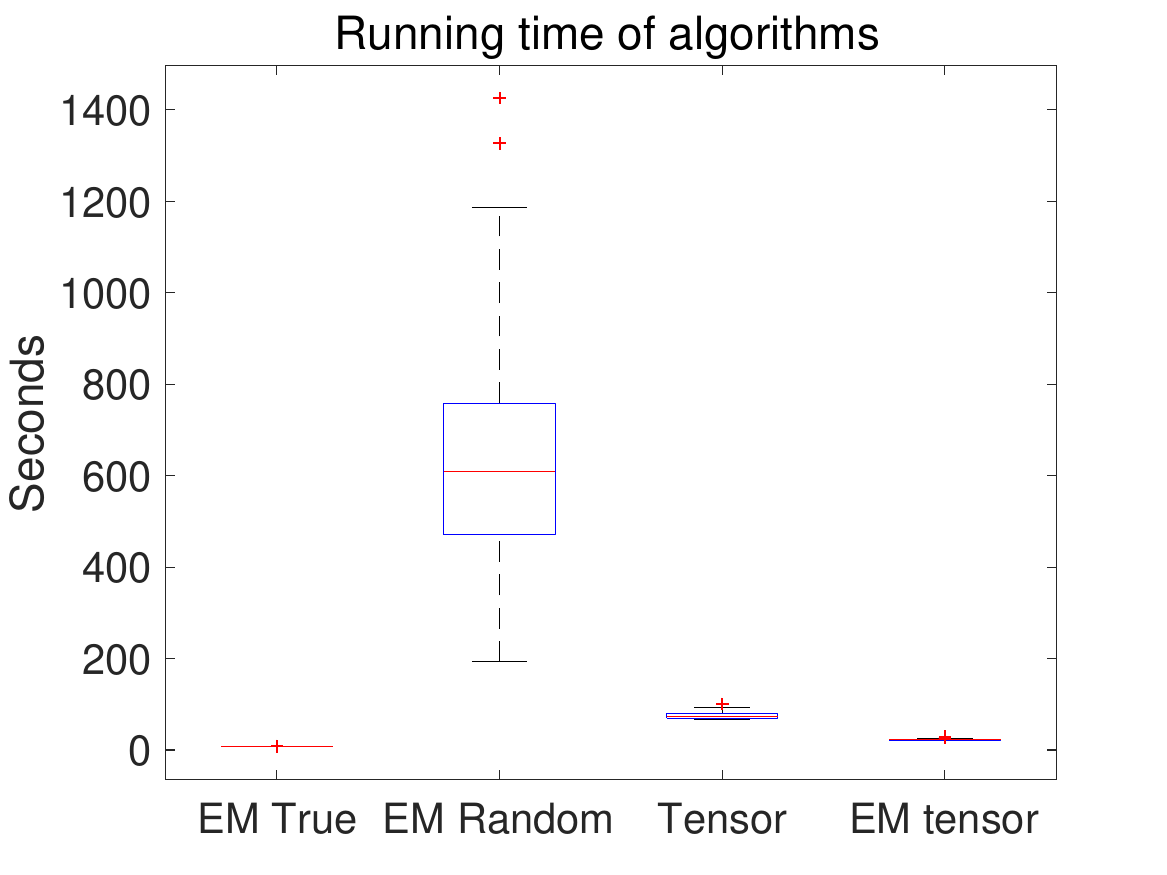}
			\end{minipage}%
		}%
		\centering
		\caption{$N = 10000, J= 200, L=10,$ item parameters $\in \{0.2,0.4,0.6,0.8\}$}
	\end{figure}

\begin{figure}[H]
	\centering
	\subfigure[MSE of item parameters]{
		\begin{minipage}[t]{0.4\linewidth}
			\centering
			\includegraphics[width=2in]{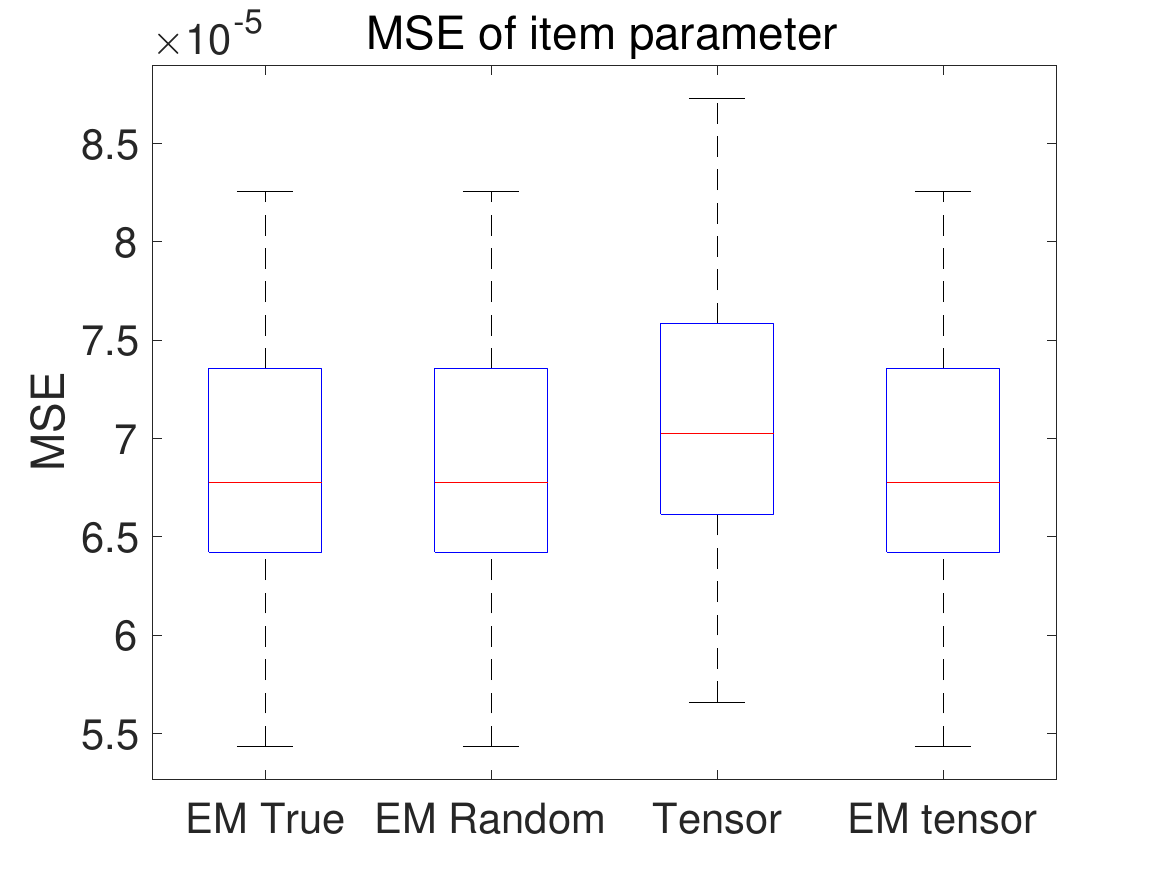}
		\end{minipage}%
	}%
	\subfigure[Running time of the algorithms]{
		\begin{minipage}[t]{0.4\linewidth}
			\centering
			\includegraphics[width=2in]{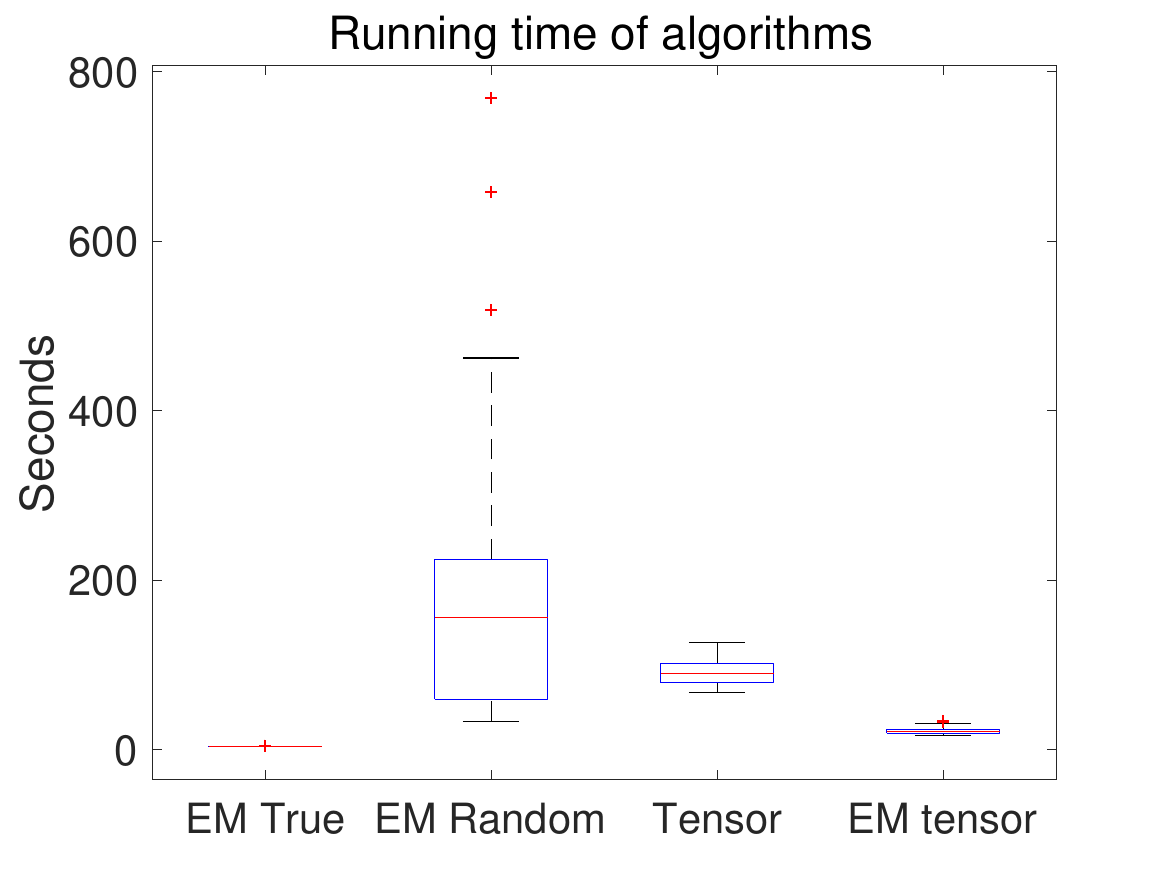}
		\end{minipage}%
	}%
	\centering
	\caption{$N = 10000, J= 200, L=5,$ item parameters $\in \{0.1,0.2,0.8,0.9\}$}
\end{figure}

\begin{figure}[H]
	\centering
	\subfigure[MSE of item parameters]{
		\begin{minipage}[t]{0.33\linewidth}
			\centering
			\includegraphics[width=2in]{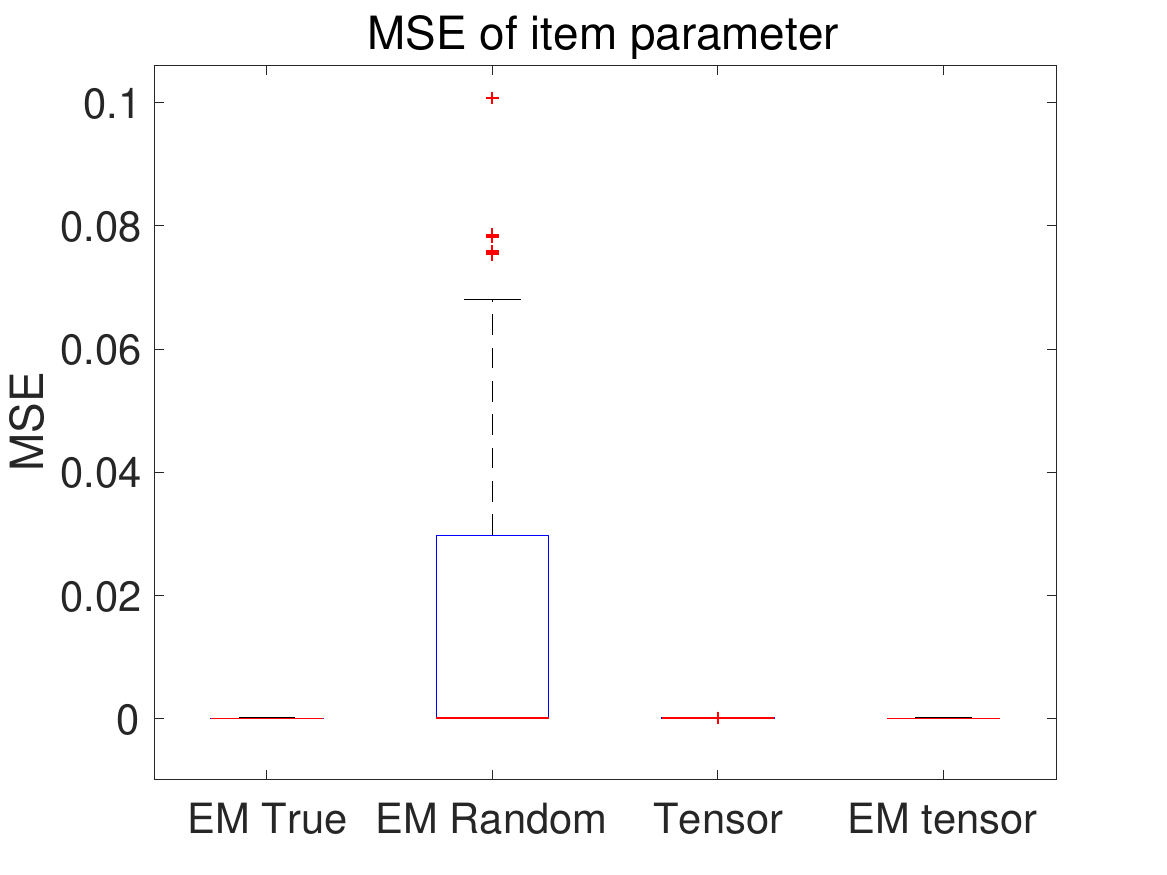}
		\end{minipage}%
	}%
		\subfigure[MSE without EM-random]{
	\begin{minipage}[t]{0.33\linewidth}
		\centering
		\includegraphics[width=2in]{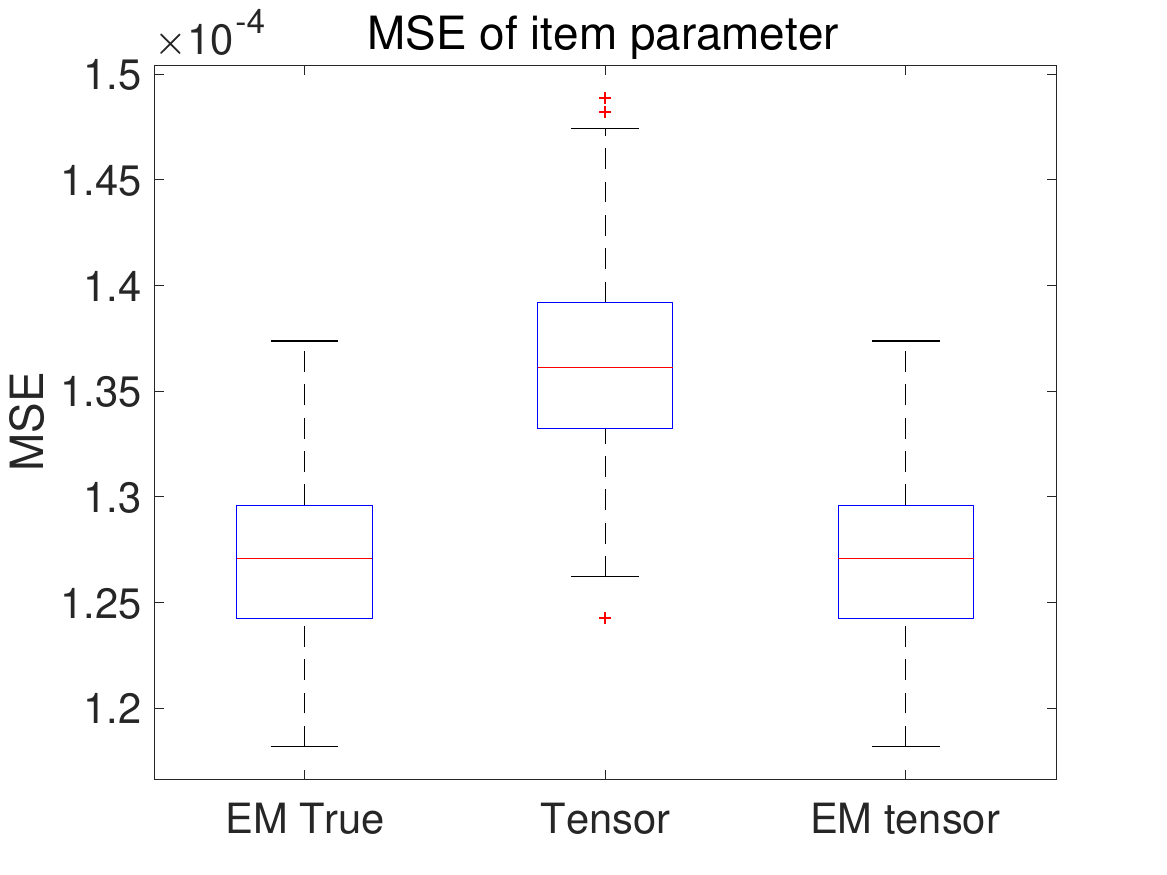}
	\end{minipage}%
}%
	\subfigure[Running time of the algorithms]{
		\begin{minipage}[t]{0.33\linewidth}
			\centering
			\includegraphics[width=2in]{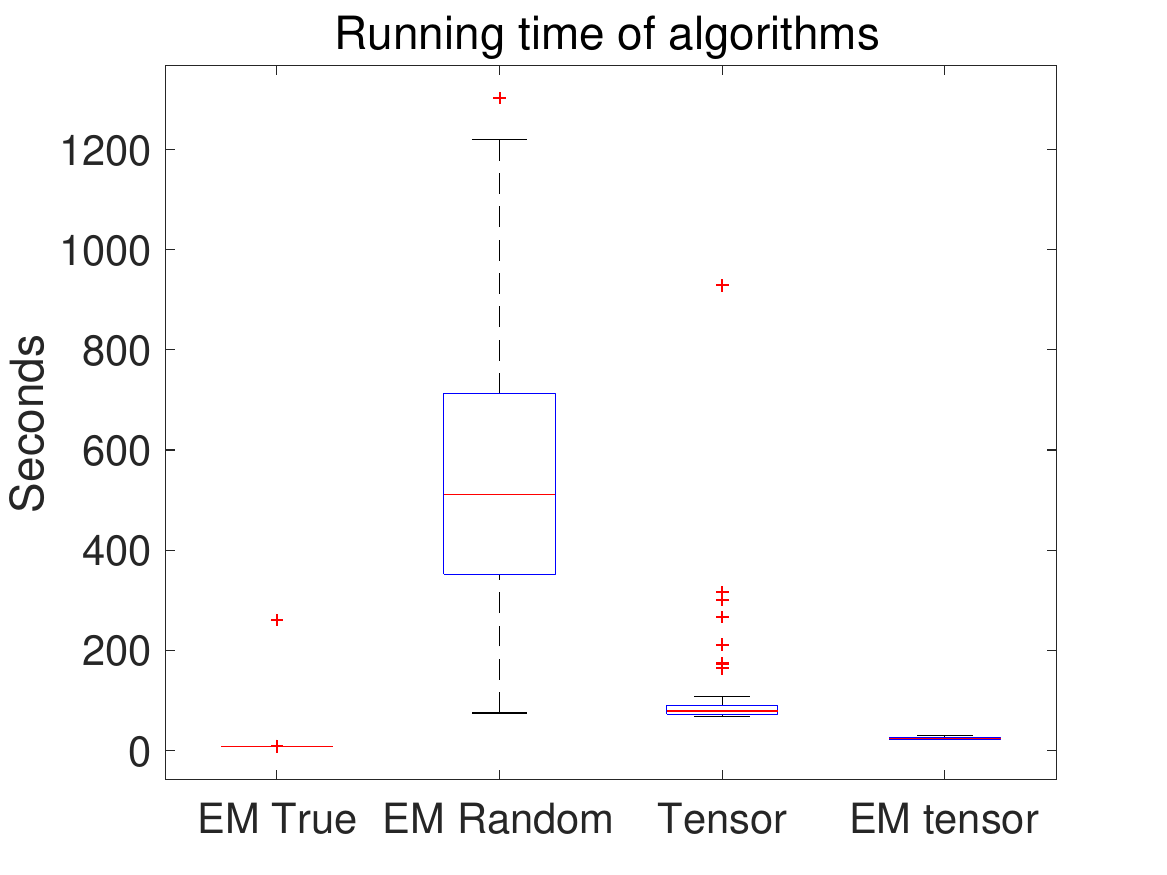}
		\end{minipage}%
	}%
	\centering
	\caption{$N = 10000, J= 200, L=10,$ item parameters $\in \{0.1,0.2,0.8,0.9\}$}
\end{figure}

\begin{figure}[H]
	\centering
	\subfigure[MSE of item parameters]{
		\begin{minipage}[t]{0.4\linewidth}
			\centering
			\includegraphics[width=2in]{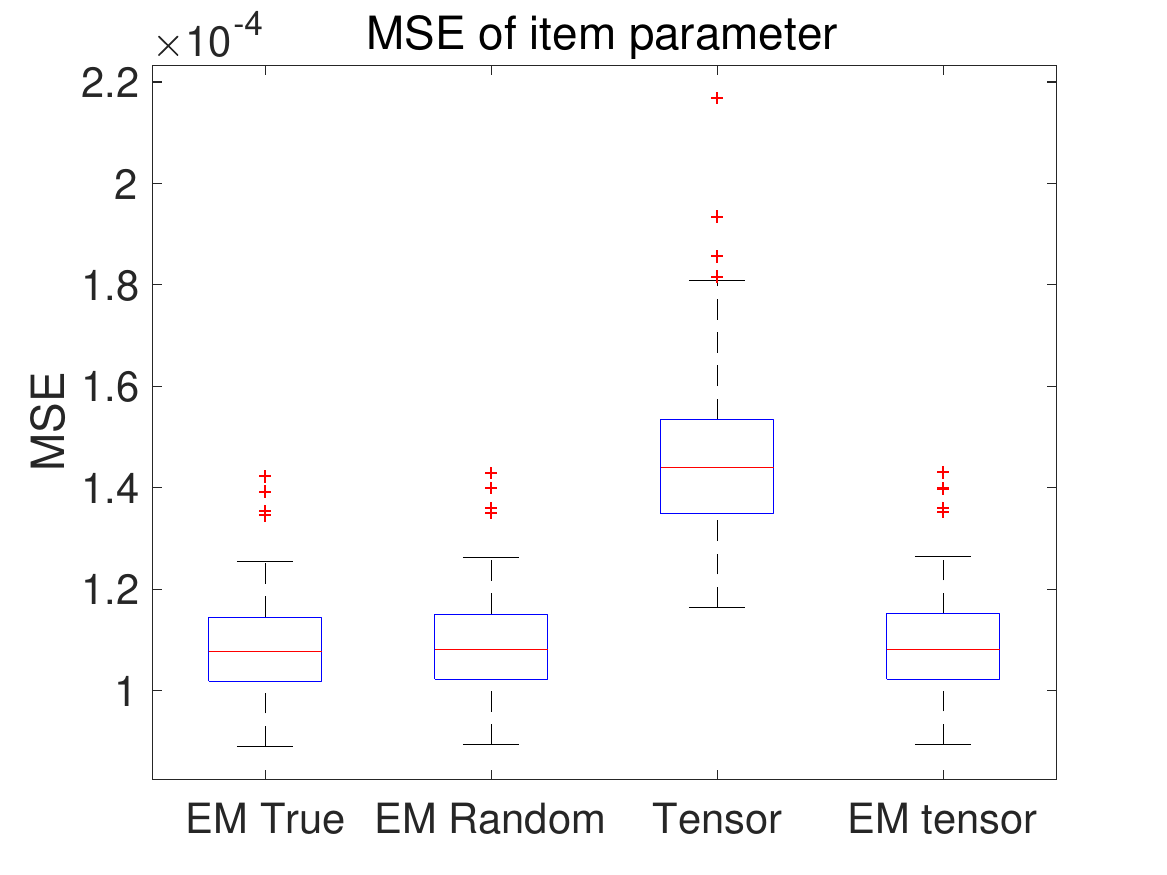}
		\end{minipage}%
	}%
	\subfigure[Running time of the algorithms]{
		\begin{minipage}[t]{0.4\linewidth}
			\centering
			\includegraphics[width=2in]{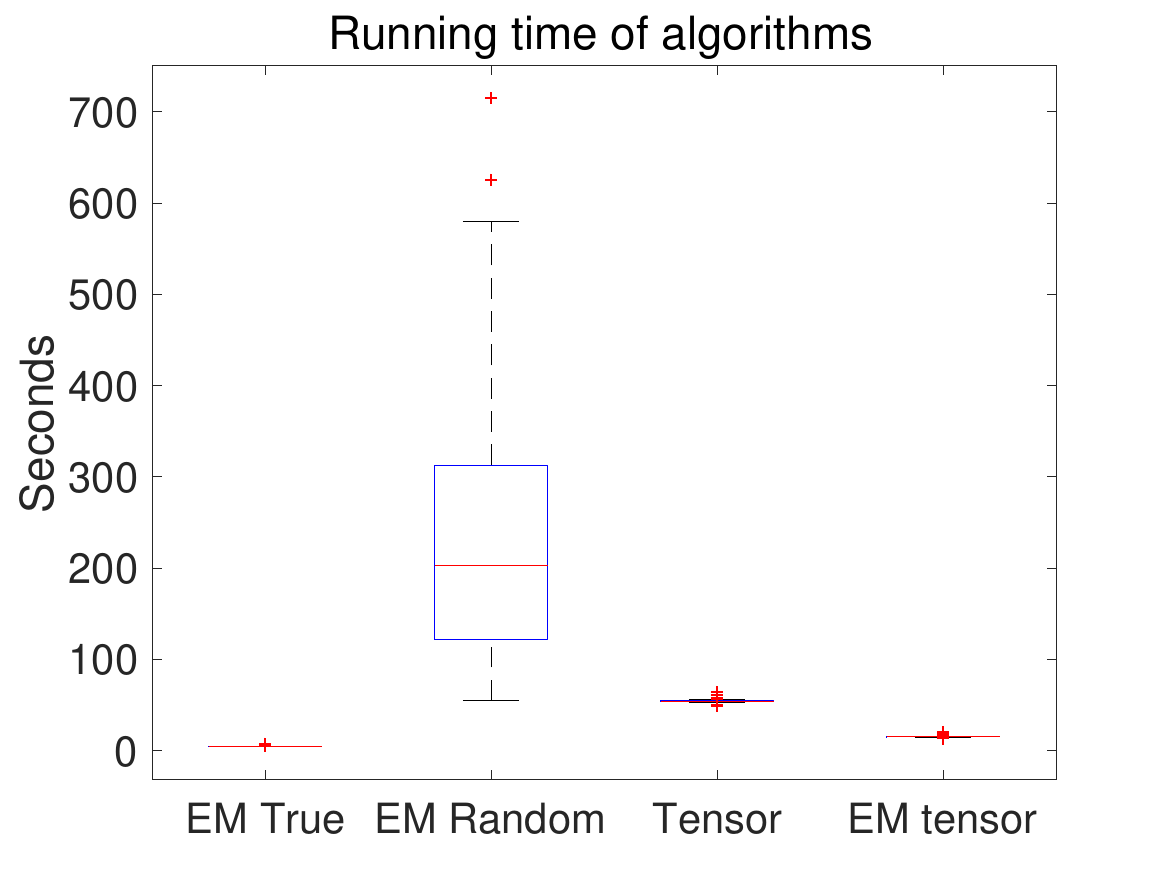}
		\end{minipage}%
	}%
	\centering
	\caption{$N = 10000, J= 100, L=5,$ item parameters $\in \{0.2,0.4,0.6,0.8\}$}
\end{figure}

\begin{figure}[H]
	\centering
	\subfigure[MSE of item parameters]{
		\begin{minipage}[t]{0.33\linewidth}
			\centering
			\includegraphics[width=2in]{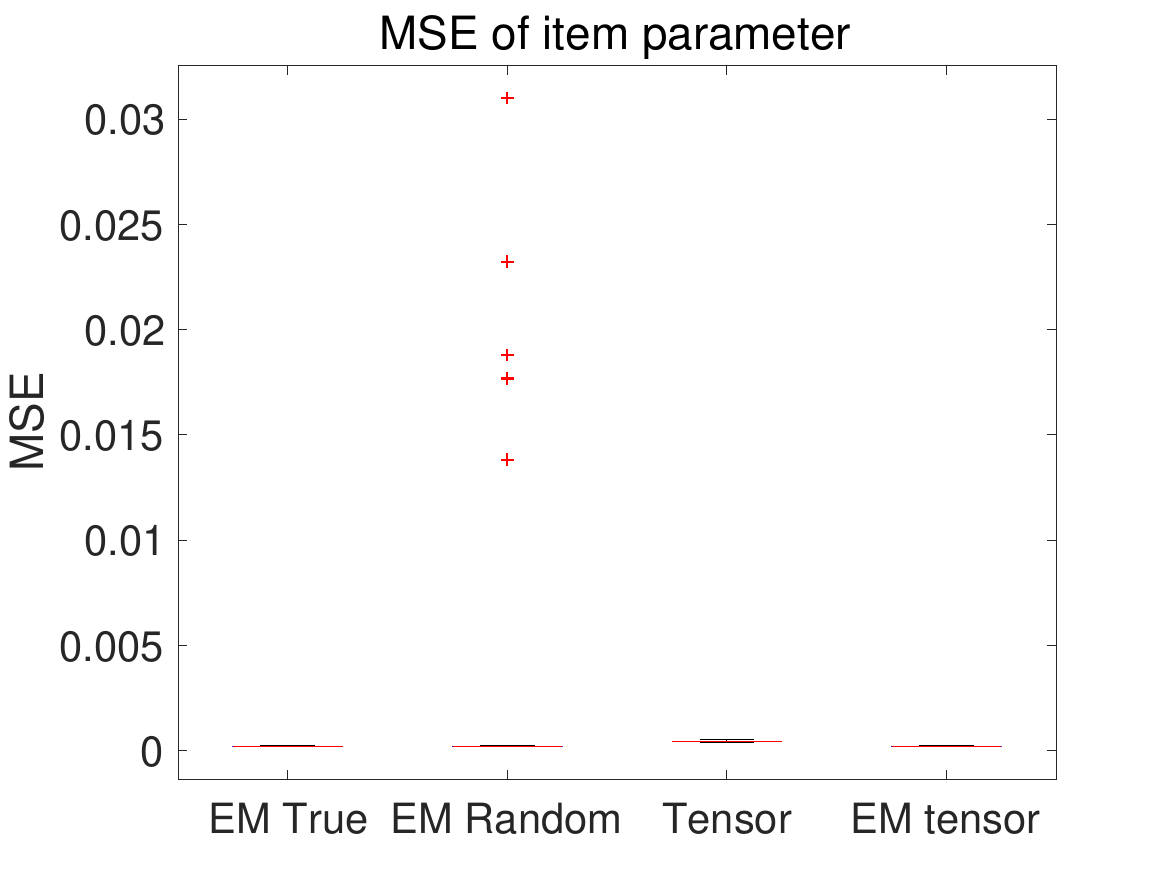}
		\end{minipage}%
	}%
		\subfigure[MSE without EM-random]{
	\begin{minipage}[t]{0.33\linewidth}
		\centering
		\includegraphics[width=2in]{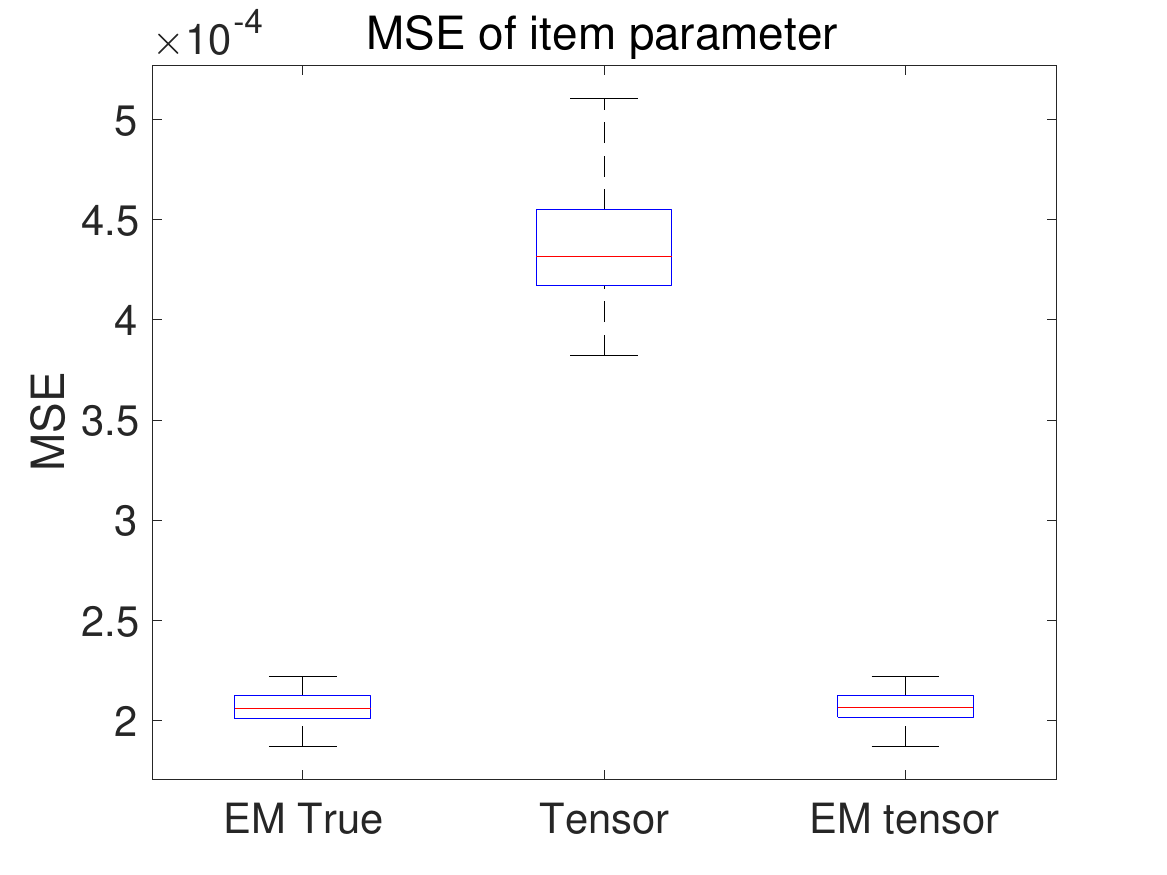}
	\end{minipage}%
}%
	\subfigure[Running time of the algorithms]{
		\begin{minipage}[t]{0.33\linewidth}
			\centering
			\includegraphics[width=2in]{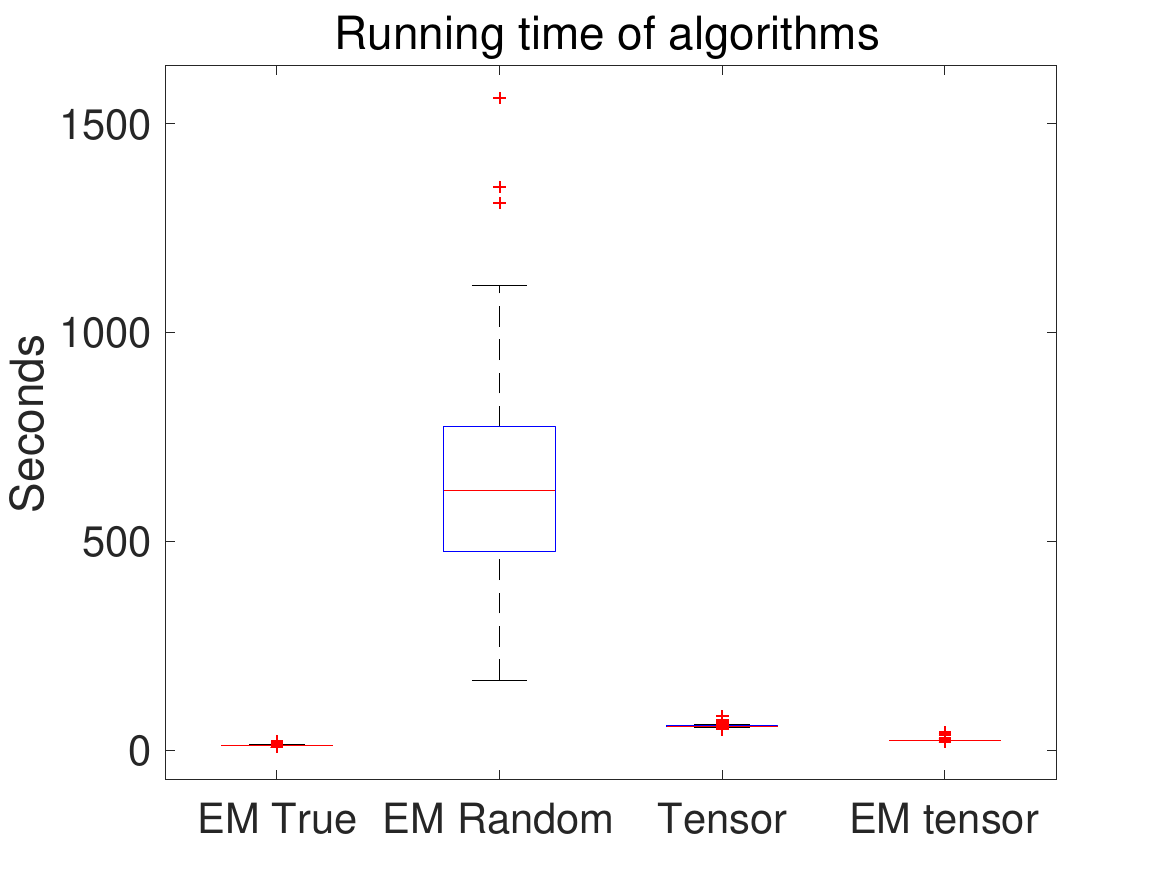}
		\end{minipage}%
	}%
	\centering
	\caption{$N = 10000, J= 100, L=10,$ item parameters $\in \{0.2,0.4,0.6,0.8\}$}
\end{figure}

\begin{figure}[H]
	\centering
	\subfigure[MSE of item parameters]{
		\begin{minipage}[t]{0.4\linewidth}
			\centering
			\includegraphics[width=2in]{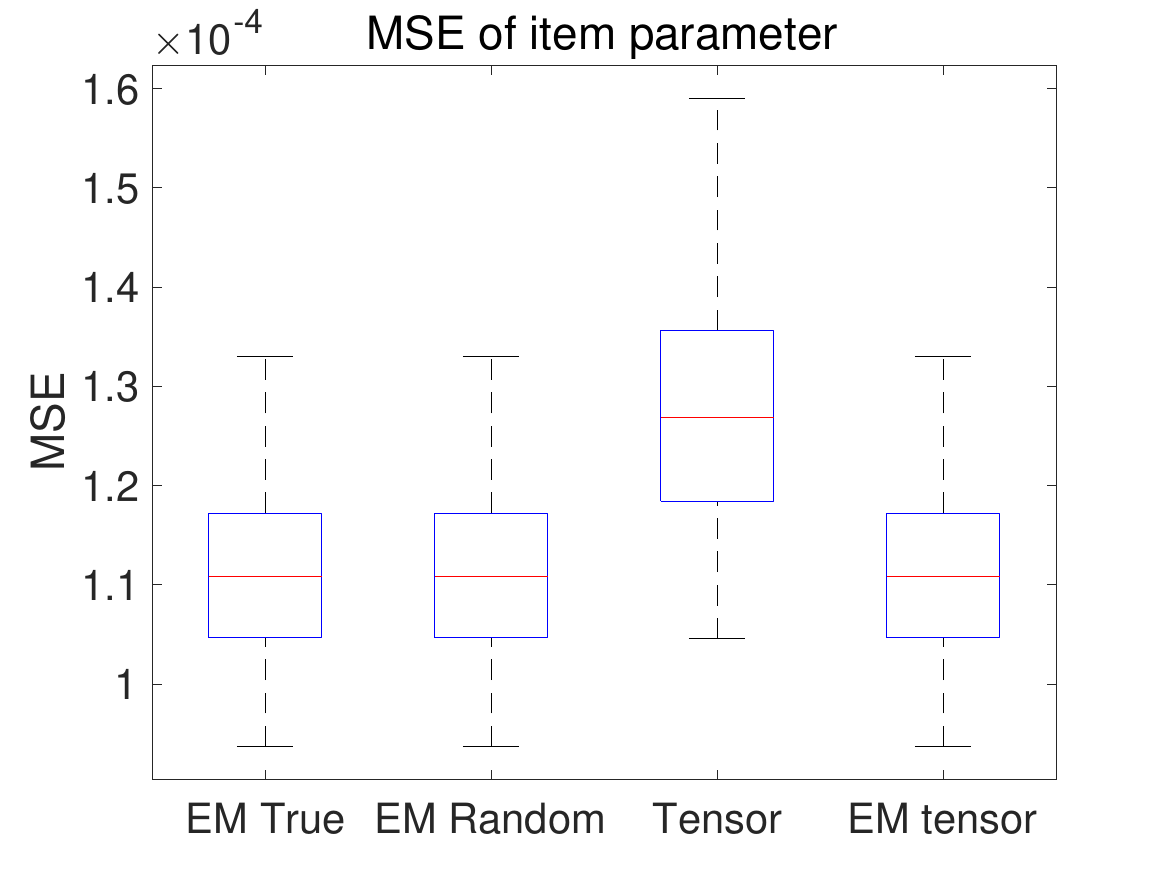}
		\end{minipage}%
	}%
	\subfigure[Running time of the algorithms]{
		\begin{minipage}[t]{0.4\linewidth}
			\centering
			\includegraphics[width=2in]{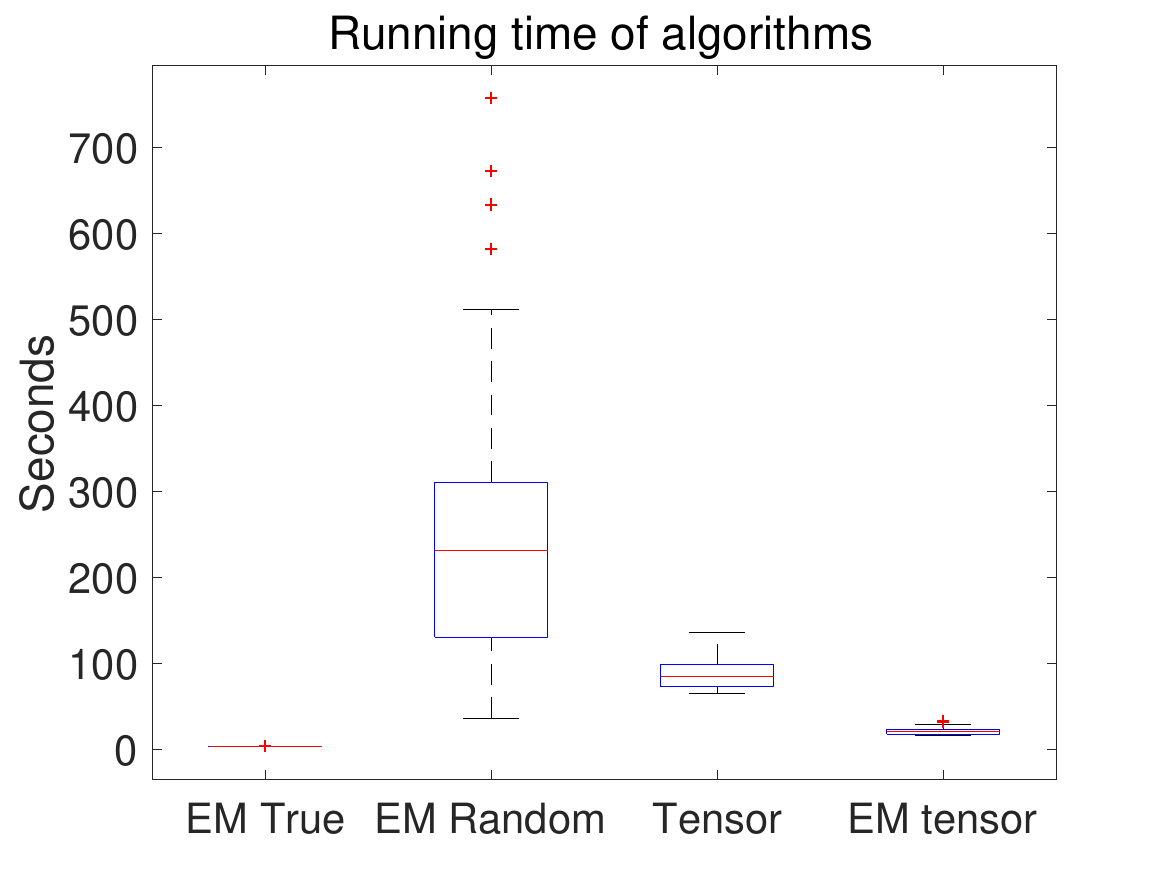}
		\end{minipage}%
	}%
	\centering
	\caption{$N = 10000, J= 200, L=5,$ item parameters $\in \{0.2,0.4,0.6,0.8\}$}
\end{figure}

\begin{figure}[H]
	\centering
	\subfigure[MSE of item parameters]{
		\begin{minipage}[t]{0.4\linewidth}
			\centering
			\includegraphics[width=2in]{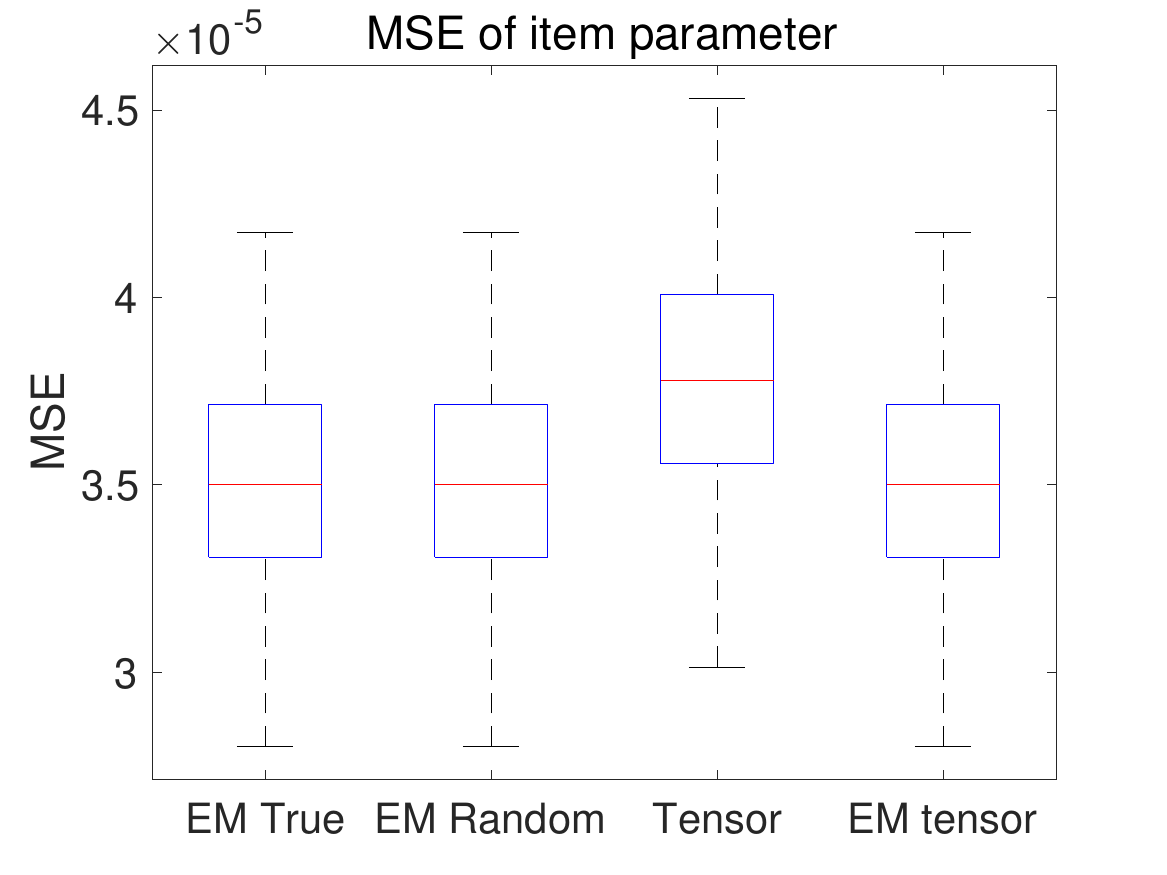}
		\end{minipage}%
	}%
	\subfigure[Running time of the algorithms]{
		\begin{minipage}[t]{0.4\linewidth}
			\centering
			\includegraphics[width=2in]{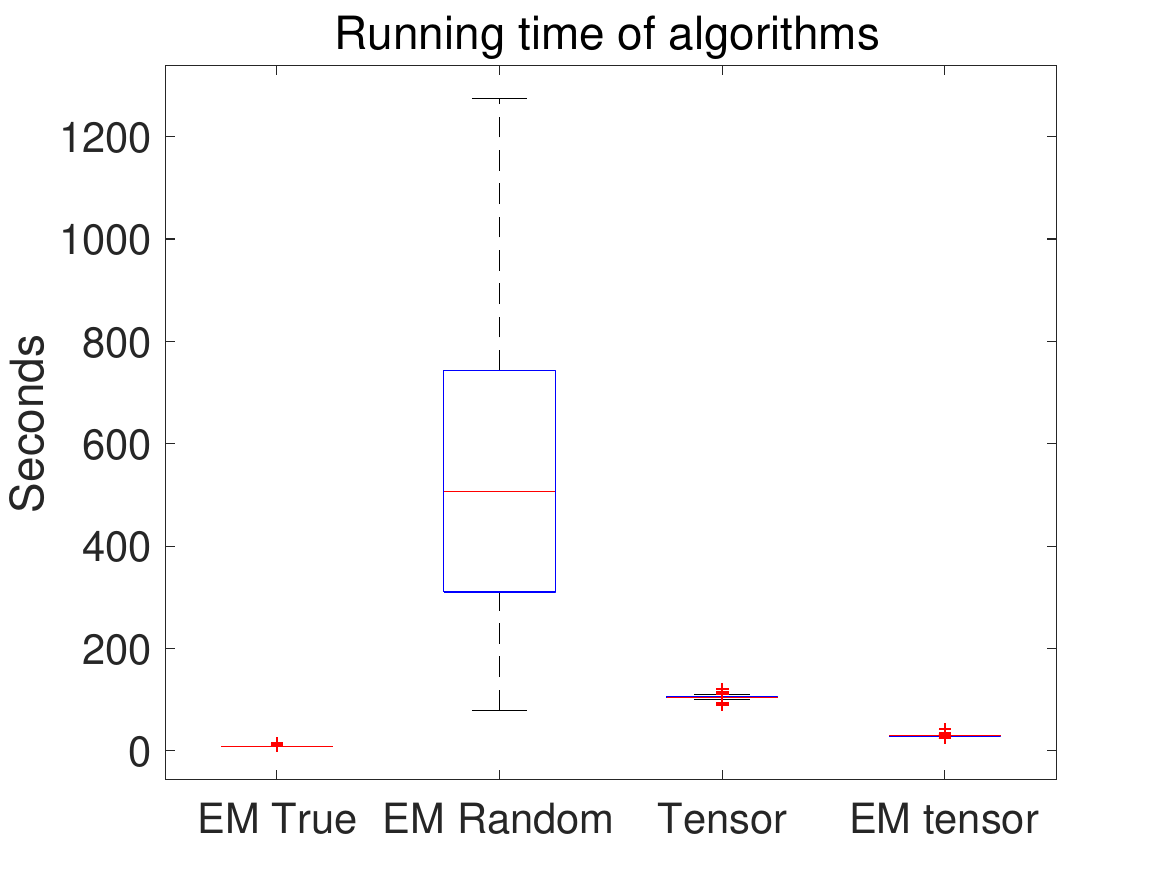}
		\end{minipage}%
	}%
	\centering
	\caption{$N = 20000, J= 100, L=5,$ item parameters $\in \{0.1,0.2,0.8,0.9\}$}
\end{figure}

\begin{figure}[H]
	\centering
	\subfigure[MSE of item parameters]{
		\begin{minipage}[t]{0.33\linewidth}
			\centering
			\includegraphics[width=2in]{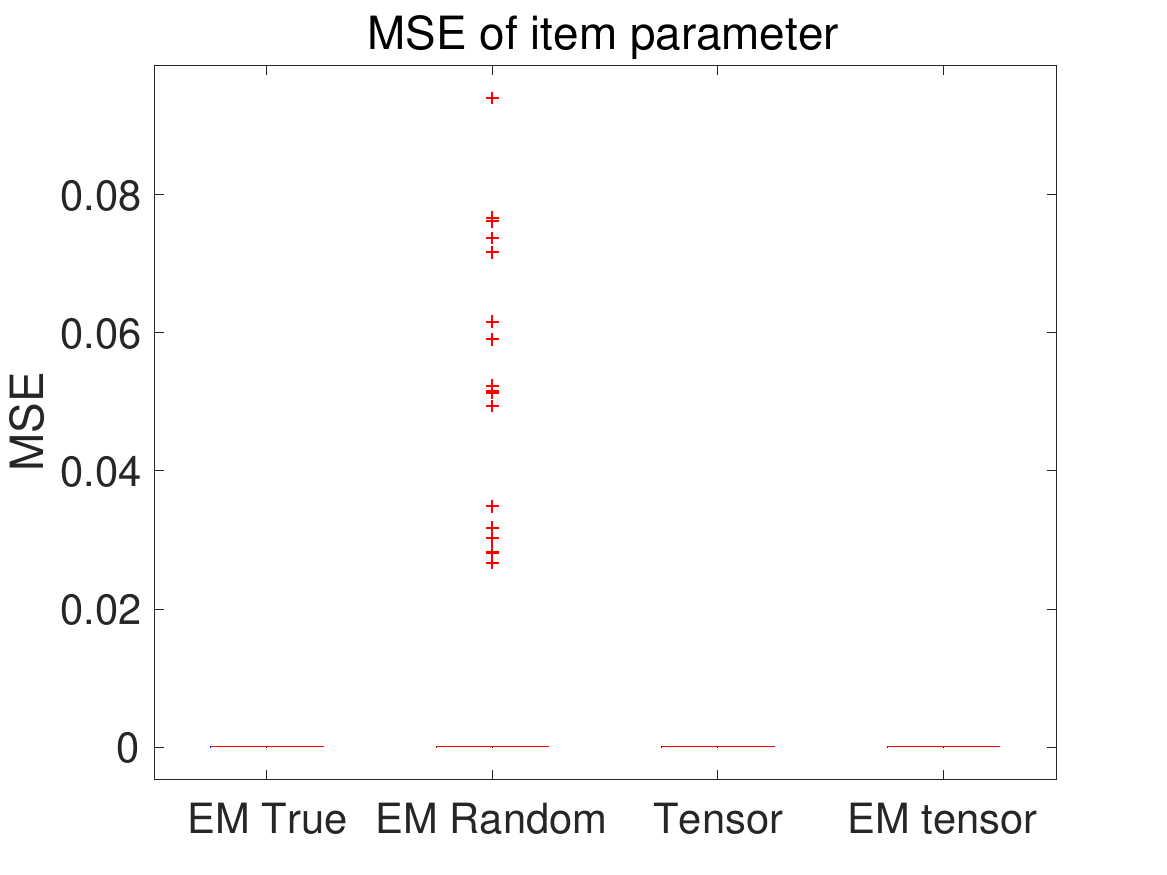}
		\end{minipage}%
	}%
		\subfigure[MSE without EM-random]{
	\begin{minipage}[t]{0.33\linewidth}
		\centering
		\includegraphics[width=2in]{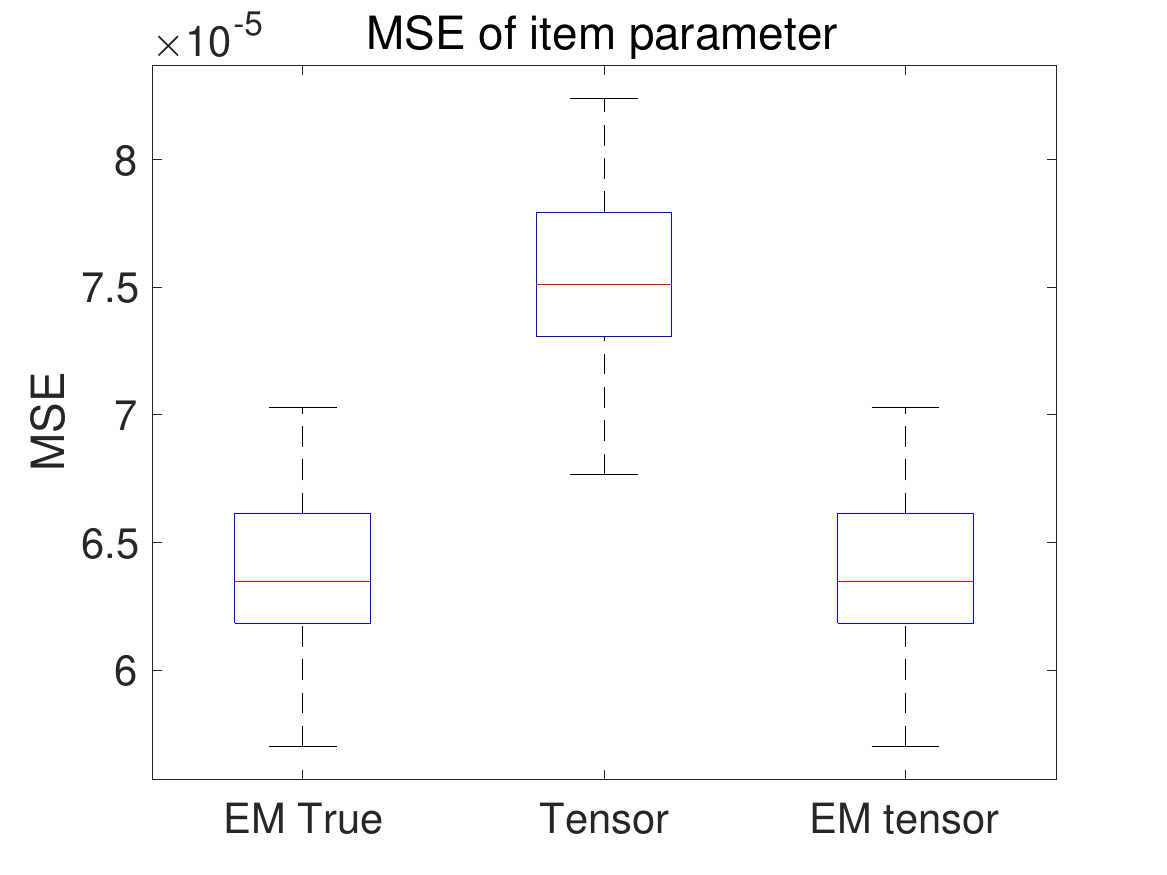}
	\end{minipage}%
}%
	\subfigure[Running time of the algorithms]{
		\begin{minipage}[t]{0.33\linewidth}
			\centering
			\includegraphics[width=2in]{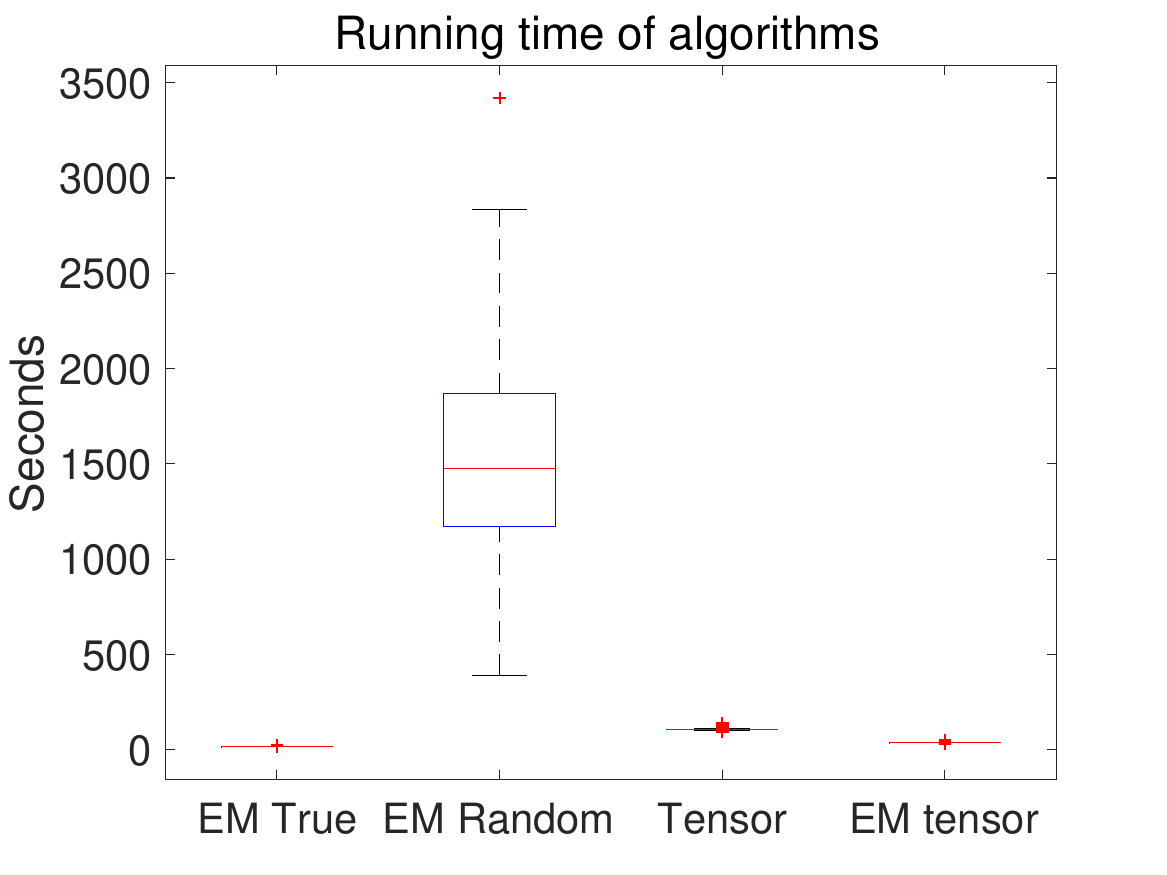}
		\end{minipage}%
	}%
	\centering
	\caption{$N = 20000, J= 100, L=10, $item parameters $\in \{0.1,0.2,0.8,0.9\}$}
\end{figure}

\begin{figure}[H]
	\centering
	\subfigure[MSE of item parameters]{
		\begin{minipage}[t]{0.4\linewidth}
			\centering
			\includegraphics[width=2in]{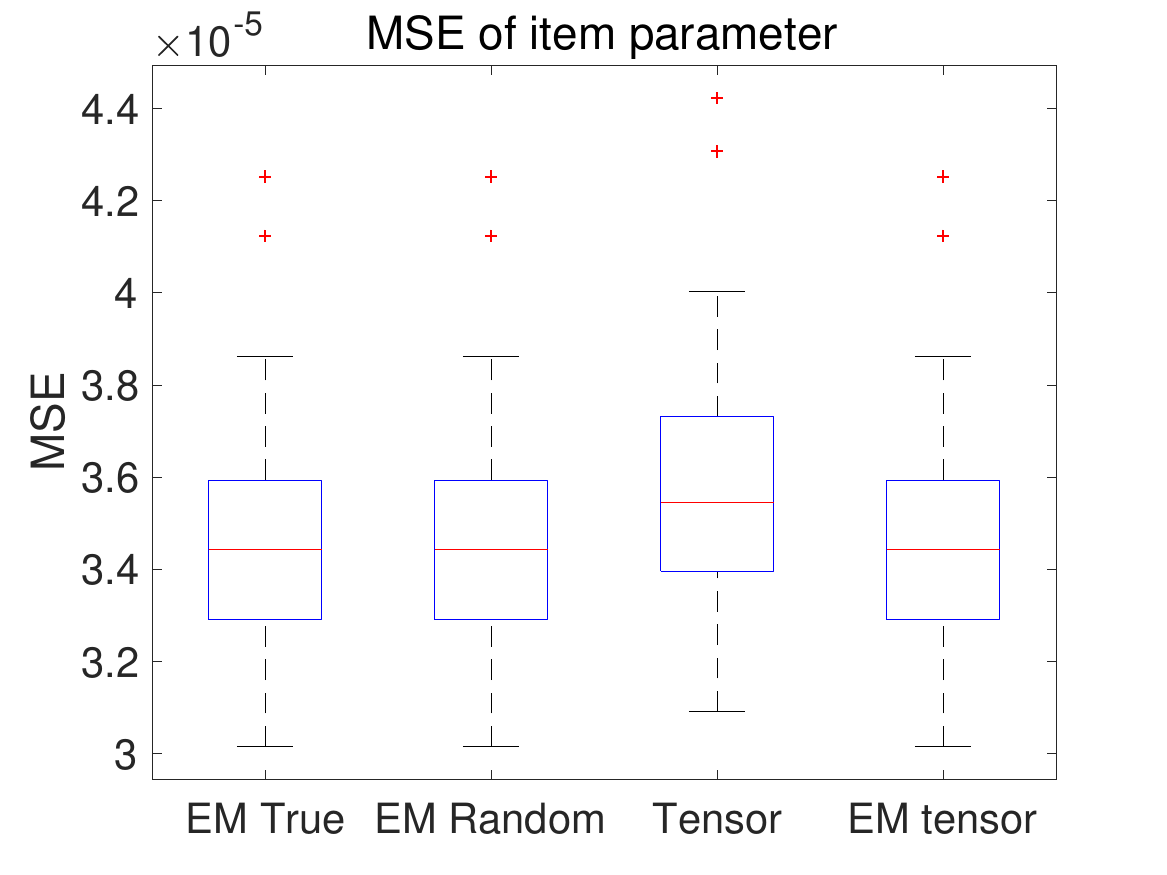}
		\end{minipage}%
	}%
	\subfigure[Running time of the algorithms]{
		\begin{minipage}[t]{0.4\linewidth}
			\centering
			\includegraphics[width=2in]{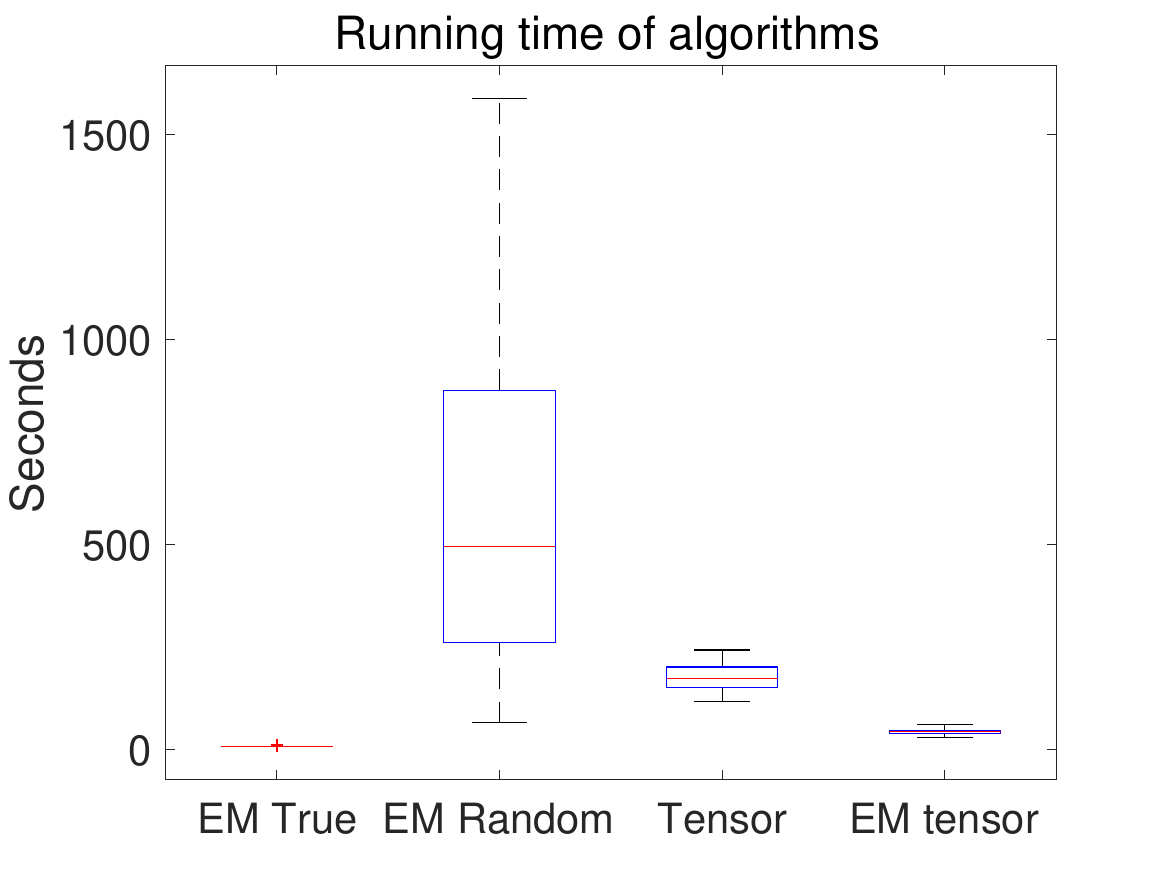}
		\end{minipage}%
	}%
	\centering
	\caption{$N = 20000, J= 200, L=5,$ item parameters $\in \{0.1,0.2,0.8,0.9\}$}
\end{figure}

\begin{figure}[H]
	\centering
	\subfigure[MSE of item parameters]{
		\begin{minipage}[t]{0.4\linewidth}
			\centering
			\includegraphics[width=2in]{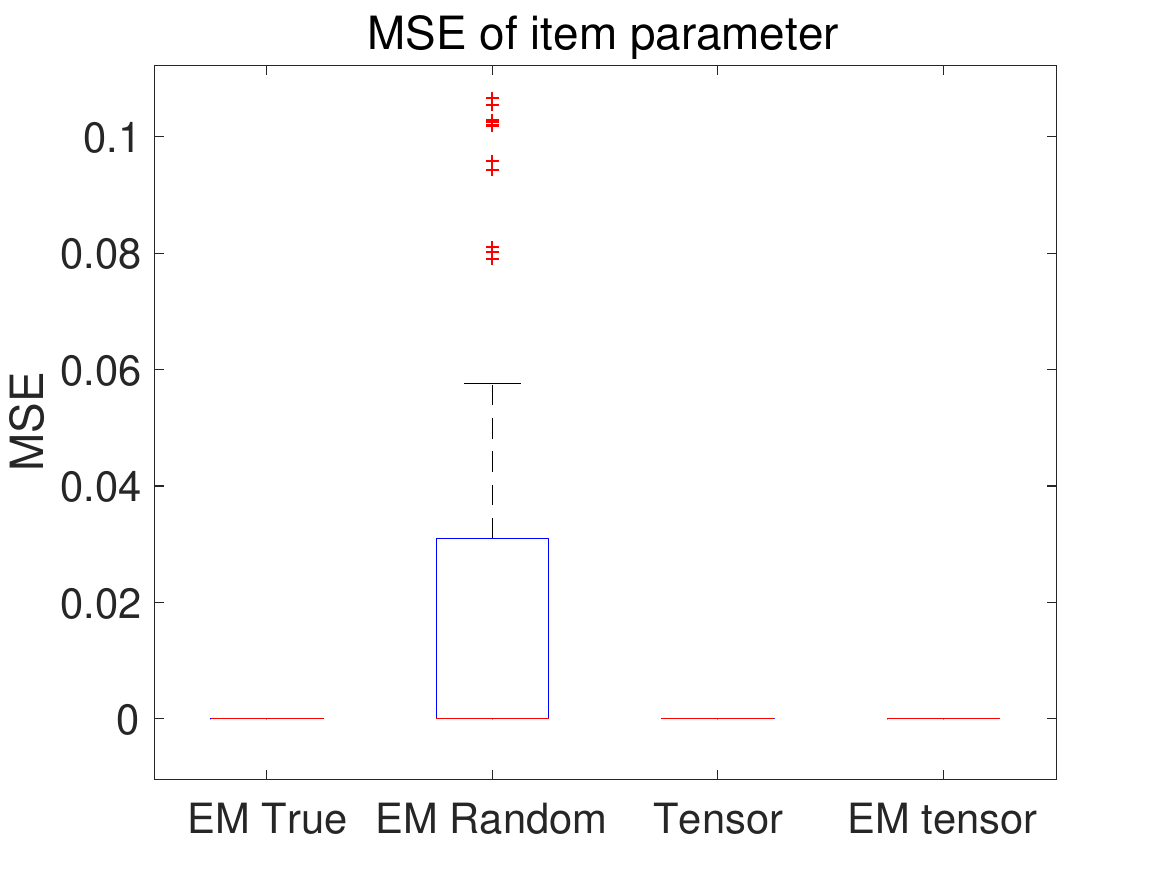}
		\end{minipage}%
	}%
	\subfigure[Running time of the algorithms]{
		\begin{minipage}[t]{0.4\linewidth}
			\centering
			\includegraphics[width=2in]{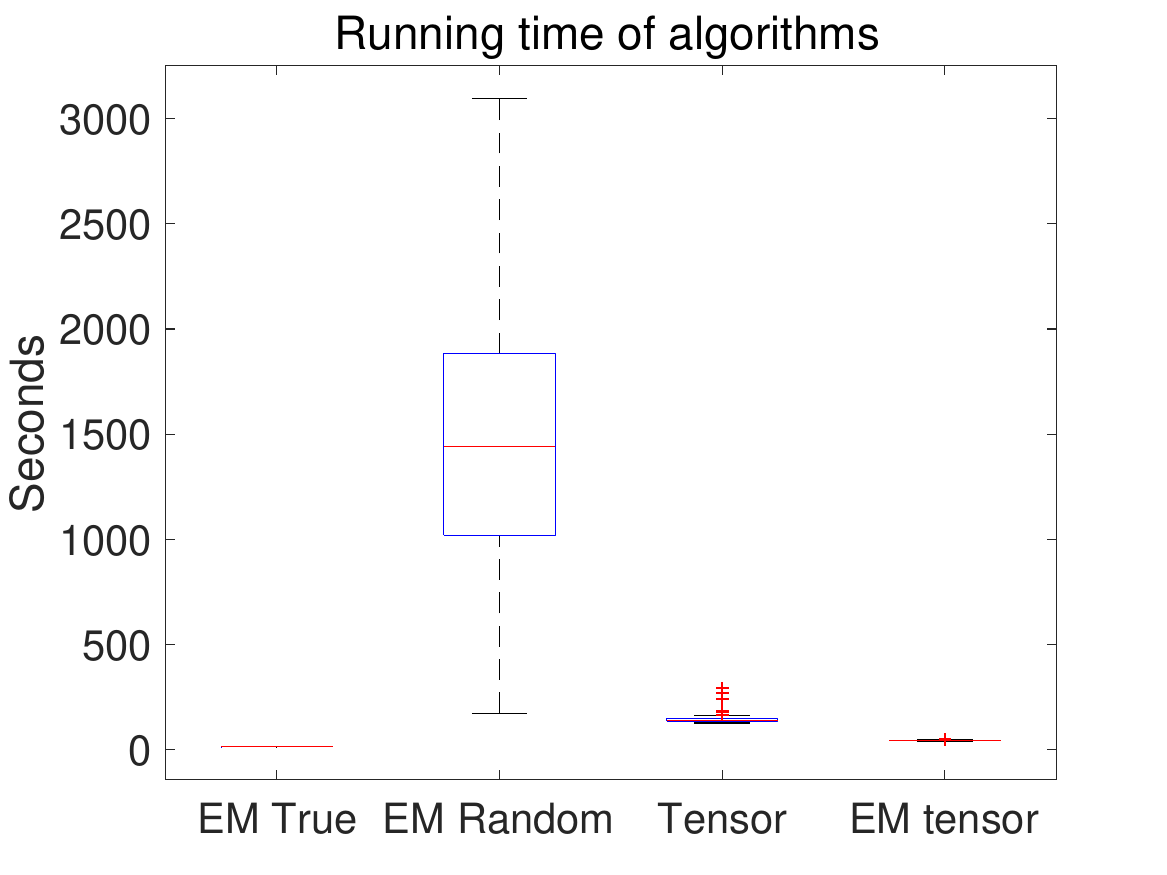}
		\end{minipage}%
	}%
	\centering
	\caption{$N = 20000, J= 200, L=10,$ item parameters $\in \{0.1,0.2,0.8,0.9\}$}
\end{figure}

\begin{figure}[H]
	\centering
	\subfigure[MSE of item parameters]{
		\begin{minipage}[t]{0.4\linewidth}
			\centering
			\includegraphics[width=2in]{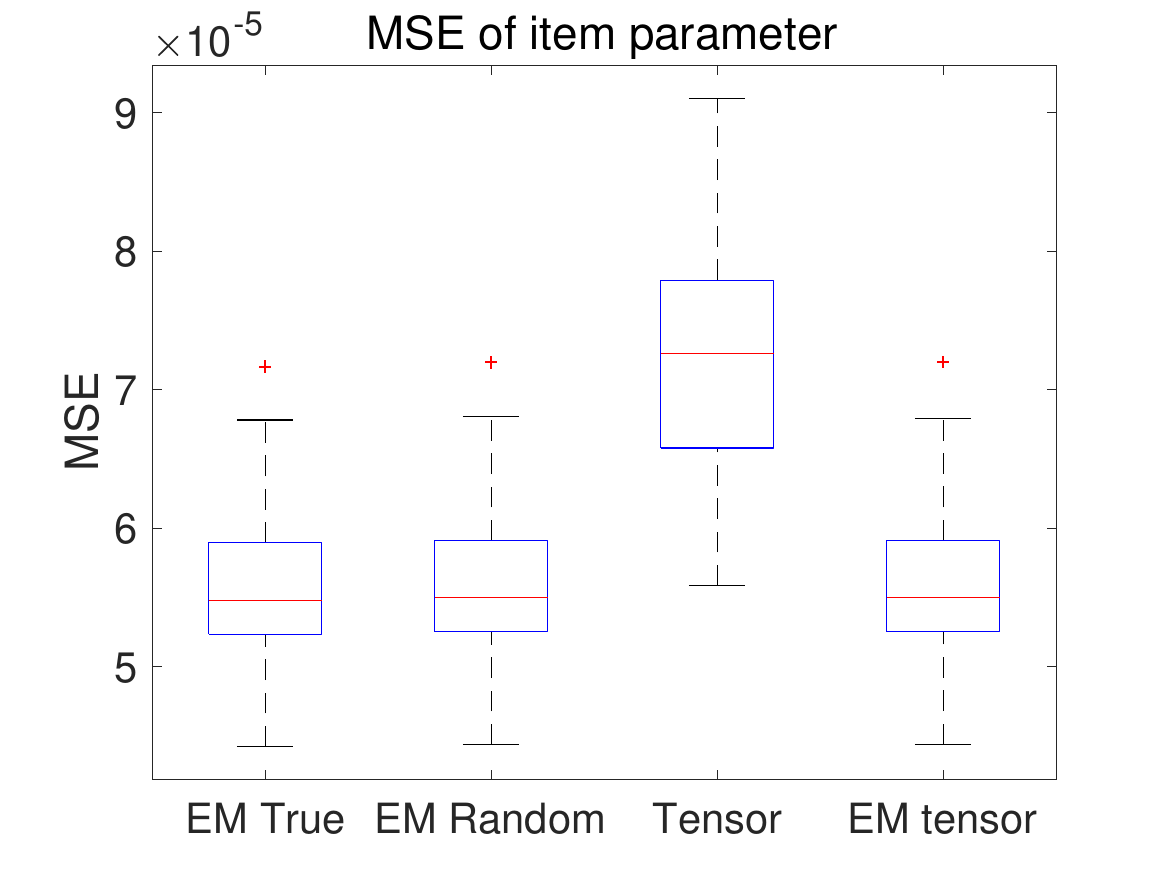}
		\end{minipage}%
	}%
	\subfigure[Running time of the algorithms]{
		\begin{minipage}[t]{0.4\linewidth}
			\centering
			\includegraphics[width=2in]{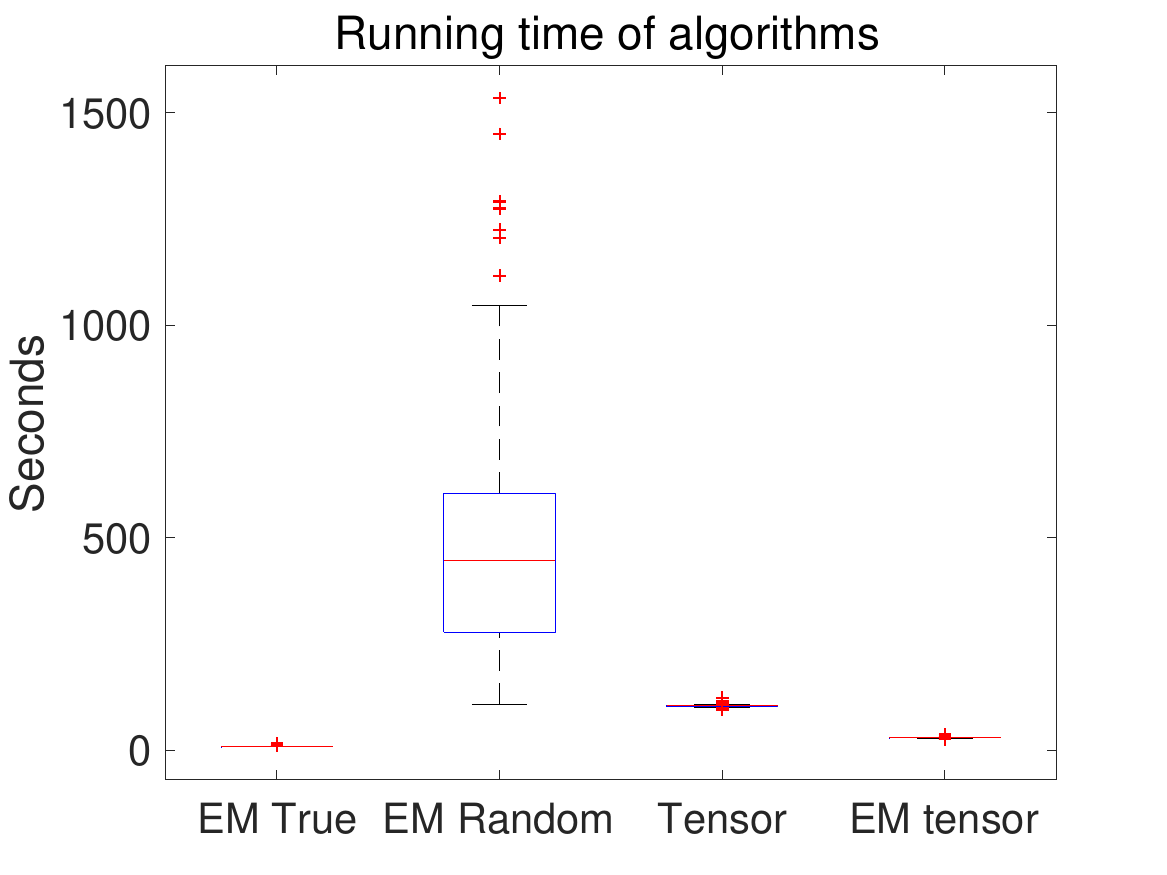}
		\end{minipage}%
	}%
	\centering
	\caption{$N = 20000, J= 100, L=5,$ item parameters $\in \{0.2,0.4,0.6,0.8\}$}
\end{figure}

\begin{figure}[H]
	\centering
	\subfigure[MSE of item parameters]{
		\begin{minipage}[t]{0.4\linewidth}
			\centering
			\includegraphics[width=2in]{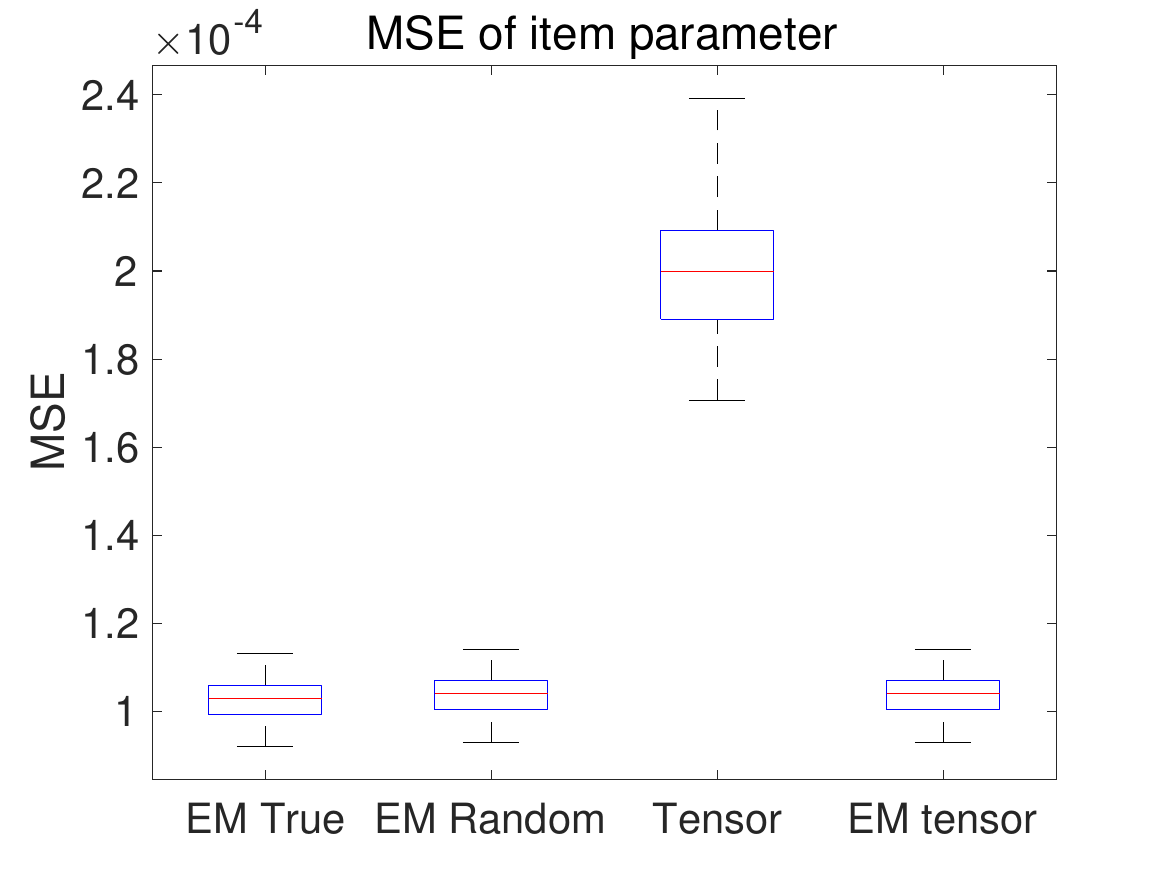}
		\end{minipage}%
	}%
	\subfigure[Running time of the algorithms]{
		\begin{minipage}[t]{0.4\linewidth}
			\centering
			\includegraphics[width=2in]{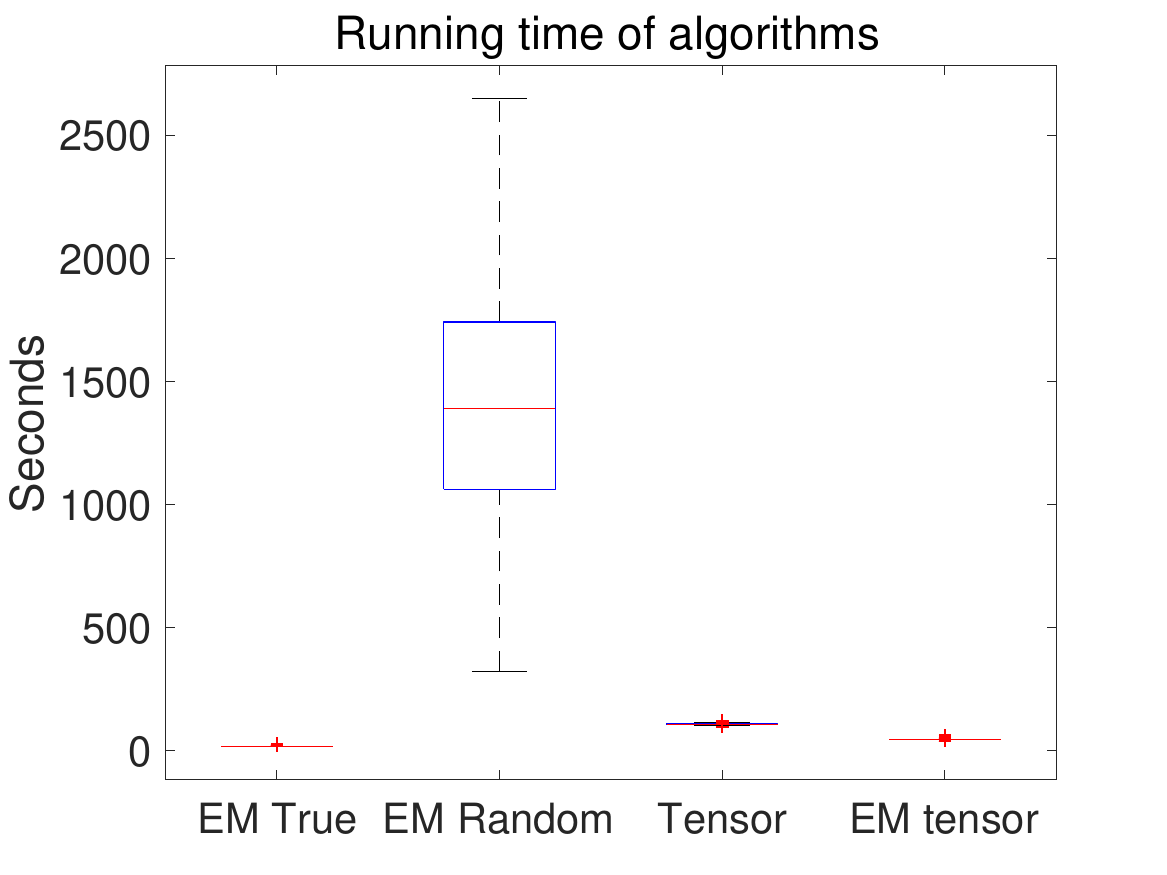}
		\end{minipage}%
	}%
	\centering
	\caption{$N = 20000, J= 100, L=10,$ item parameters $\in \{0.2,0.4,0.6,0.8\}$}
\end{figure}

\begin{figure}[H]
	\centering
	\subfigure[MSE of item parameters]{
		\begin{minipage}[t]{0.4\linewidth}
			\centering
			\includegraphics[width=2in]{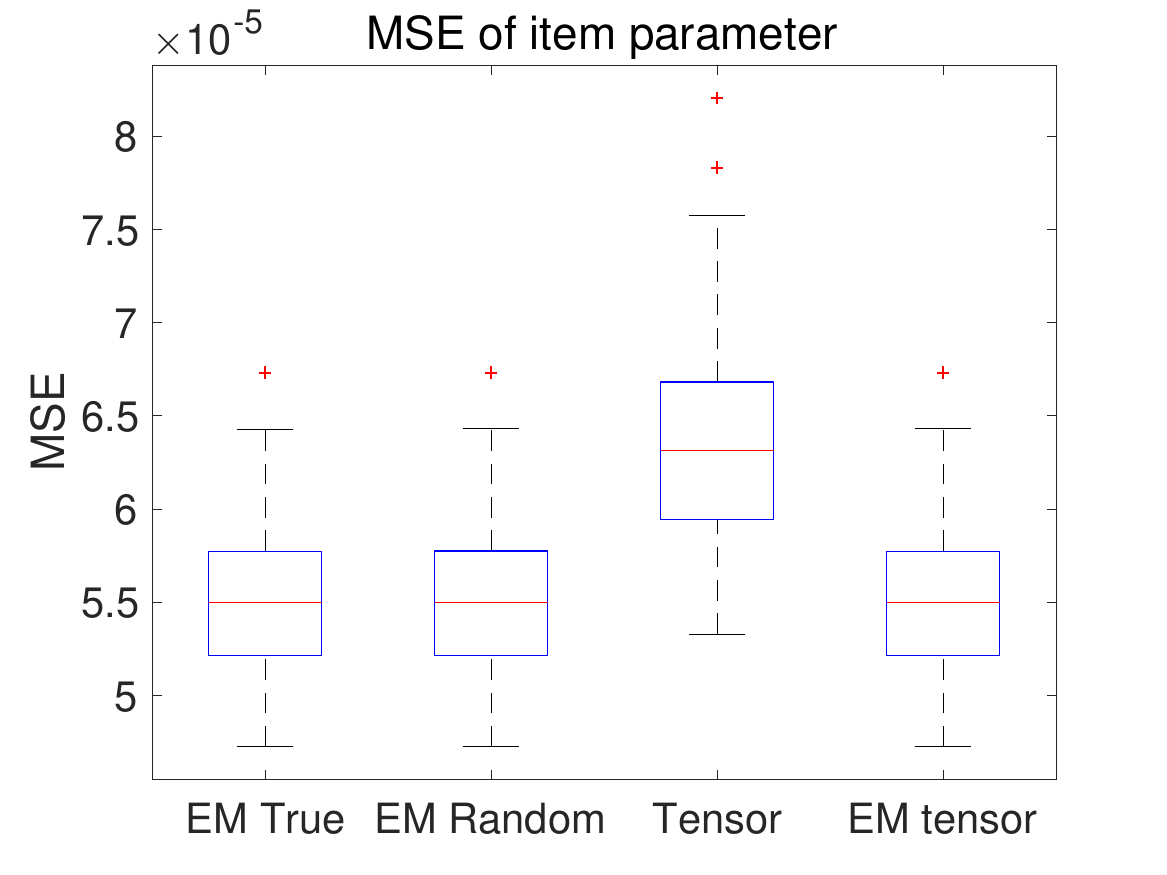}
		\end{minipage}%
	}%
	\subfigure[Running time of the algorithms]{
		\begin{minipage}[t]{0.4\linewidth}
			\centering
			\includegraphics[width=2in]{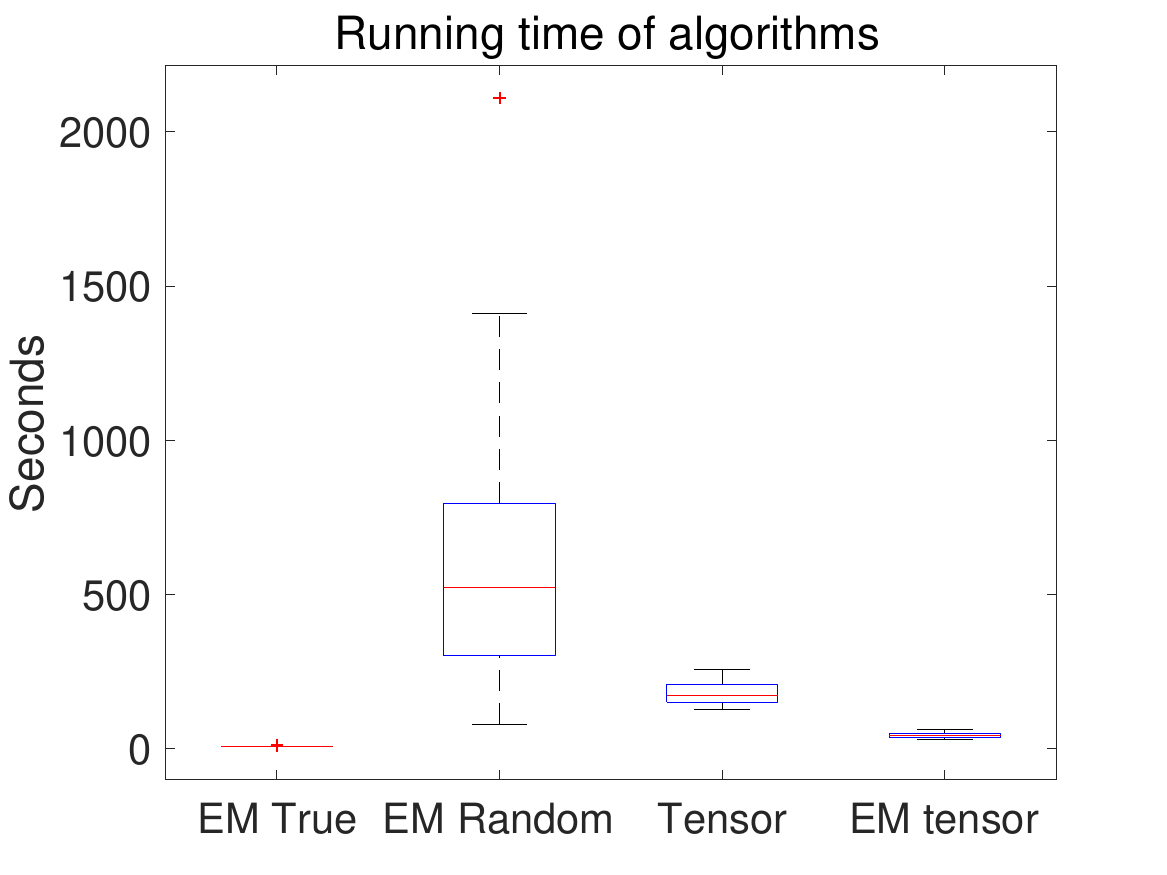}
		\end{minipage}%
	}%
	\centering
	\caption{$N = 20000, J= 200, L=5,$ item parameters $\in \{0.2,0.4,0.6,0.8\}$}
\end{figure}

\begin{figure}[H]
	\centering
	\subfigure[MSE of item parameters]{
		\begin{minipage}[t]{0.33\textwidth}
			\centering
			\includegraphics[width=2in]{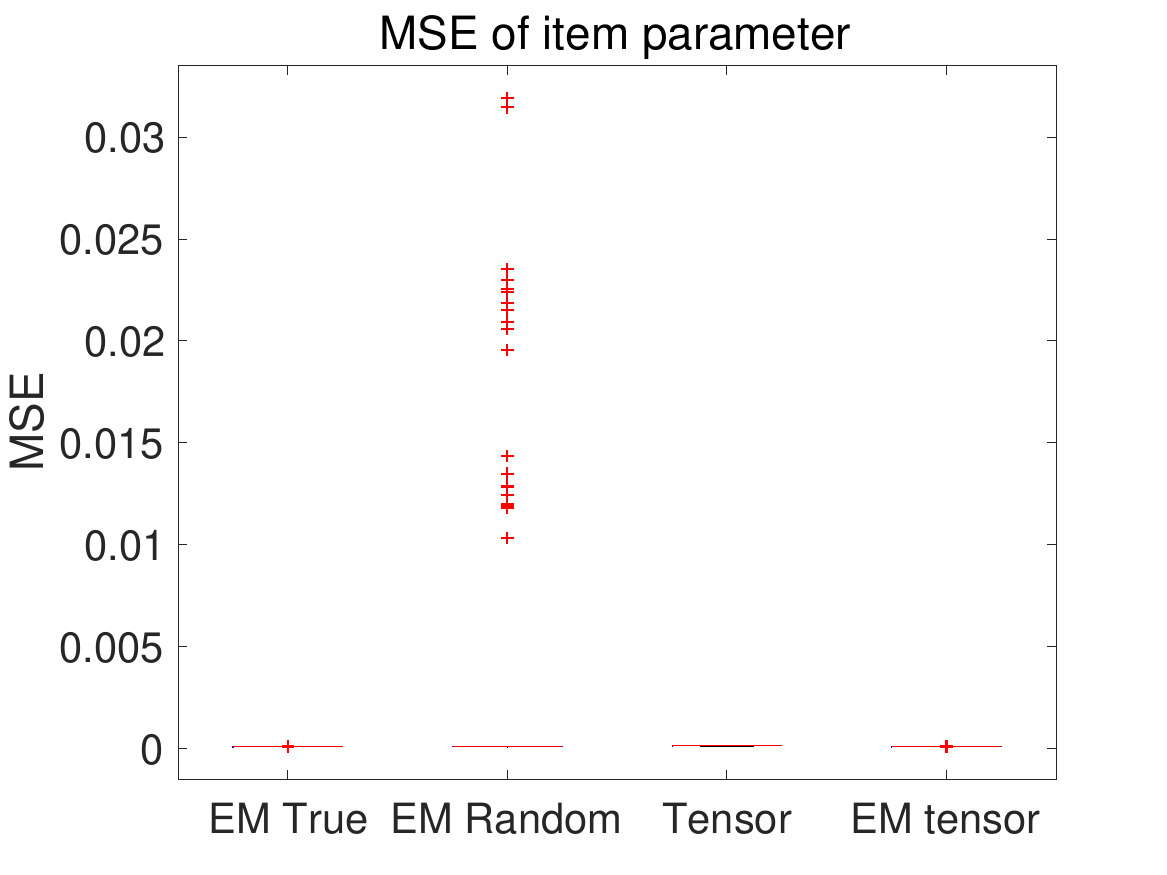}
		\end{minipage}%
	}%
		\subfigure[MSE without EM-random]{
	\begin{minipage}[t]{0.33\textwidth}
		\centering
		\includegraphics[width=2in]{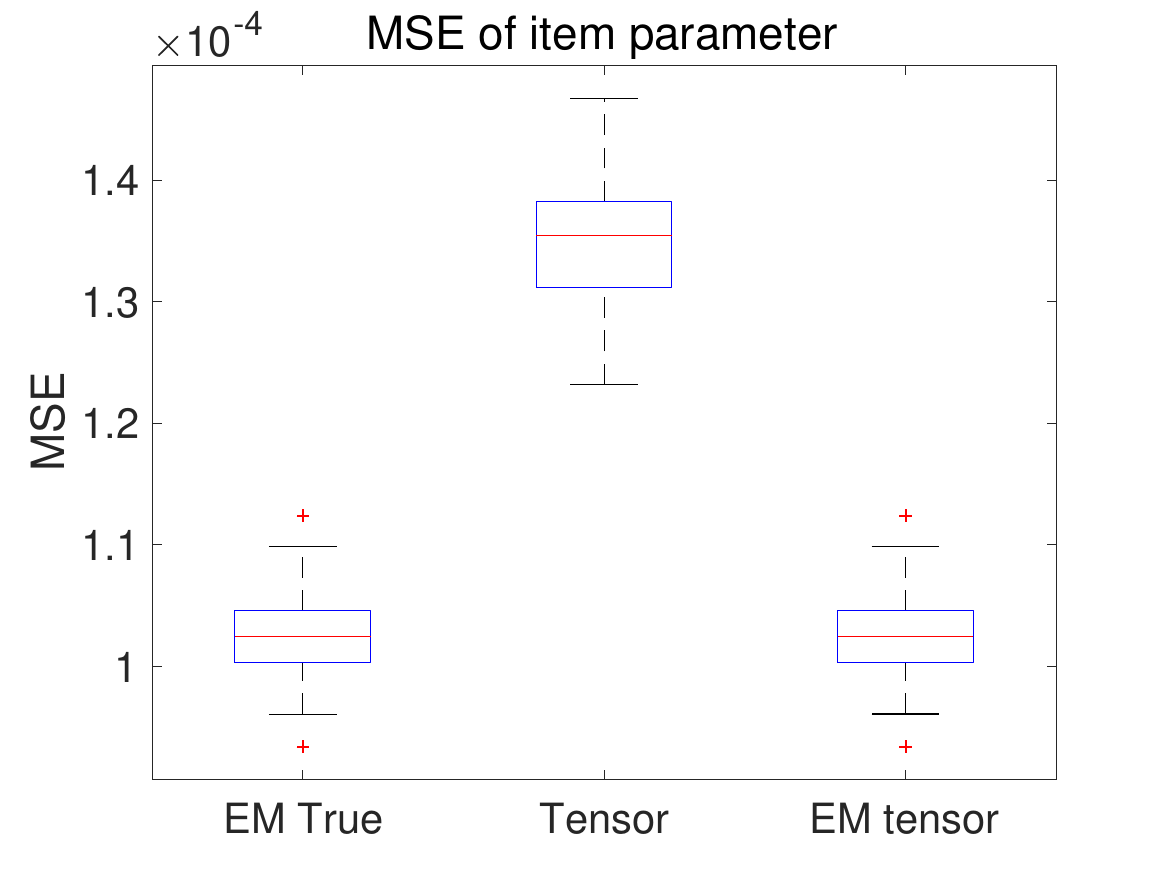}
	\end{minipage}%
}%
	\subfigure[Running time of the algorithms]{
		\begin{minipage}[t]{0.33\textwidth}
			\centering
			\includegraphics[width=2in]{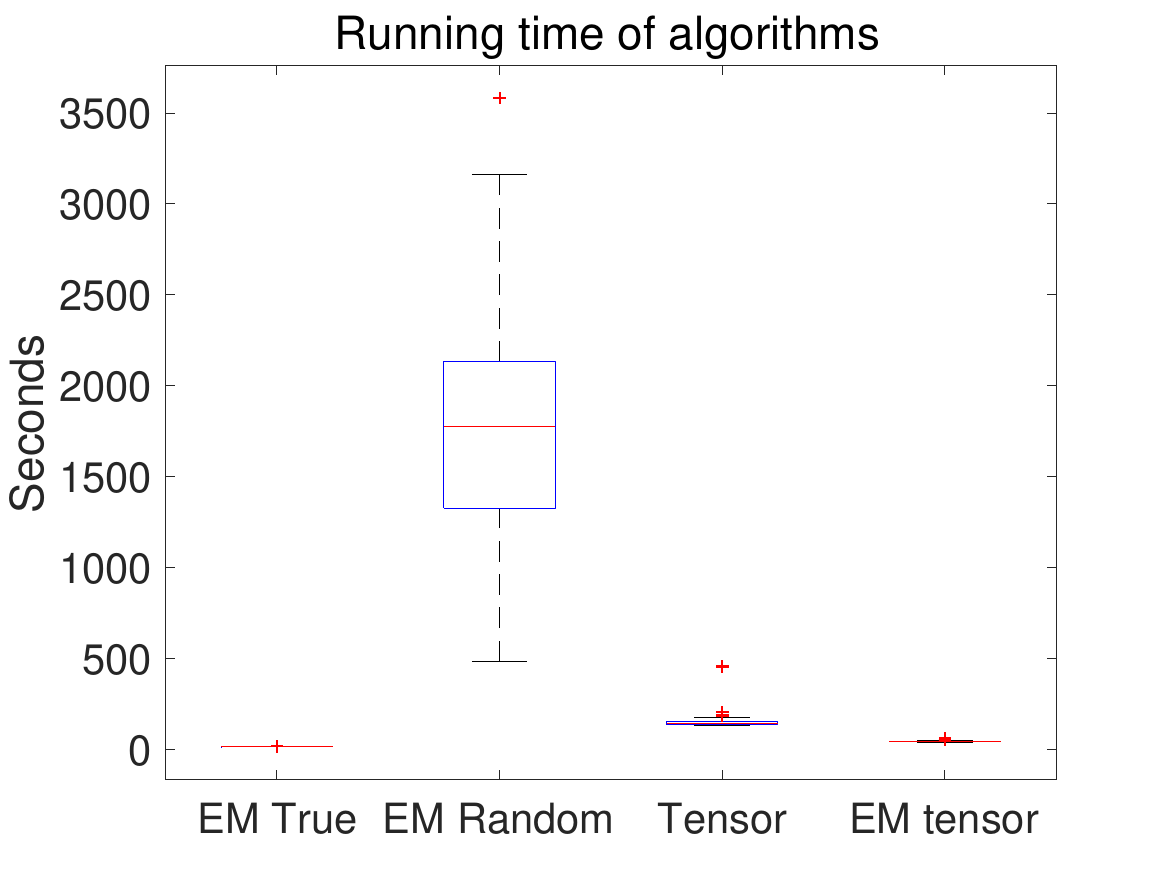}
		\end{minipage}%
	}%
	\centering
	\caption{$N = 20000, J= 200, L=10,$ item parameters $\in \{0.2,0.4,0.6,0.8\}$}
\end{figure}

\subsection*{Fixed-effect LCM}

Then the simulation results of fixed-effect LCM are presented.

\begin{figure}[H]
	\centering
	\subfigure[MSE of item parameters]{
		\begin{minipage}[t]{0.33\linewidth}
			\centering
			\includegraphics[width=2in]{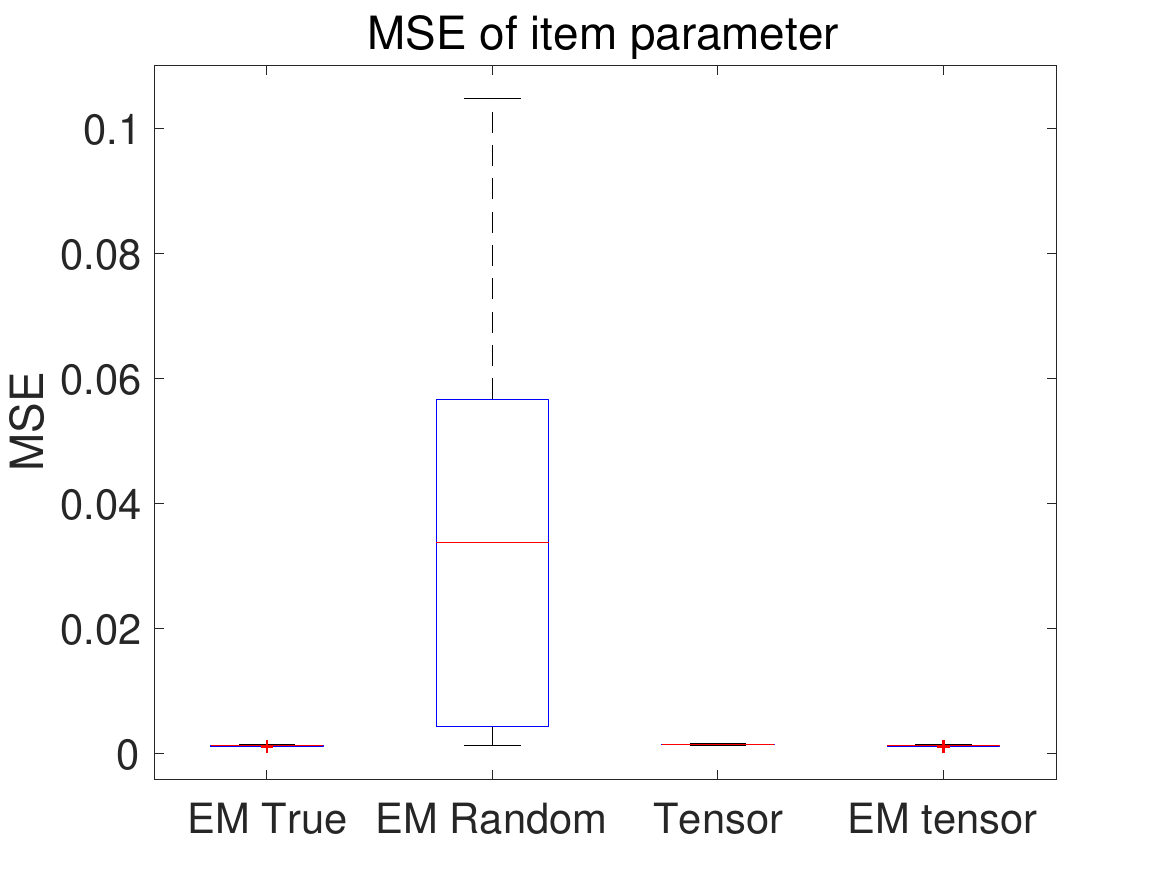}
		\end{minipage}%
	}%
	\subfigure[MSE without EM-random]{
		\begin{minipage}[t]{0.33\linewidth}
			\centering
			\includegraphics[width=2in]{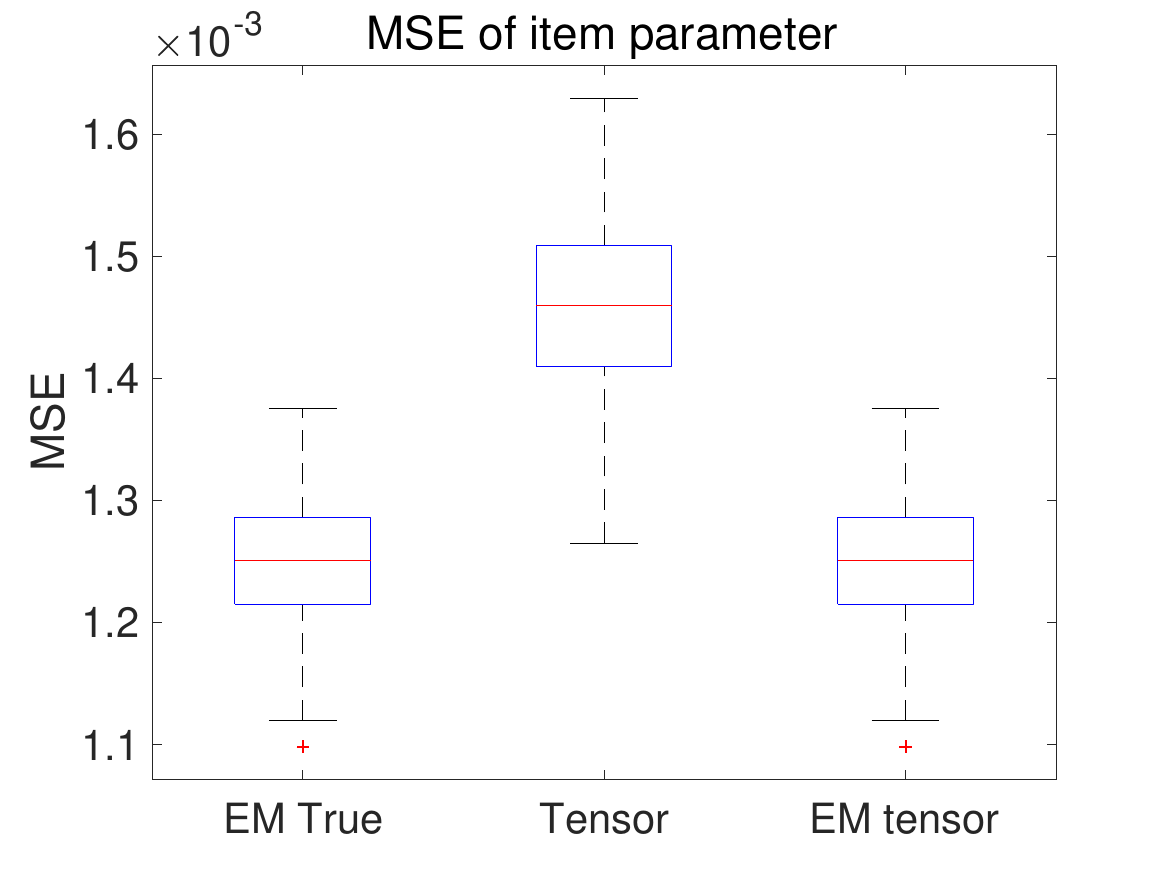}
		\end{minipage}%
	}%
	\subfigure[Running time of the algorithms]{
		\begin{minipage}[t]{0.33\linewidth}
			\centering
			\includegraphics[width=2in]{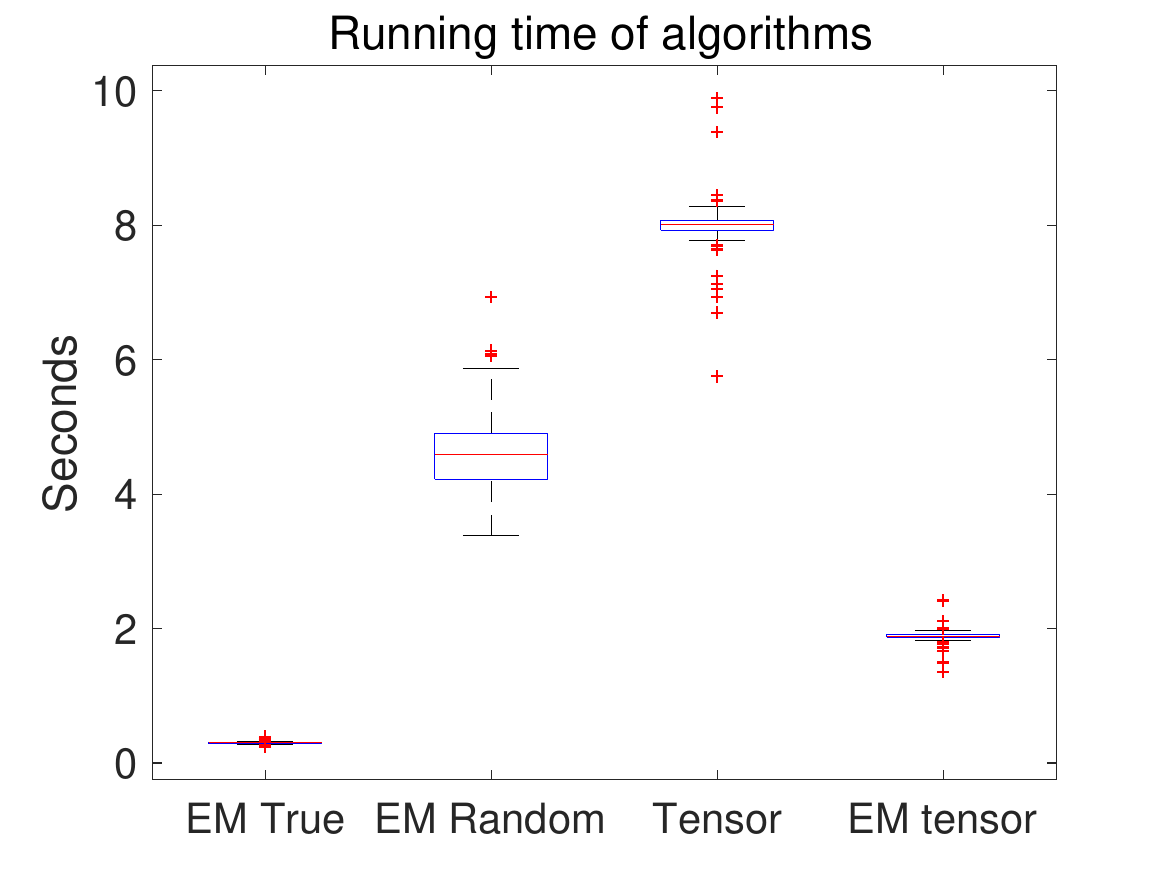}
		\end{minipage}%
	}%
	\centering
	\caption{$N = 1000, J= 100, L=10,$ item parameters $\in \{0.1,0.2,0.8,0.9\}$}
\end{figure}

\begin{figure}[H]
	\centering
	\subfigure[MSE of item parameters]{
		\begin{minipage}[t]{0.33\linewidth}
			\centering
			\includegraphics[width=2in]{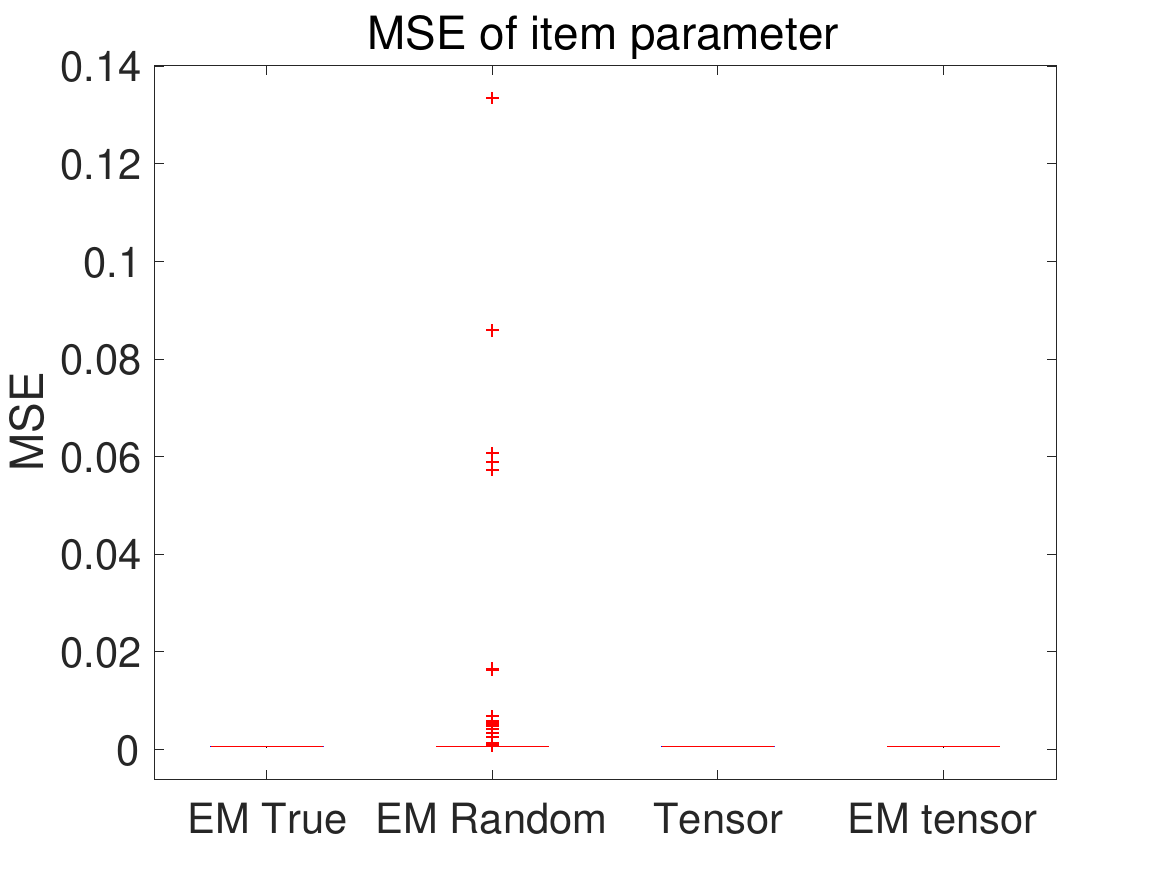}
		\end{minipage}%
	}%
	\subfigure[MSE without EM-random]{
		\begin{minipage}[t]{0.33\linewidth}
			\centering
			\includegraphics[width=2in]{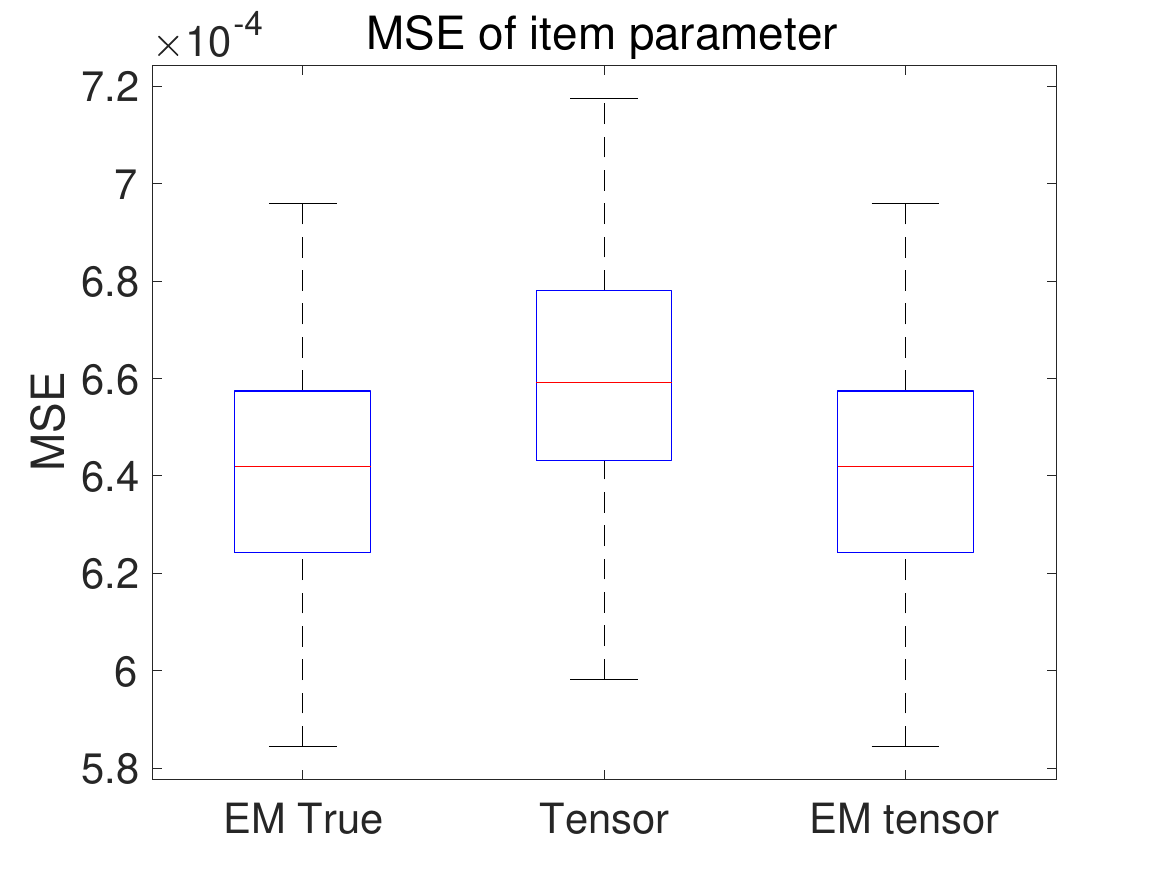}
		\end{minipage}%
	}%
	\subfigure[Running time of the algorithms]{
		\begin{minipage}[t]{0.33\linewidth}
			\centering
			\includegraphics[width=2in]{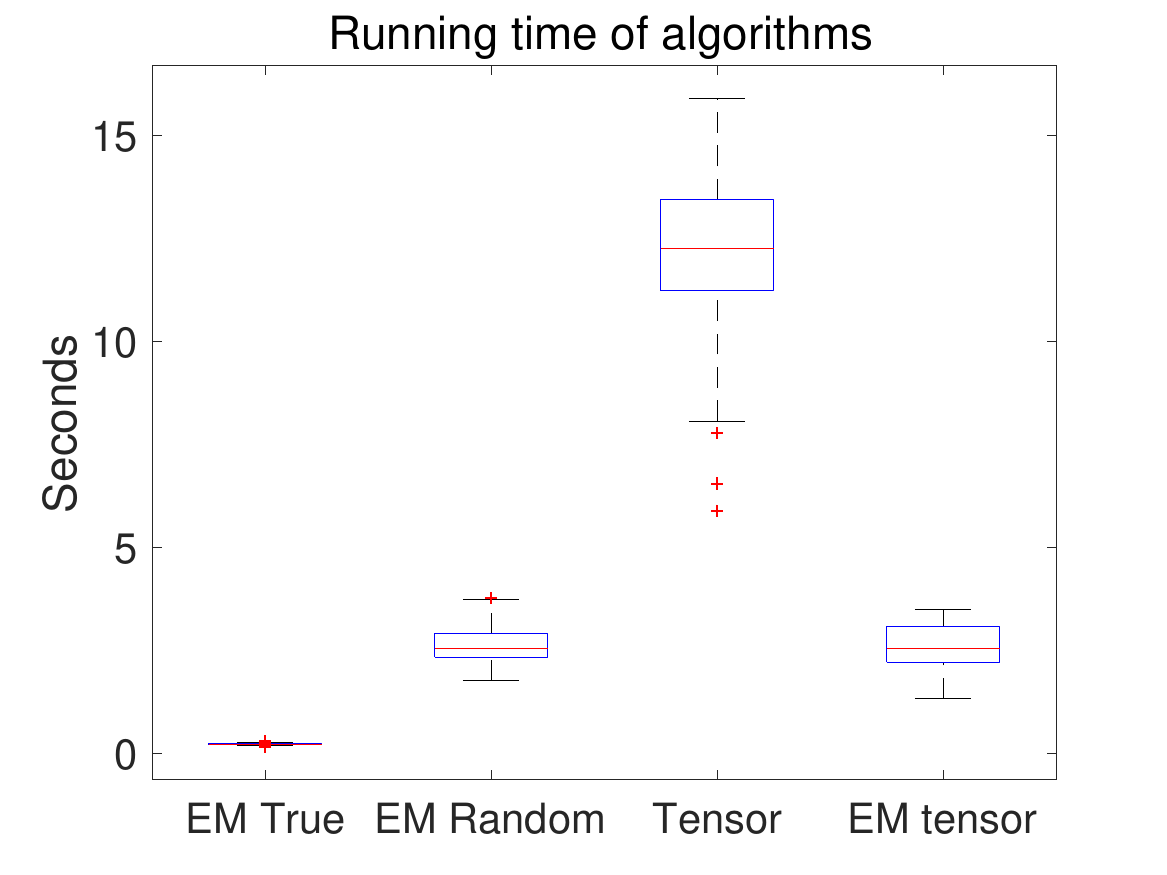}
		\end{minipage}%
	}%
	\centering
	\caption{$N = 1000, J= 200, L=5,$ item parameters $\in \{0.1,0.2,0.8,0.9\}$}
\end{figure}

\begin{figure}[H]
	\centering
	\subfigure[MSE of item parameters]{
		\begin{minipage}[t]{0.33\linewidth}
			\centering
			\includegraphics[width=2in]{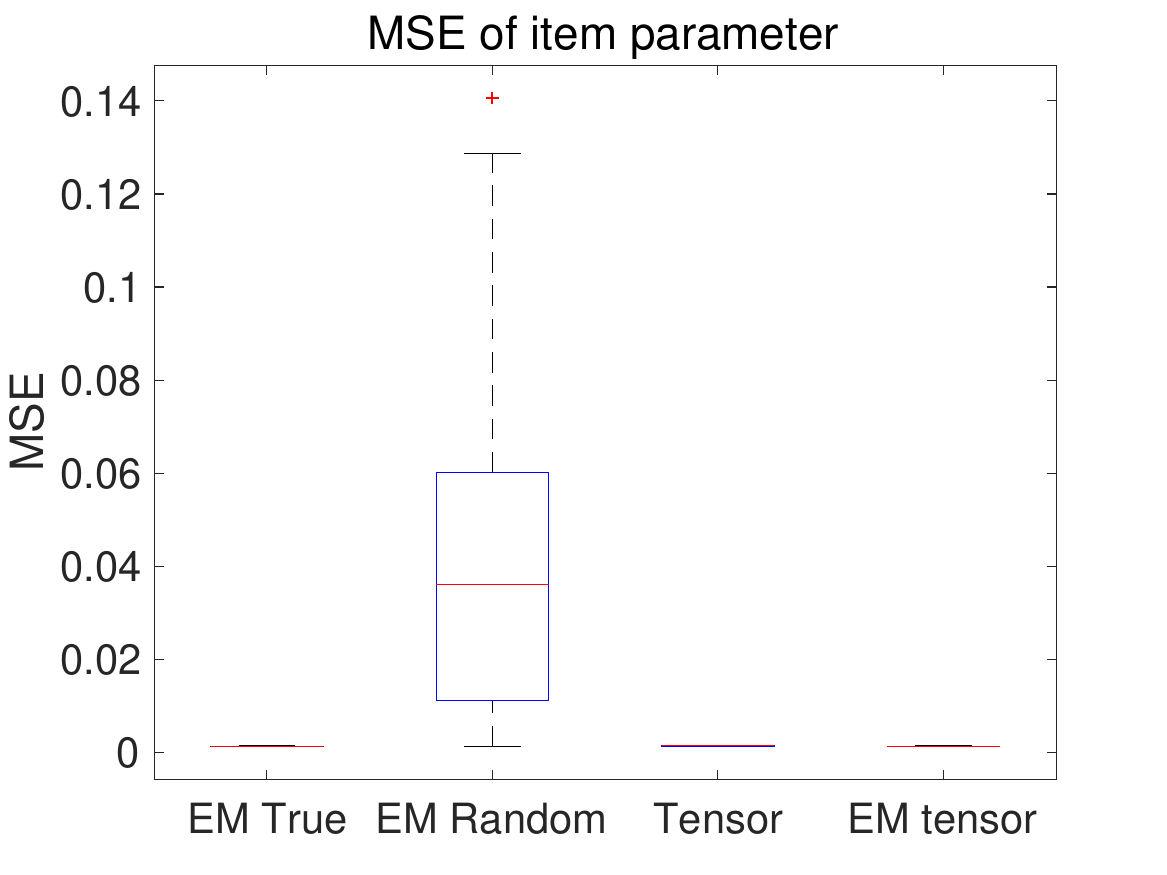}
		\end{minipage}%
	}%
	\subfigure[MSE without EM-random]{
		\begin{minipage}[t]{0.33\linewidth}
			\centering
			\includegraphics[width=2in]{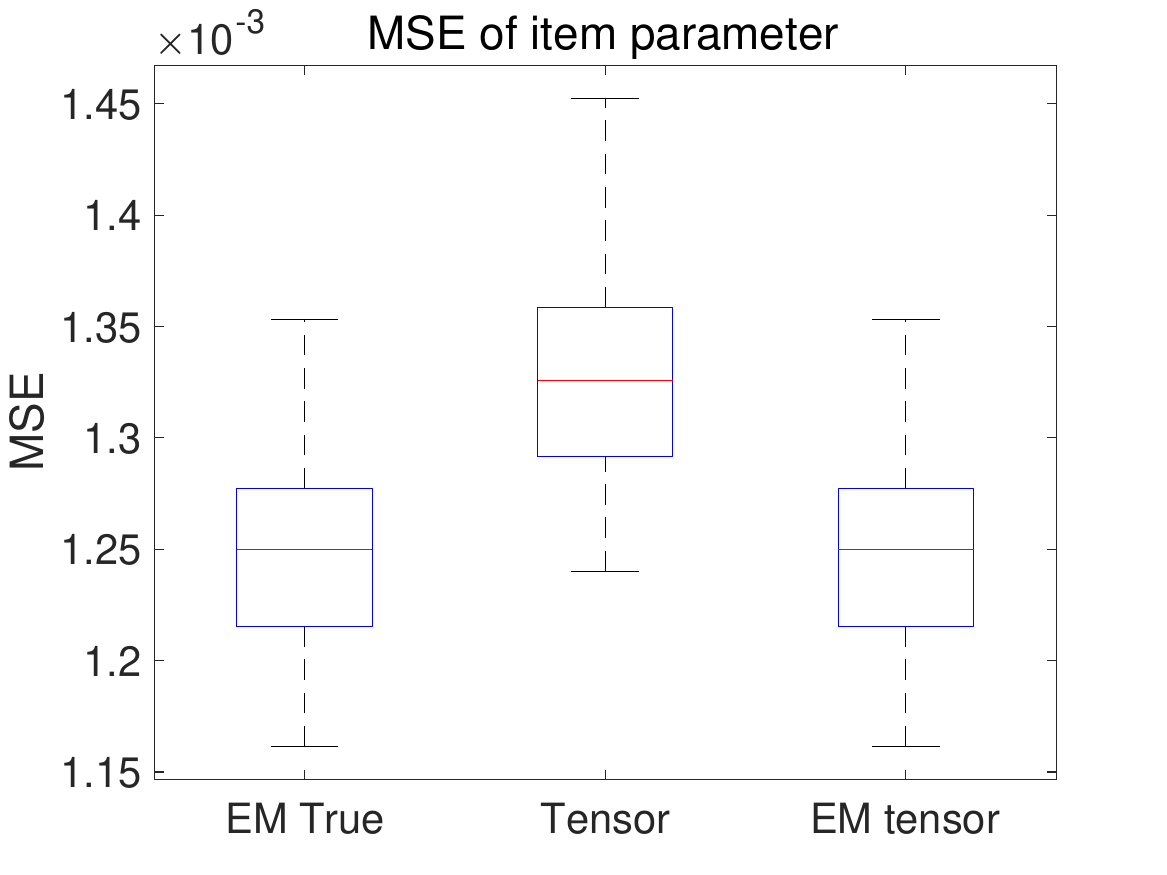}
		\end{minipage}%
	}%
	\subfigure[Running time of the algorithms]{
		\begin{minipage}[t]{0.33\linewidth}
			\centering
			\includegraphics[width=2in]{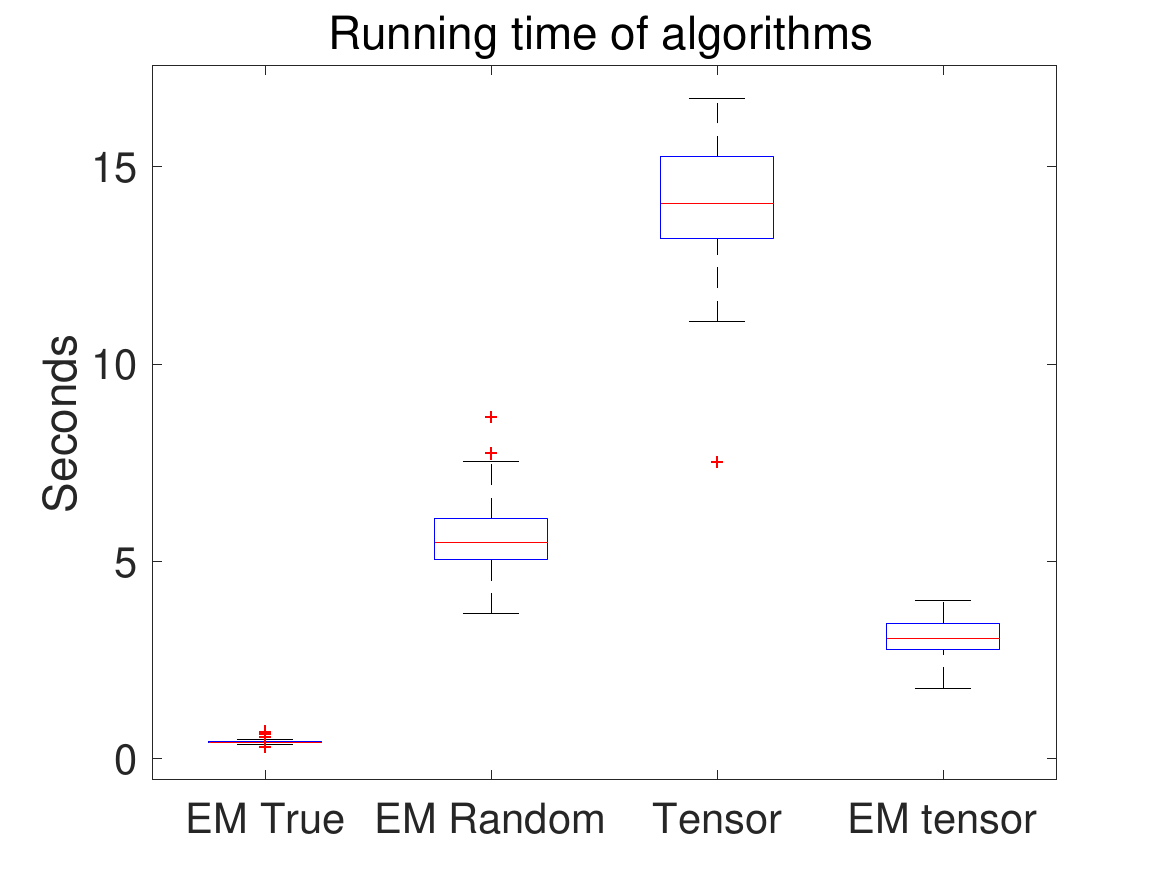}
		\end{minipage}%
	}%
	\centering
	\caption{$N = 1000, J= 200, L=10,$ item parameters $\in \{0.1,0.2,0.8,0.9\}$}
\end{figure}

\begin{figure}[H]
	\centering
	\subfigure[MSE of item parameters]{
		\begin{minipage}[t]{0.4\linewidth}
			\centering
			\includegraphics[width=2in]{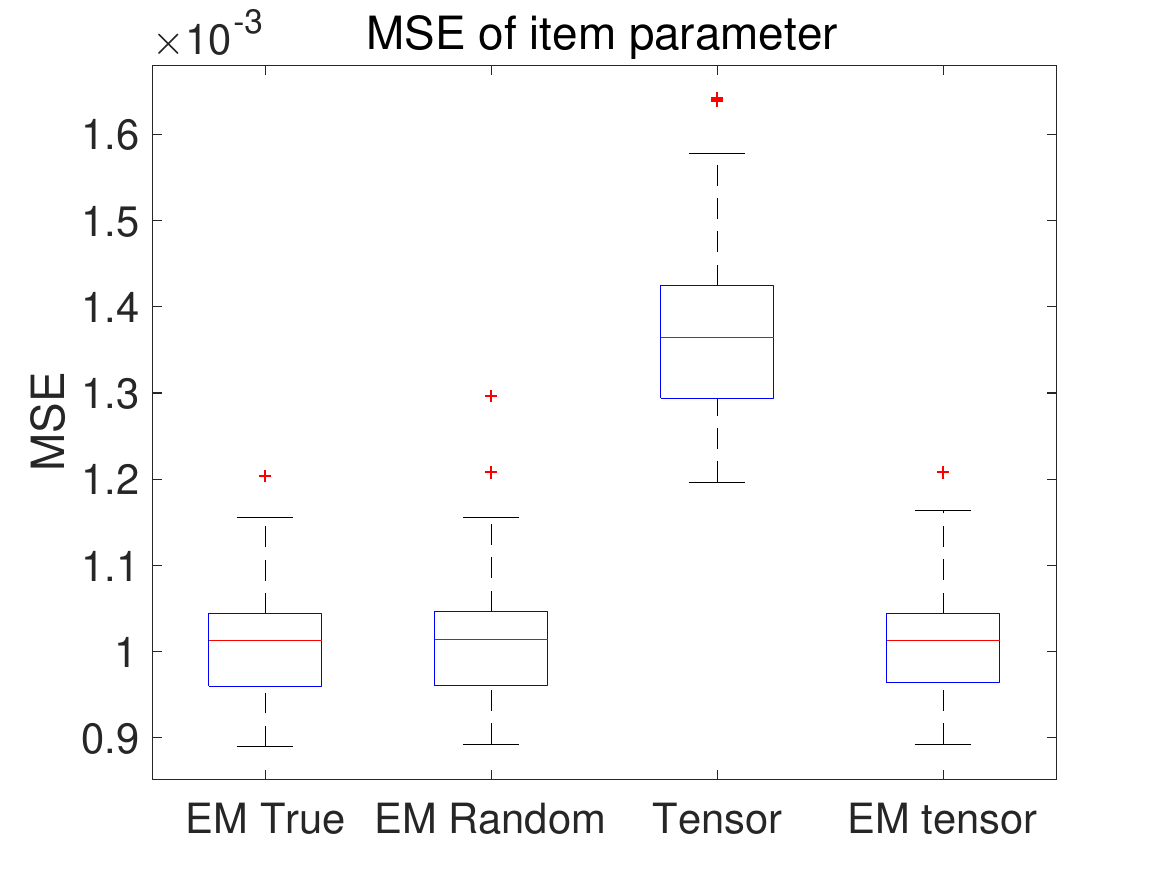}
		\end{minipage}%
	}%
	\subfigure[Running time of the algorithms]{
		\begin{minipage}[t]{0.4\linewidth}
			\centering
			\includegraphics[width=2in]{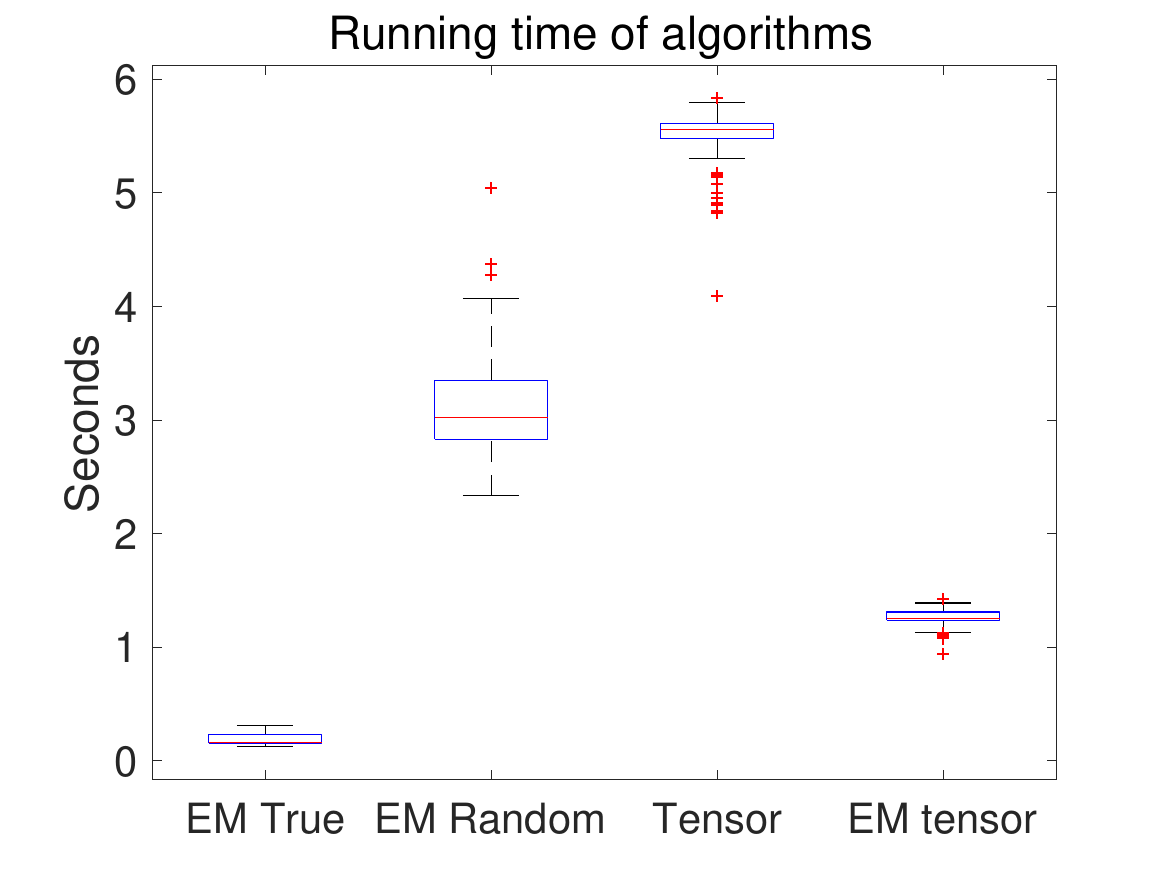}
		\end{minipage}%
	}%
	\centering
	\caption{$N = 1000, J= 100, L=5,$ item parameters $\in \{0.2,0.4,0.6,0.8\}$}
\end{figure}

\begin{figure}[H]
	\centering
	\subfigure[MSE of item parameters]{
		\begin{minipage}[t]{0.33\linewidth}
			\centering
			\includegraphics[width=2in]{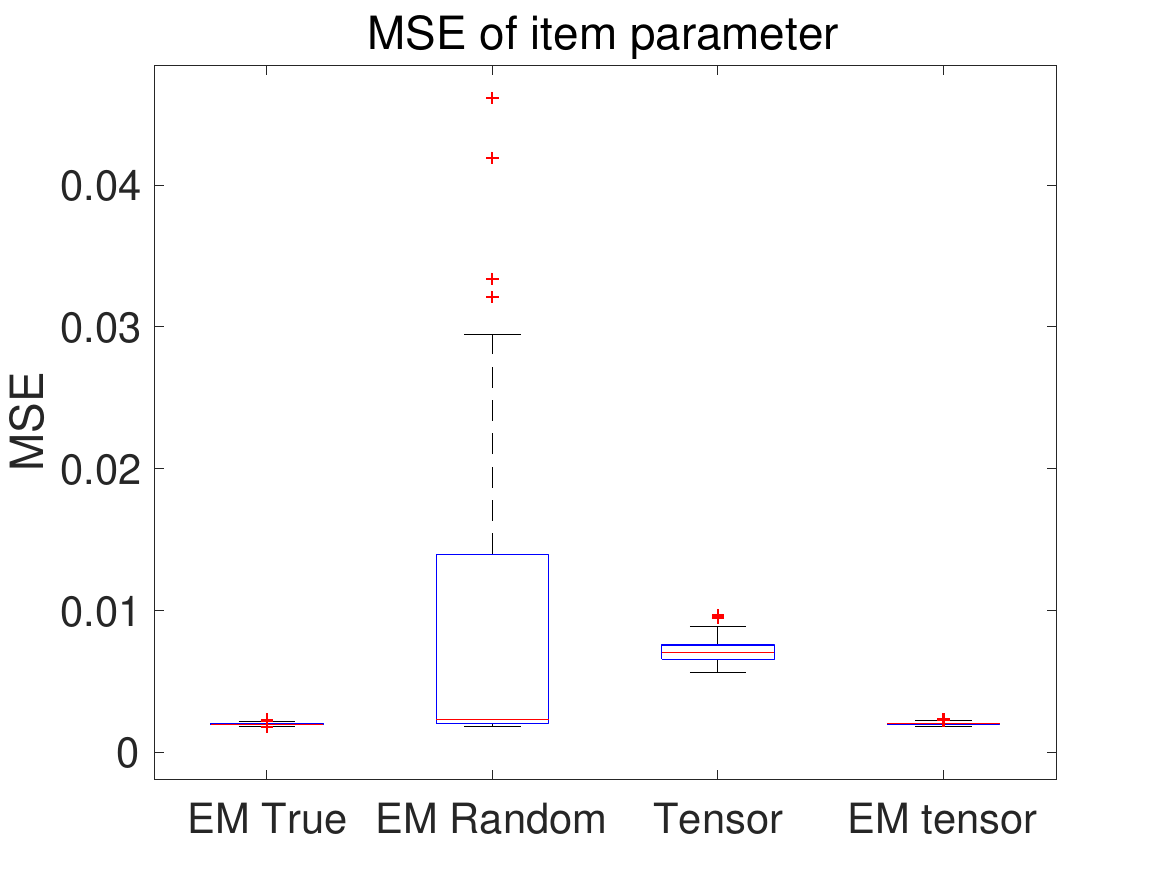}
		\end{minipage}%
	}%
	\subfigure[MSE without EM-random]{
		\begin{minipage}[t]{0.33\linewidth}
			\centering
			\includegraphics[width=2in]{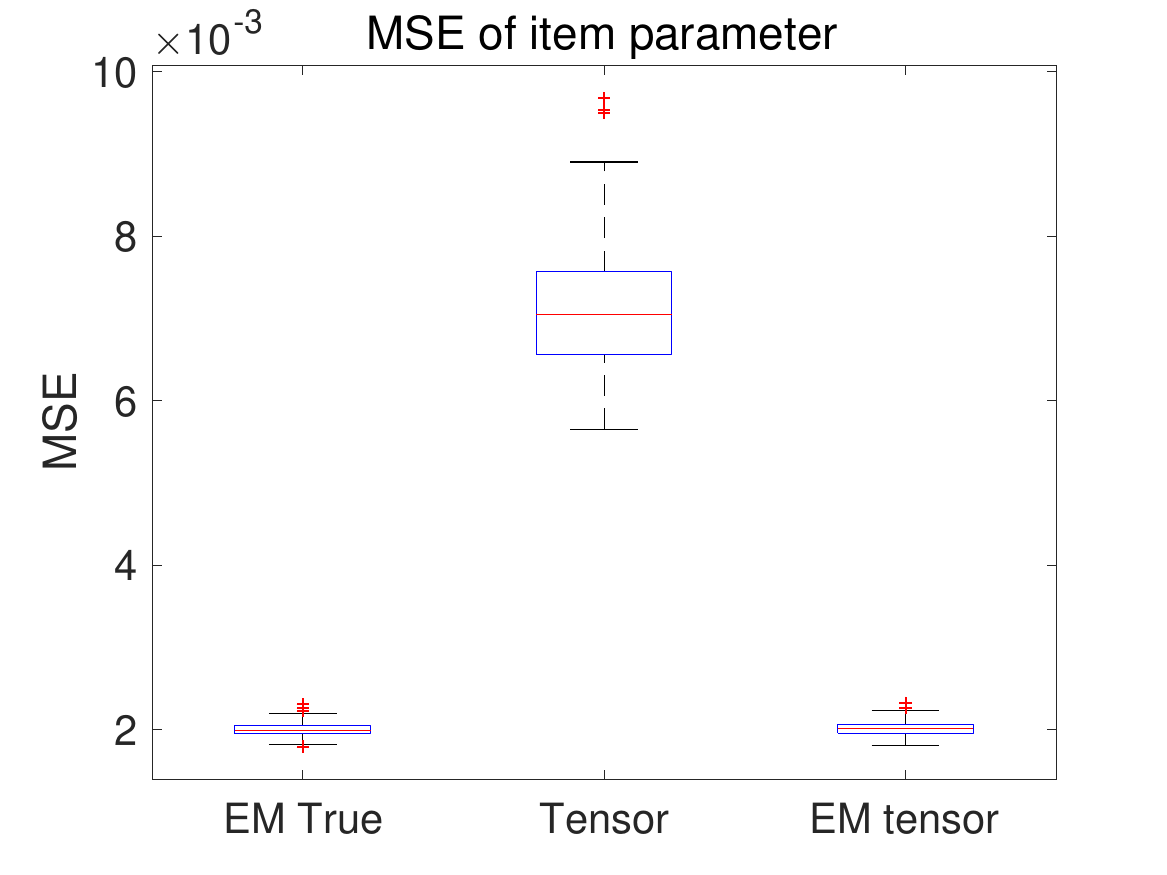}
		\end{minipage}%
	}%
	\subfigure[Running time of the algorithms]{
		\begin{minipage}[t]{0.33\linewidth}
			\centering
			\includegraphics[width=2in]{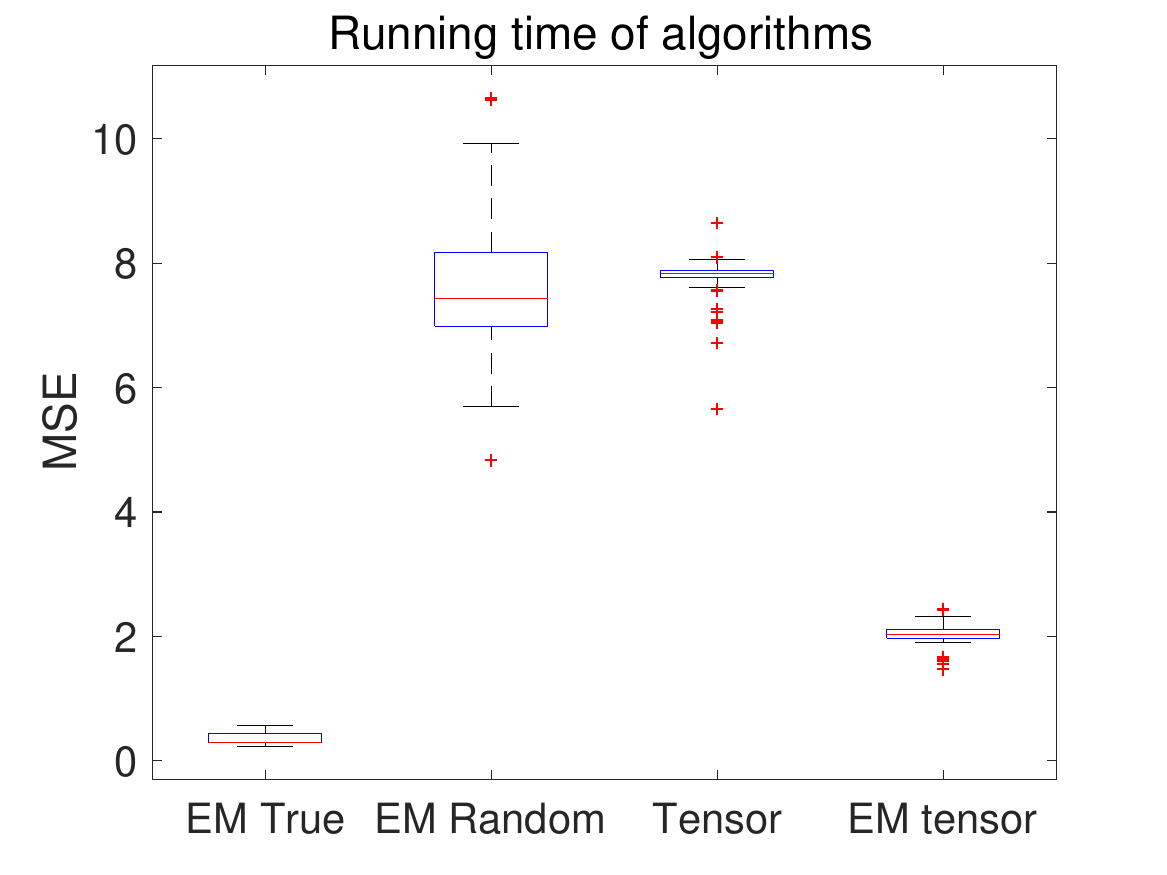}
		\end{minipage}%
	}%
	\centering
	\caption{$N = 1000, J= 100, L=10,$ item parameters $\in \{0.2,0.4,0.6,0.8\}$}
\end{figure}

\begin{figure}[H]
	\centering
	\subfigure[MSE of item parameters]{
		\begin{minipage}[t]{0.4\linewidth}
			\centering
			\includegraphics[width=2in]{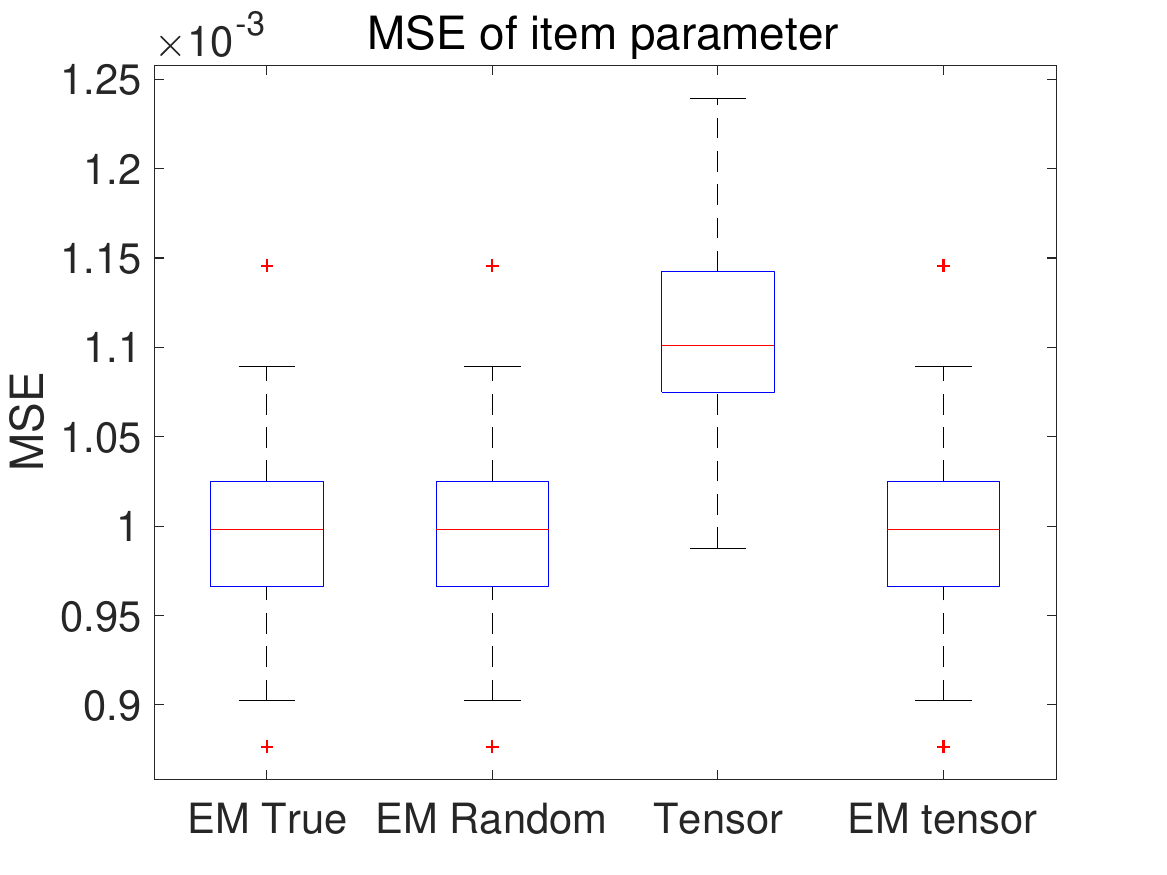}
		\end{minipage}%
	}%
	\subfigure[Running time of the algorithms]{
		\begin{minipage}[t]{0.4\linewidth}
			\centering
			\includegraphics[width=2in]{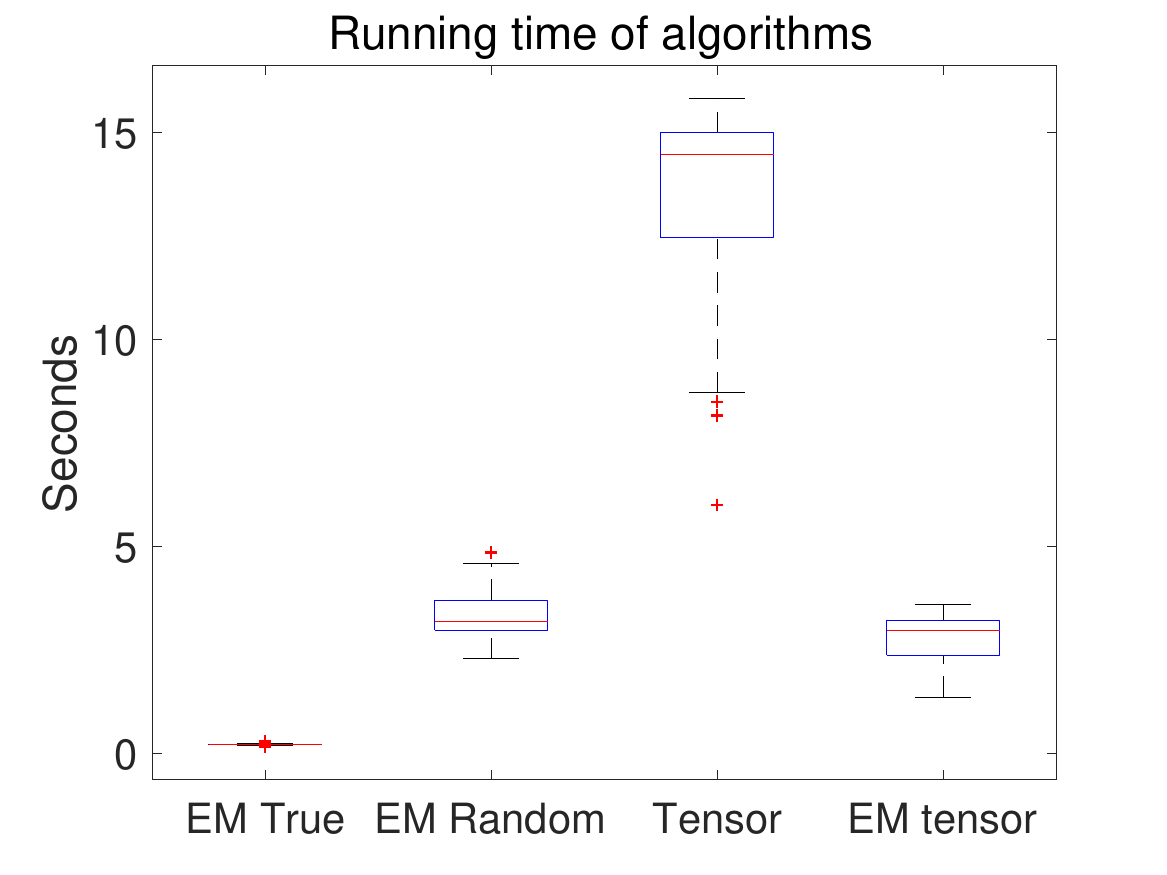}
		\end{minipage}%
	}%
	\centering
	\caption{$N = 1000, J= 200, L=5,$ item parameters $\in \{0.2,0.4,0.6,0.8\}$}
\end{figure}

\begin{figure}[H]
	\centering
	\subfigure[MSE of item parameters]{
		\begin{minipage}[t]{0.33\linewidth}
			\centering
			\includegraphics[width=2in]{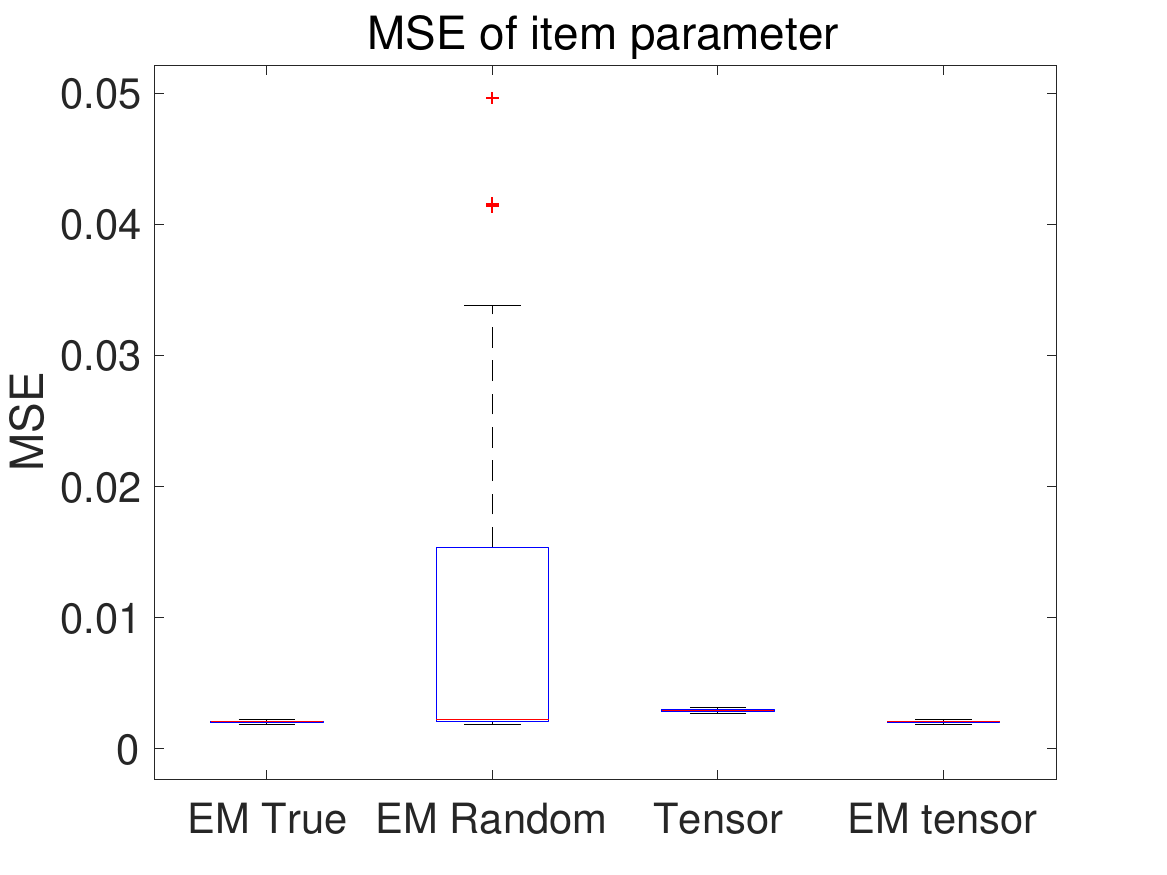}
		\end{minipage}%
	}%
	\subfigure[MSE without EM-random]{
		\begin{minipage}[t]{0.33\linewidth}
			\centering
			\includegraphics[width=2in]{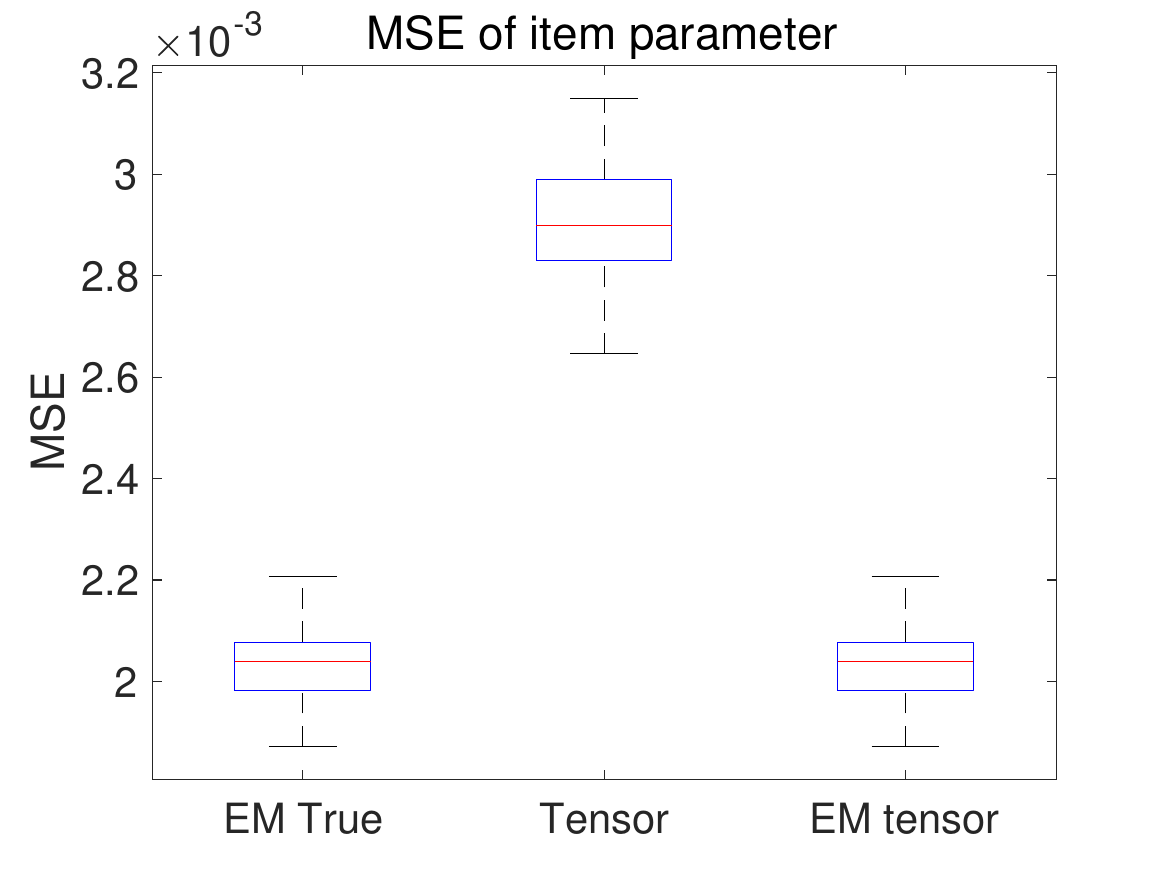}
		\end{minipage}%
	}%
	\subfigure[Running time of the algorithms]{
		\begin{minipage}[t]{0.33\linewidth}
			\centering
			\includegraphics[width=2in]{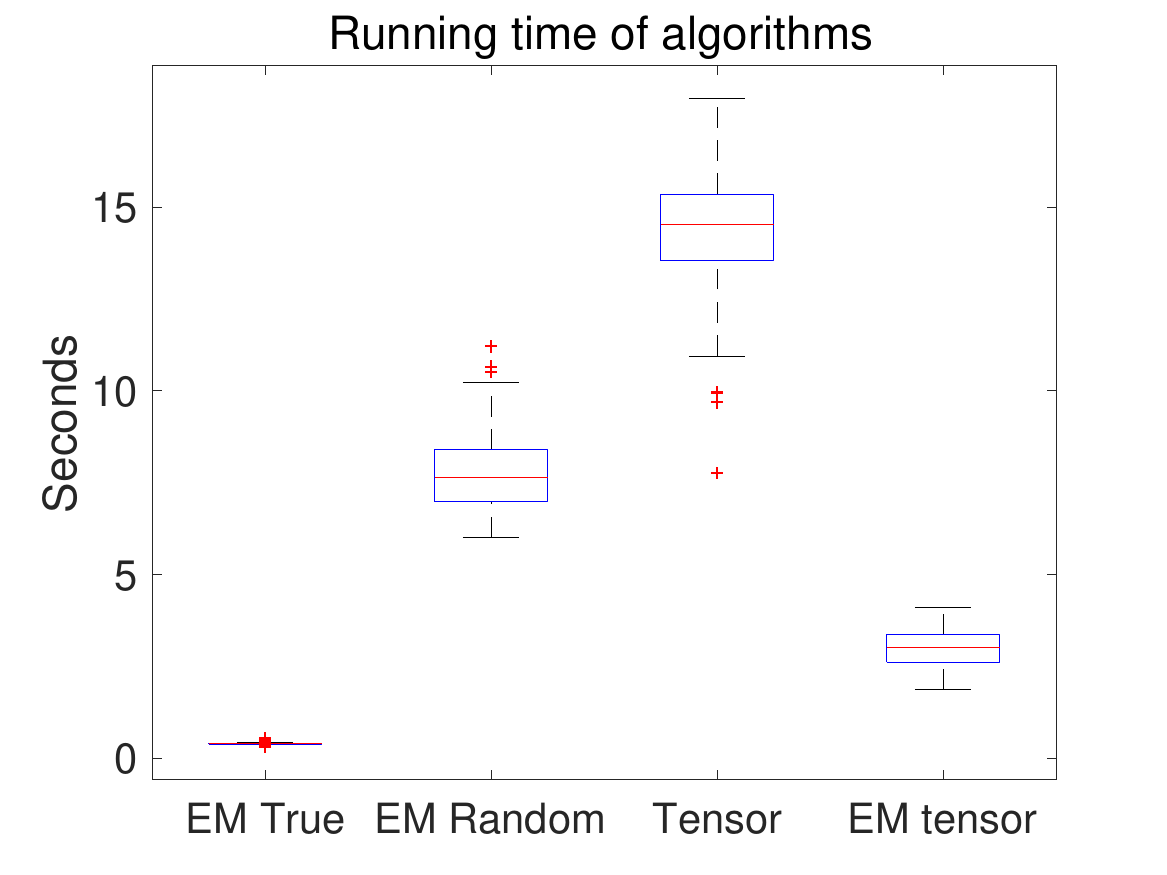}
		\end{minipage}%
	}%
	\centering
	\caption{$N = 1000, J= 200, L=10,$ item parameters $\in \{0.2,0.4,0.6,0.8\}$}
\end{figure}

\begin{figure}[H]
	\centering
	\subfigure[MSE of item parameters]{
		\begin{minipage}[t]{0.4\linewidth}
			\centering
			\includegraphics[width=2in]{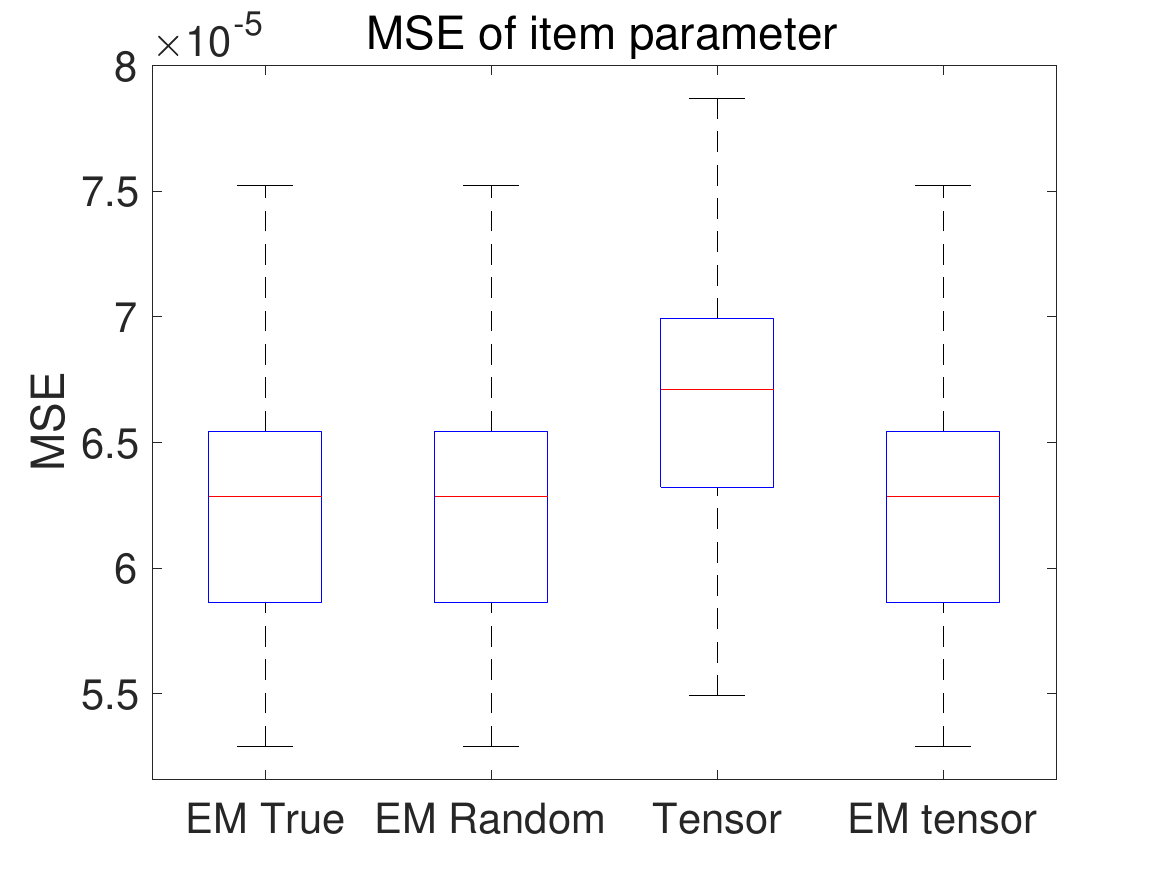}
		\end{minipage}%
	}%
	\subfigure[Running time of the algorithms]{
		\begin{minipage}[t]{0.4\linewidth}
			\centering
			\includegraphics[width=2in]{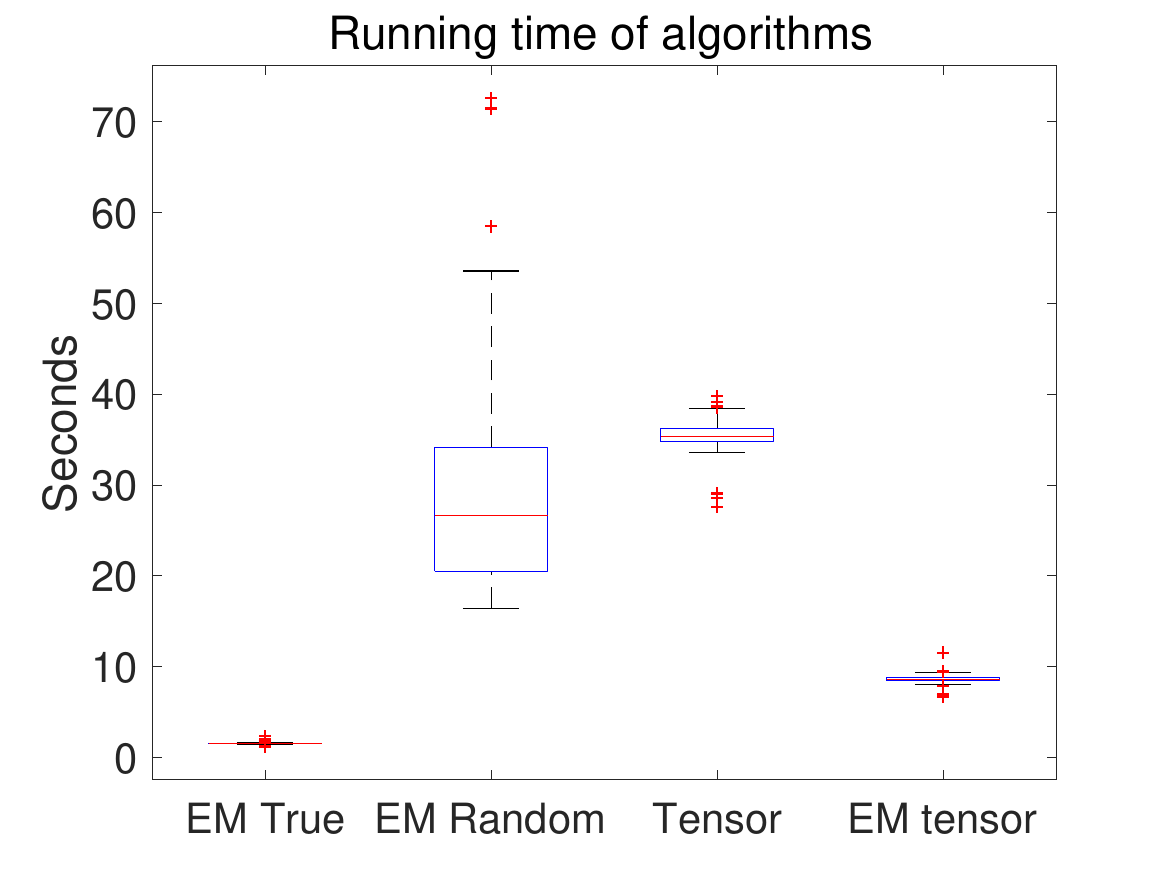}
		\end{minipage}%
	}%
	\centering
	\caption{$N = 10000, J= 100, L=5,$ item parameters $\in \{0.1,0.2,0.8,0.9\}$}
\end{figure}

\begin{figure}[H]
	\centering
	\subfigure[MSE of item parameters]{
		\begin{minipage}[t]{0.33\linewidth}
			\centering
			\includegraphics[width=2in]{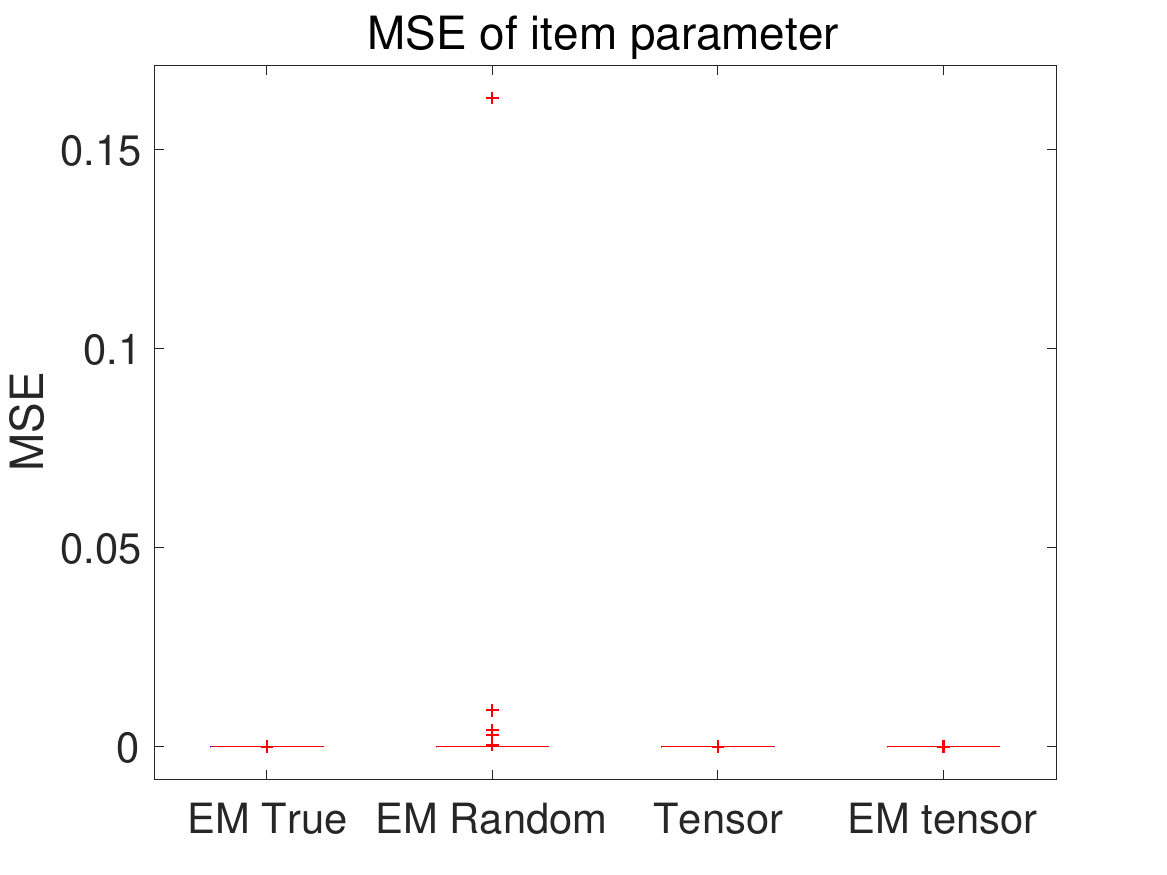}
		\end{minipage}%
	}%
	\subfigure[MSE without EM-random]{
		\begin{minipage}[t]{0.33\linewidth}
			\centering
			\includegraphics[width=2in]{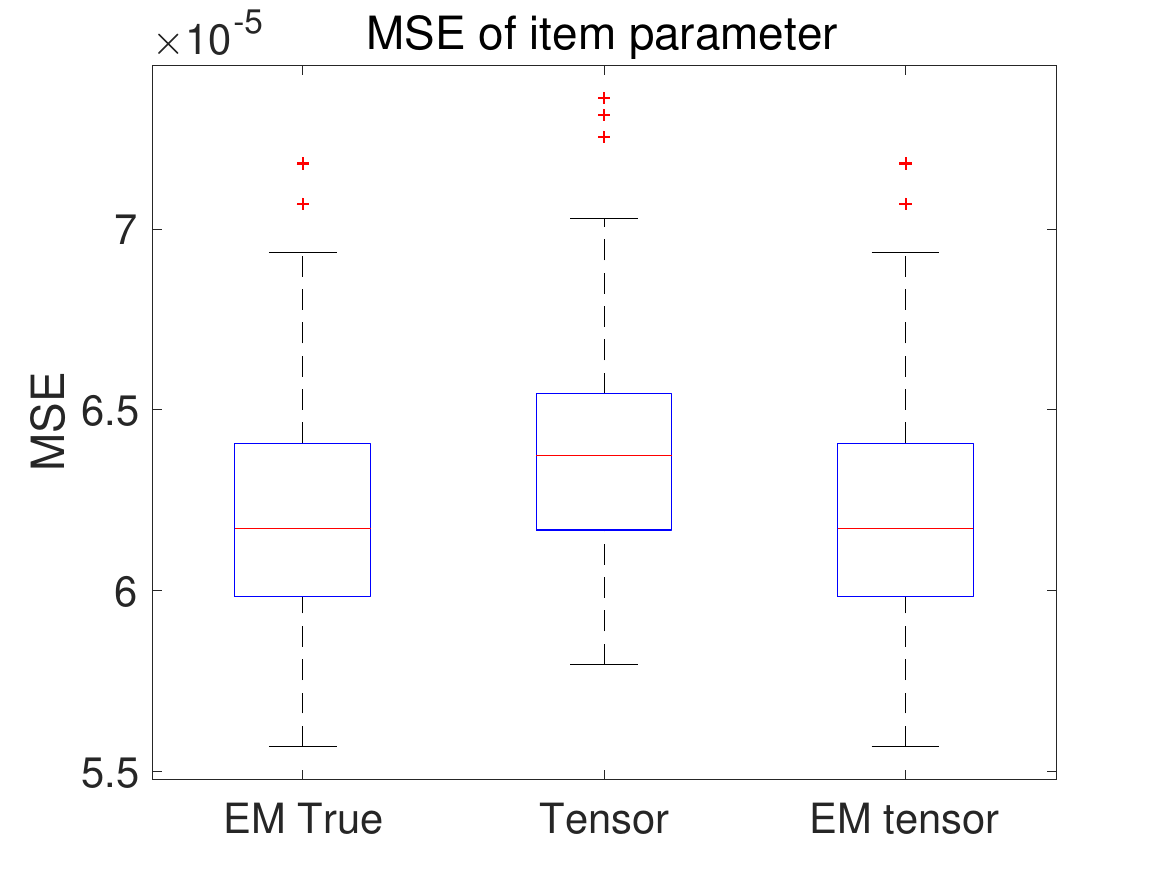}
		\end{minipage}%
	}%
	\subfigure[Running time of the algorithms]{
		\begin{minipage}[t]{0.33\linewidth}
			\centering
			\includegraphics[width=2in]{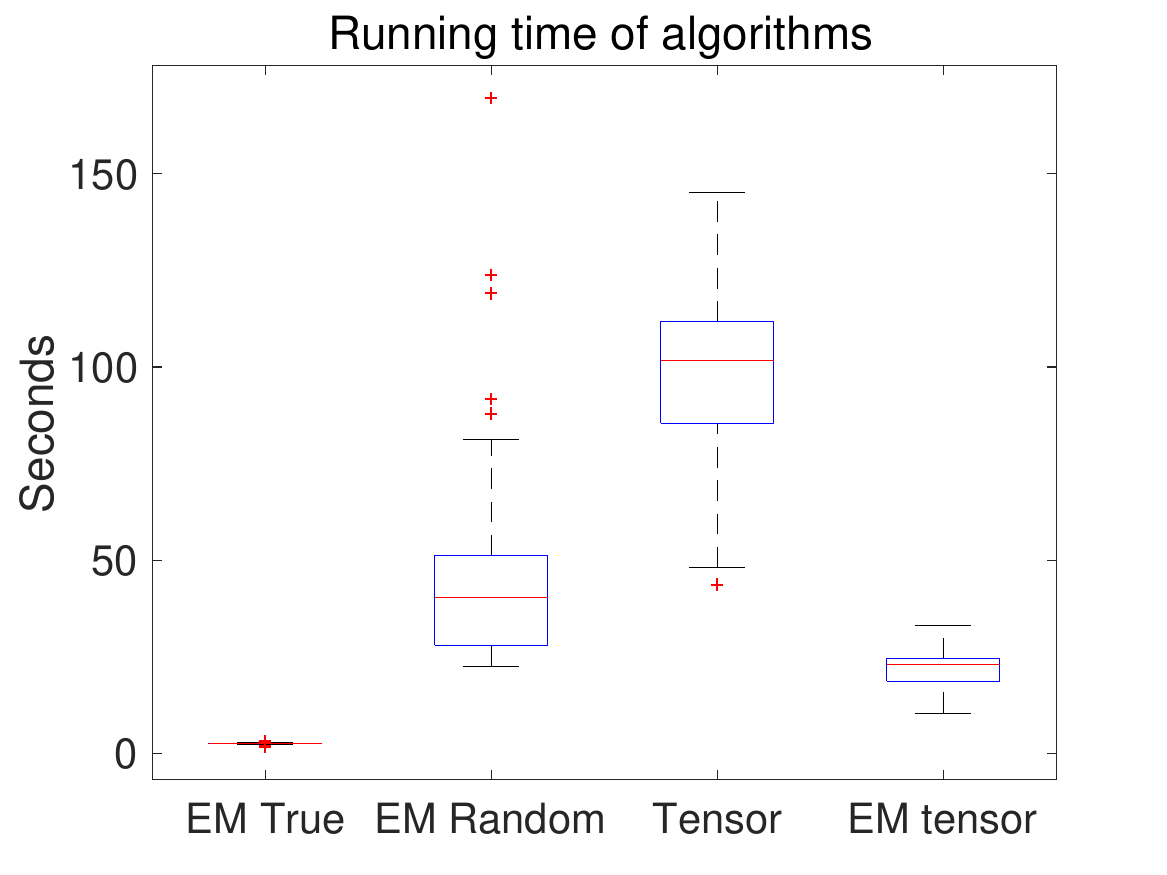}
		\end{minipage}%
	}%
	\centering
	\caption{$N = 10000, J= 200, L=5,$ item parameters $\in \{0.1,0.2,0.8,0.9\}$}
\end{figure}

\begin{figure}[H]
	\centering
	\subfigure[MSE of item parameters]{
		\begin{minipage}[t]{0.33\linewidth}
			\centering
			\includegraphics[width=2in]{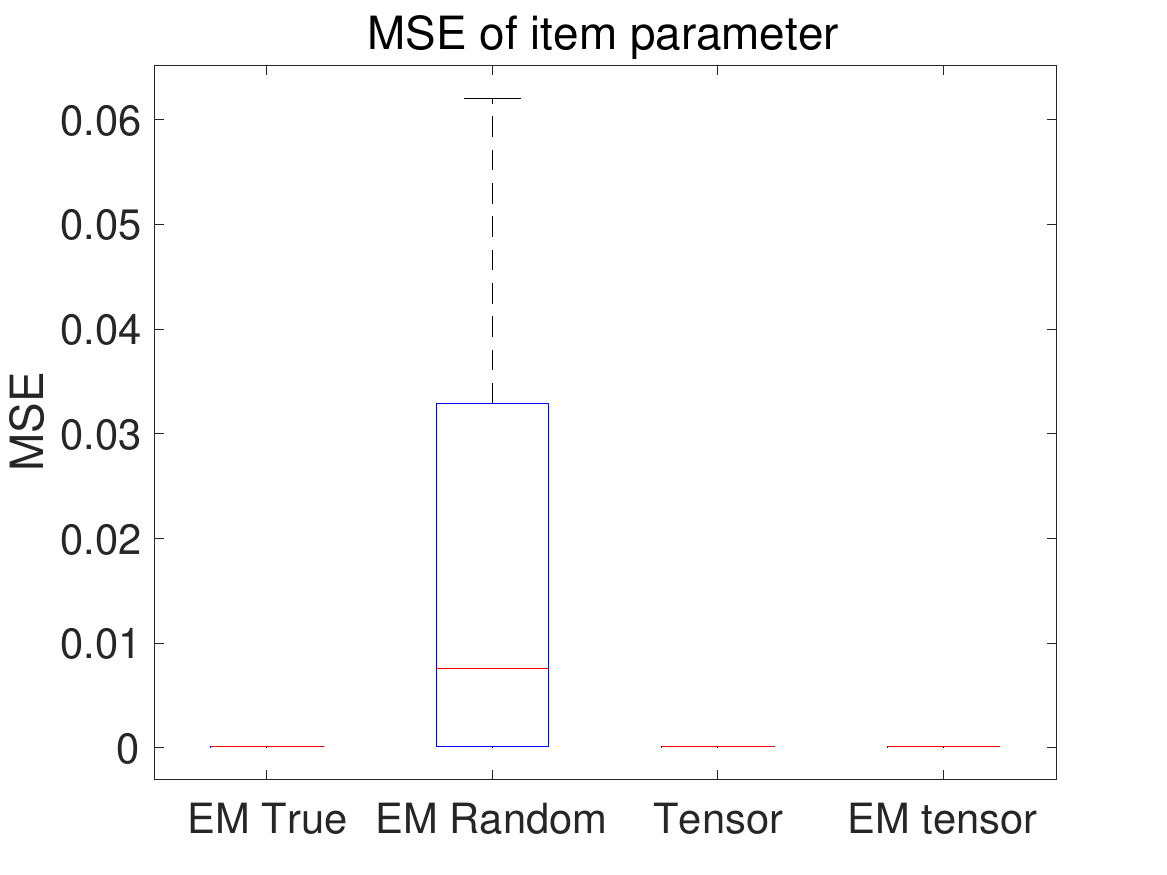}
		\end{minipage}%
	}%
	\subfigure[MSE without EM-random]{
		\begin{minipage}[t]{0.33\linewidth}
			\centering
			\includegraphics[width=2in]{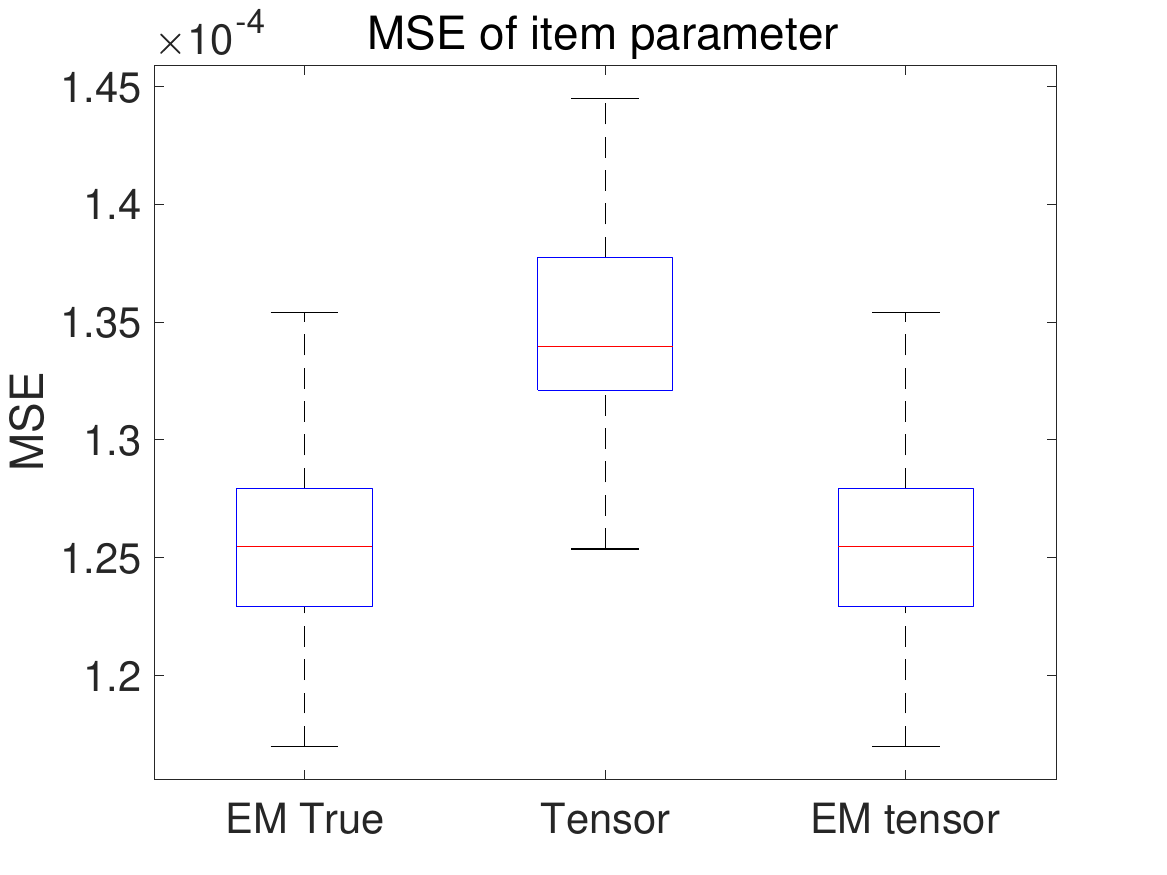}
		\end{minipage}%
	}%
	\subfigure[Running time of the algorithms]{
		\begin{minipage}[t]{0.33\linewidth}
			\centering
			\includegraphics[width=2in]{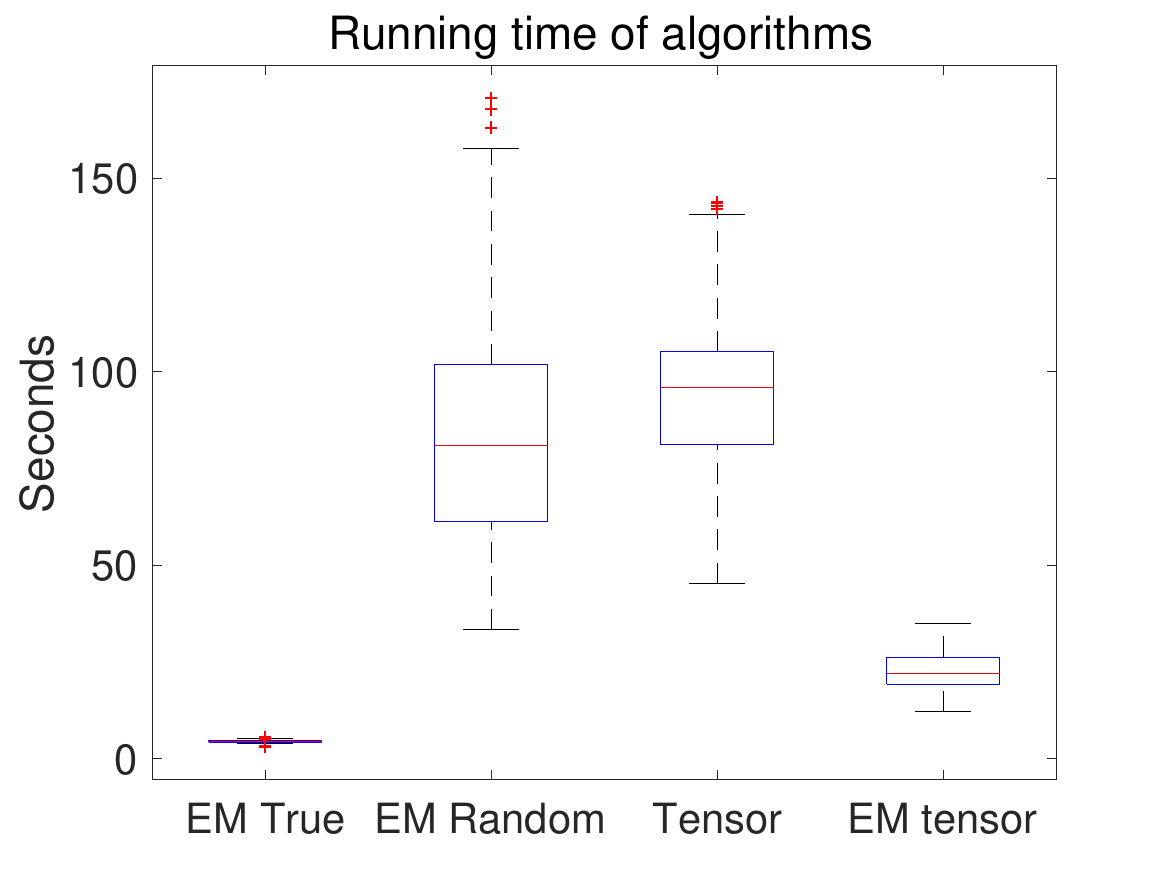}
		\end{minipage}%
	}%
	\centering
	\caption{$N = 10000, J= 200, L=10,$ item parameters $\in \{0.1,0.2,0.8,0.9\}$}
\end{figure}

\begin{figure}[H]
	\centering
	\subfigure[MSE of item parameters]{
		\begin{minipage}[t]{0.4\linewidth}
			\centering
			\includegraphics[width=2in]{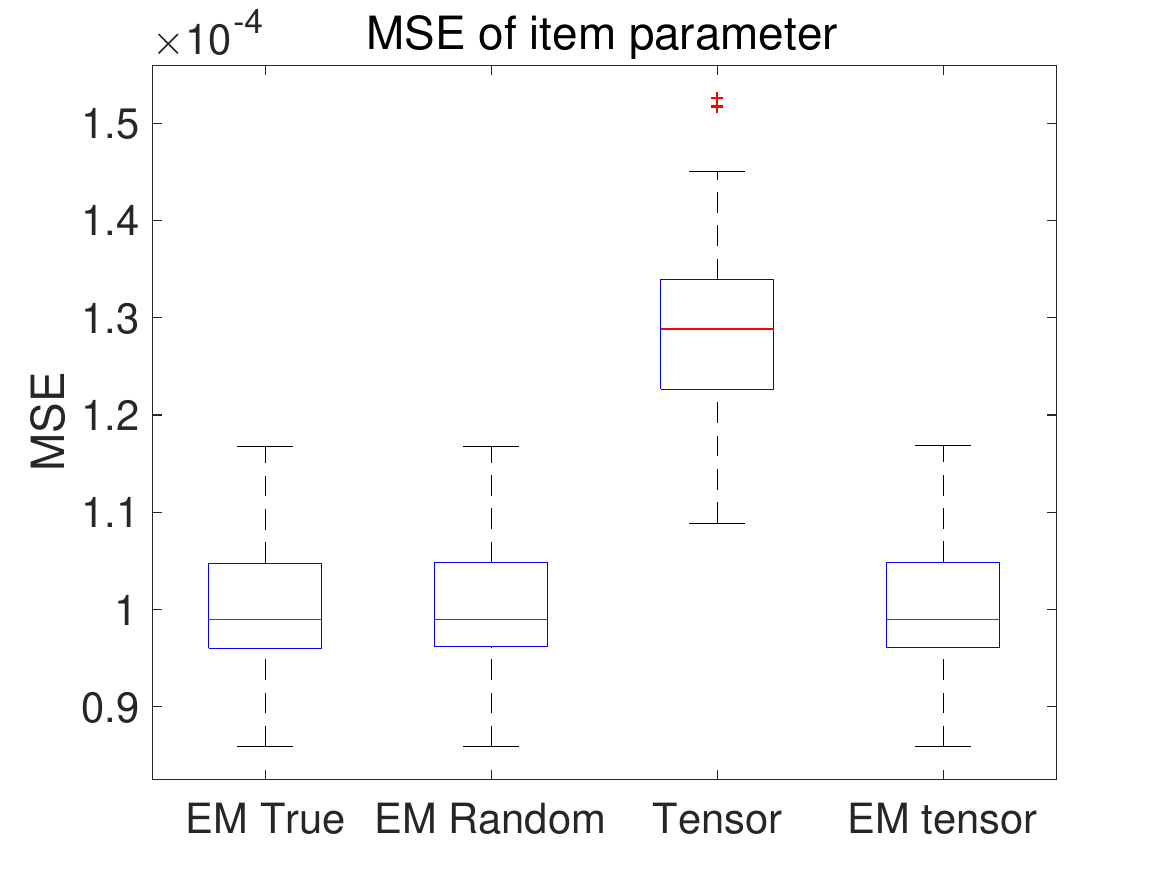}
		\end{minipage}%
	}%
	\subfigure[Running time of the algorithms]{
		\begin{minipage}[t]{0.4\linewidth}
			\centering
			\includegraphics[width=2in]{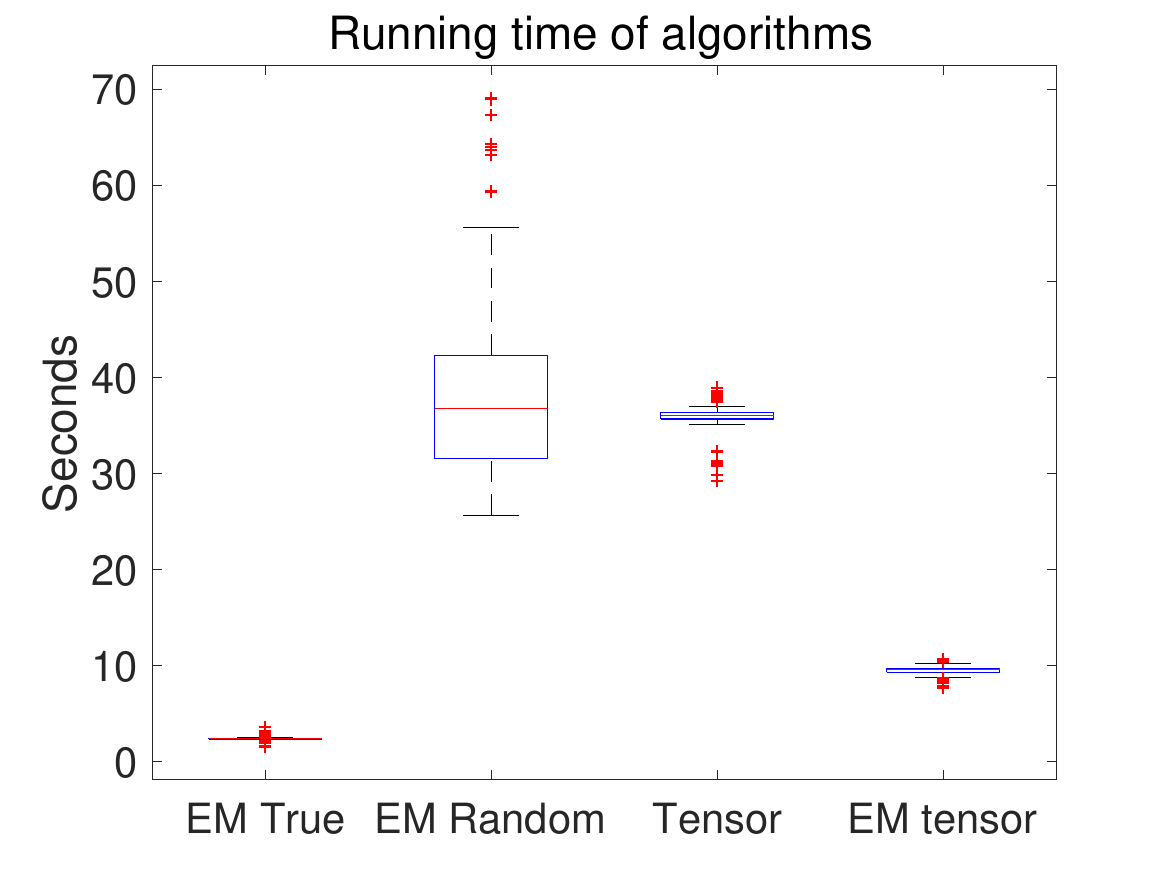}
		\end{minipage}%
	}%
	\centering
	\caption{$N = 10000, J= 100, L=5,$ item parameters $\in \{0.2,0.4,0.6,0.8\}$}
\end{figure}

\begin{figure}[H]
	\centering
	\subfigure[MSE of item parameters]{
		\begin{minipage}[t]{0.33\linewidth}
			\centering
			\includegraphics[width=2in]{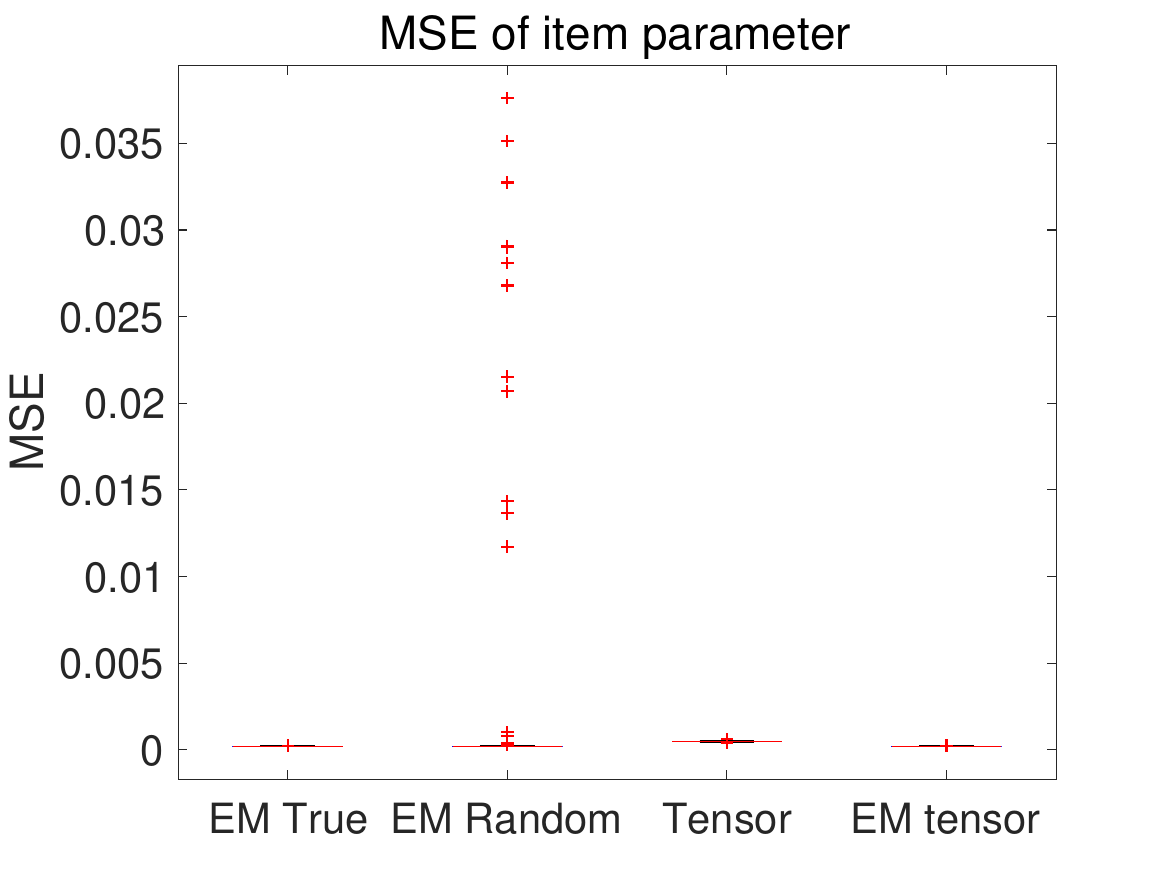}
		\end{minipage}%
	}%
	\subfigure[MSE without EM-random]{
		\begin{minipage}[t]{0.33\linewidth}
			\centering
			\includegraphics[width=2in]{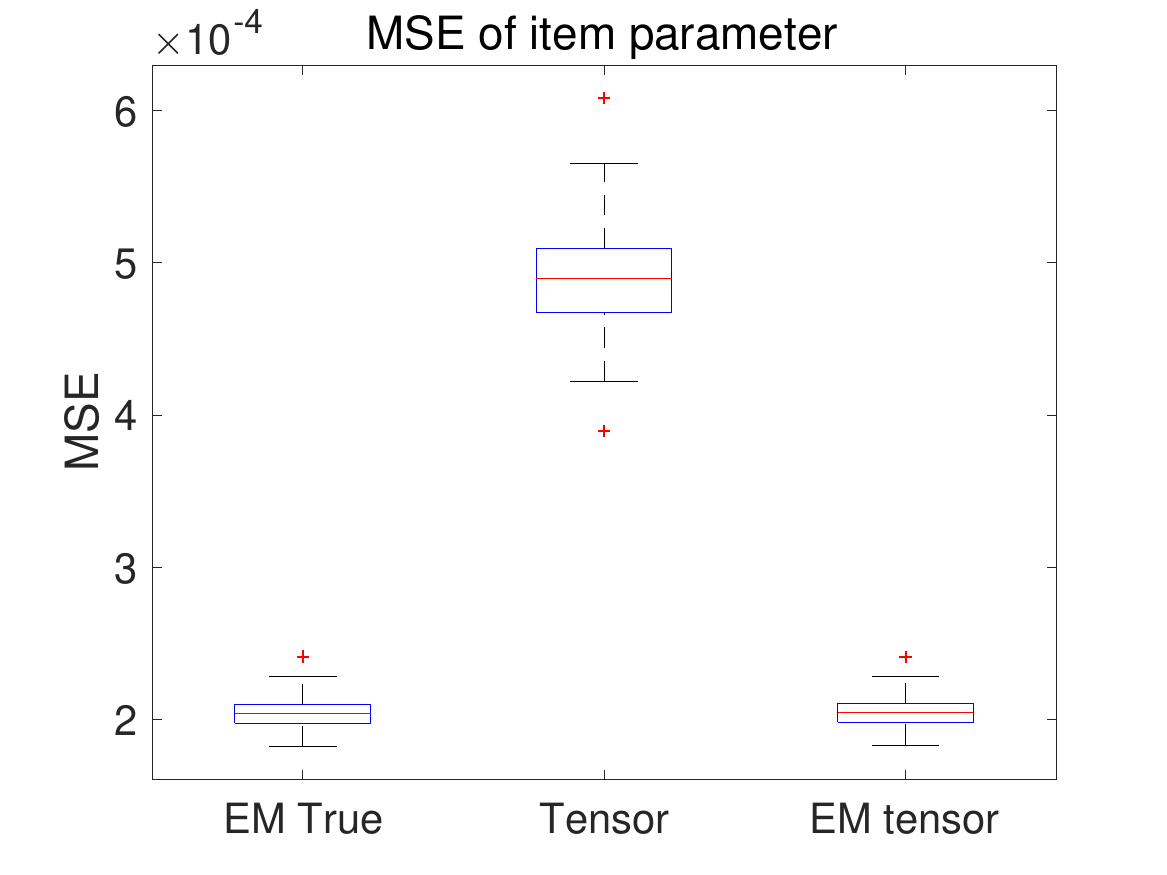}
		\end{minipage}%
	}%
	\subfigure[Running time of the algorithms]{
		\begin{minipage}[t]{0.33\linewidth}
			\centering
			\includegraphics[width=2in]{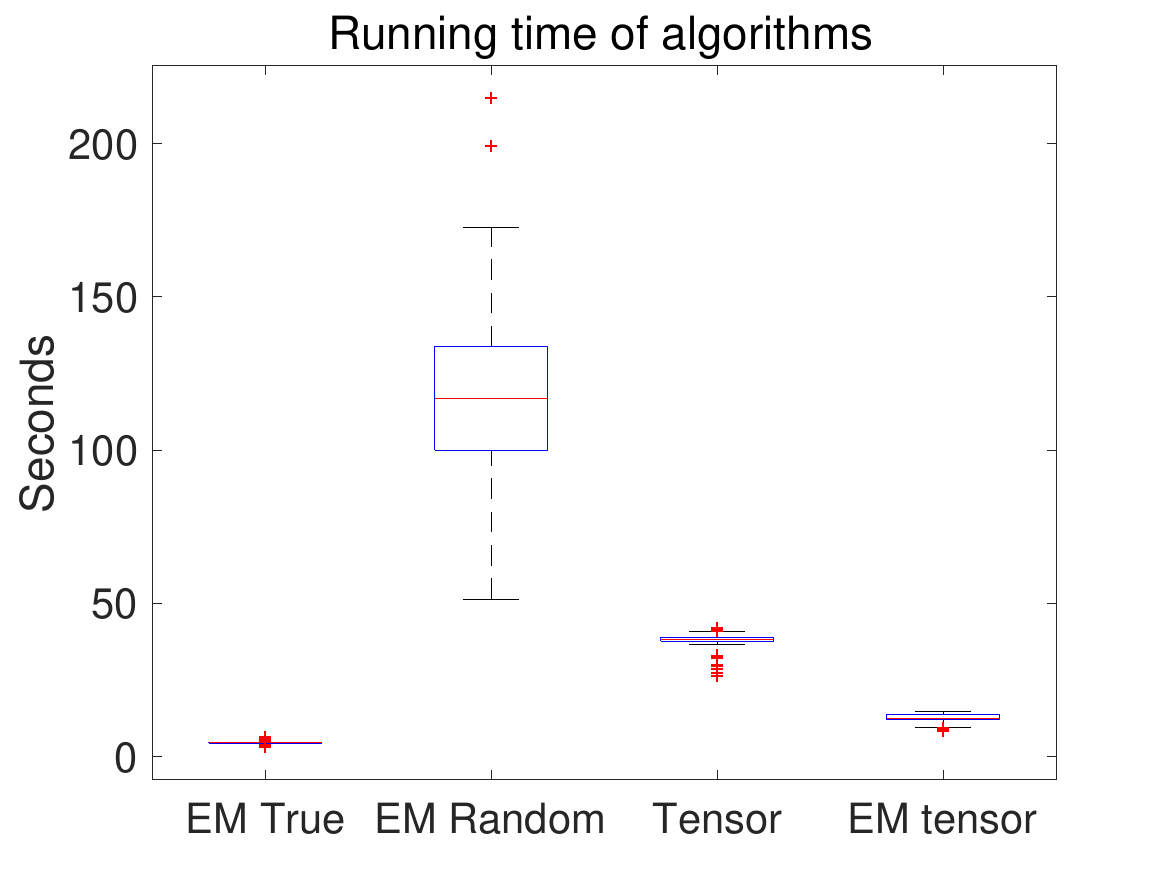}
		\end{minipage}%
	}%
	\centering
	\caption{$N = 10000, J= 100, L=10, $item parameters $\in \{0.2,0.4,0.6,0.8\}$}
\end{figure}

\begin{figure}[H]
	\centering
	\subfigure[MSE of item parameters]{
		\begin{minipage}[t]{0.4\linewidth}
			\centering
			\includegraphics[width=2in]{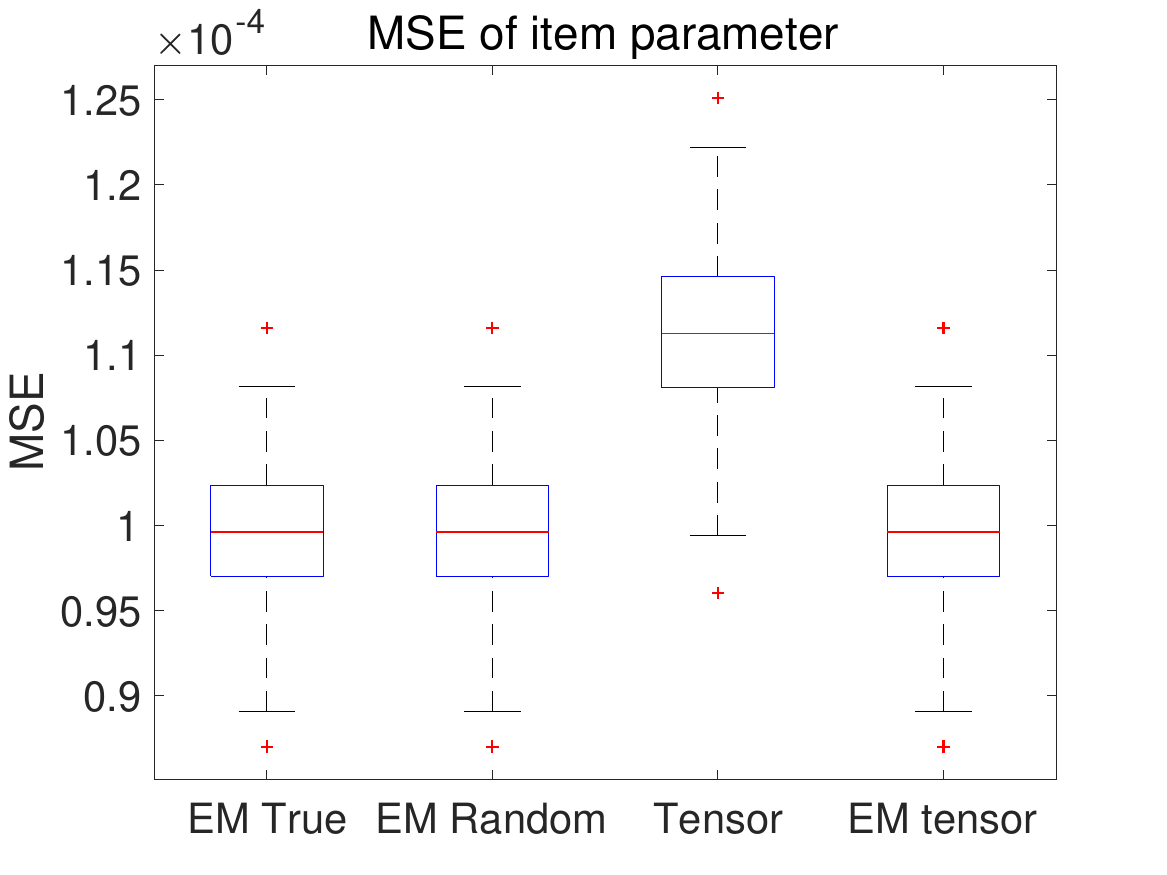}
		\end{minipage}%
	}%
	\subfigure[Running time of the algorithms]{
		\begin{minipage}[t]{0.4\linewidth}
			\centering
			\includegraphics[width=2in]{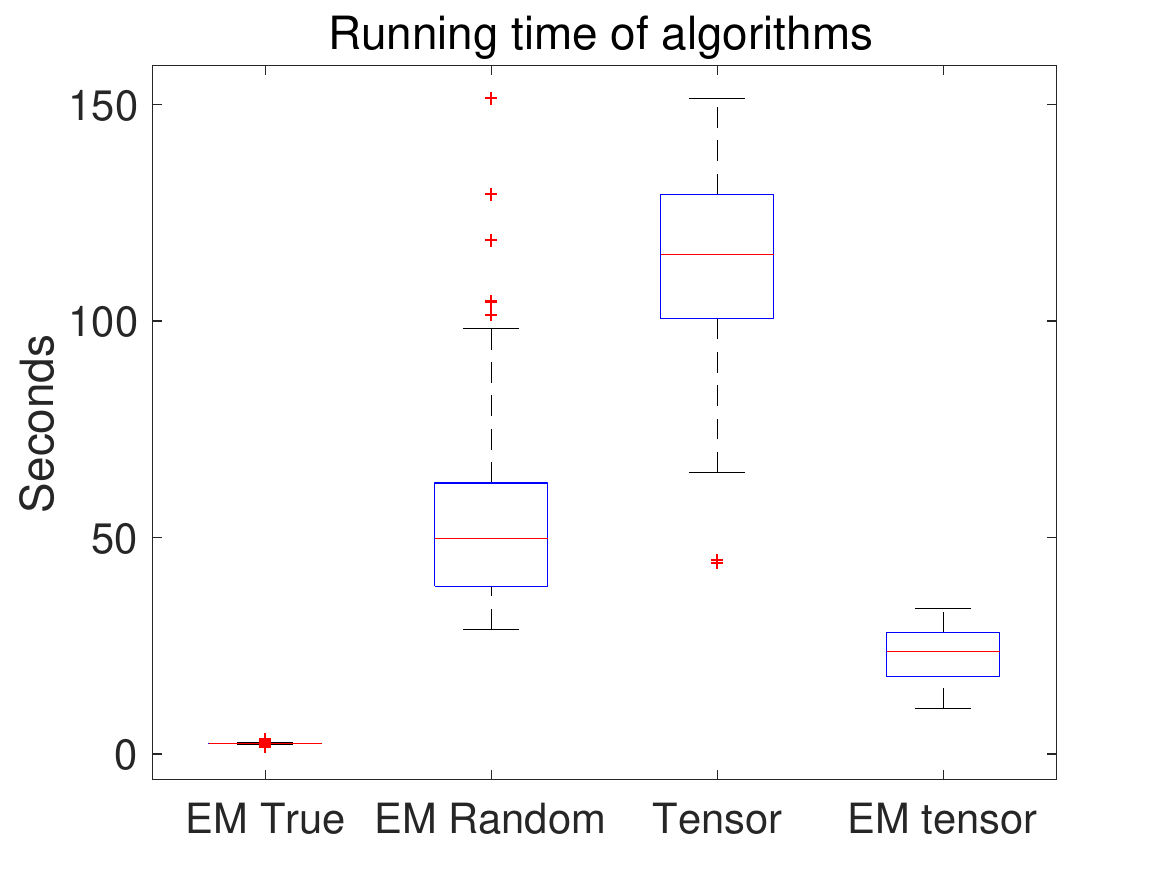}
		\end{minipage}%
	}%
	\centering
	\caption{$N = 10000, J= 200, L=5,$ item parameters $\in \{0.2,0.4,0.6,0.8\}$}
\end{figure}

\begin{figure}[H]
		\centering
		\subfigure[MSE of item parameters]{
			\begin{minipage}[t]{0.33\linewidth}
				\centering
				\includegraphics[width=2in]{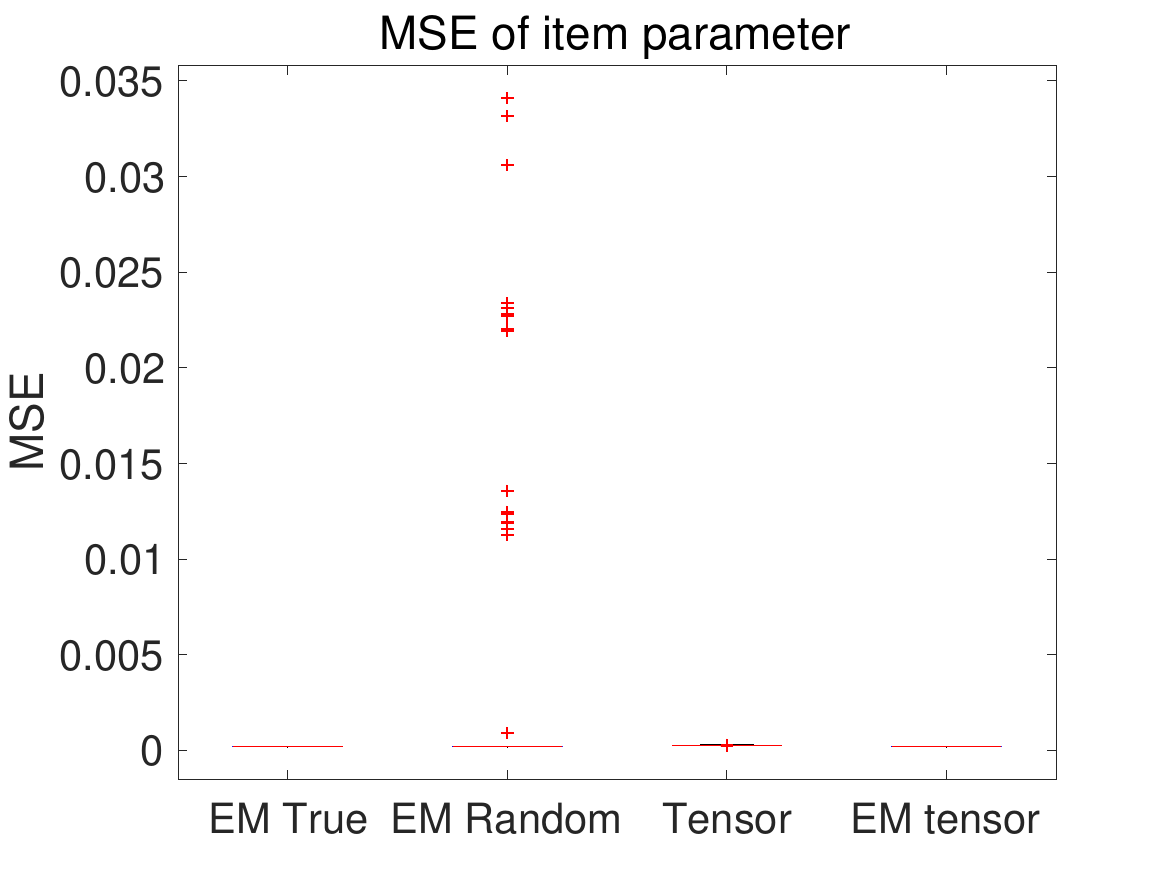}
			\end{minipage}%
		}%
		\subfigure[MSE without EM-random]{
		\begin{minipage}[t]{0.33\linewidth}
			\centering
			\includegraphics[width=2in]{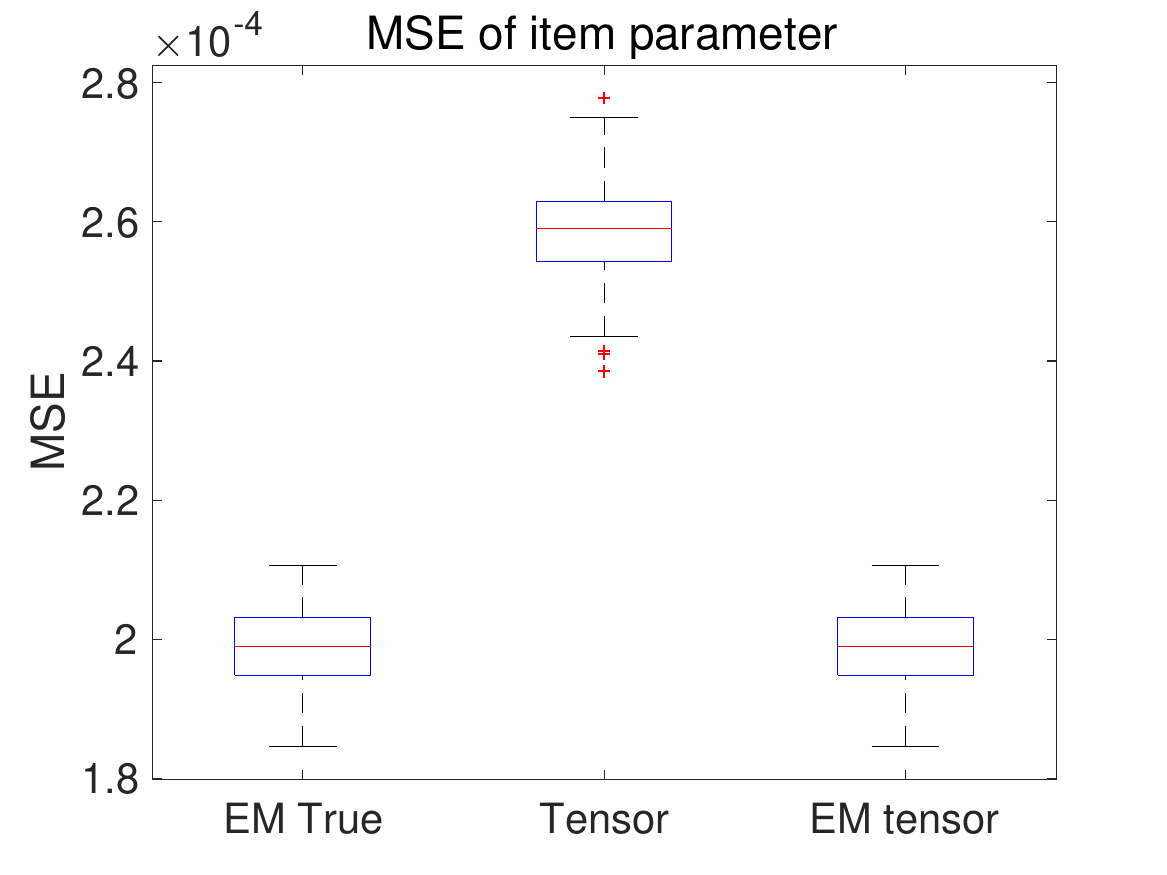}
		\end{minipage}%
	}%
		\subfigure[Running time of the algorithms]{
			\begin{minipage}[t]{0.33\linewidth}
				\centering
				\includegraphics[width=2in]{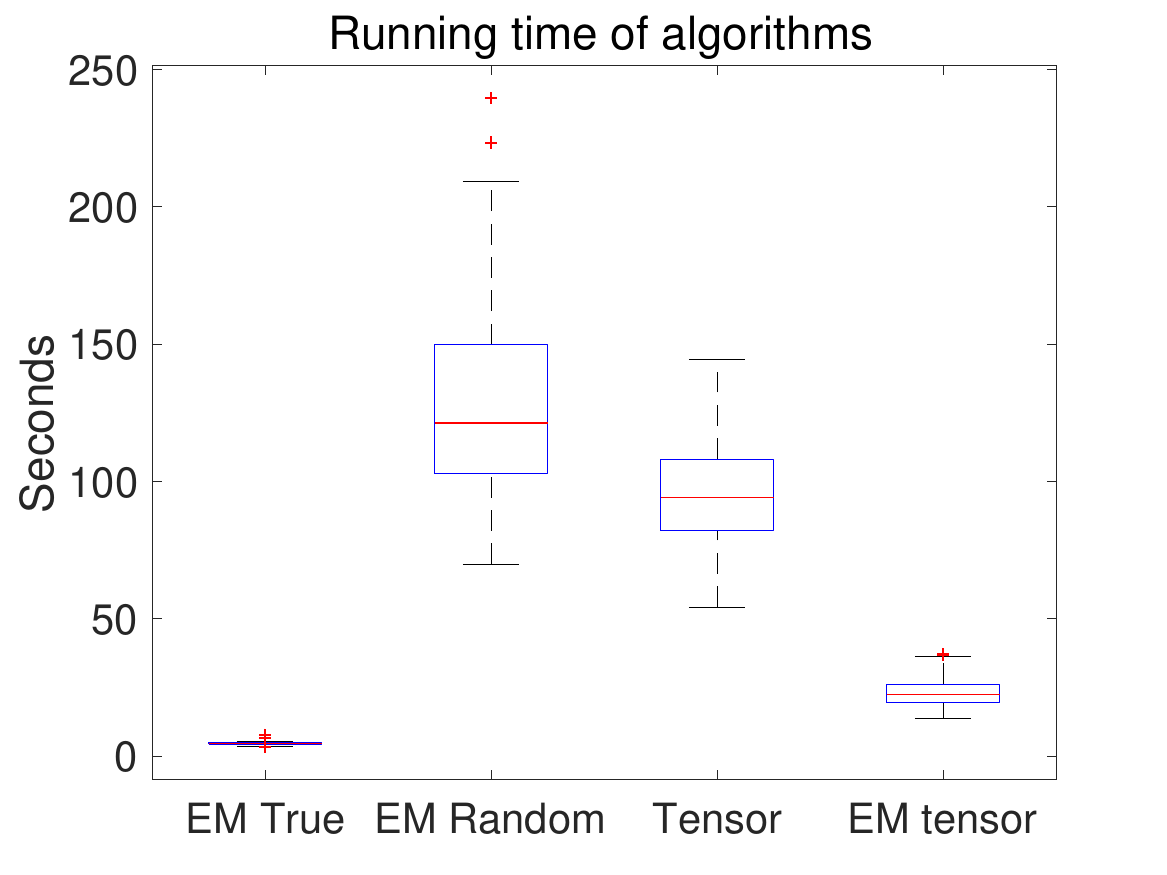}
			\end{minipage}%
		}%
		\centering
		\caption{$N = 10000, J= 200, L=10,$ item parameters $\in \{0.2,0.4,0.6,0.8\}$}
	\label{fig-fix-big2}
	\end{figure}
	
\begin{figure}[H]
	\centering
	\subfigure[MSE of item parameters]{
		\begin{minipage}[t]{0.4\linewidth}
			\centering
			\includegraphics[width=2in]{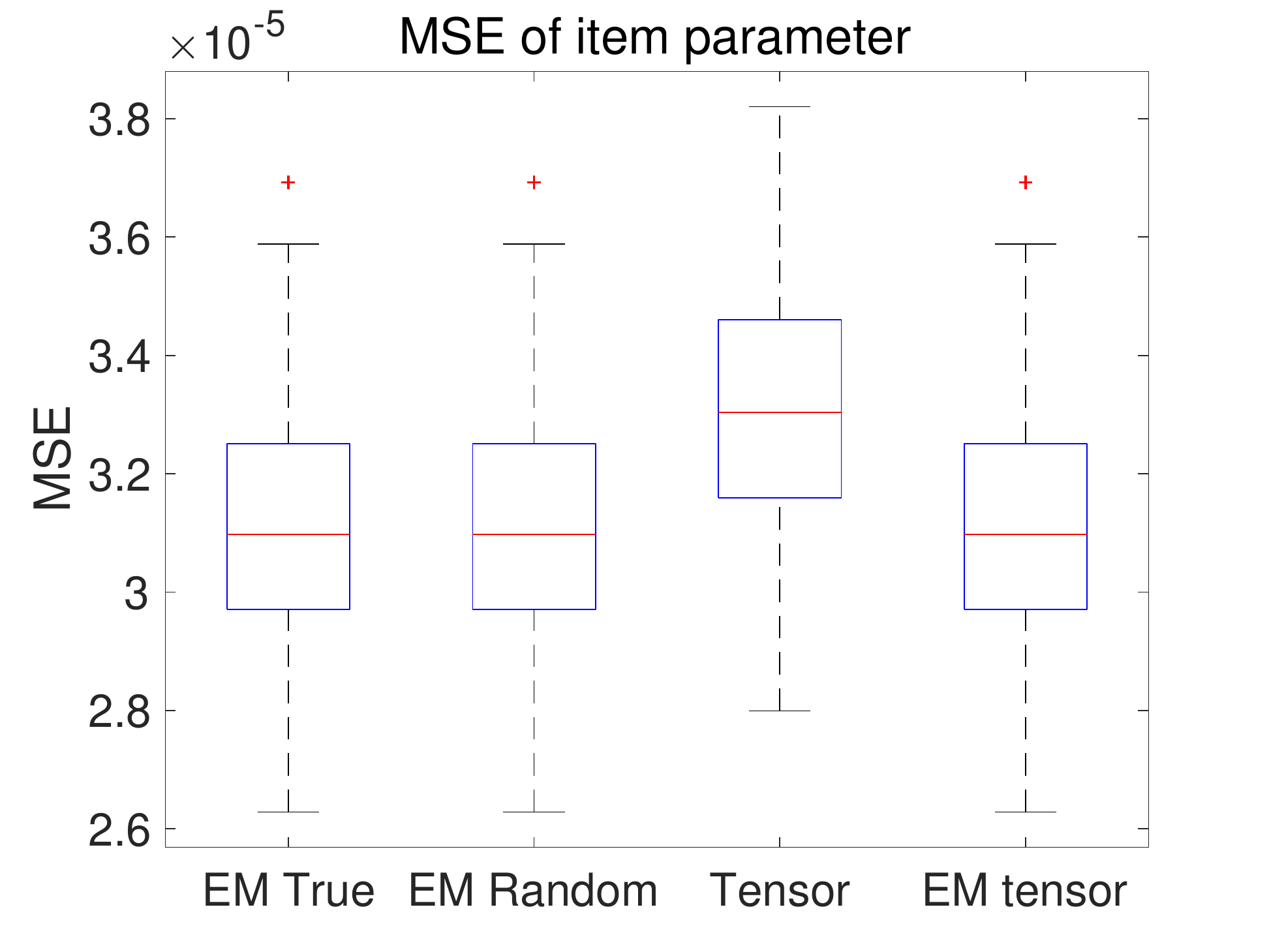}
		\end{minipage}%
	}%
	\subfigure[Running time of the algorithms]{
		\begin{minipage}[t]{0.4\linewidth}
			\centering
			\includegraphics[width=2in]{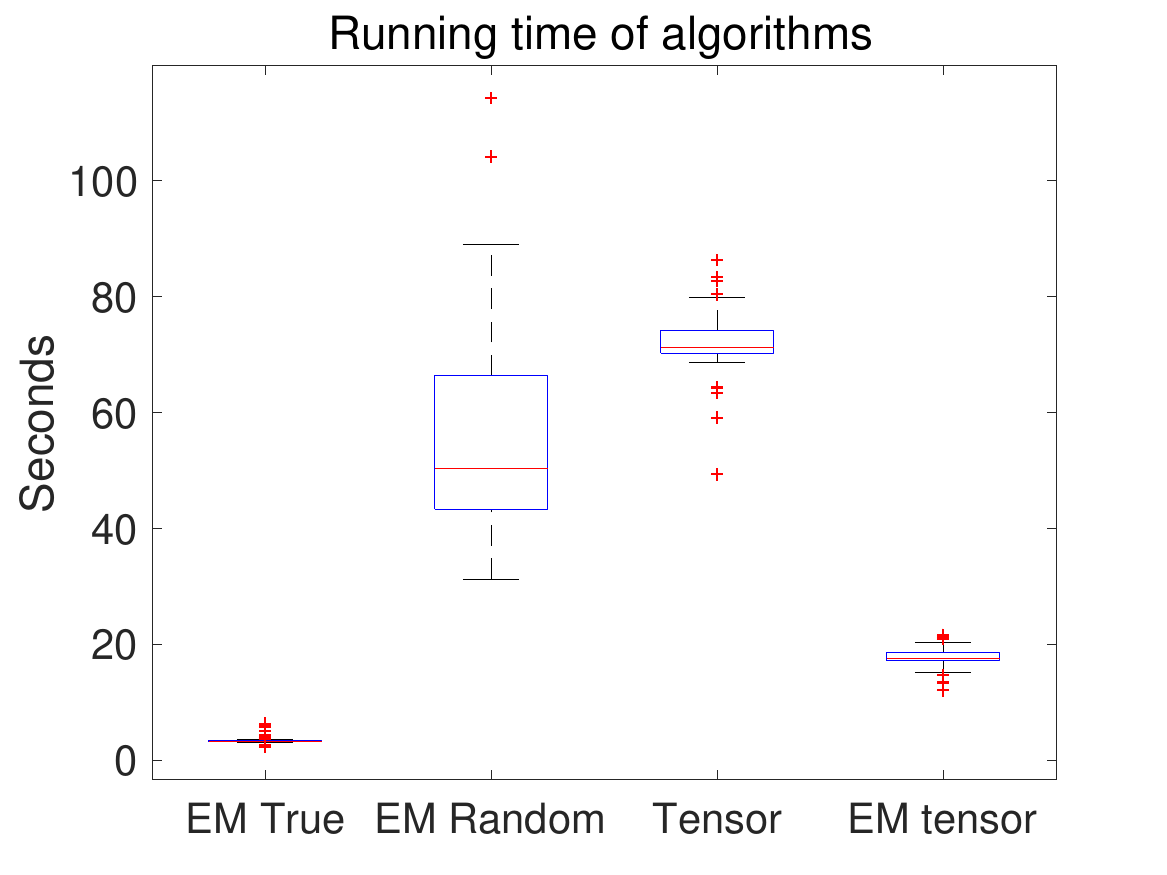}
		\end{minipage}%
	}%
	\centering
	\caption{$N = 20000, J= 100, L=5,$ item parameters $\in \{0.1,0.2,0.8,0.9\}$}
\end{figure}

\begin{figure}[H]
	\centering
	\subfigure[MSE of item parameters]{
		\begin{minipage}[t]{0.33\linewidth}
			\centering
			\includegraphics[width=2in]{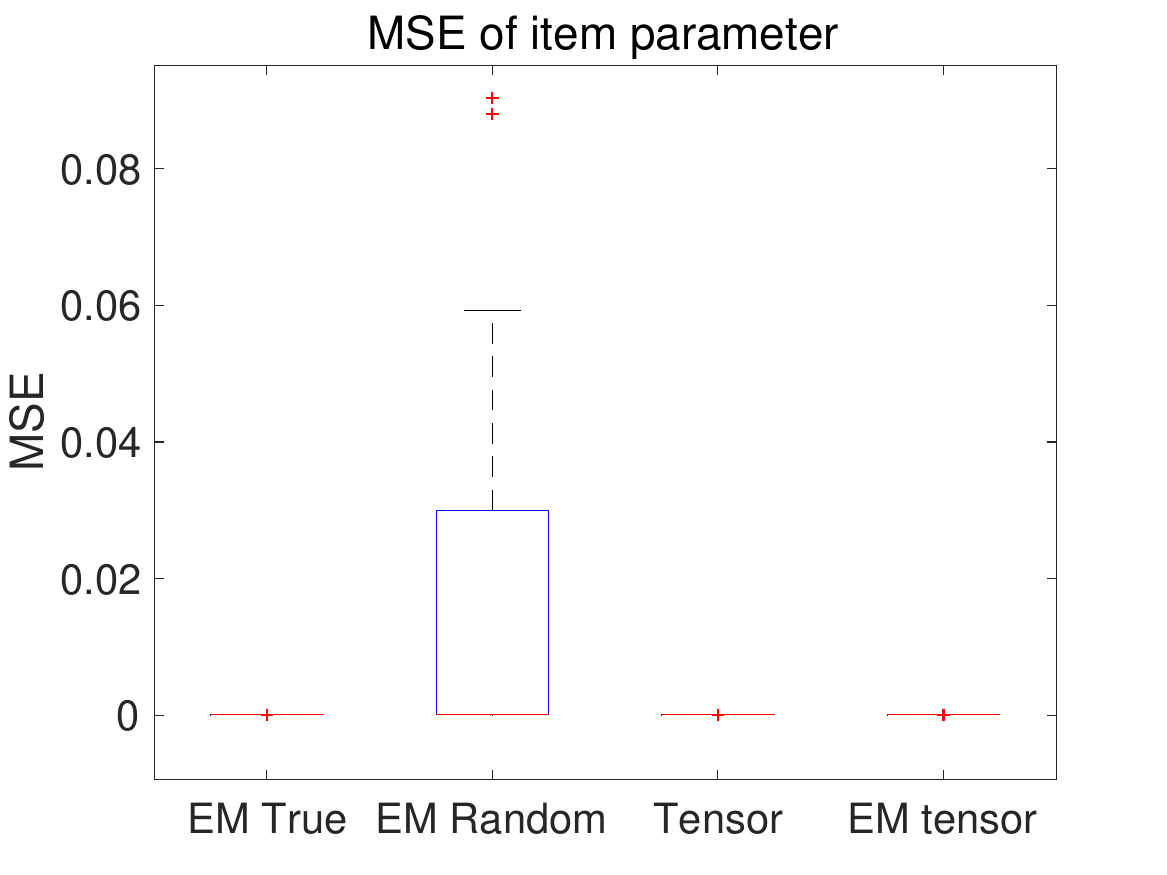}
		\end{minipage}%
	}%
	\subfigure[MSE without EM-random]{
		\begin{minipage}[t]{0.33\linewidth}
			\centering
			\includegraphics[width=2in]{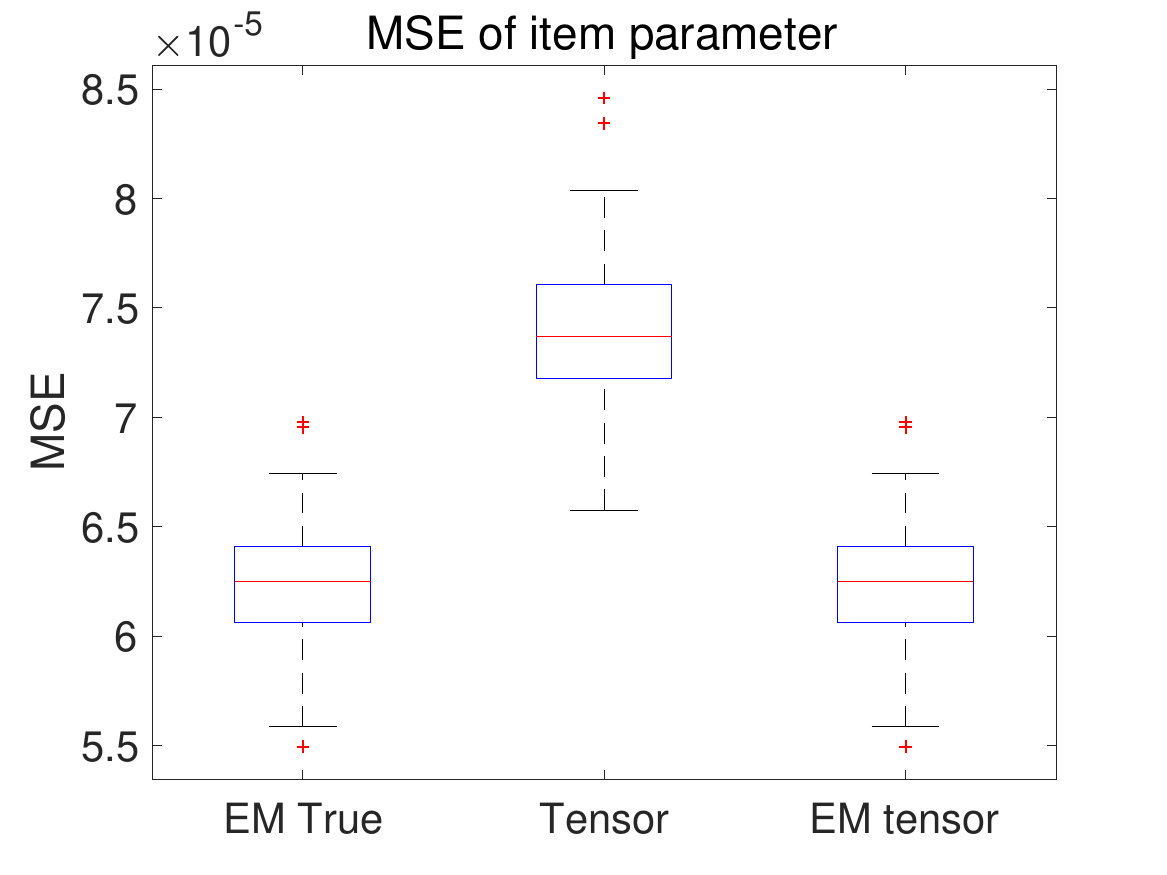}
		\end{minipage}%
	}%
		\subfigure[MSE without EM-random]{
	\begin{minipage}[t]{0.33\linewidth}
		\centering
		\includegraphics[width=2in]{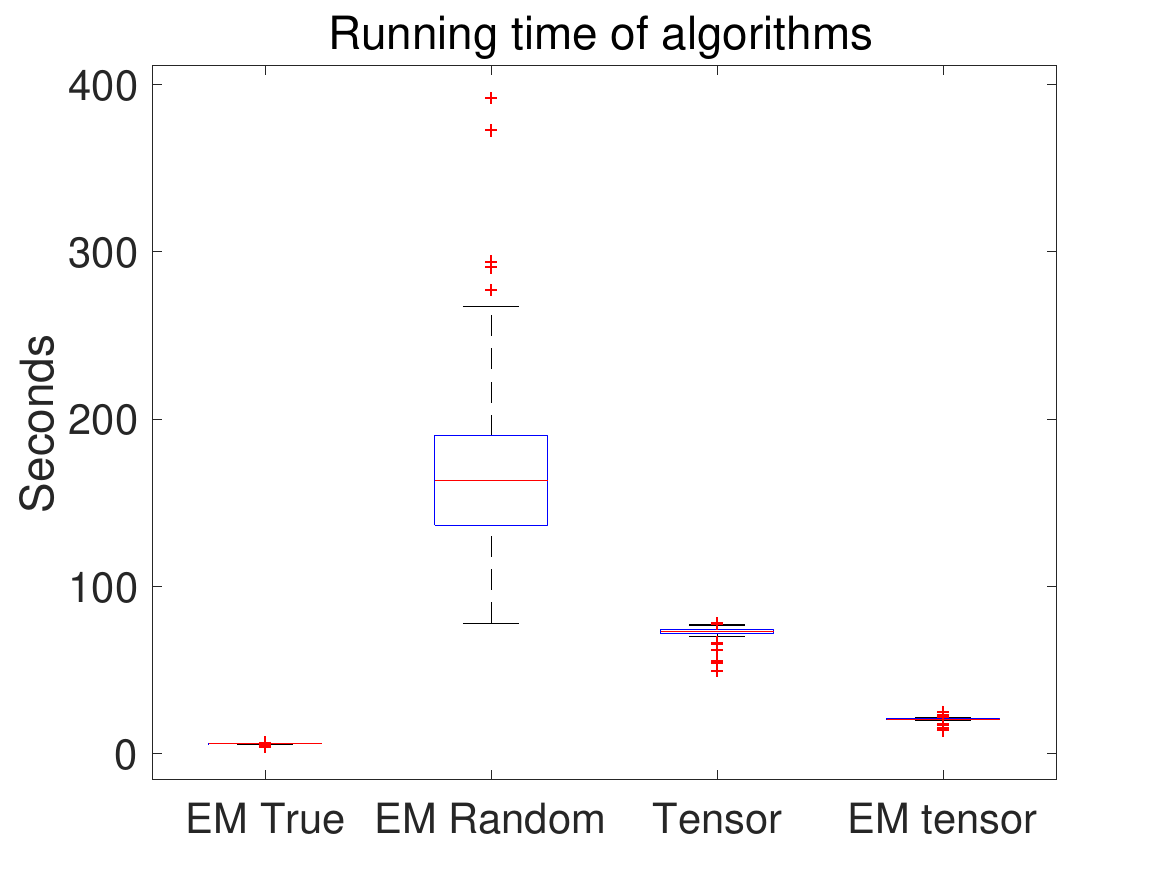}
	\end{minipage}%
}%
	\centering
	\caption{$N = 20000, J= 100, L=10, $ item parameters $\in \{0.1,0.2,0.8,0.9\}$}
\end{figure}

\begin{figure}[H]
	\centering
	\subfigure[MSE of item parameters]{
		\begin{minipage}[t]{0.4\linewidth}
			\centering
			\includegraphics[width=2in]{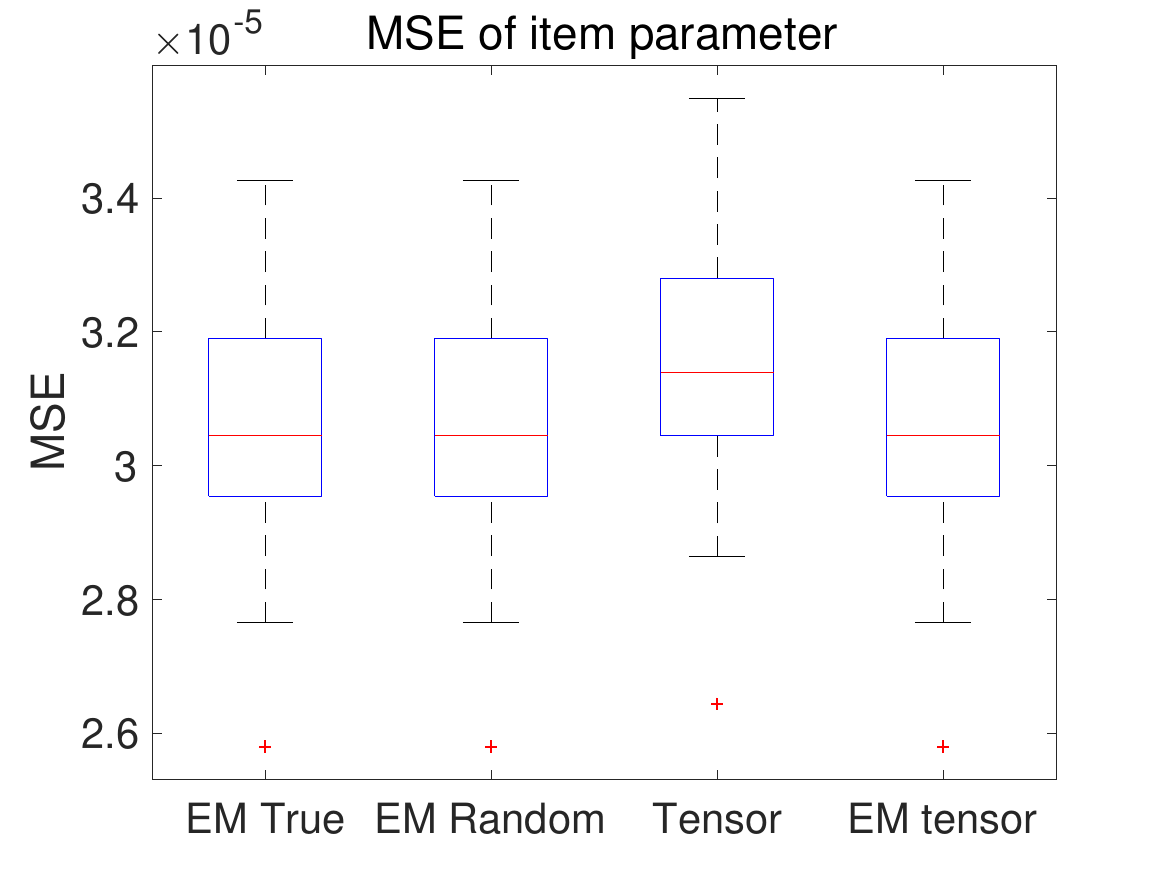}
		\end{minipage}%
	}%
	\subfigure[Running time of the algorithms]{
		\begin{minipage}[t]{0.4\linewidth}
			\centering
			\includegraphics[width=2in]{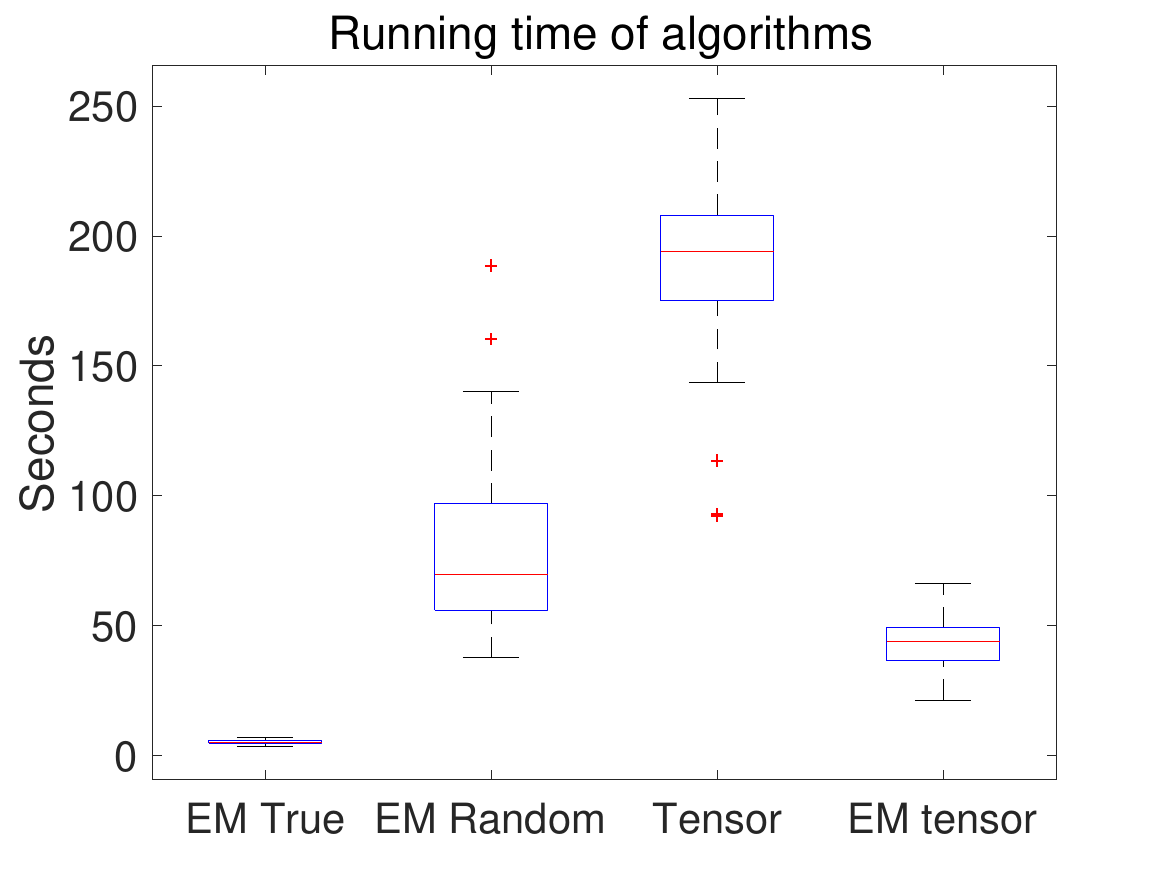}
		\end{minipage}%
	}%
	\centering
	\caption{$N = 20000, J= 200, L=5,$ item parameters $\in \{0.1,0.2,0.8,0.9\}$}
\end{figure}

\begin{figure}[H]
	\centering
	\subfigure[MSE of item parameters]{
		\begin{minipage}[t]{0.33\linewidth}
			\centering
			\includegraphics[width=2in]{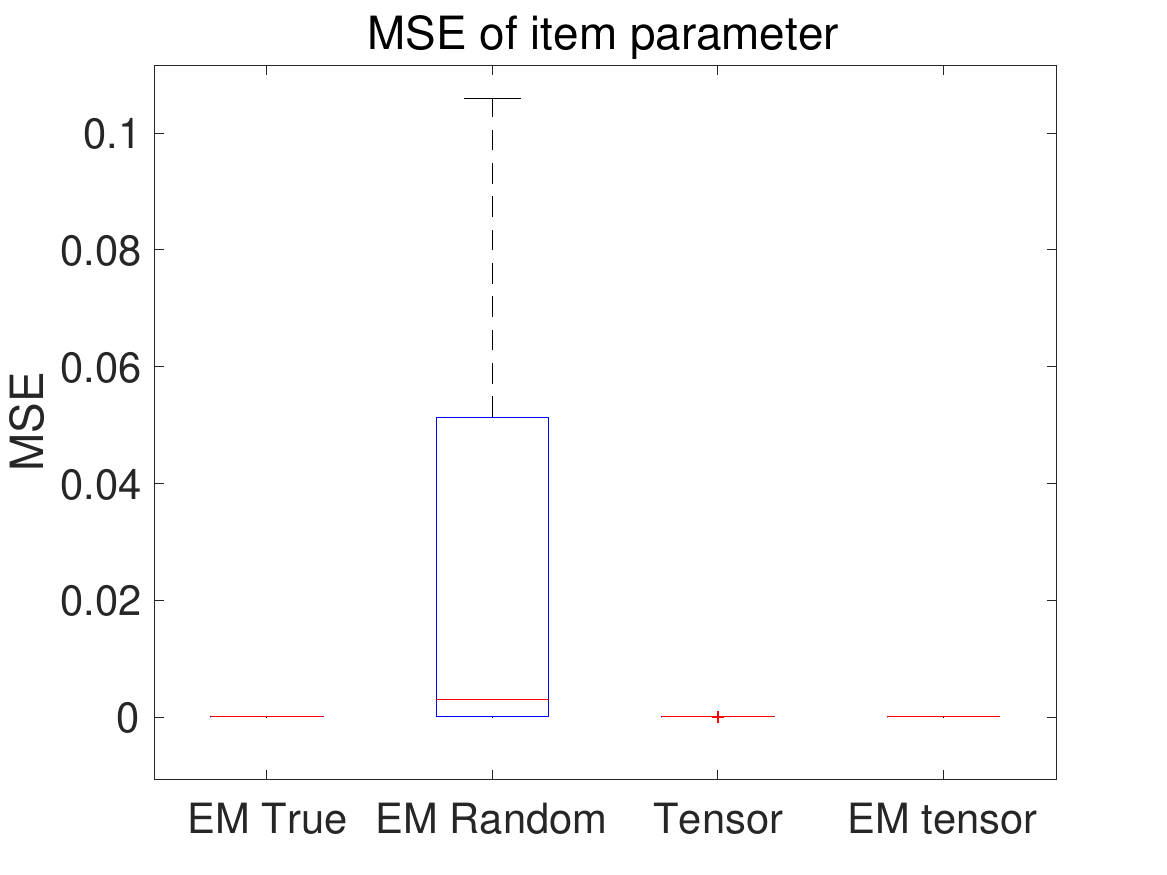}
		\end{minipage}%
	}%
	\subfigure[MSE without EM-random]{
		\begin{minipage}[t]{0.33\linewidth}
			\centering
			\includegraphics[width=2in]{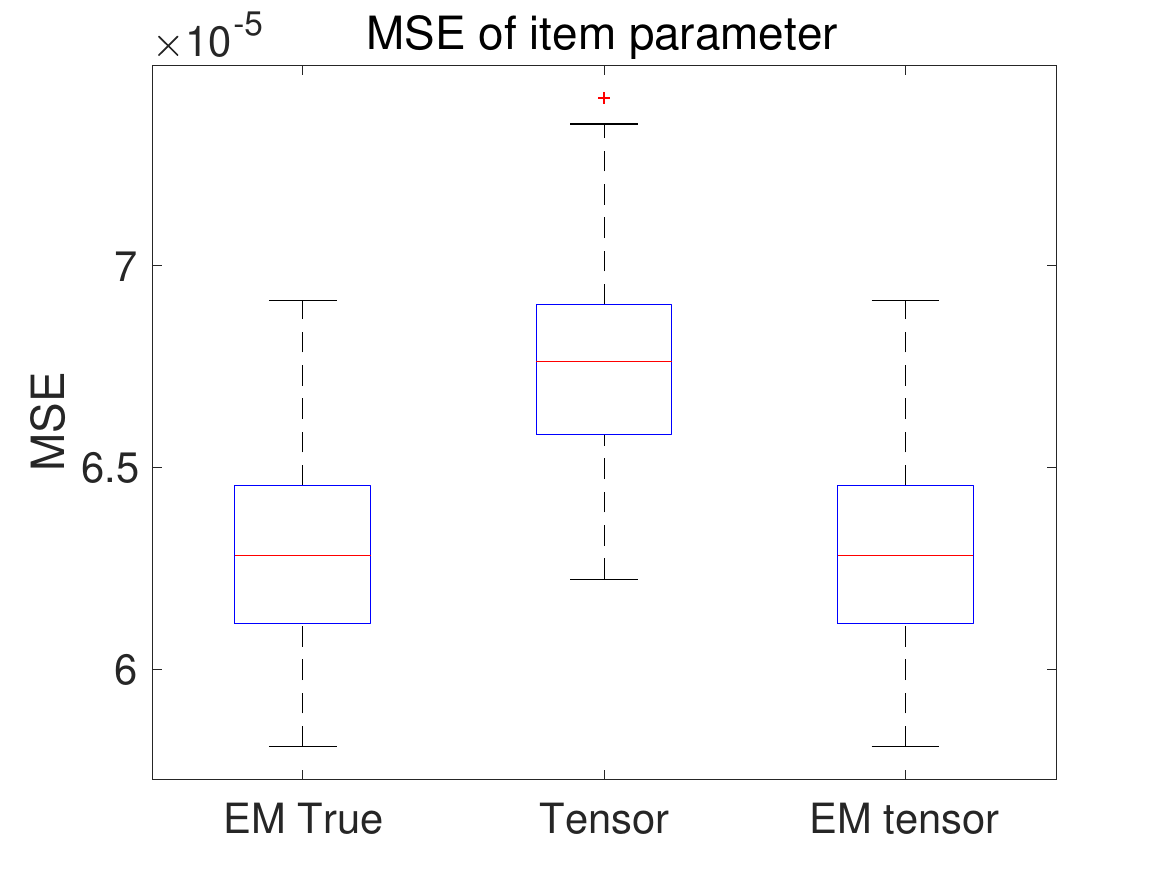}
		\end{minipage}%
	}%
	\subfigure[Running time of the algorithms]{
		\begin{minipage}[t]{0.33\linewidth}
			\centering
			\includegraphics[width=2in]{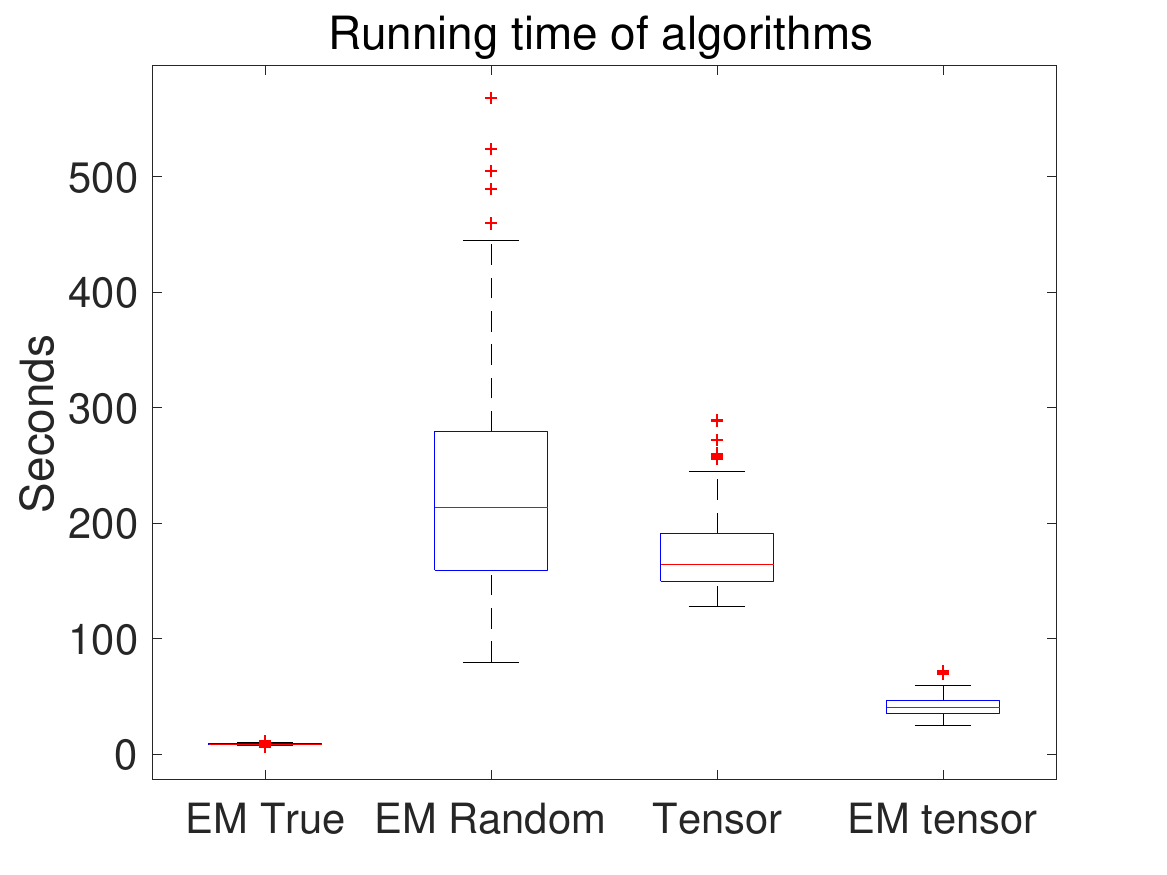}
		\end{minipage}%
	}%
	\centering
	\caption{$N = 20000, J= 200, L=10,$ item parameters $\in \{0.1,0.2,0.8,0.9\}$}
\end{figure}

\begin{figure}[H]
	\centering
	\subfigure[MSE of item parameters]{
		\begin{minipage}[t]{0.4\linewidth}
			\centering
			\includegraphics[width=2in]{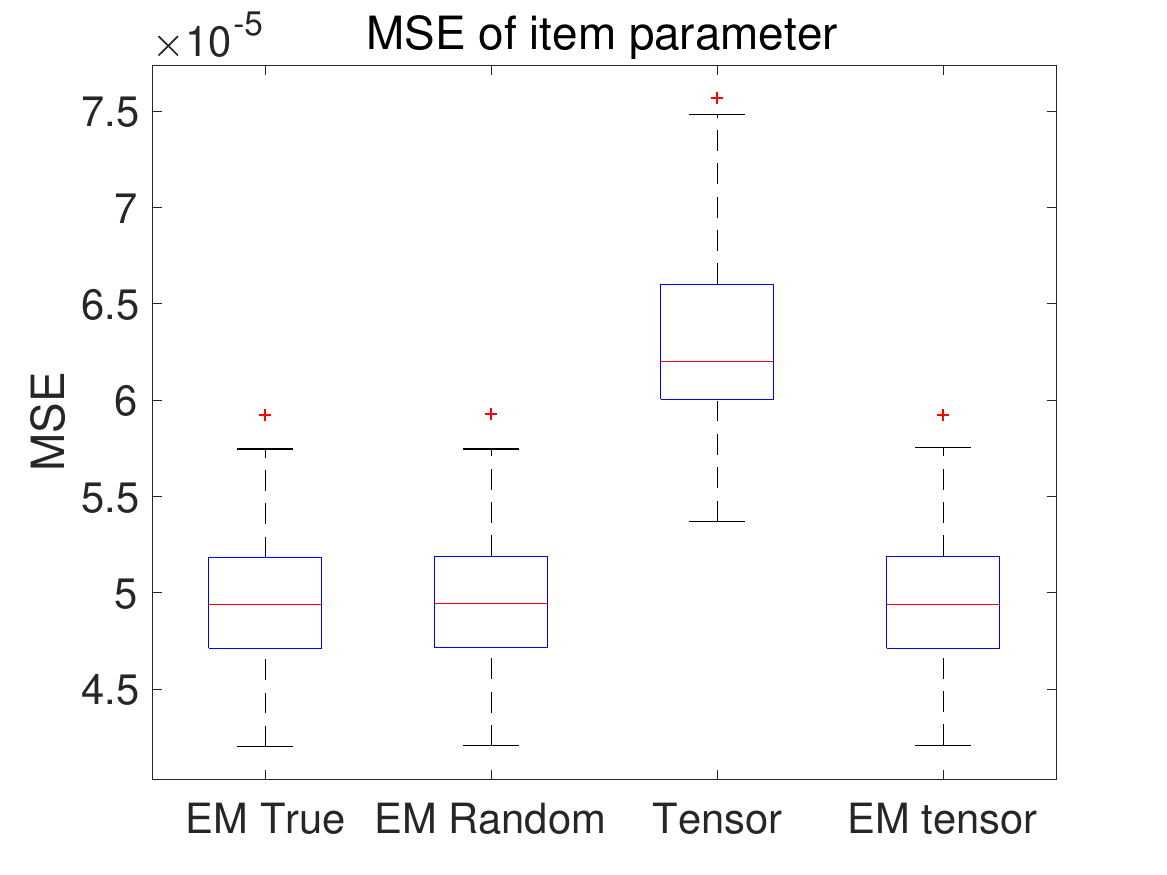}
		\end{minipage}%
	}%
	\subfigure[Running time of the algorithms]{
		\begin{minipage}[t]{0.4\linewidth}
			\centering
			\includegraphics[width=2in]{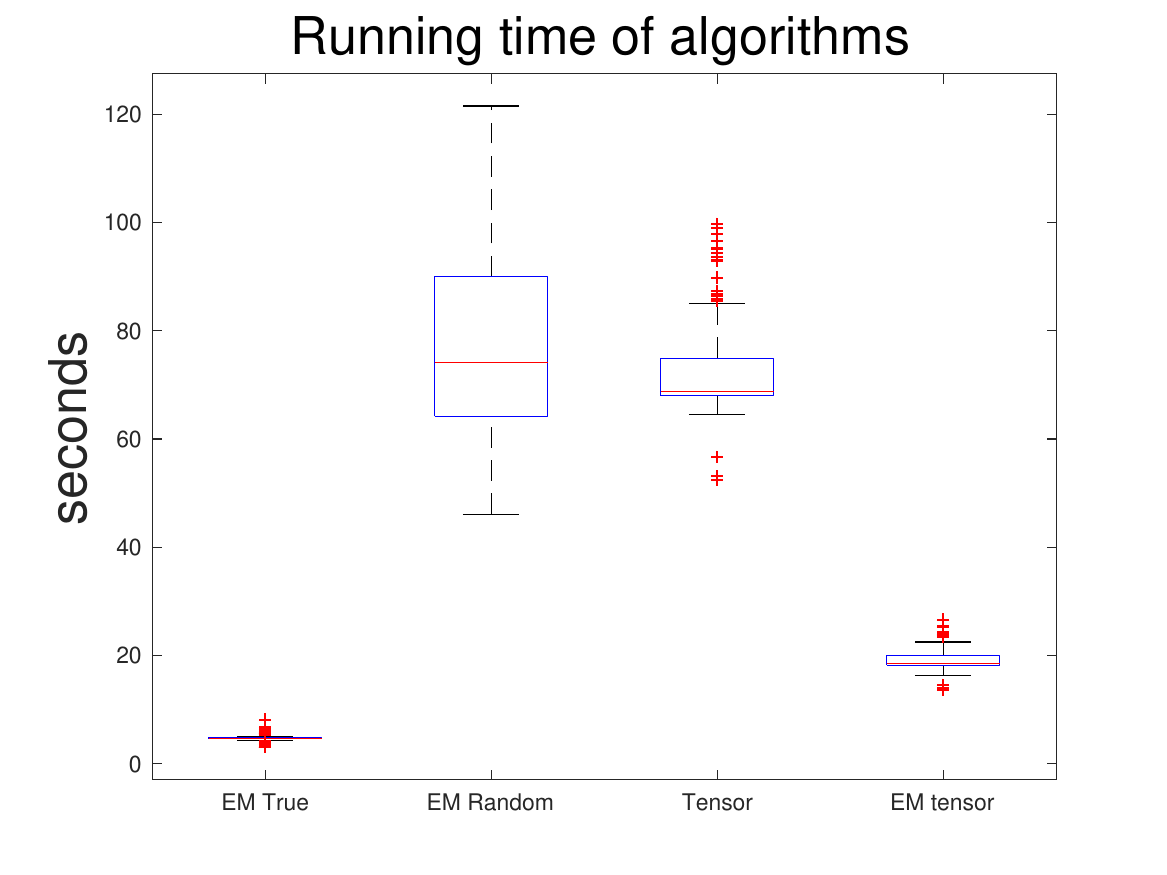}
		\end{minipage}%
	}%
	\centering
	\caption{$N = 20000, J= 100, L=5,$ item parameters $\in \{0.2,0.4,0.6,0.8\}$}
\end{figure}

\begin{figure}[H]
	\centering
	\subfigure[MSE of item parameters]{
		\begin{minipage}[t]{0.33\linewidth}
			\centering
			\includegraphics[width=2in]{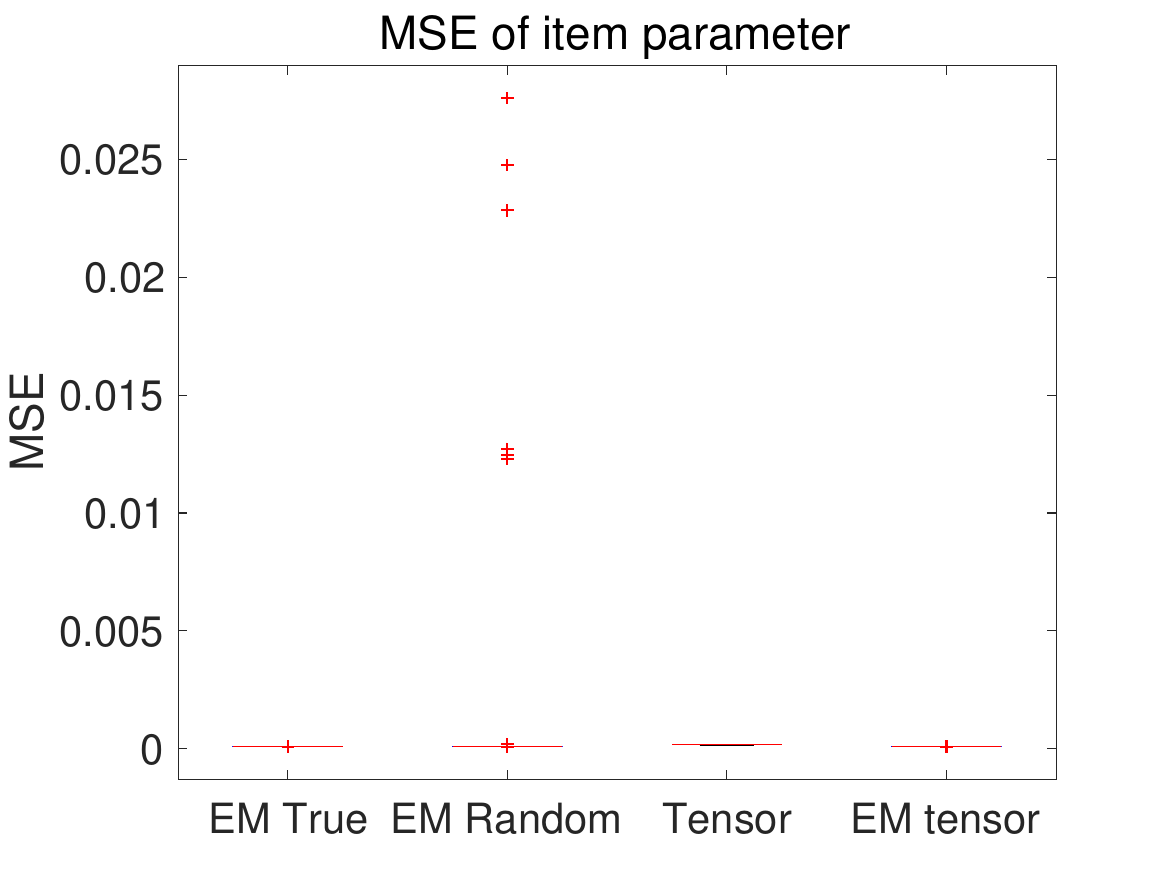}
		\end{minipage}%
	}%
	\subfigure[MSE without EM-random]{
		\begin{minipage}[t]{0.33\linewidth}
			\centering
			\includegraphics[width=2in]{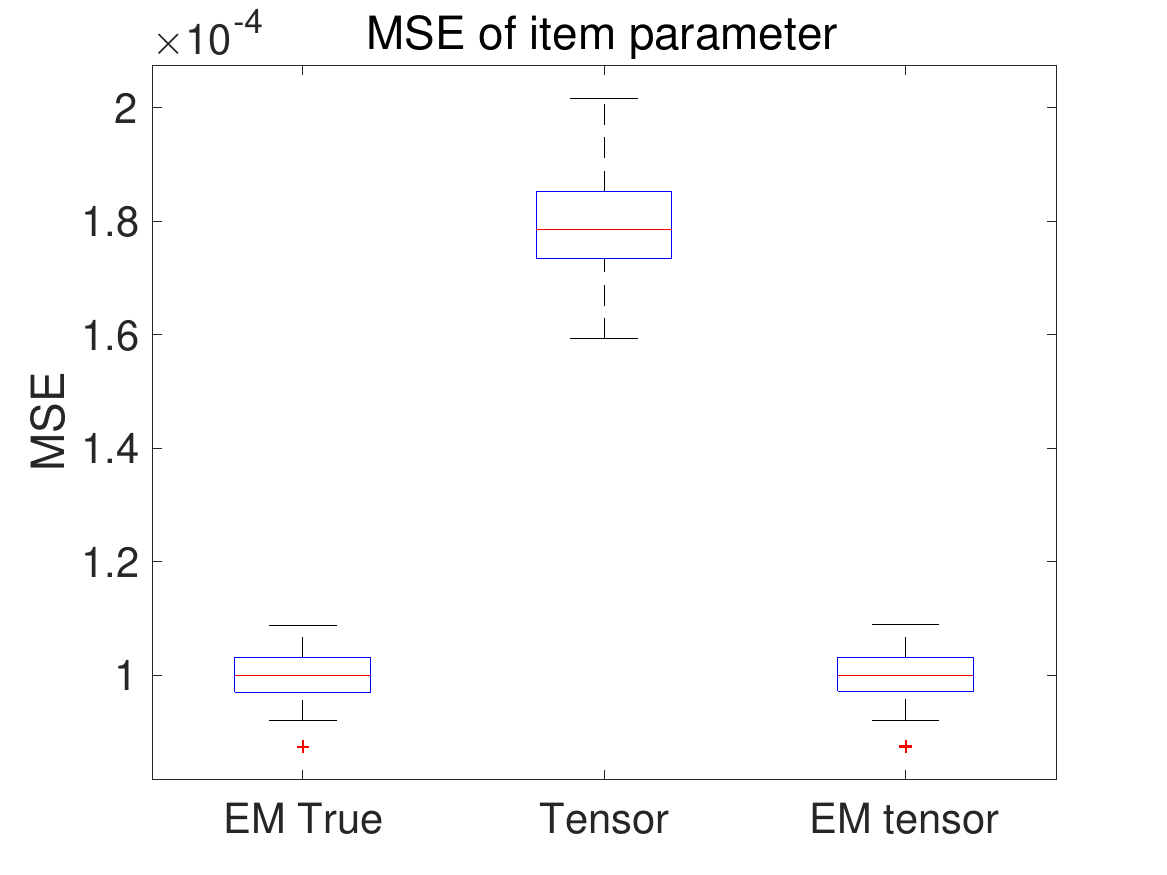}
		\end{minipage}%
	}%
	\subfigure[Running time of the algorithms]{
		\begin{minipage}[t]{0.33\linewidth}
			\centering
			\includegraphics[width=2in]{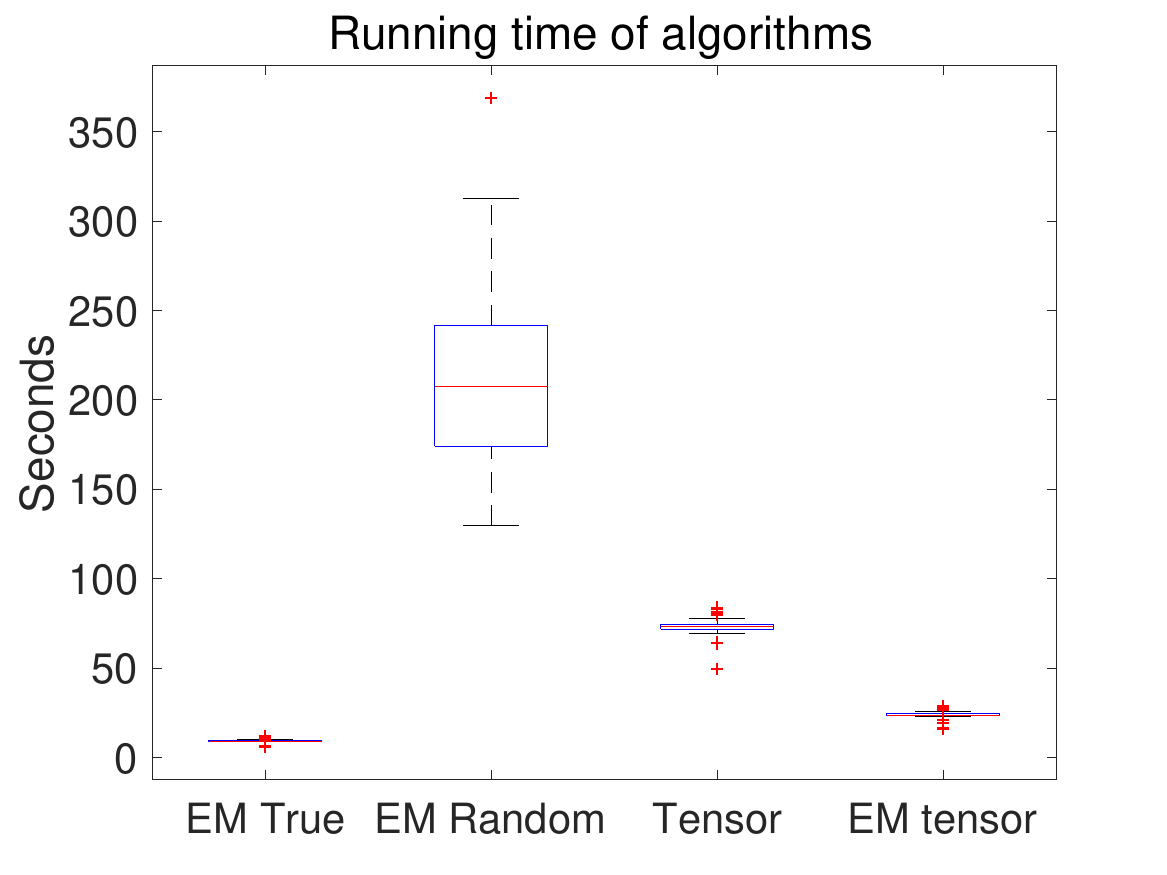}
		\end{minipage}%
	}%
	\centering
	\caption{$N = 20000, J= 100, L=10,$ item parameters $\in \{0.2,0.4,0.6,0.8\}$}
\end{figure}

\begin{figure}[H]
	\centering
	\subfigure[MSE of item parameters]{
		\begin{minipage}[t]{0.4\linewidth}
			\centering
			\includegraphics[width=2in]{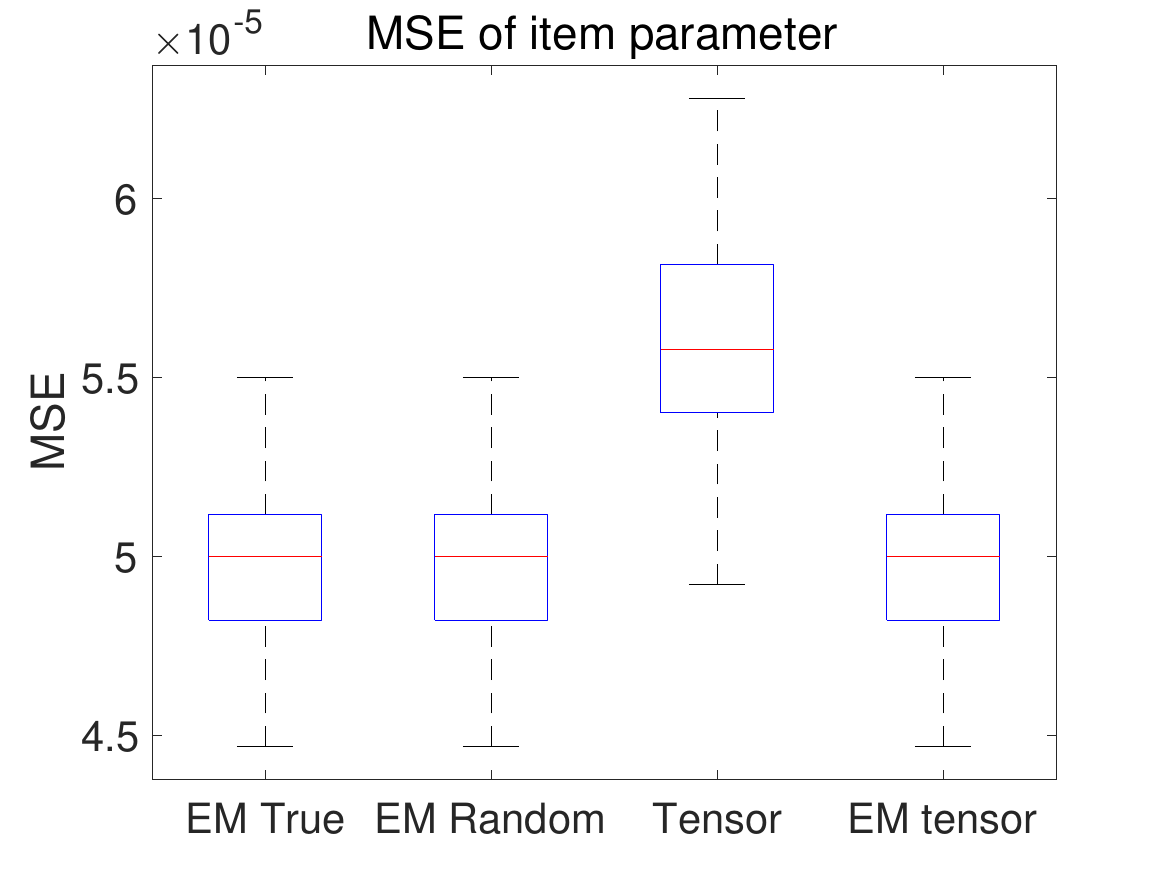}
		\end{minipage}%
	}%
	\subfigure[Running time of the algorithms]{
		\begin{minipage}[t]{0.4\linewidth}
			\centering
			\includegraphics[width=2in]{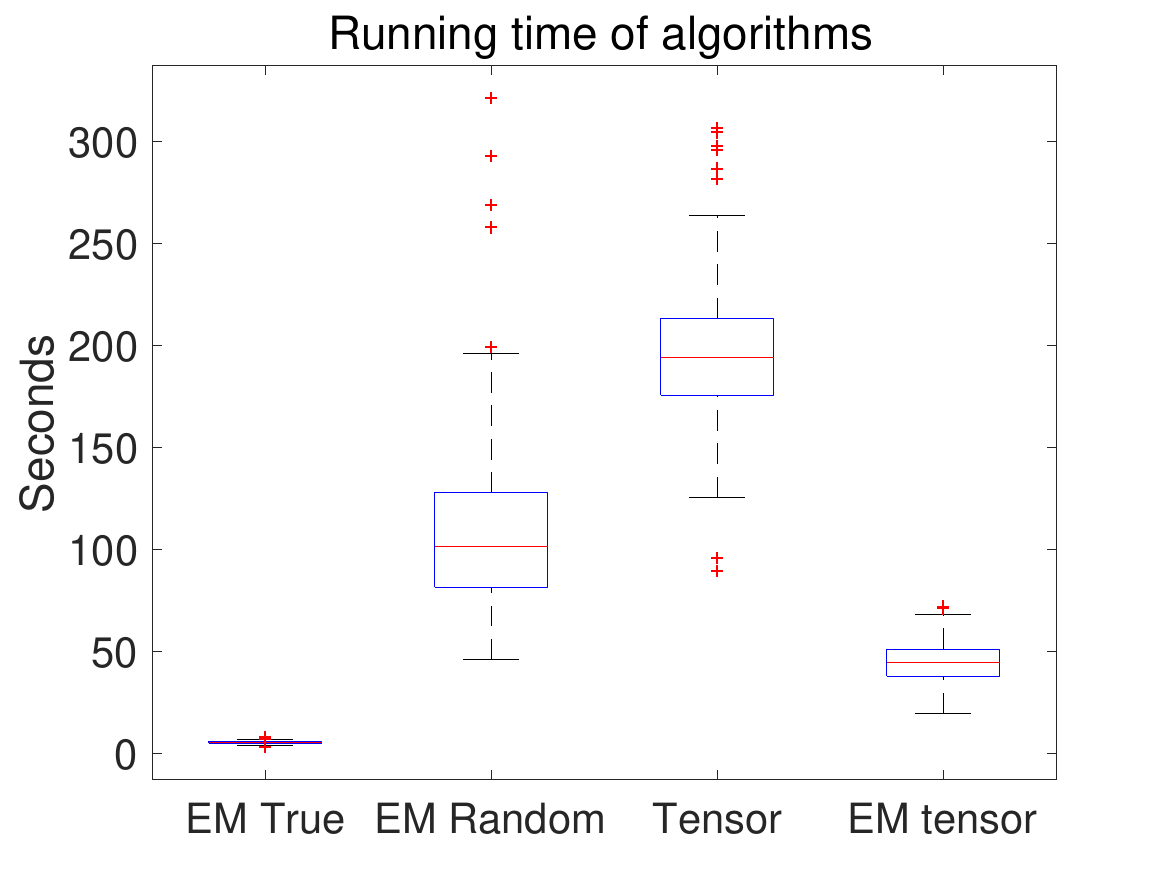}
		\end{minipage}%
	}%
	\centering
	\caption{$N = 20000, J= 200, L=5,$ item parameters $\in \{0.2,0.4,0.6,0.8\}$}
\end{figure}

\begin{figure}[H]
	\centering
	\subfigure[MSE of item parameters]{
		\begin{minipage}[t]{0.33\textwidth}
			\centering
			\includegraphics[width=2in]{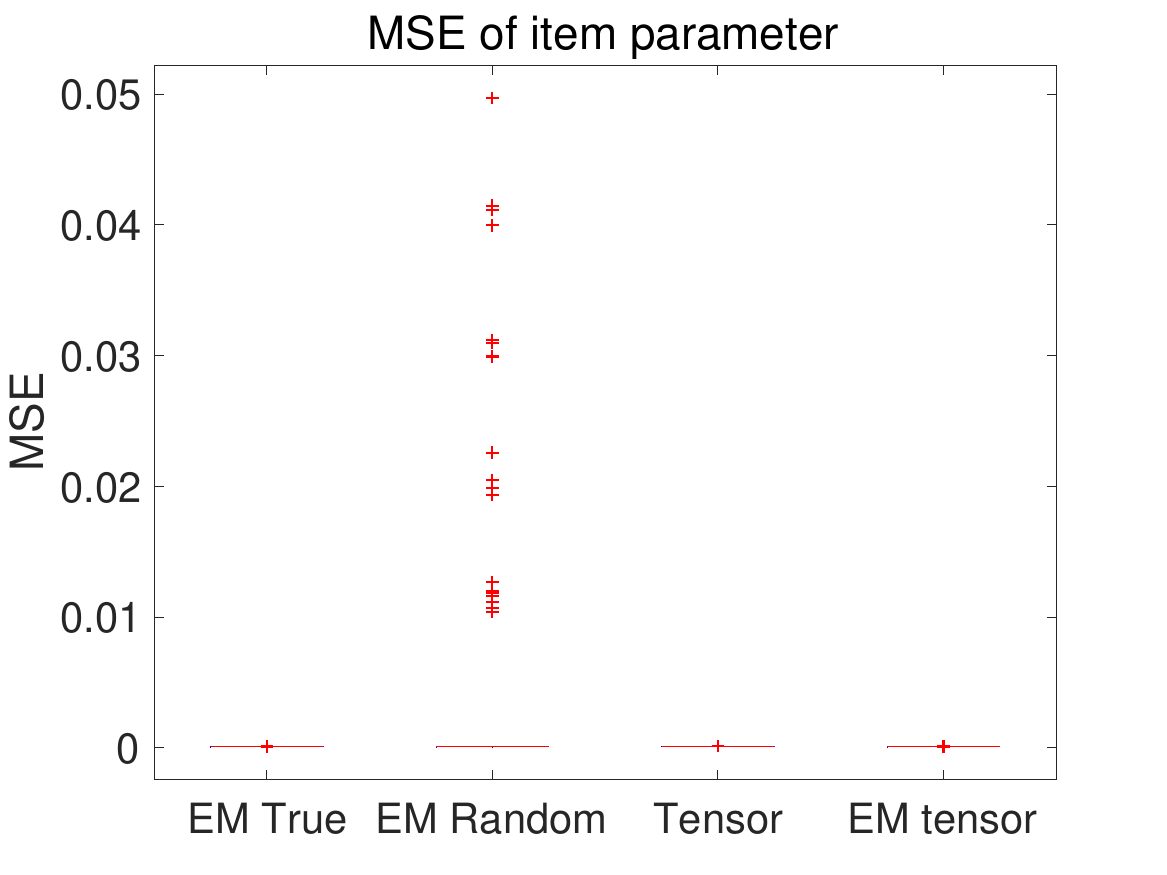}
		\end{minipage}%
	}%
		\subfigure[MSE without EM-random]{
	\begin{minipage}[t]{0.33\textwidth}
		\centering
		\includegraphics[width=2in]{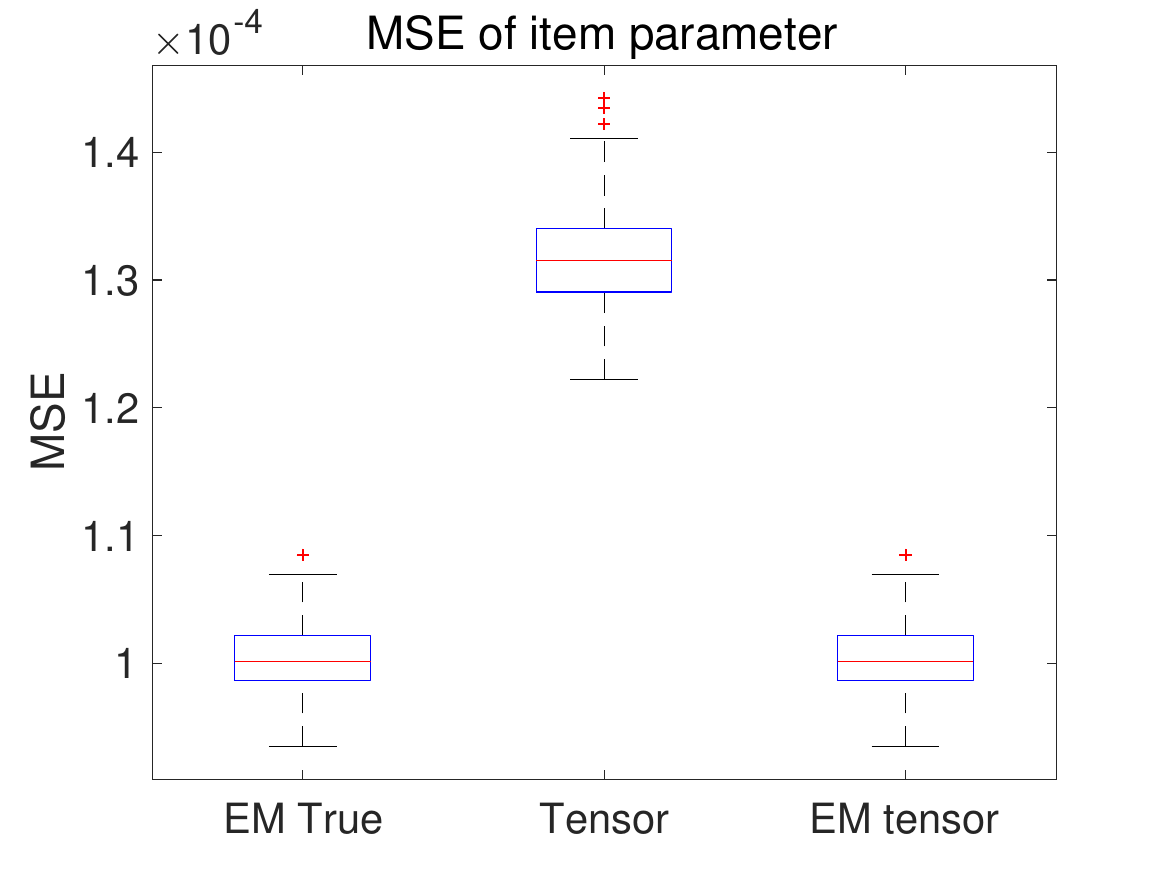}
	\end{minipage}%
}%
	\subfigure[Running time of the algorithms]{
		\begin{minipage}[t]{0.33\textwidth}
			\centering
			\includegraphics[width=2in]{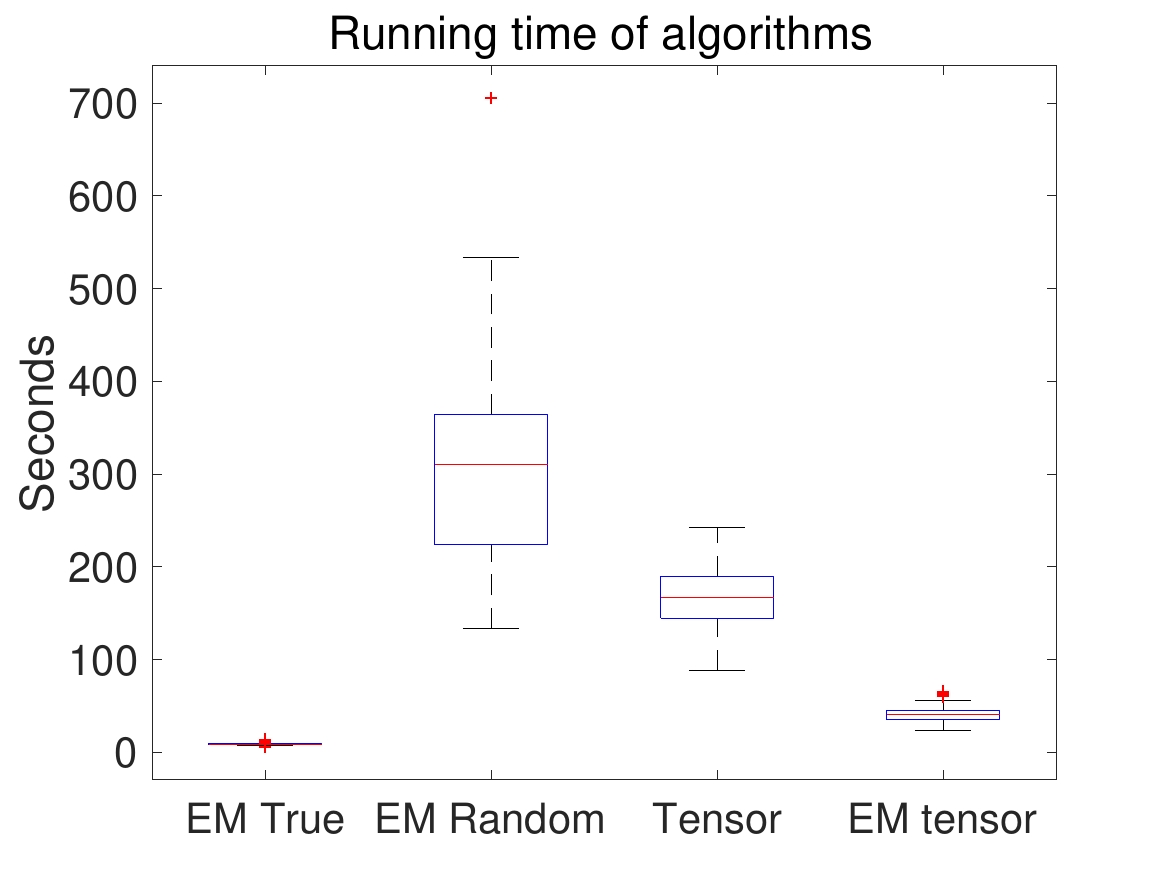}
		\end{minipage}%
	}%
	\centering
	\caption{$N = 20000, J= 200, L=10,$ item parameters $\in \{0.2,0.4,0.6,0.8\}$}
\end{figure}

\subsection*{EM-random and Tensor-EM with same initializations}

We then present the simulation results when EM-random and Tensor-EM are implemented with same initializations together with the ``smarter" version of EM-random.

\begin{figure}[H]
	\centering
	\subfigure[MSE of item parameters]{
		\begin{minipage}[t]{0.4\linewidth}
			\centering
			\includegraphics[width=2in]{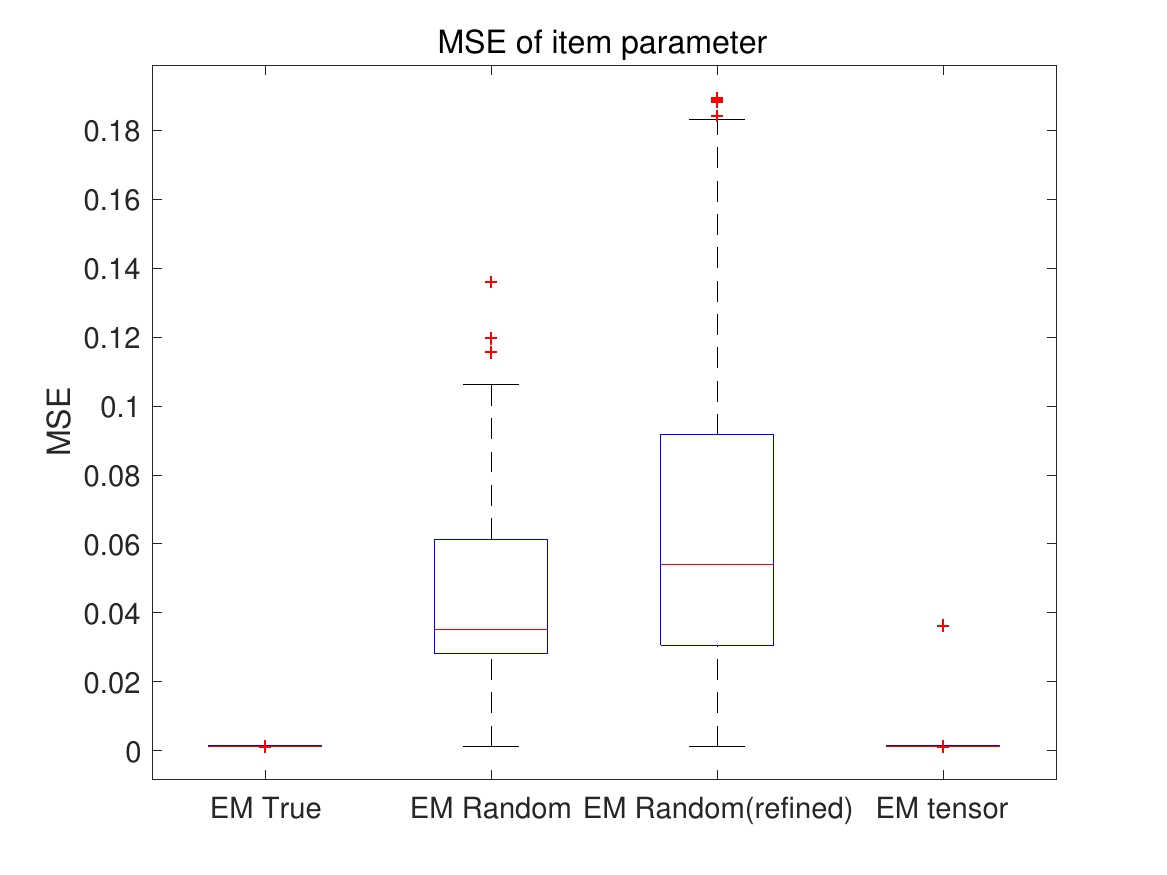}
		\end{minipage}%
	}%
	\subfigure[Running time of the algorithms]{
		\begin{minipage}[t]{0.4\linewidth}
			\centering
			\includegraphics[width=2in]{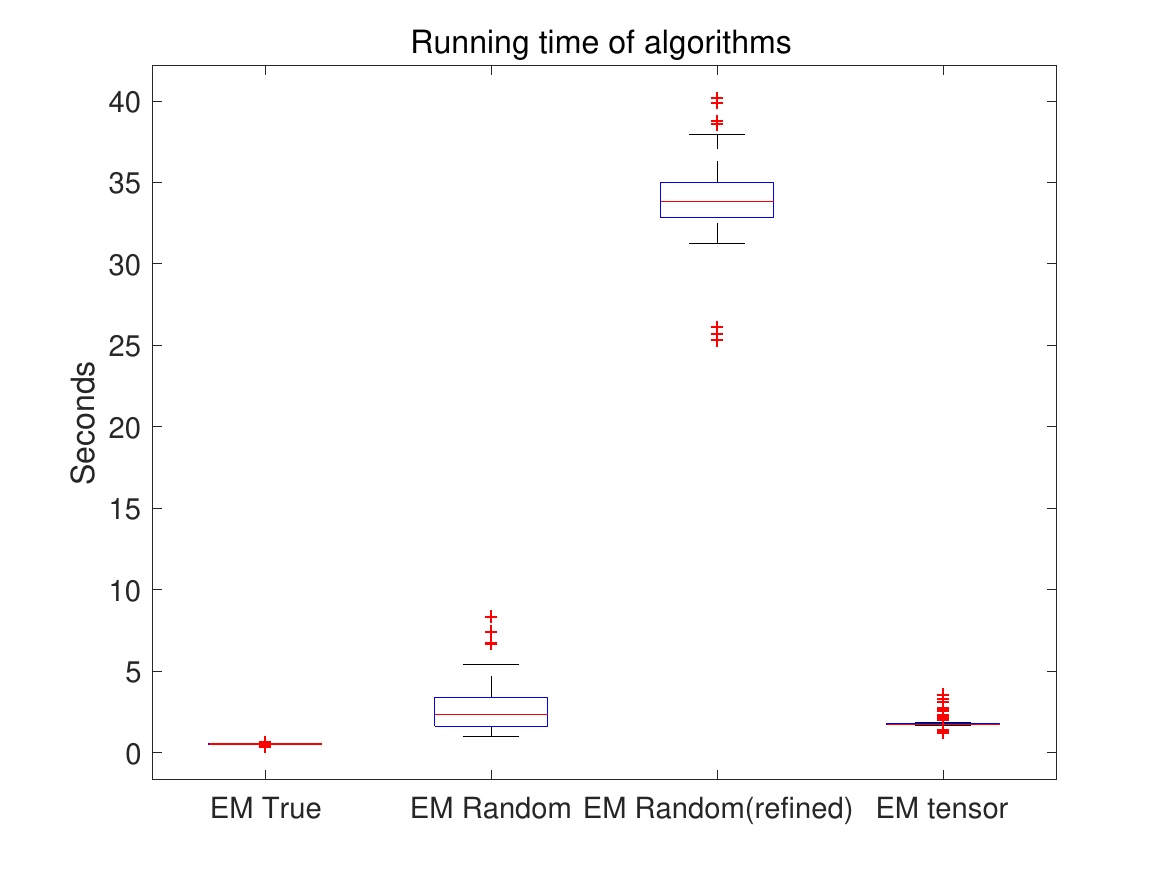}
		\end{minipage}%
	}%
	\centering
	\caption{$N = 1000, J= 100, L=10,$ item parameters $\in \{0.1,0.2,0.8,0.9\}$}
\end{figure}

\begin{figure}[H]
	\centering
	\subfigure[MSE of item parameters]{
		\begin{minipage}[t]{0.4\linewidth}
			\centering
			\includegraphics[width=2in]{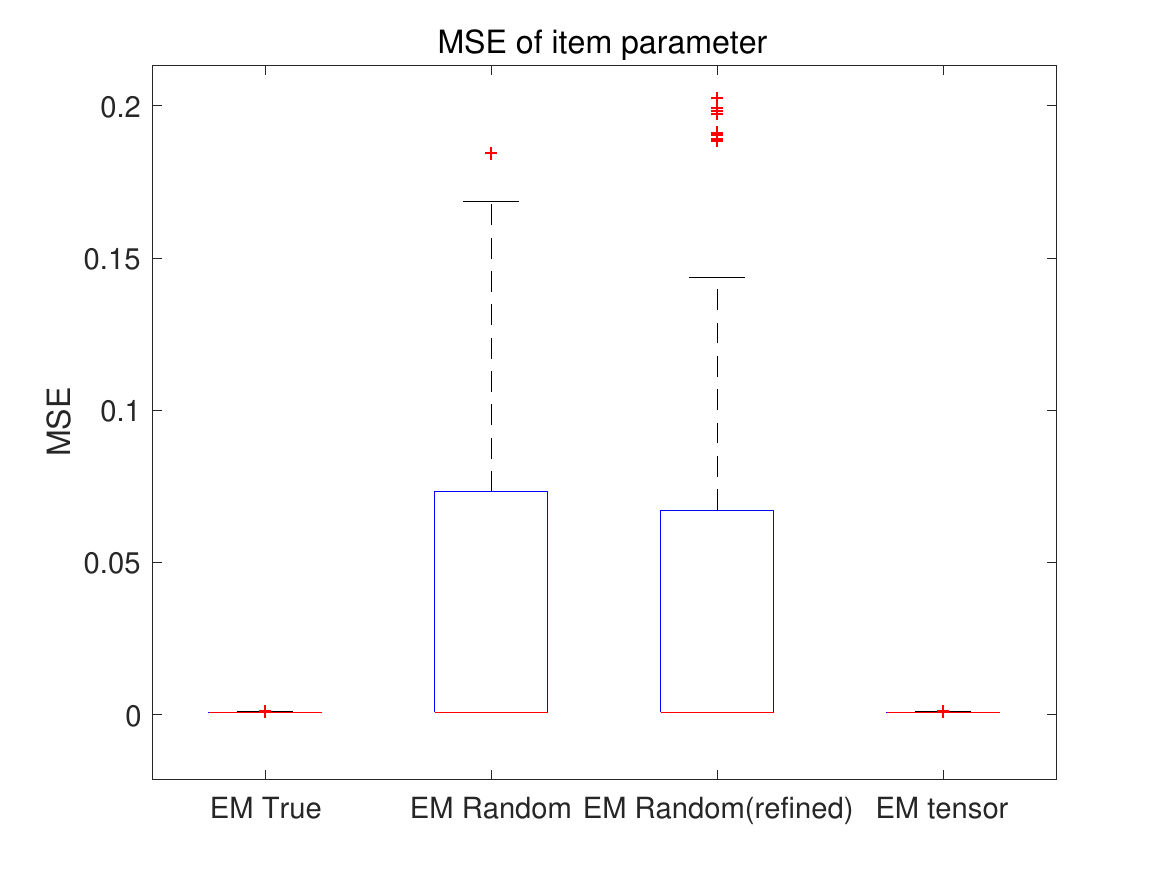}
		\end{minipage}%
	}%
	\subfigure[Running time of the algorithms]{
		\begin{minipage}[t]{0.4\linewidth}
			\centering
			\includegraphics[width=2in]{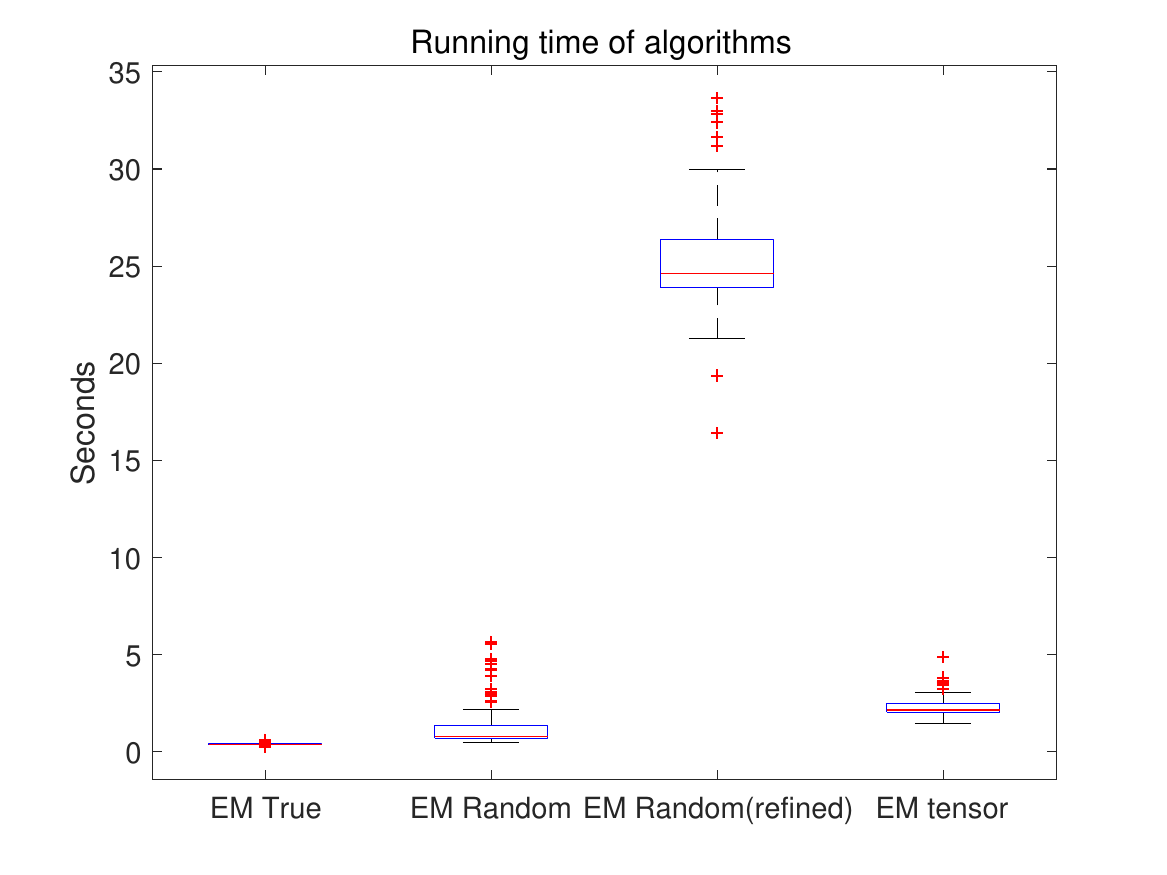}
		\end{minipage}%
	}%
	\centering
	\caption{$N = 1000, J= 200, L=5,$ item parameters $\in \{0.1,0.2,0.8,0.9\}$}
\end{figure}

\begin{figure}[H]
	\centering
	\subfigure[MSE of item parameters]{
		\begin{minipage}[t]{0.4\linewidth}
			\centering
			\includegraphics[width=2in]{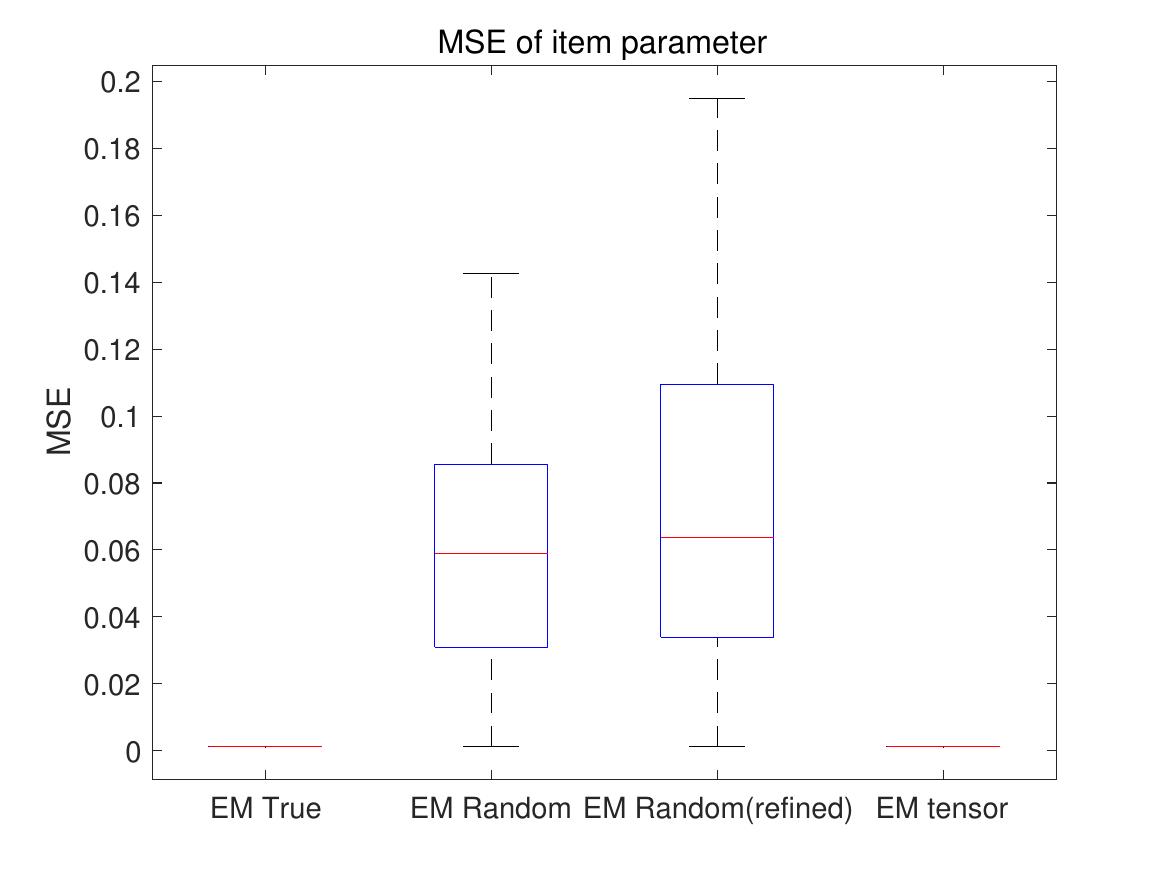}
		\end{minipage}%
	}%
	\subfigure[Running time of the algorithms]{
		\begin{minipage}[t]{0.4\linewidth}
			\centering
			\includegraphics[width=2in]{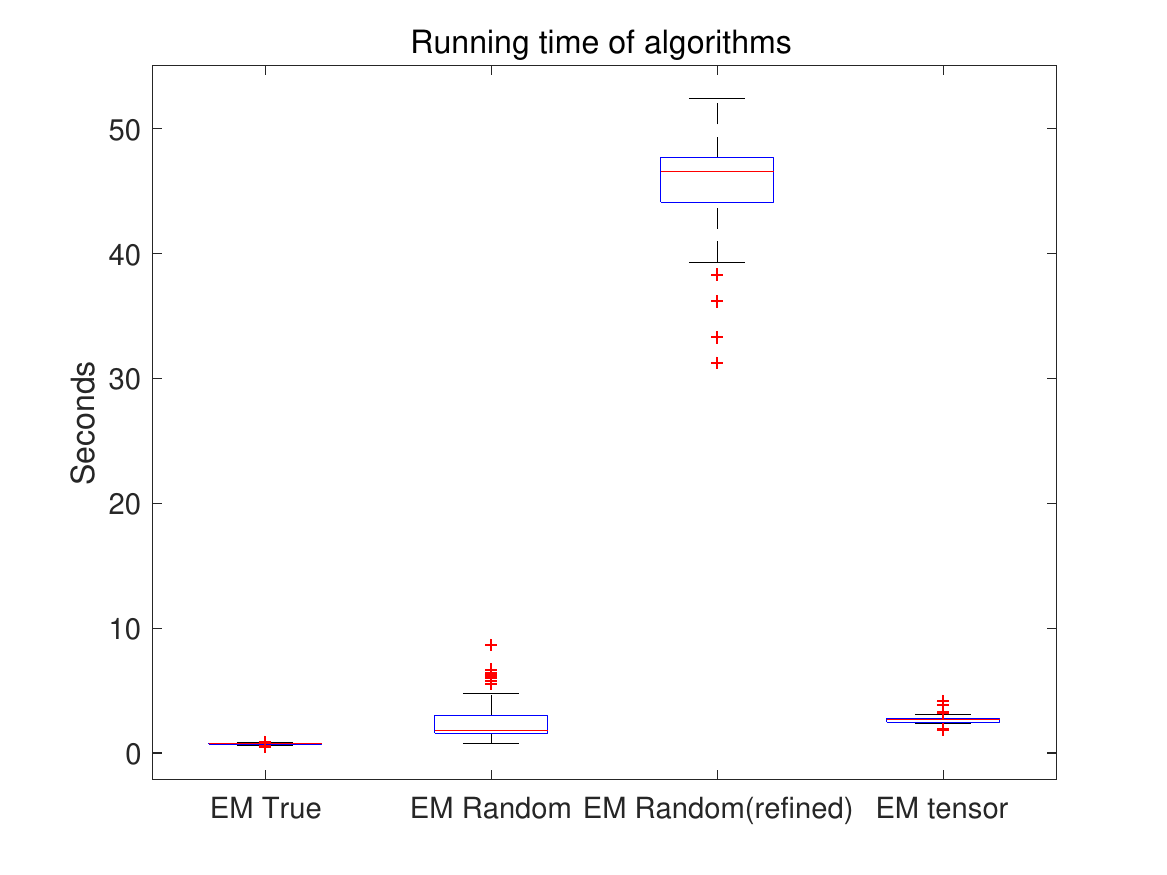}
		\end{minipage}%
	}%
	\centering
	\caption{$N = 1000, J= 200, L=10,$ item parameters $\in \{0.1,0.2,0.8,0.9\}$}
\end{figure}

\begin{figure}[H]
	\centering
	\subfigure[MSE of item parameters]{
		\begin{minipage}[t]{0.4\linewidth}
			\centering
			\includegraphics[width=2in]{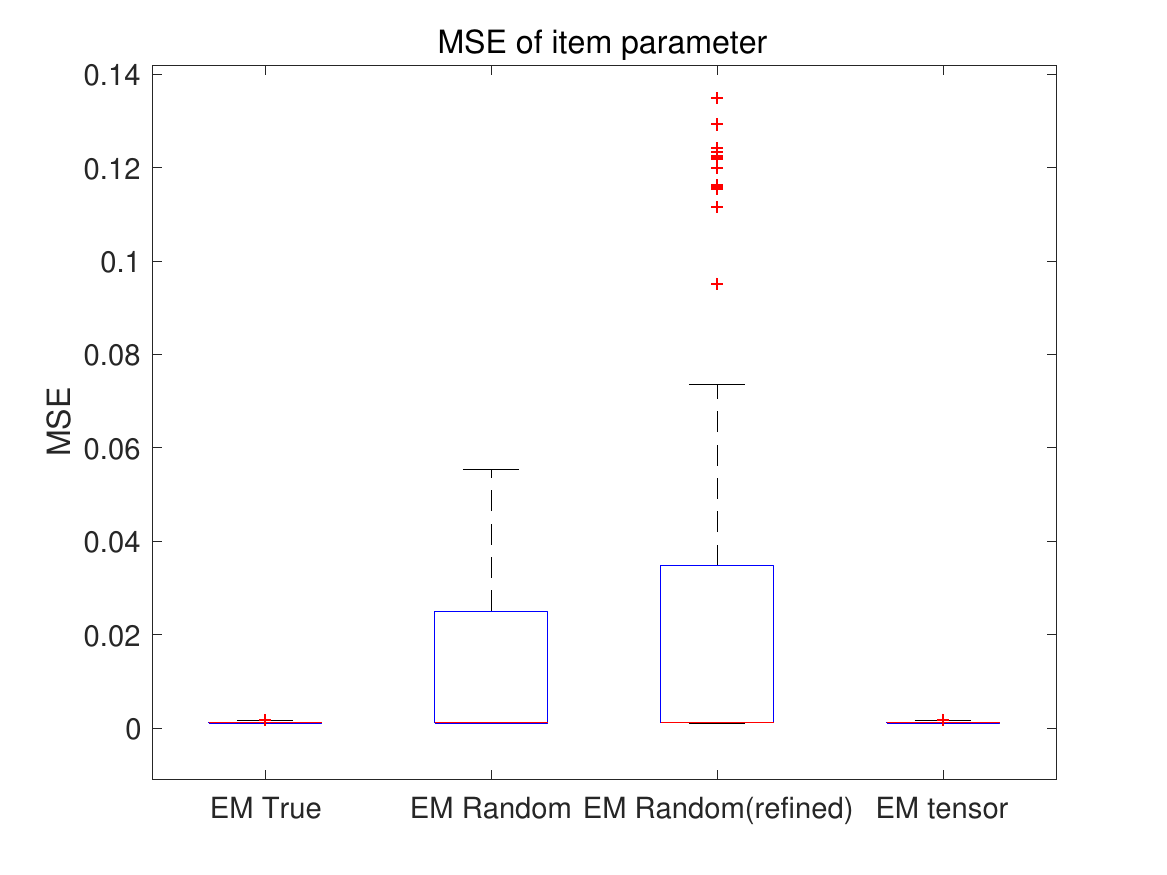}
		\end{minipage}%
	}%
	\subfigure[Running time of the algorithms]{
		\begin{minipage}[t]{0.4\linewidth}
			\centering
			\includegraphics[width=2in]{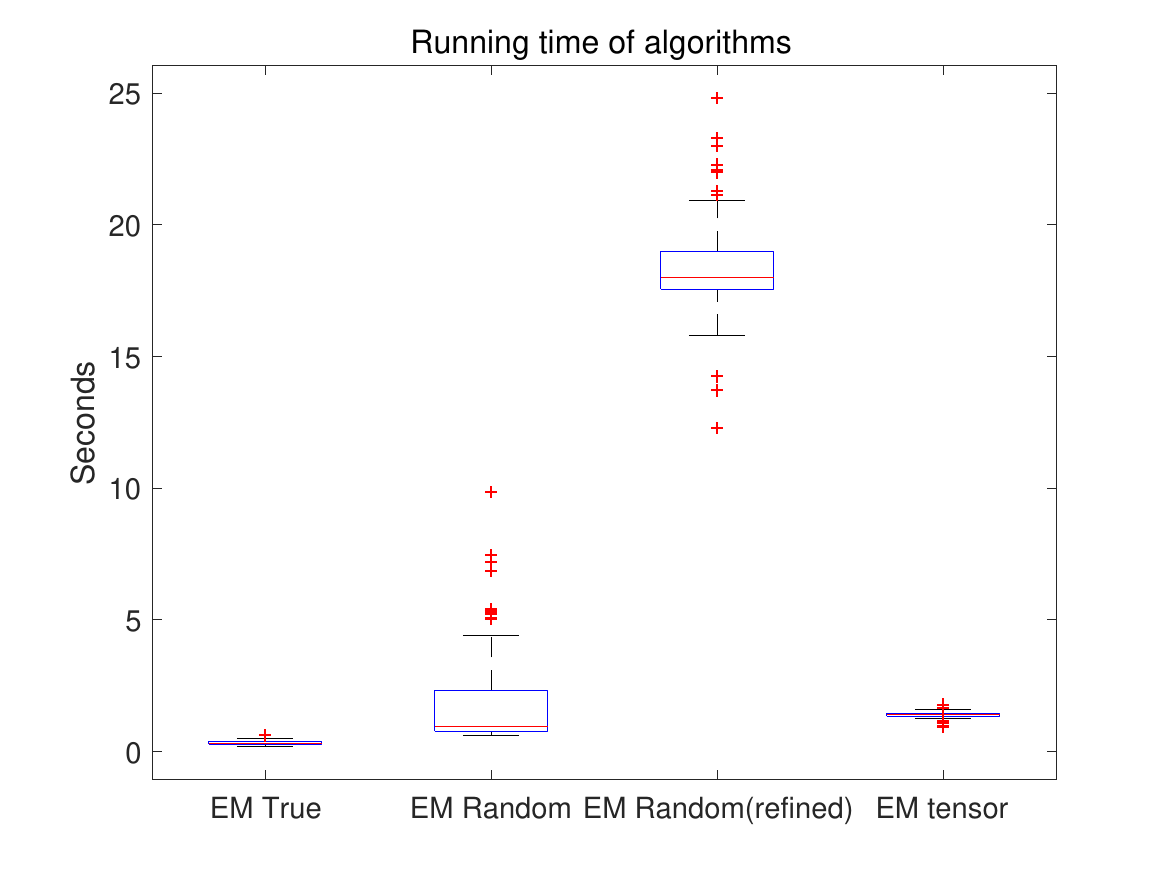}
		\end{minipage}%
	}%
	\centering
	\caption{$N = 1000, J= 100, L=5,$ item parameters $\in \{0.2,0.4,0.6,0.8\}$}
\end{figure}

\begin{figure}[H]
	\centering
	\subfigure[MSE of item parameters]{
		\begin{minipage}[t]{0.4\linewidth}
			\centering
			\includegraphics[width=2in]{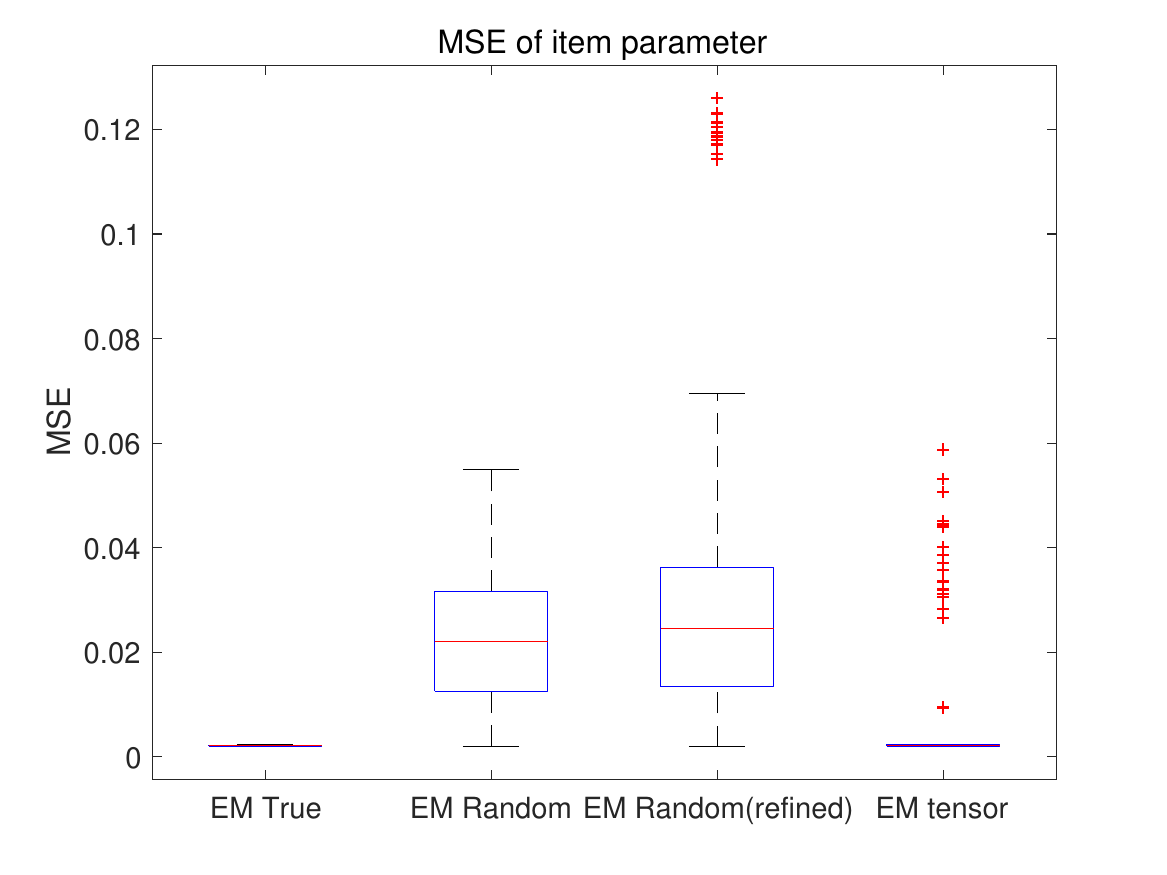}
		\end{minipage}%
	}%
	\subfigure[Running time of the algorithms]{
		\begin{minipage}[t]{0.4\linewidth}
			\centering
			\includegraphics[width=2in]{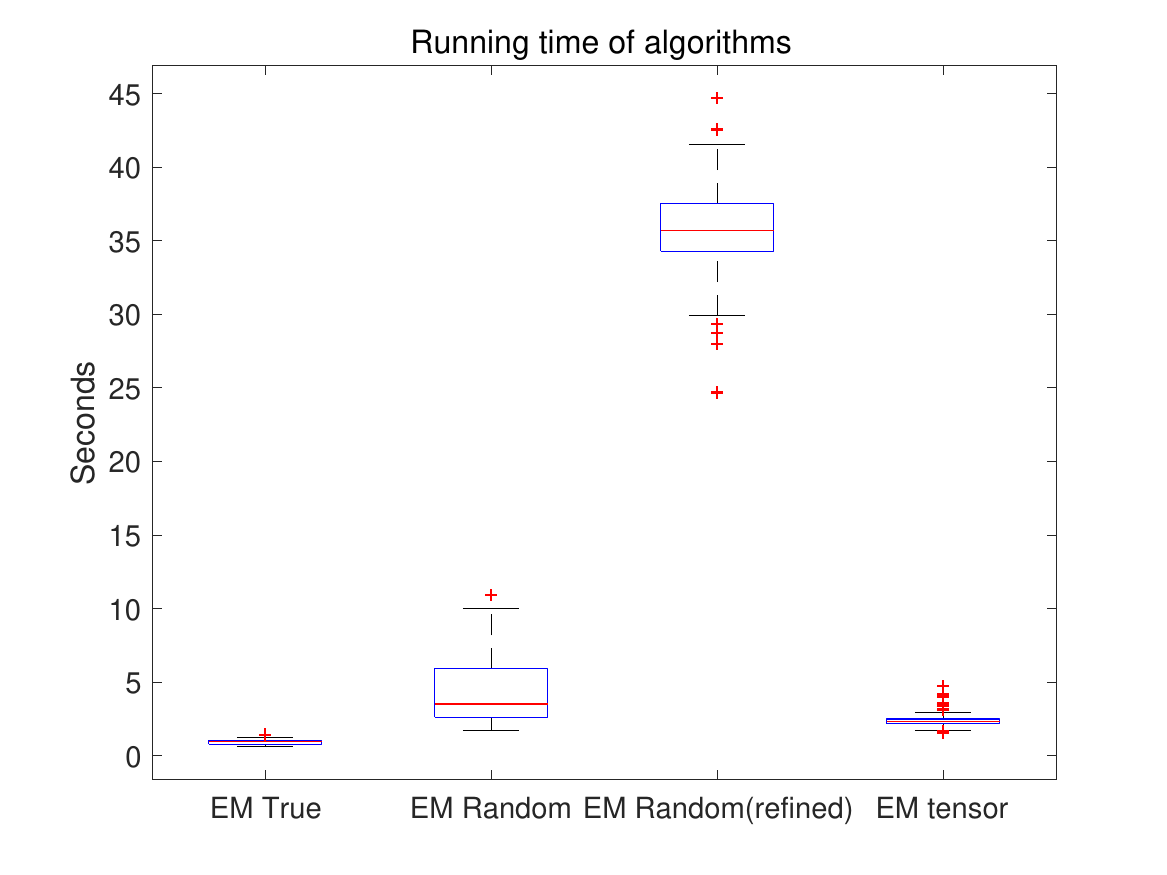}
		\end{minipage}%
	}%
	\centering
	\caption{$N = 1000, J= 100, L=10,$ item parameters $\in \{0.2,0.4,0.6,0.8\}$}
\end{figure}

\begin{figure}[H]
	\centering
	\subfigure[MSE of item parameters]{
		\begin{minipage}[t]{0.4\linewidth}
			\centering
			\includegraphics[width=2in]{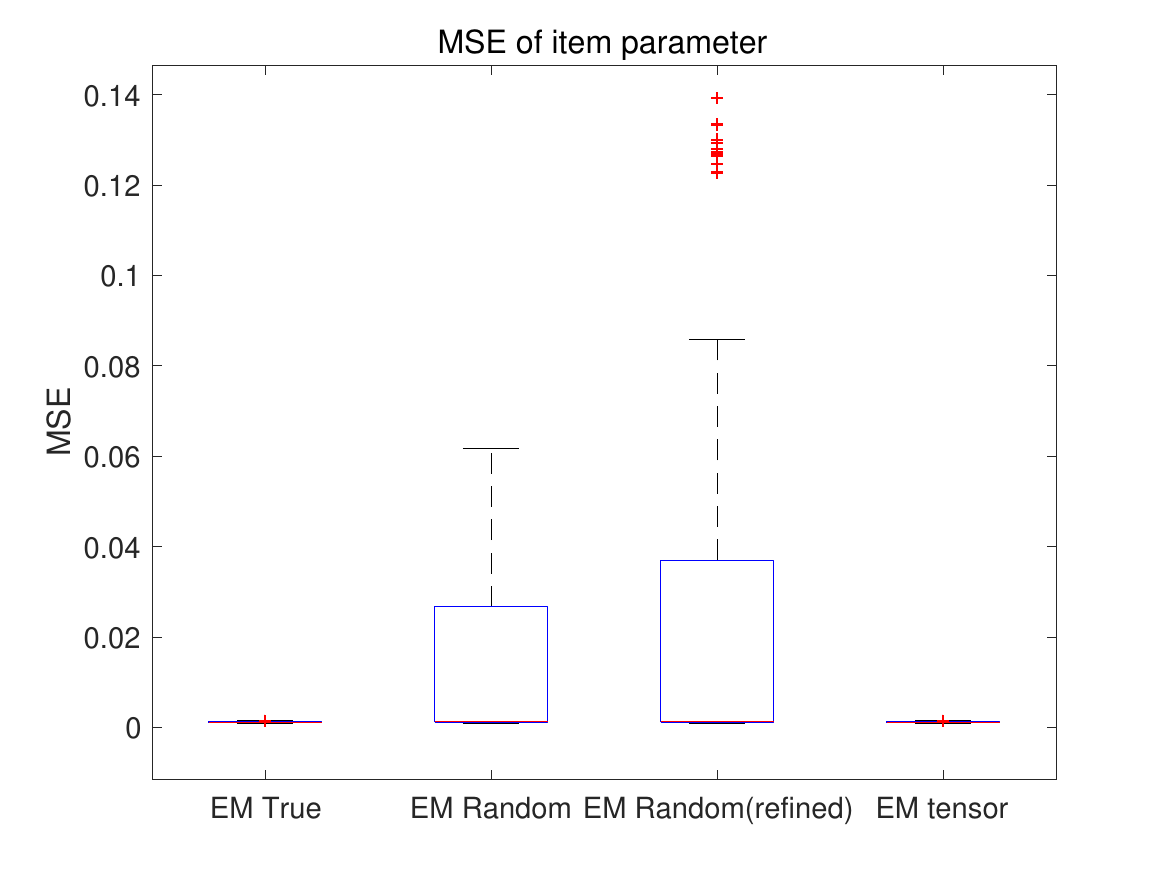}
		\end{minipage}%
	}%
	\subfigure[Running time of the algorithms]{
		\begin{minipage}[t]{0.4\linewidth}
			\centering
			\includegraphics[width=2in]{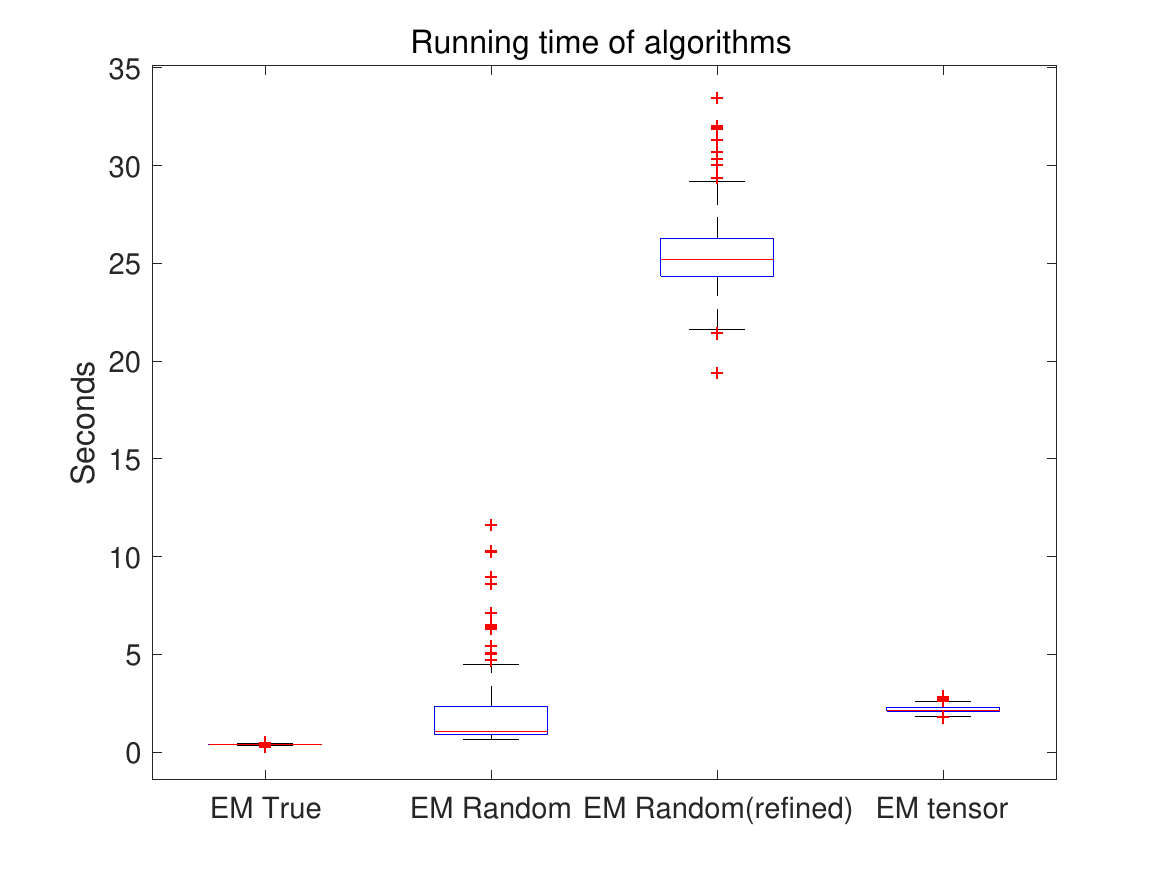}
		\end{minipage}%
	}%
	\centering
	\caption{$N = 1000, J= 200, L=5,$ item parameters $\in \{0.2,0.4,0.6,0.8\}$}
\end{figure}

\begin{figure}[H]
	\centering
	\subfigure[MSE of item parameters]{
		\begin{minipage}[t]{0.4\linewidth}
			\centering
			\includegraphics[width=2in]{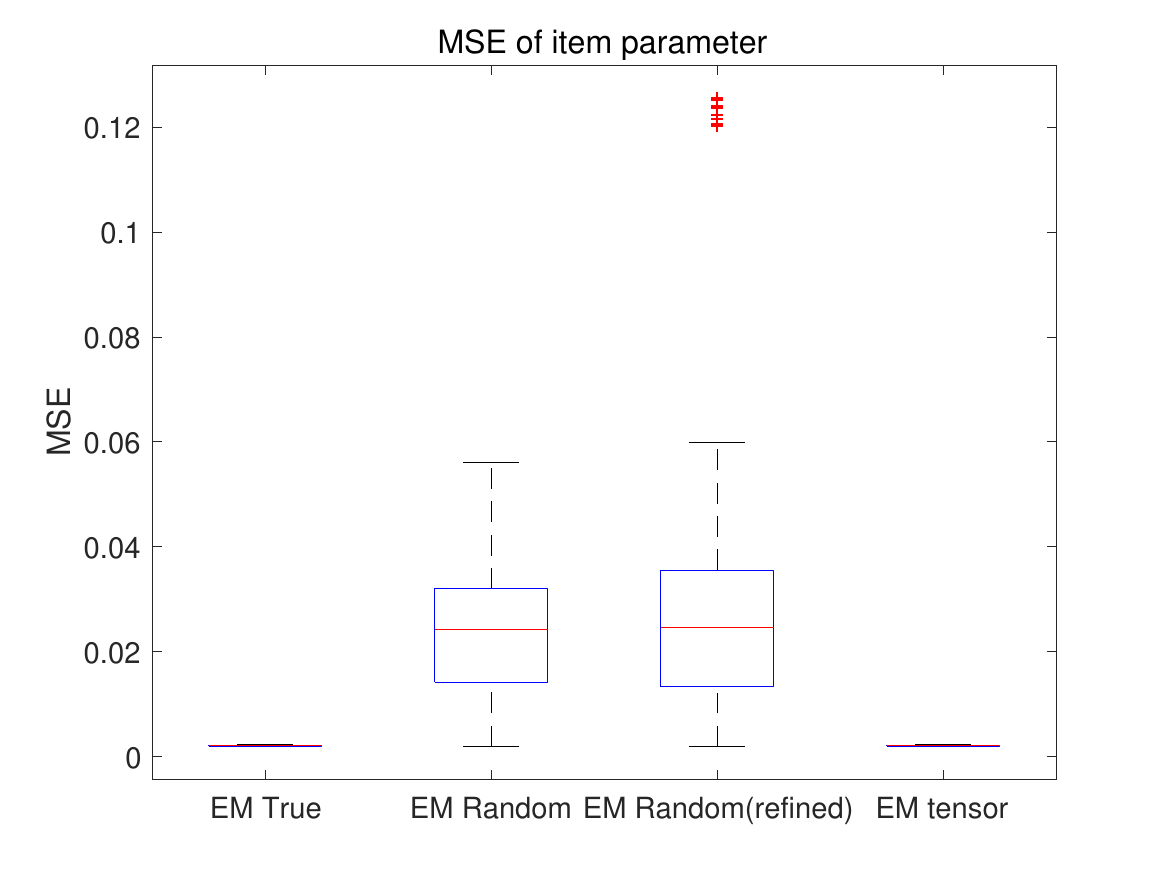}
		\end{minipage}%
	}%
	\subfigure[Running time of the algorithms]{
		\begin{minipage}[t]{0.4\linewidth}
			\centering
			\includegraphics[width=2in]{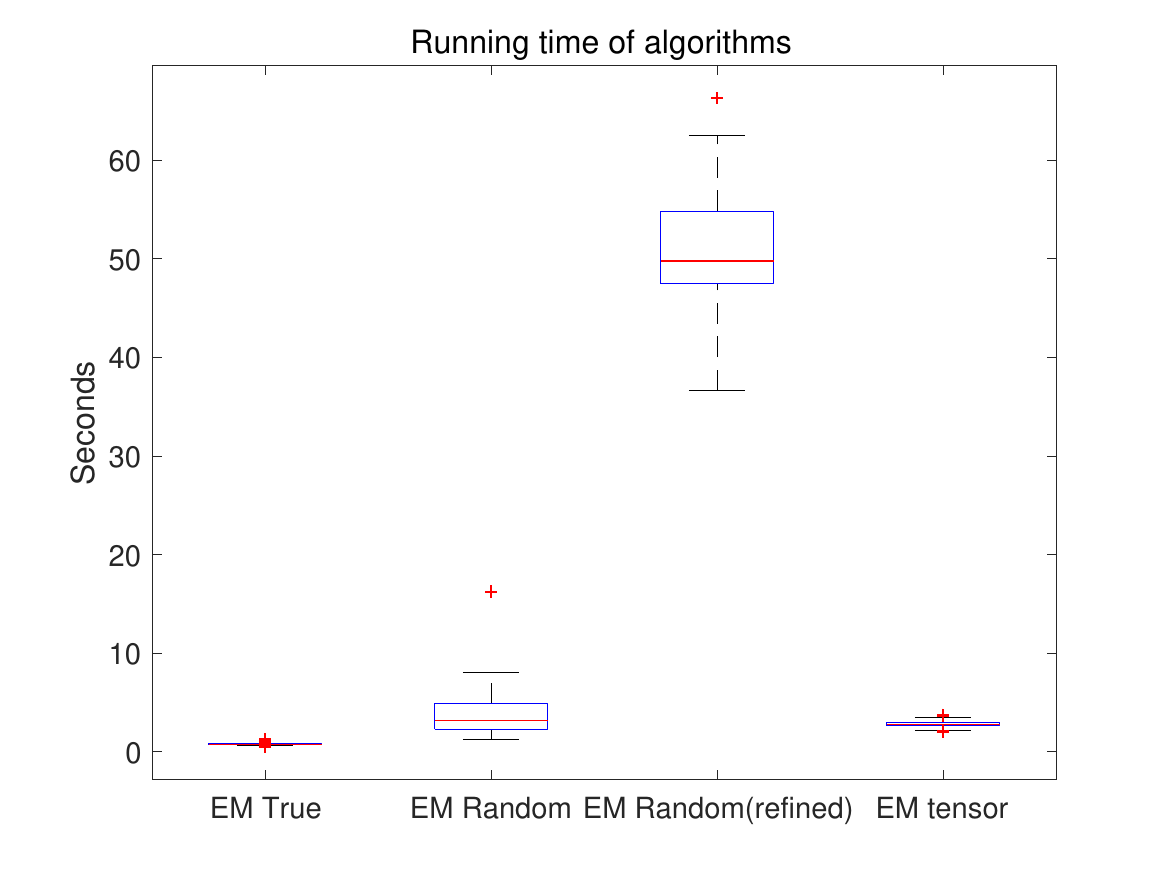}
		\end{minipage}%
	}%
	\centering
	\caption{$N = 1000, J= 200, L=10,$ item parameters $\in \{0.2,0.4,0.6,0.8\}$}
\end{figure}

\begin{figure}[H]
	\centering
	\subfigure[MSE of item parameters]{
		\begin{minipage}[t]{0.4\linewidth}
			\centering
			\includegraphics[width=2in]{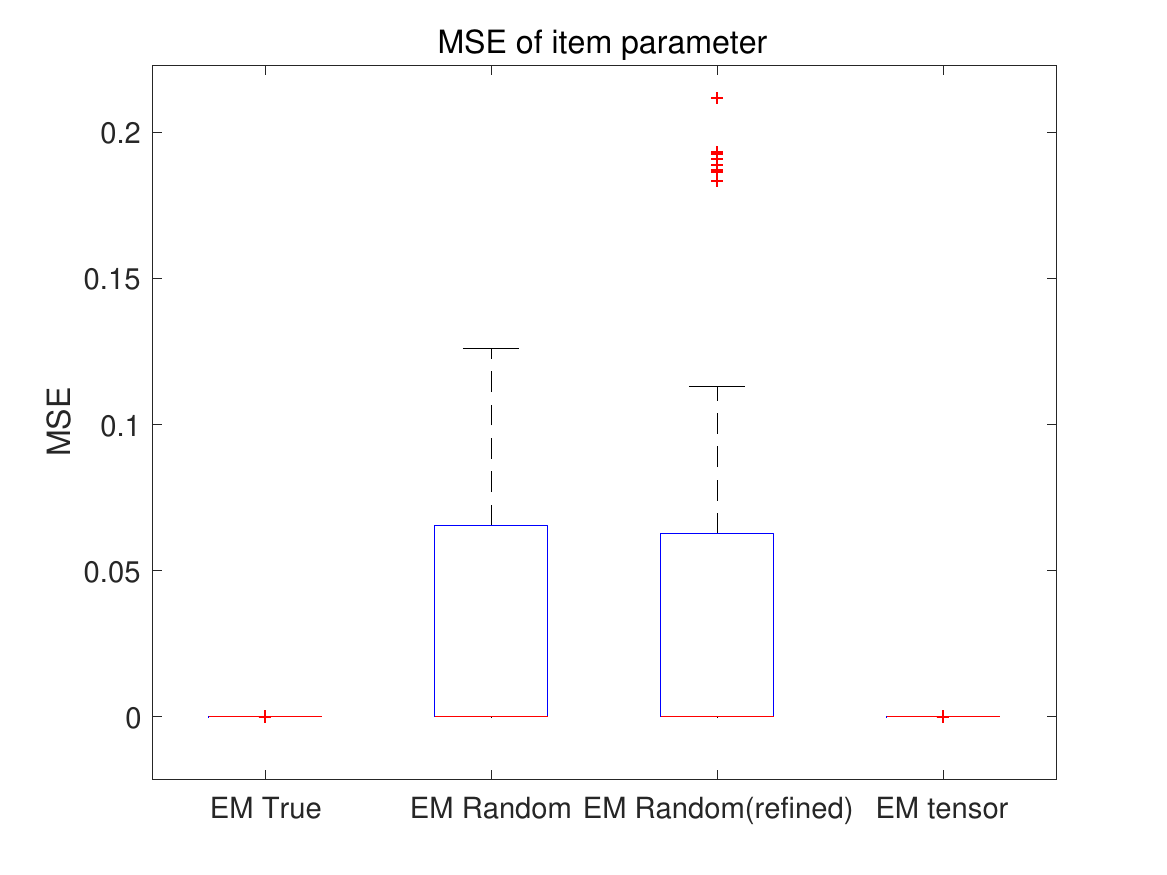}
		\end{minipage}%
	}%
	\subfigure[Running time of the algorithms]{
		\begin{minipage}[t]{0.4\linewidth}
			\centering
			\includegraphics[width=2in]{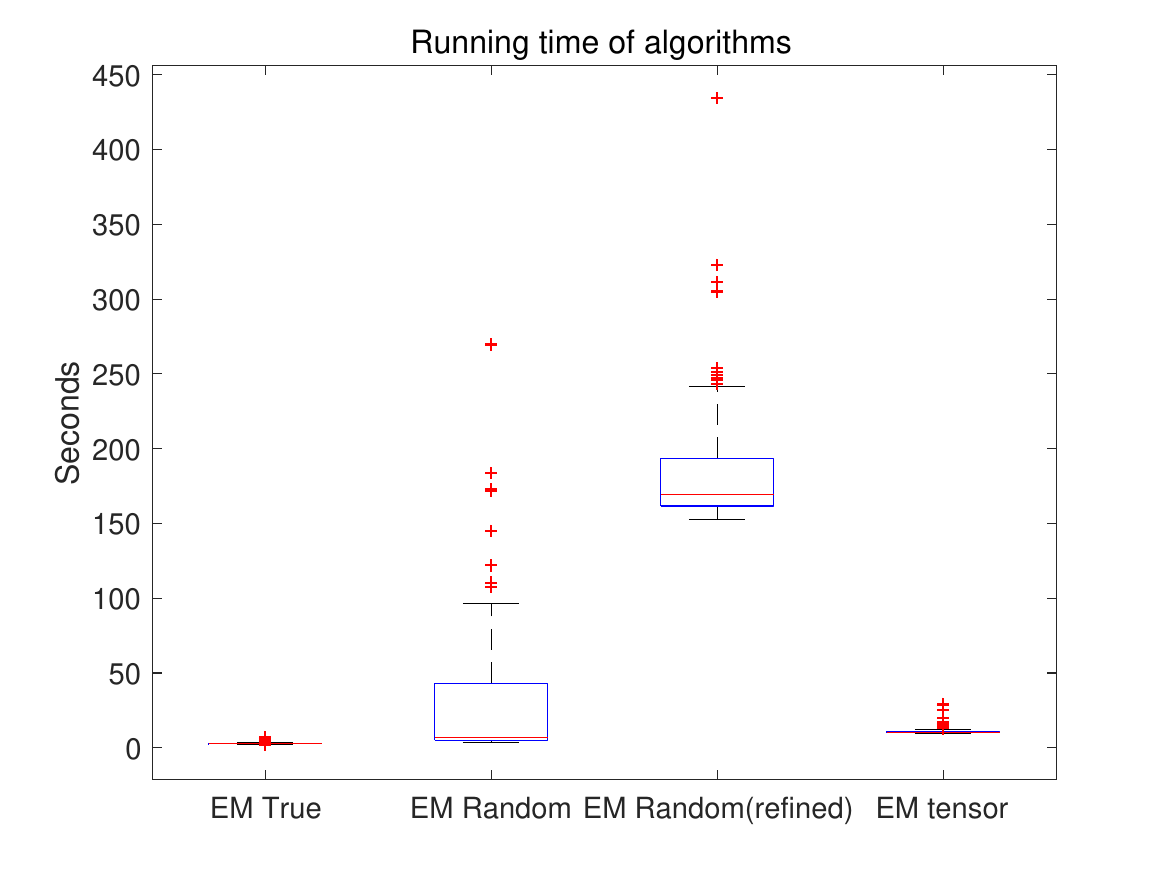}
		\end{minipage}%
	}%
	\centering
	\caption{$N = 10000, J= 100, L=5,$ item parameters $\in \{0.1,0.2,0.8,0.9\}$}
\end{figure}

\begin{figure}[H]
	\centering
	\subfigure[MSE of item parameters]{
		\begin{minipage}[t]{0.4\linewidth}
			\centering
			\includegraphics[width=2in]{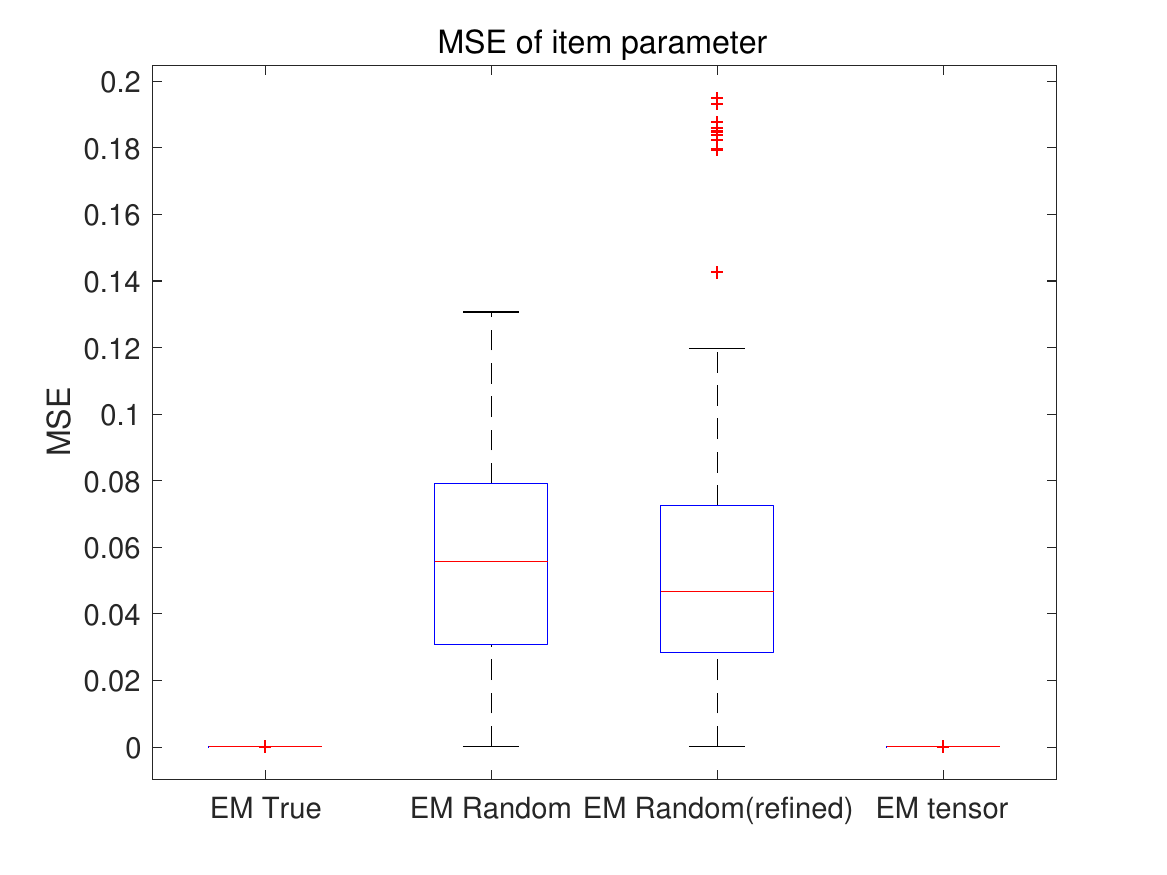}
		\end{minipage}%
	}%
	\subfigure[Running time of the algorithms]{
		\begin{minipage}[t]{0.4\linewidth}
			\centering
			\includegraphics[width=2in]{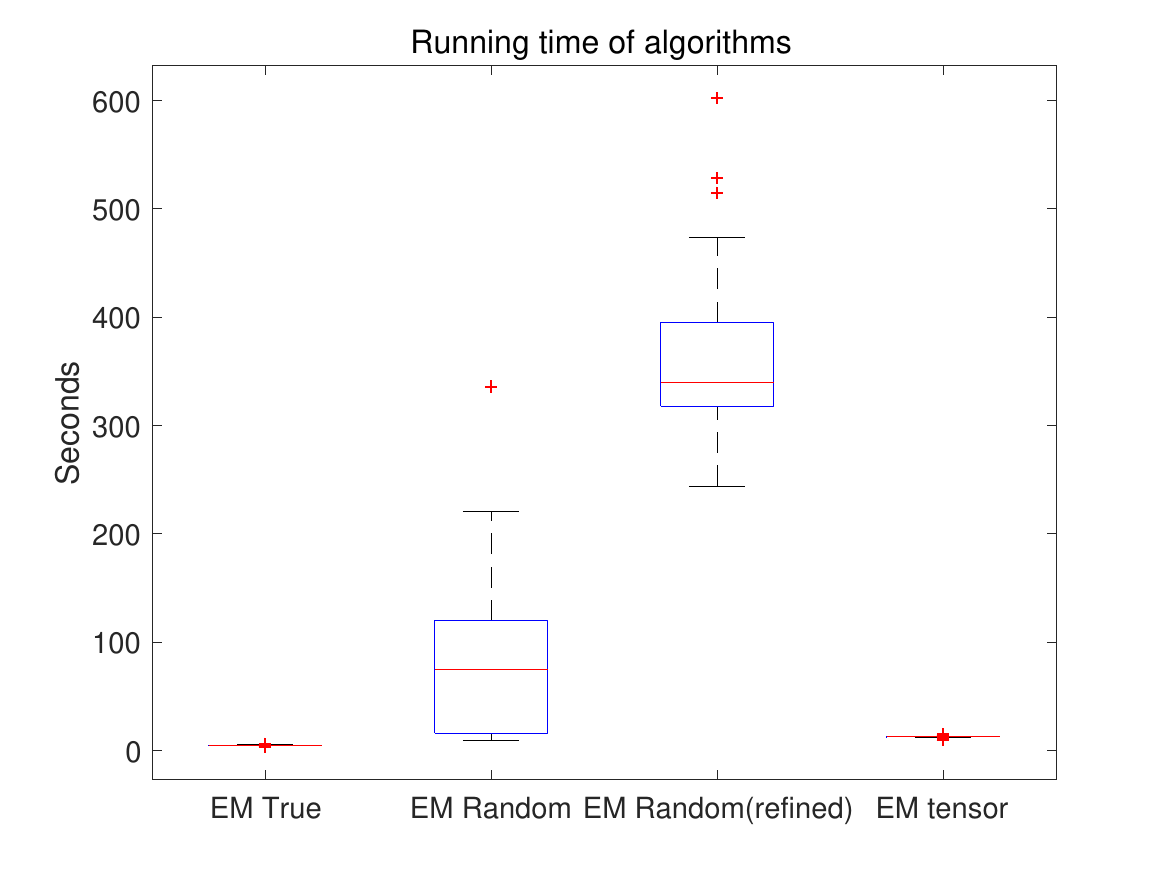}
		\end{minipage}%
	}%
	\centering
	\caption{$N = 10000, J= 100, L=10,$ item parameters $\in \{0.1,0.2,0.8,0.9\}$}
\end{figure}

\begin{figure}[H]
	\centering
	\subfigure[MSE of item parameters]{
		\begin{minipage}[t]{0.4\linewidth}
			\centering
			\includegraphics[width=2in]{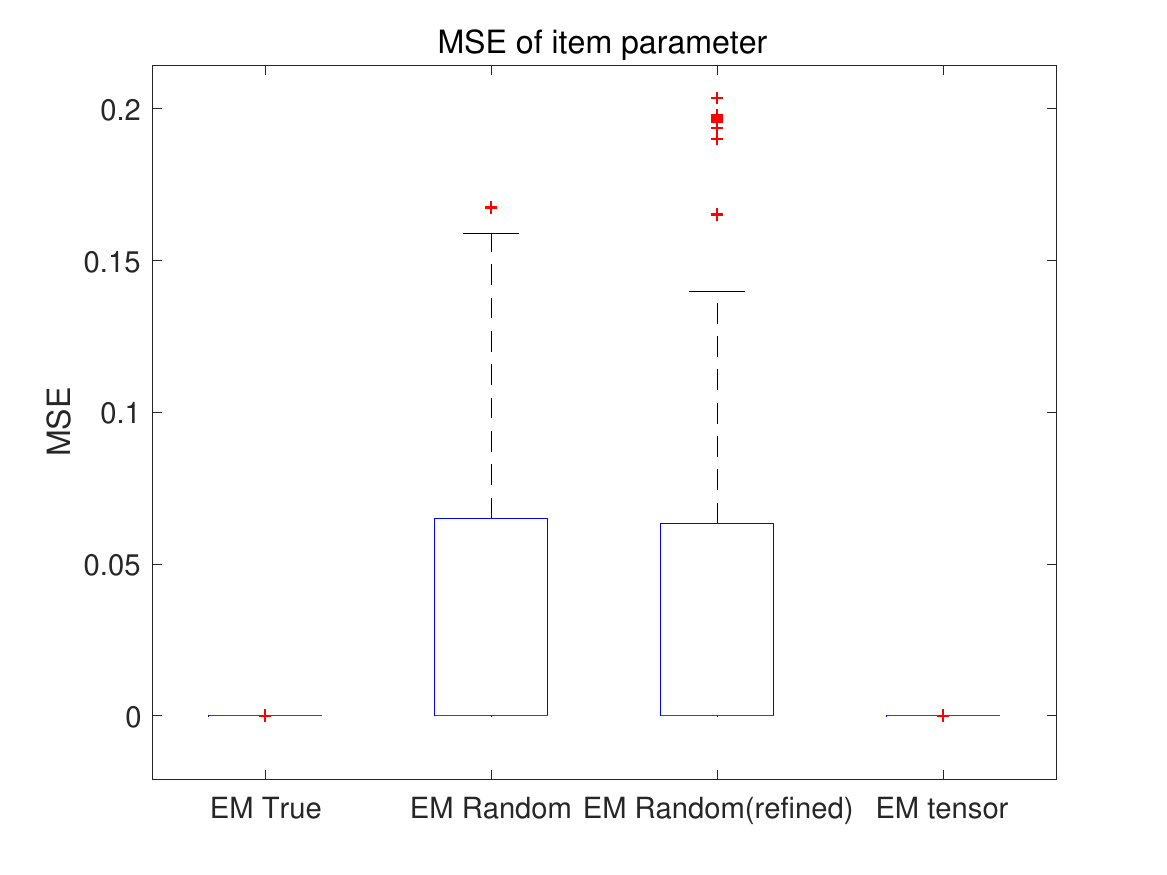}
		\end{minipage}%
	}%
	\subfigure[Running time of the algorithms]{
		\begin{minipage}[t]{0.4\linewidth}
			\centering
			\includegraphics[width=2in]{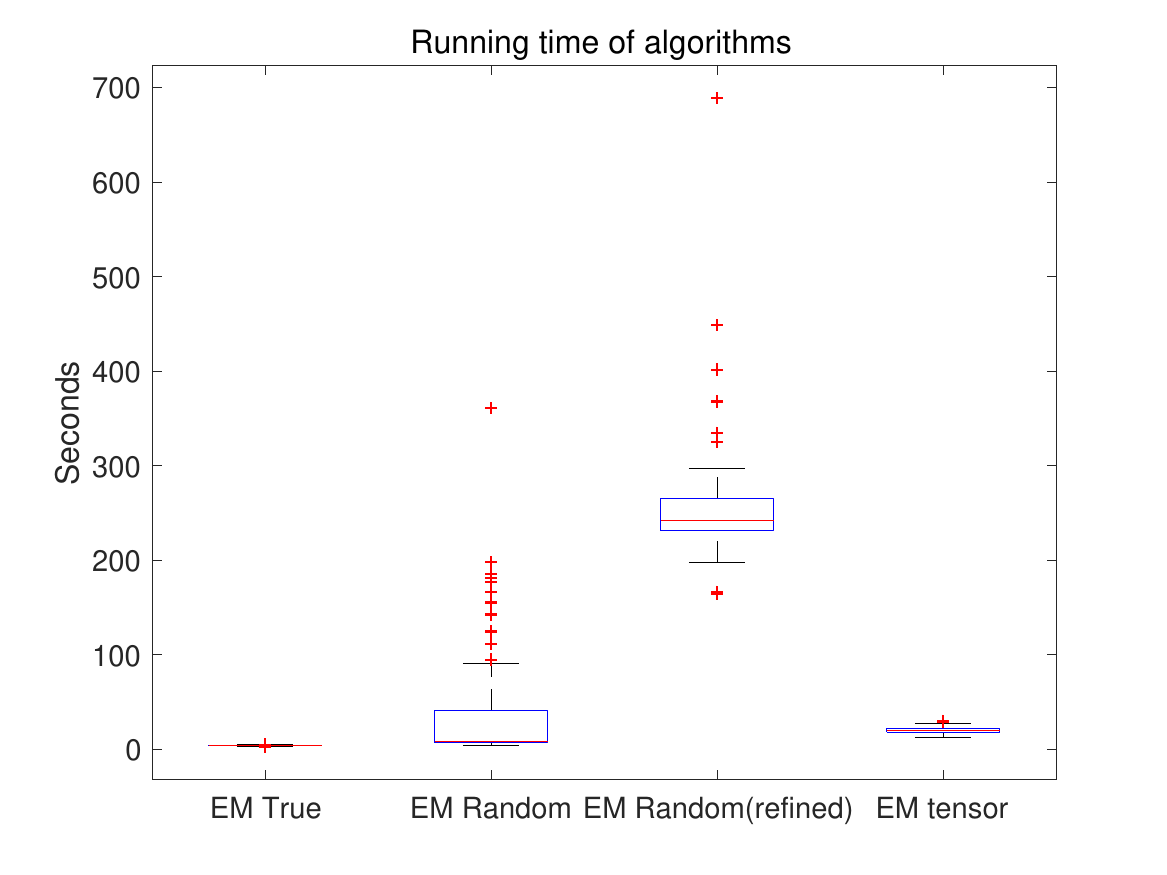}
		\end{minipage}%
	}%
	\centering
	\caption{$N = 10000, J= 200, L=5,$ item parameters $\in \{0.1,0.2,0.8,0.9\}$}
\end{figure}

\begin{figure}[H]
	\centering
	\subfigure[MSE of item parameters]{
		\begin{minipage}[t]{0.4\linewidth}
			\centering
			\includegraphics[width=2in]{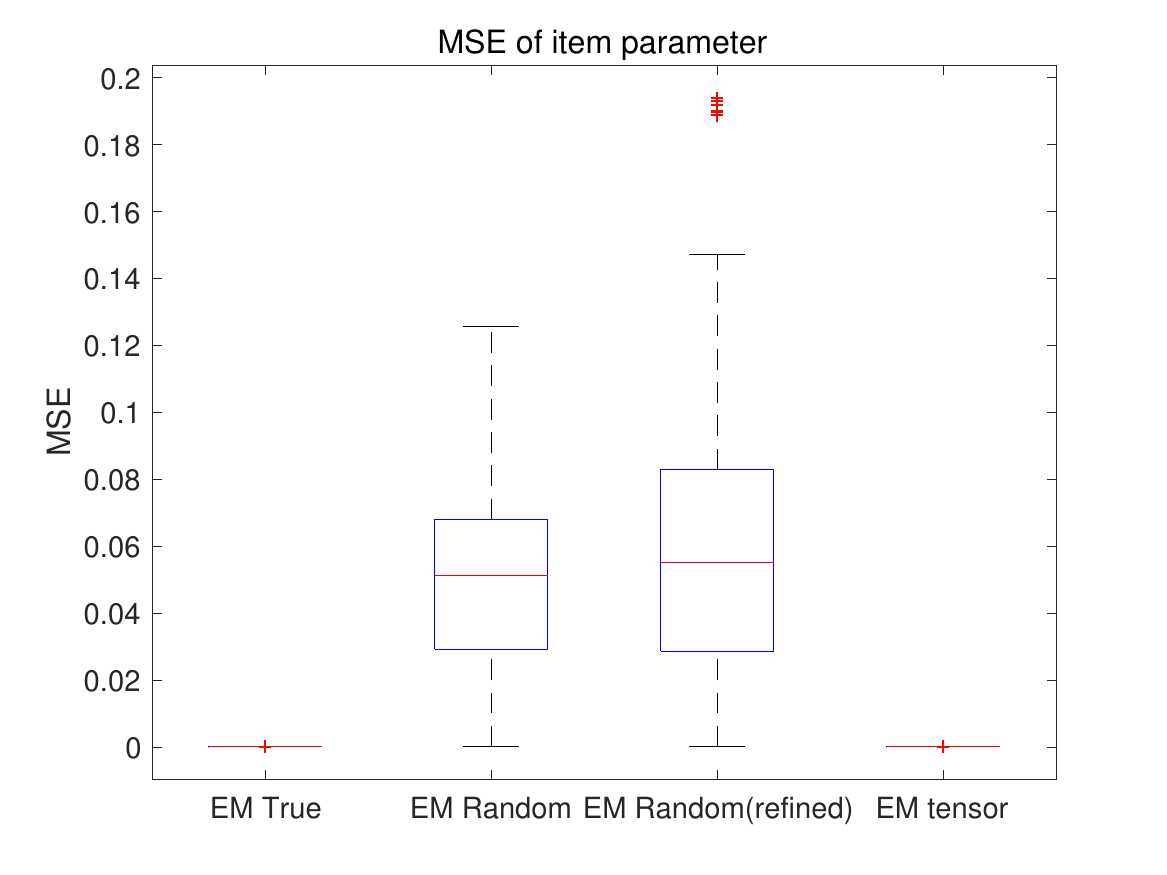}
		\end{minipage}%
	}%
	\subfigure[Running time of the algorithms]{
		\begin{minipage}[t]{0.4\linewidth}
			\centering
			\includegraphics[width=2in]{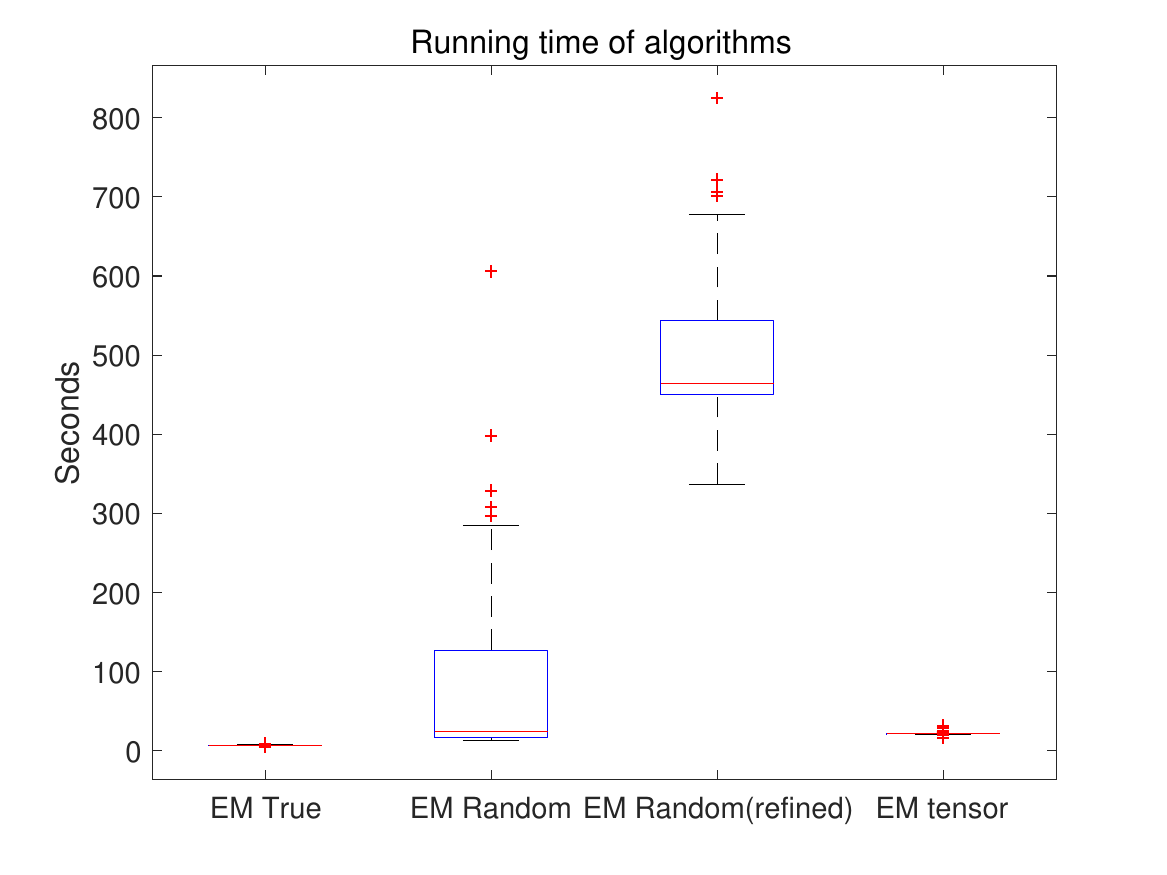}
		\end{minipage}%
	}%
	\centering
	\caption{$N = 10000, J= 200, L=10,$ item parameters $\in \{0.1,0.2,0.8,0.9\}$}
\end{figure}

\begin{figure}[H]
	\centering
	\subfigure[MSE of item parameters]{
		\begin{minipage}[t]{0.4\linewidth}
			\centering
			\includegraphics[width=2in]{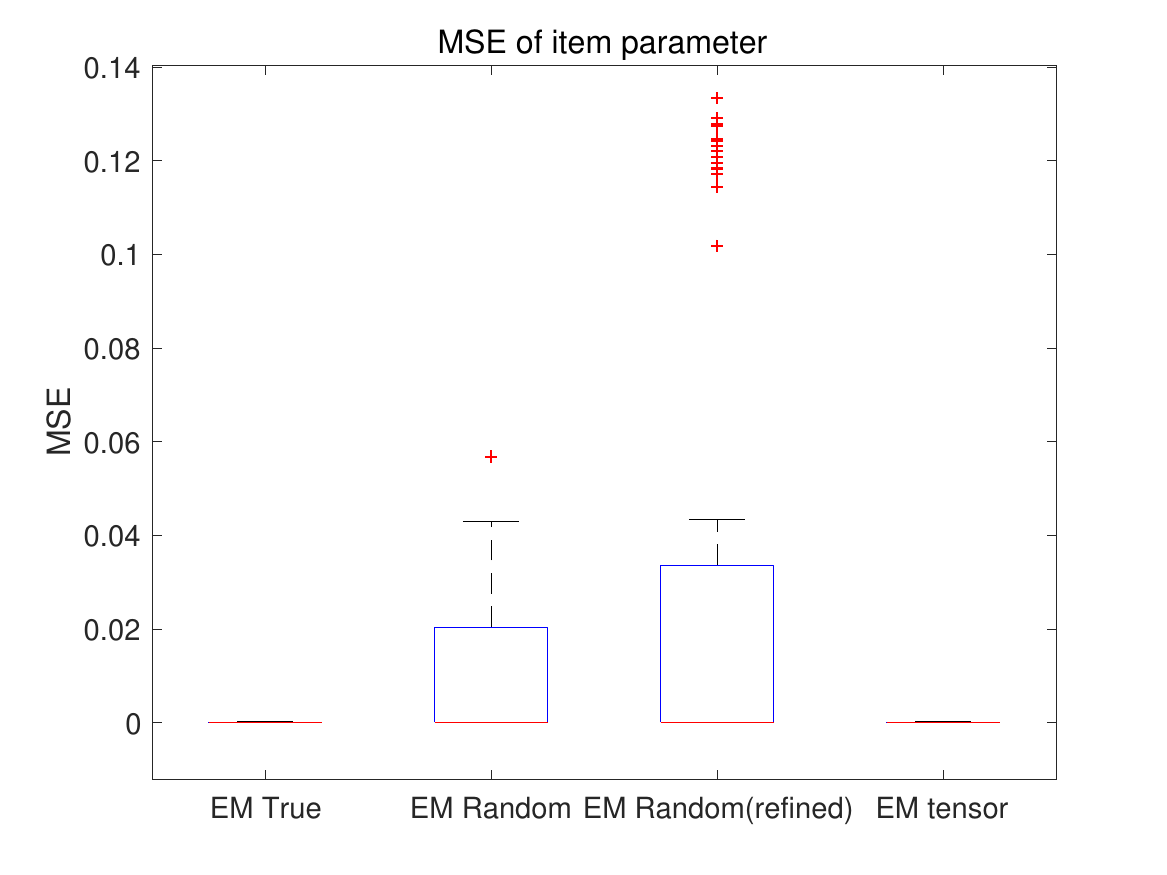}
		\end{minipage}%
	}%
	\subfigure[Running time of the algorithms]{
		\begin{minipage}[t]{0.4\linewidth}
			\centering
			\includegraphics[width=2in]{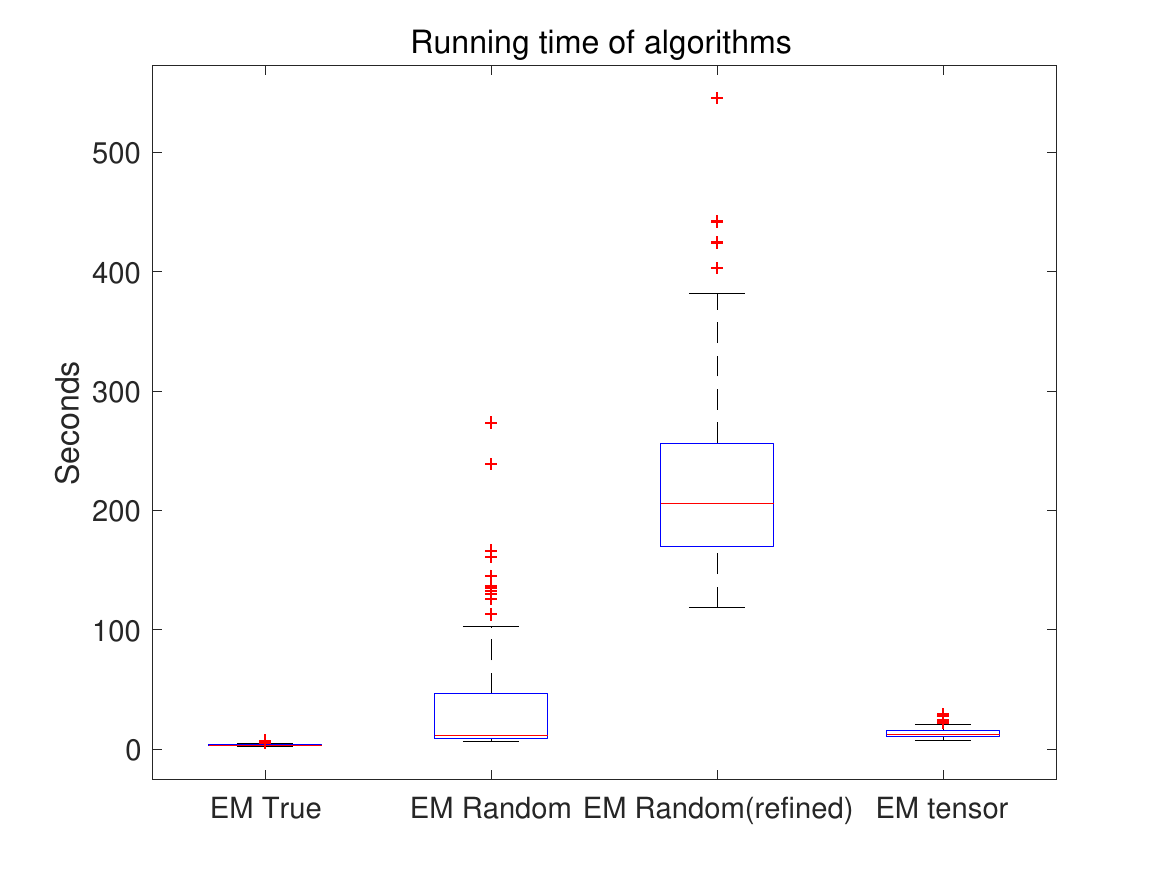}
		\end{minipage}%
	}%
	\centering
	\caption{$N = 10000, J= 100, L=5,$ item parameters $\in \{0.2,0.4,0.6,0.8\}$}
\end{figure}

\begin{figure}[H]
	\centering
	\subfigure[MSE of item parameters]{
		\begin{minipage}[t]{0.4\linewidth}
			\centering
			\includegraphics[width=2in]{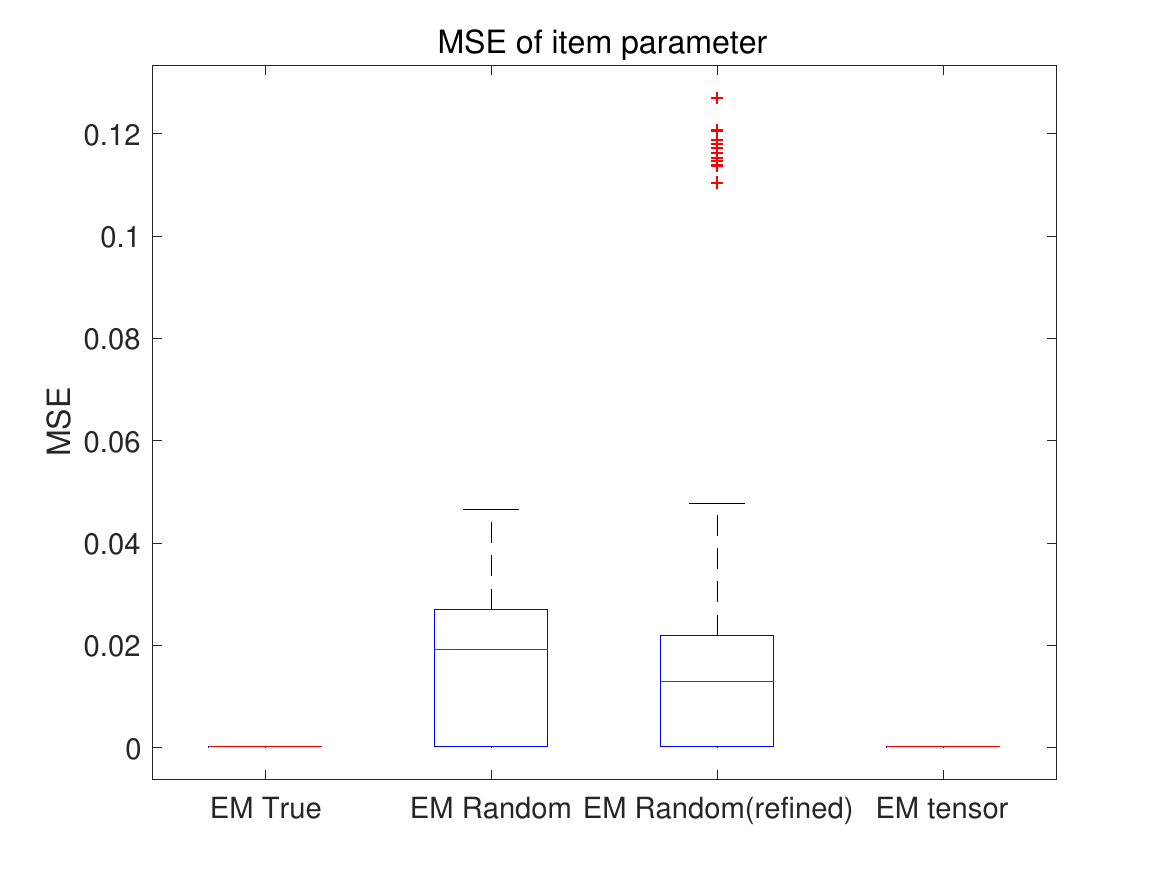}
		\end{minipage}%
	}%
	\subfigure[Running time of the algorithms]{
		\begin{minipage}[t]{0.4\linewidth}
			\centering
			\includegraphics[width=2in]{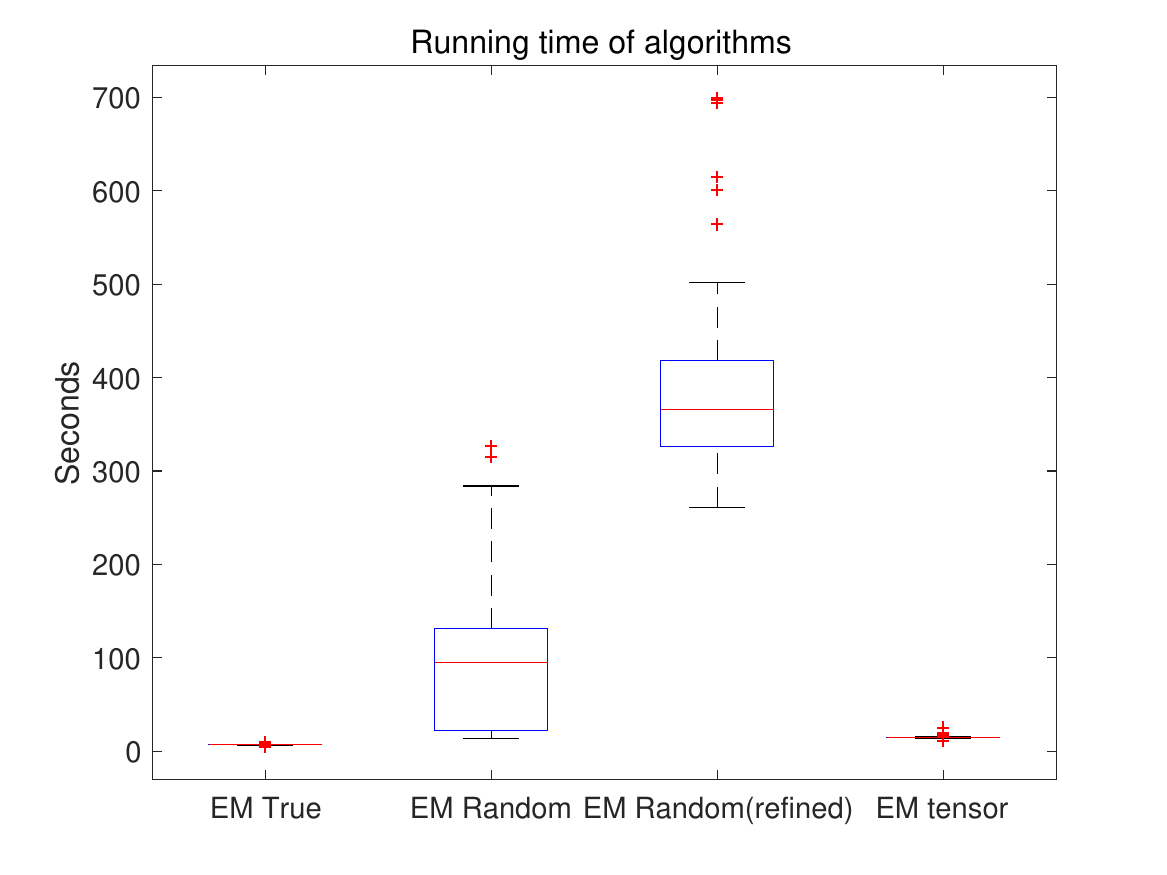}
		\end{minipage}%
	}%
	\centering
	\caption{$N = 10000, J= 100, L=10,$ item parameters $\in \{0.2,0.4,0.6,0.8\}$}
\end{figure}

\begin{figure}[H]
	\centering
	\subfigure[MSE of item parameters]{
		\begin{minipage}[t]{0.4\linewidth}
			\centering
			\includegraphics[width=2in]{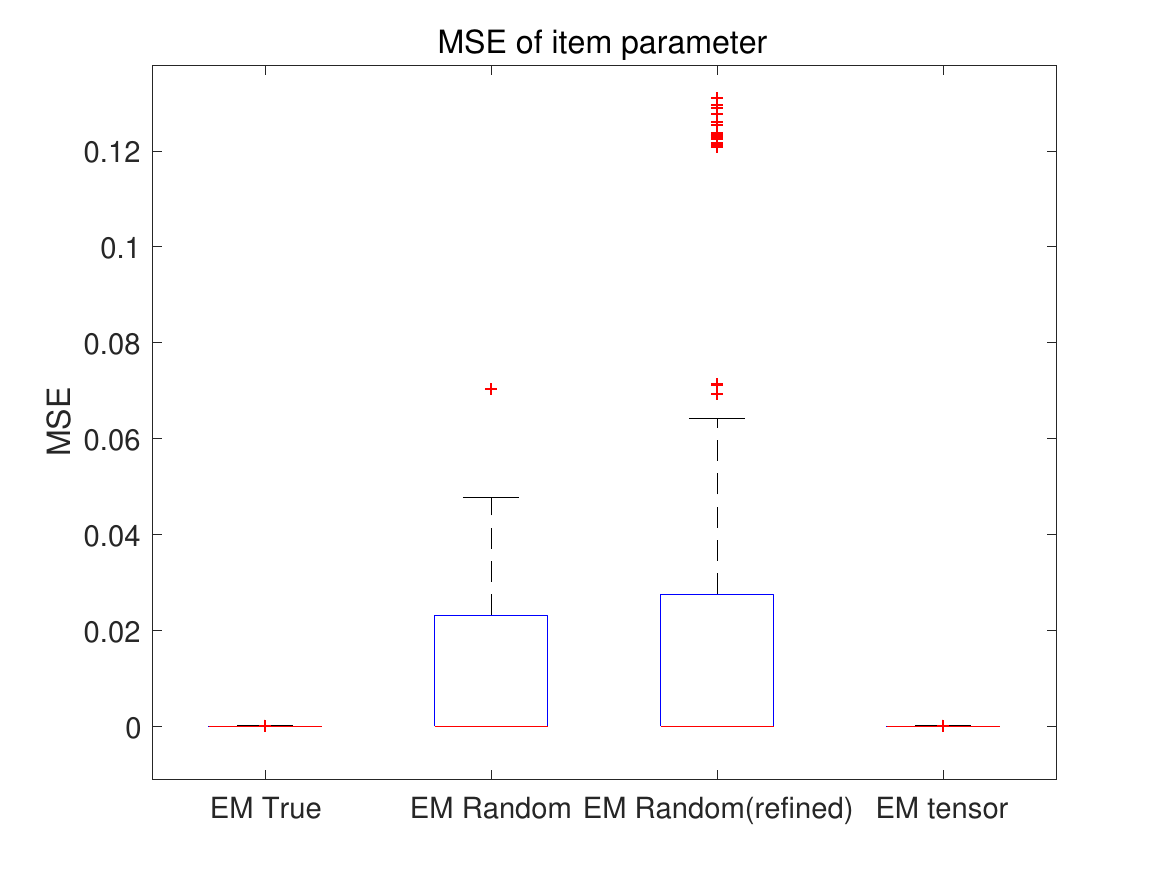}
		\end{minipage}%
	}%
	\subfigure[Running time of the algorithms]{
		\begin{minipage}[t]{0.4\linewidth}
			\centering
			\includegraphics[width=2in]{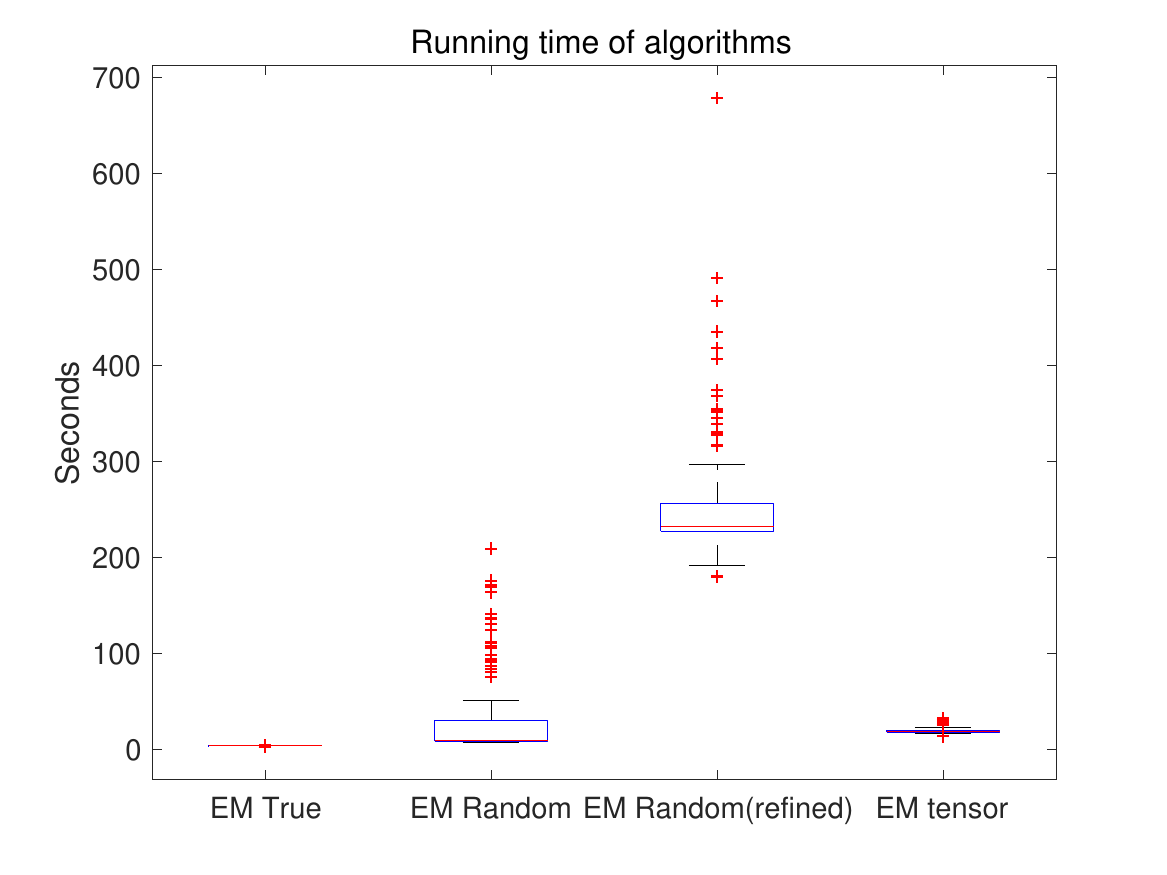}
		\end{minipage}%
	}%
	\centering
	\caption{$N = 10000, J= 200, L=5,$ item parameters $\in \{0.2,0.4,0.6,0.8\}$}
\end{figure}

\subsection*{Simulations under local dependence}

\begin{figure}[H]
	\centering
	\subfigure[MSE of item parameters $\rho=0.3$]{
		\begin{minipage}[t]{0.45\linewidth}
			\centering
			\includegraphics[width=2in]{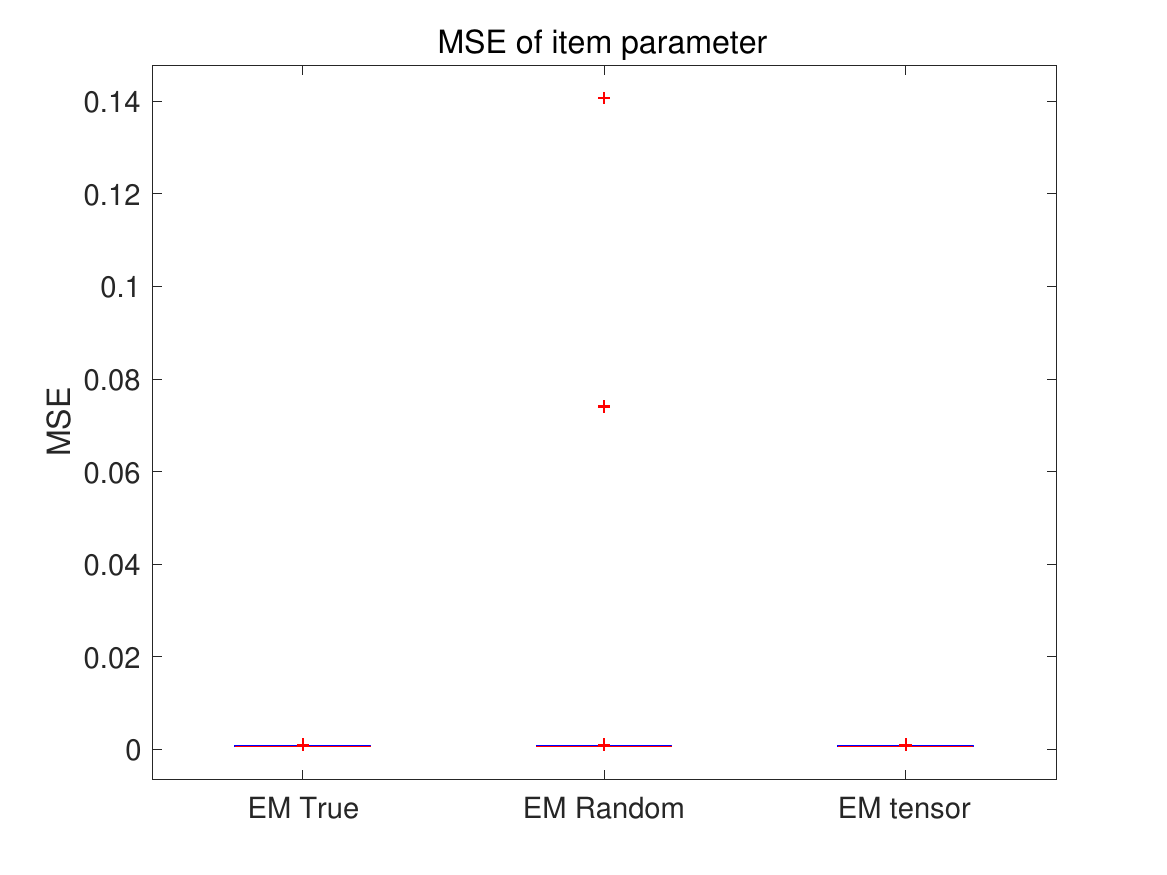}
		\end{minipage}%
	}%
	\subfigure[MSE of item parameters $\rho=0.7$]{
		\begin{minipage}[t]{0.45\linewidth}
			\centering
			\includegraphics[width=2in]{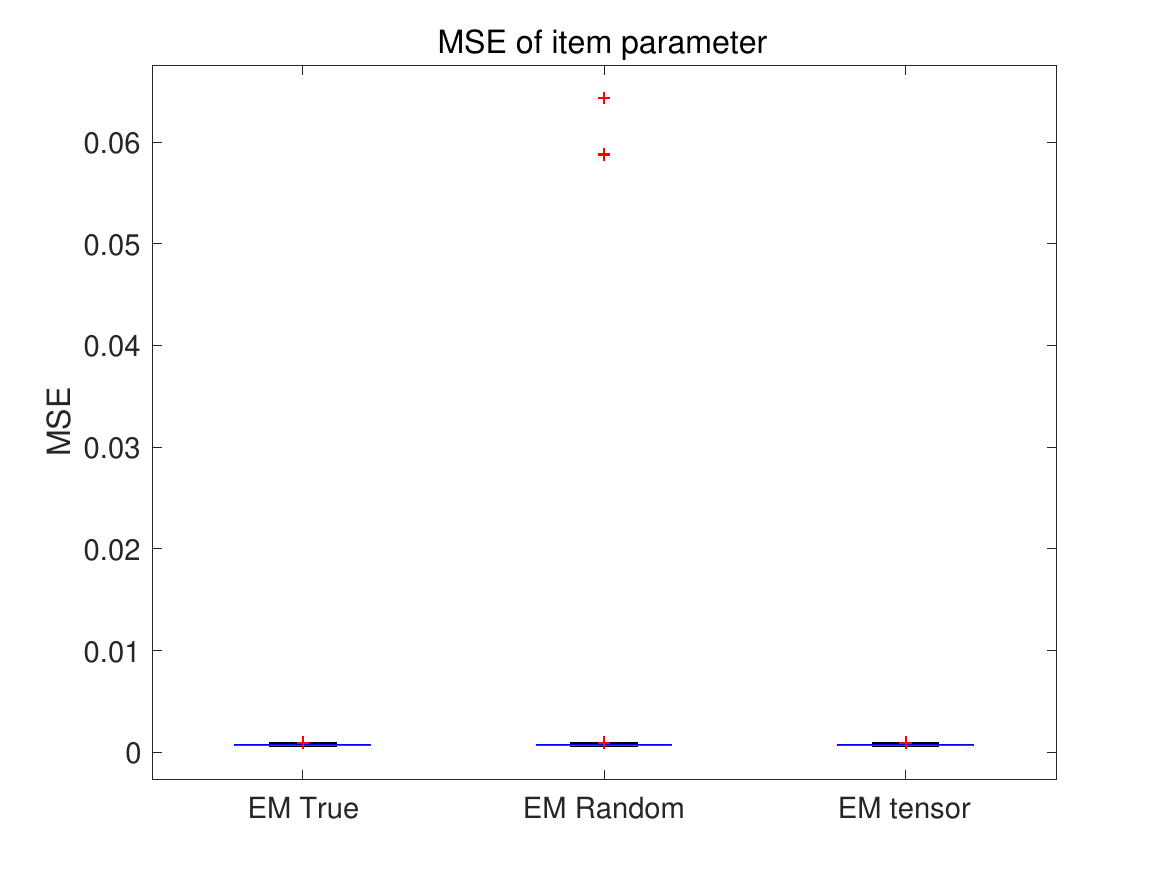}
		\end{minipage}%
	}%
	\centering
	\caption{Random-effect LCM, $N = 1000, J= 100, L=5, \theta_{j,a}\in \{0.1,0.2,0.8,0.9\}$}
\end{figure}

\begin{figure}[H]
	\centering
	\subfigure[MSE of item parameters $\rho=0.3$]{
		\begin{minipage}[t]{0.45\linewidth}
			\centering
			\includegraphics[width=2in]{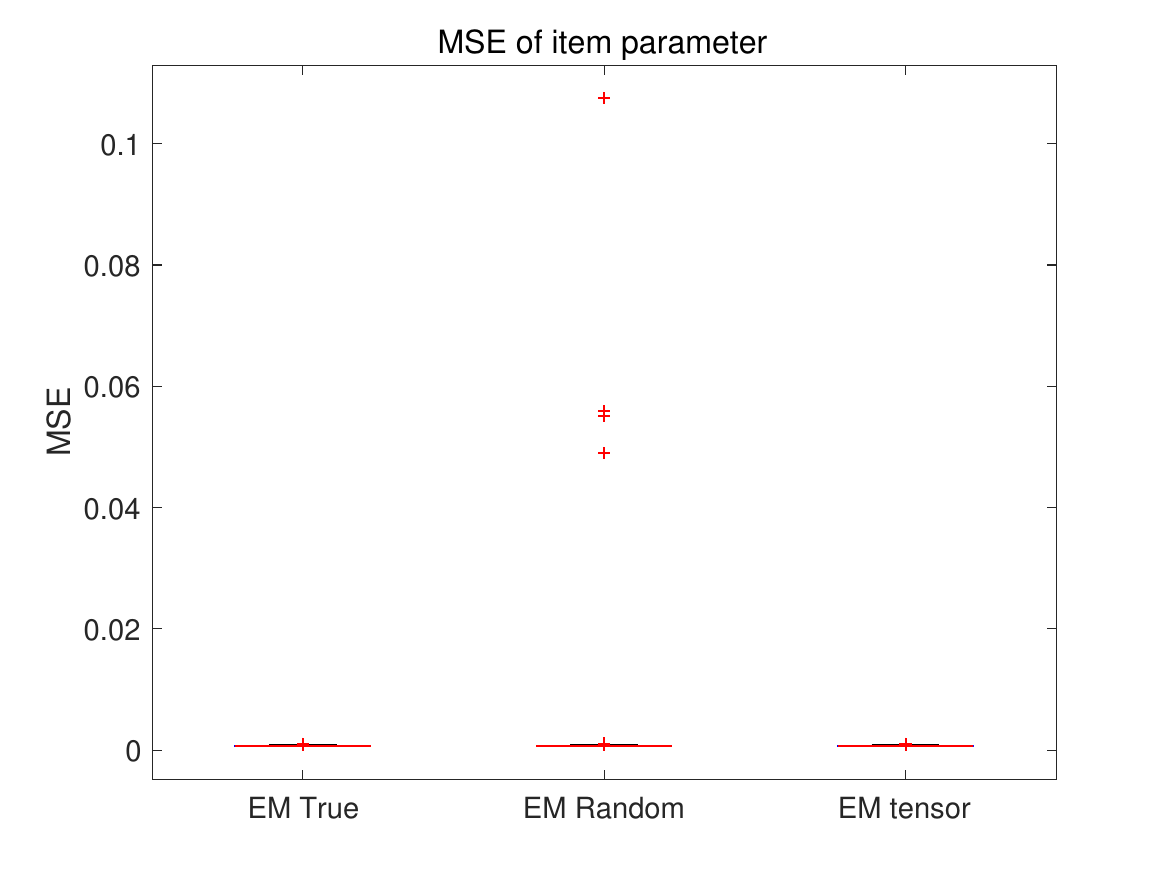}
		\end{minipage}%
	}%
	\subfigure[MSE of item parameters $\rho=0.7$]{
		\begin{minipage}[t]{0.45\linewidth}
			\centering
			\includegraphics[width=2in]{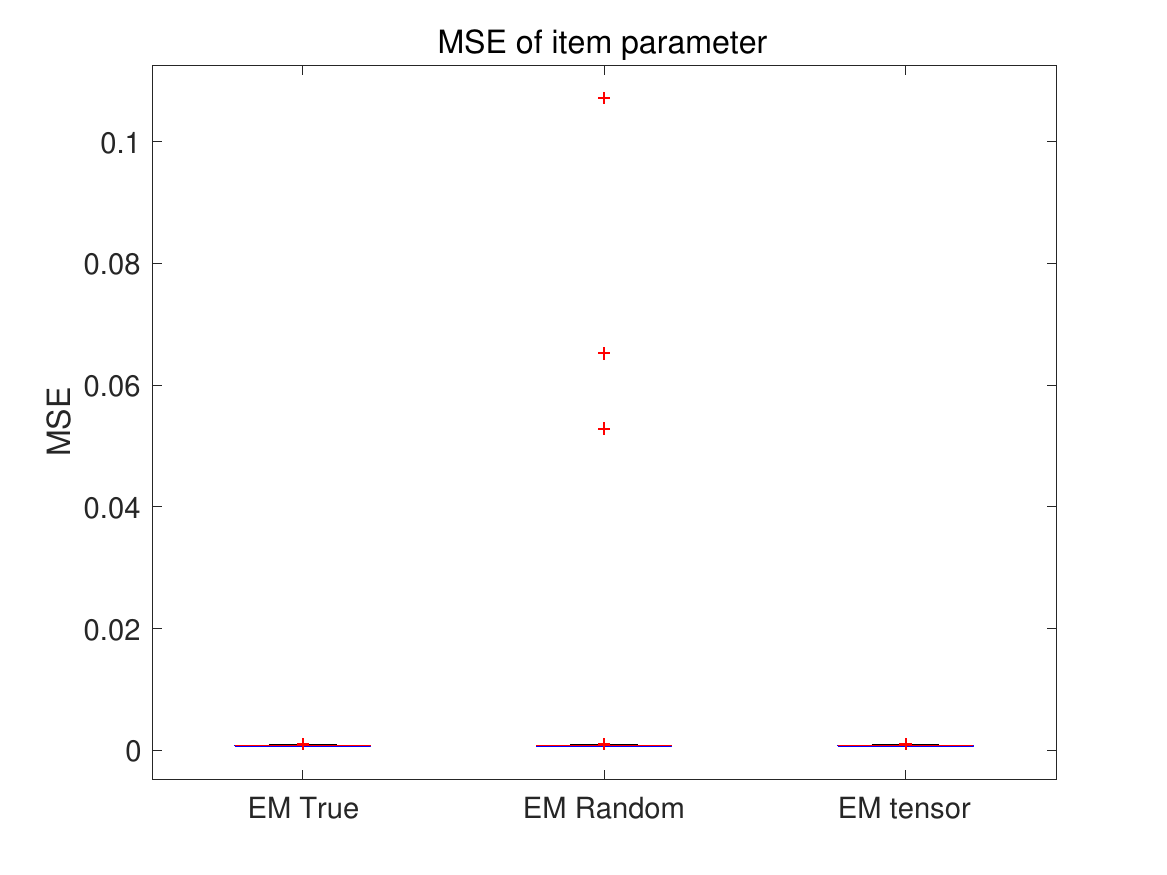}
		\end{minipage}%
	}%
	\centering
	\caption{Random-effect LCM, $N = 1000, J= 200, L=5, \theta_{j,a}\in \{0.1,0.2,0.8,0.9\}$}
\end{figure}

\begin{figure}[H]
	\centering
	\subfigure[MSE of item parameters $\rho=0.3$]{
		\begin{minipage}[t]{0.45\linewidth}
			\centering
			\includegraphics[width=2in]{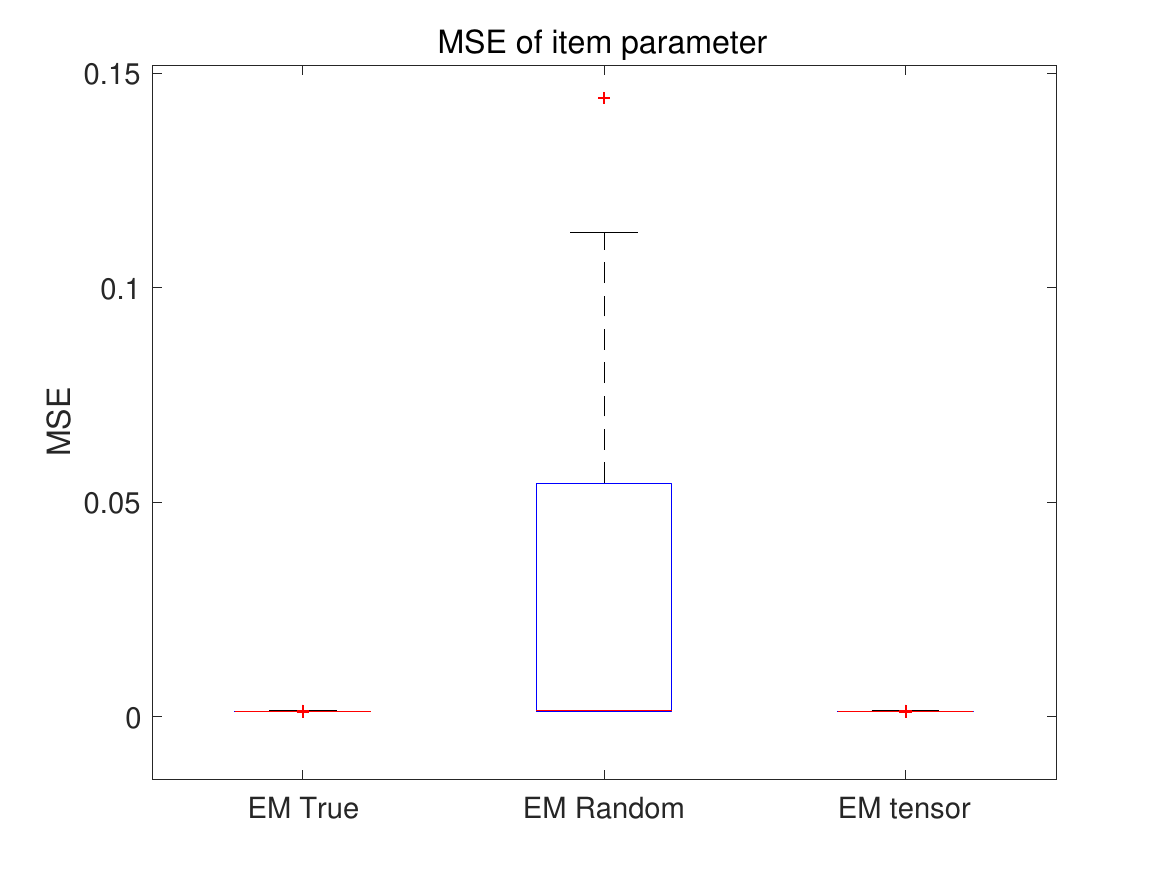}
		\end{minipage}%
	}%
	\subfigure[MSE of item parameters $\rho=0.7$]{
		\begin{minipage}[t]{0.45\linewidth}
			\centering
			\includegraphics[width=2in]{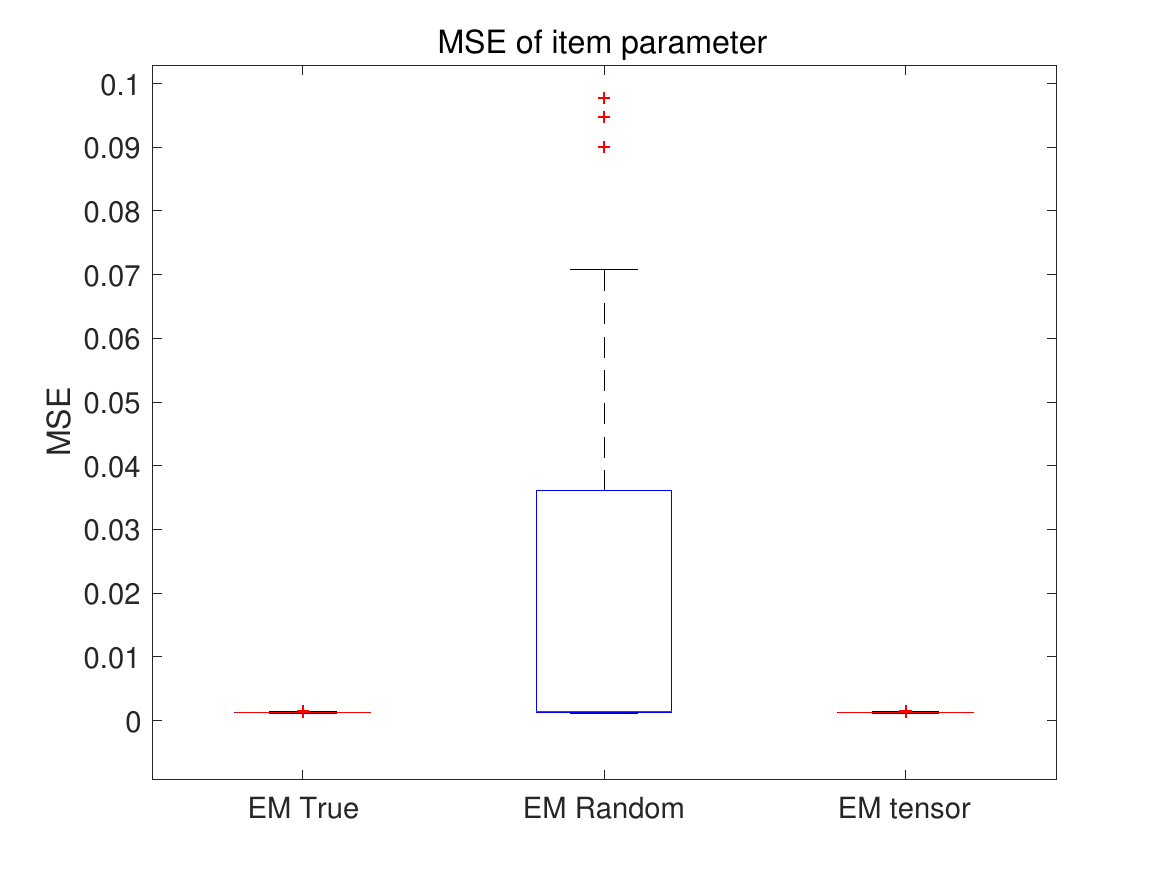}
		\end{minipage}%
	}%
	\centering
	\caption{Random-effect LCM, $N = 1000, J= 200, L=10, \theta_{j,a}\in \{0.1,0.2,0.8,0.9\}$}
\end{figure}

\begin{figure}[H]
	\centering
	\subfigure[MSE of item parameters $\rho=0.3$]{
		\begin{minipage}[t]{0.45\linewidth}
			\centering
			\includegraphics[width=2in]{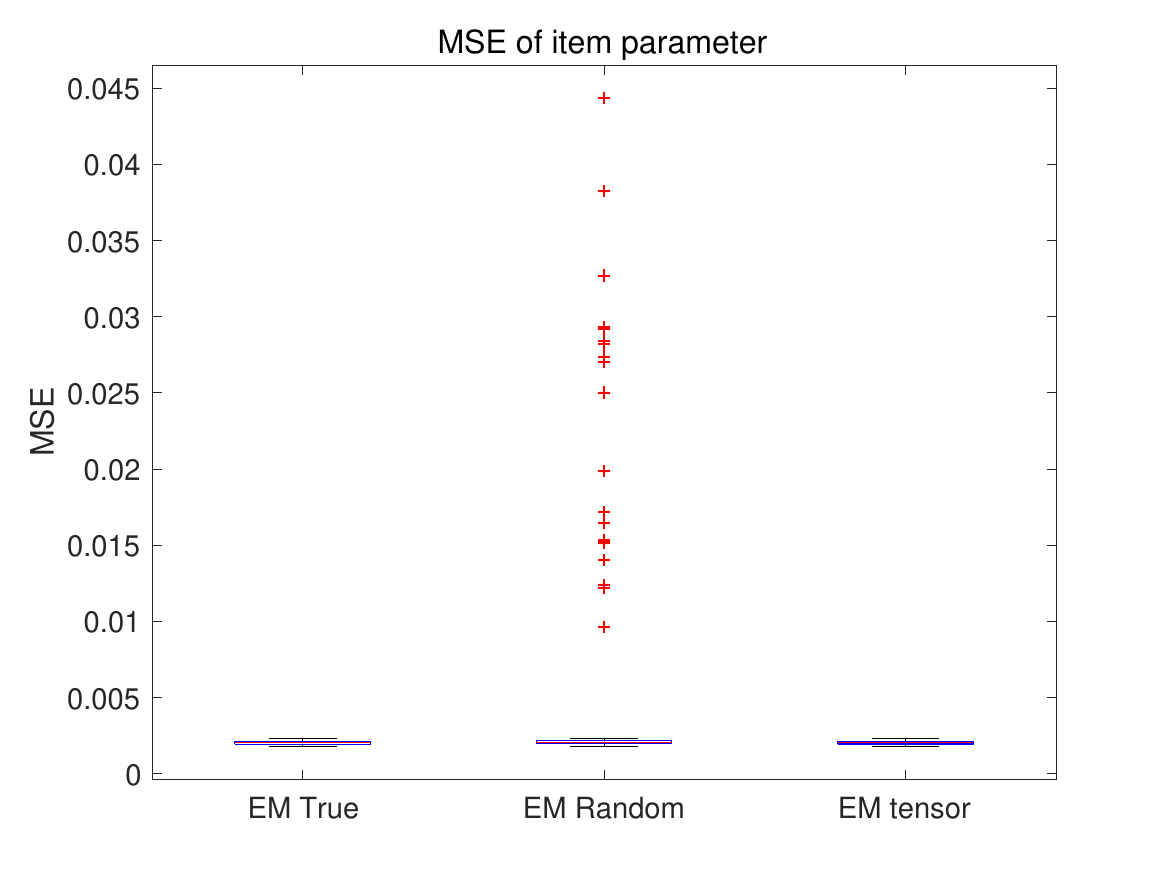}
		\end{minipage}%
	}%
	\subfigure[MSE of item parameters $\rho=0.7$]{
		\begin{minipage}[t]{0.45\linewidth}
			\centering
			\includegraphics[width=2in]{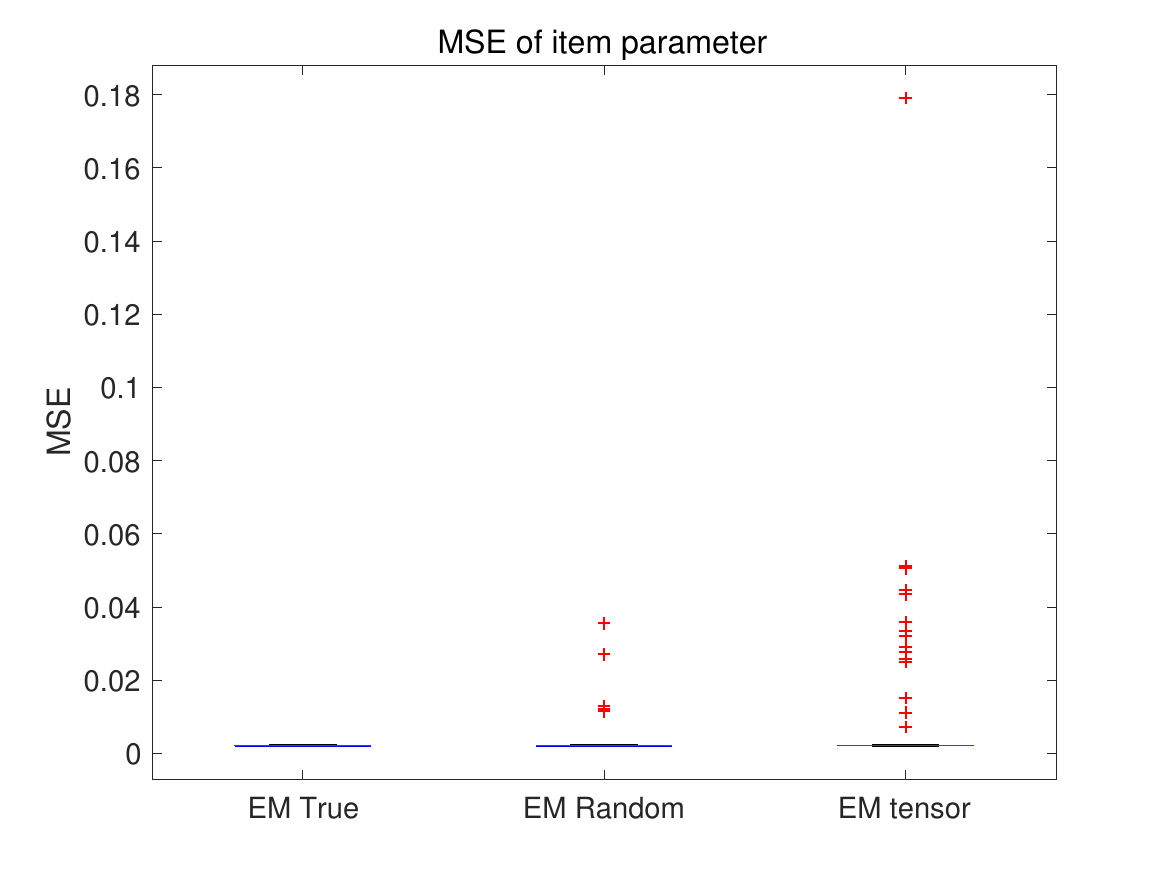}
		\end{minipage}%
	}%
	\centering
	\caption{Random-effect LCM, $N = 1000, J= 100, L=10, \theta_{j,a}\in \{0.2,0.4,0.6,0.8\}$}
\end{figure}

\begin{figure}[H]
	\centering
	\subfigure[MSE of item parameters $\rho=0.3$]{
		\begin{minipage}[t]{0.45\linewidth}
			\centering
			\includegraphics[width=2in]{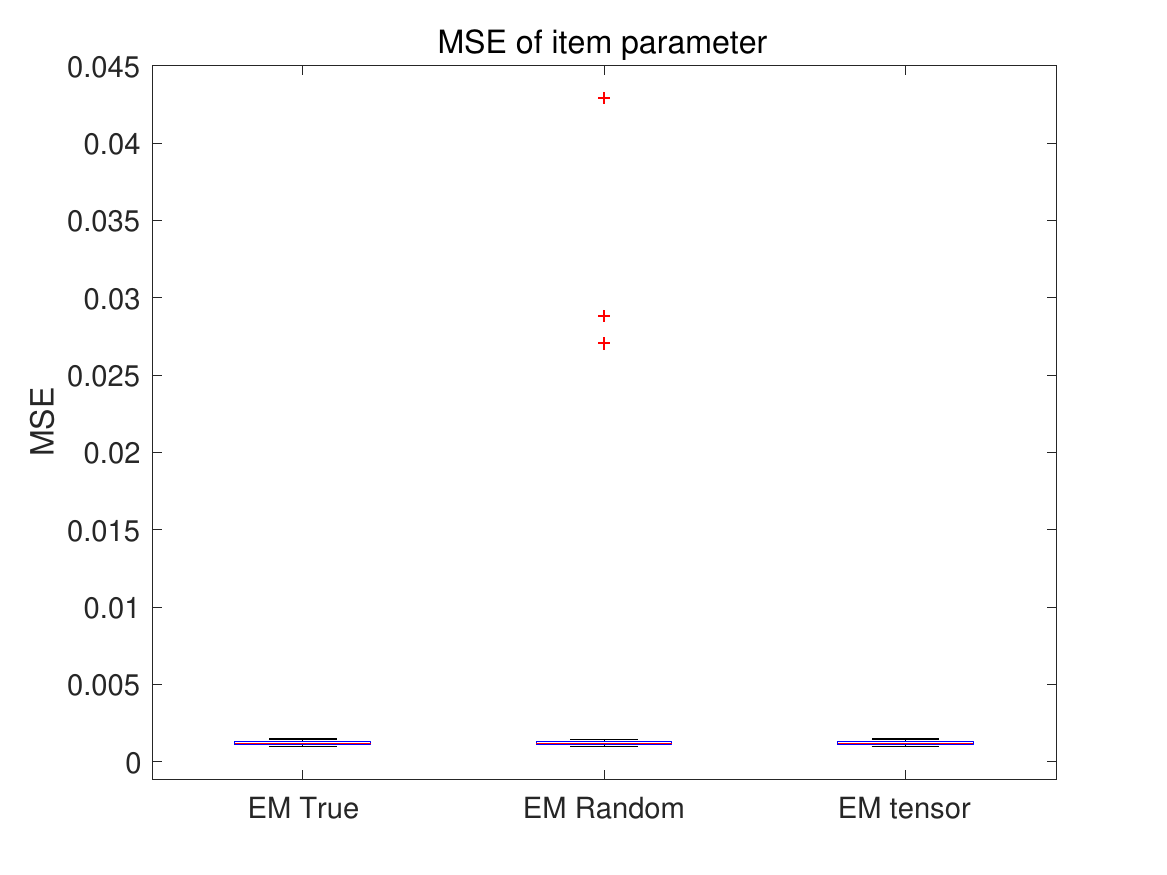}
		\end{minipage}%
	}%
	\subfigure[MSE of item parameters $\rho=0.7$]{
		\begin{minipage}[t]{0.45\linewidth}
			\centering
			\includegraphics[width=2in]{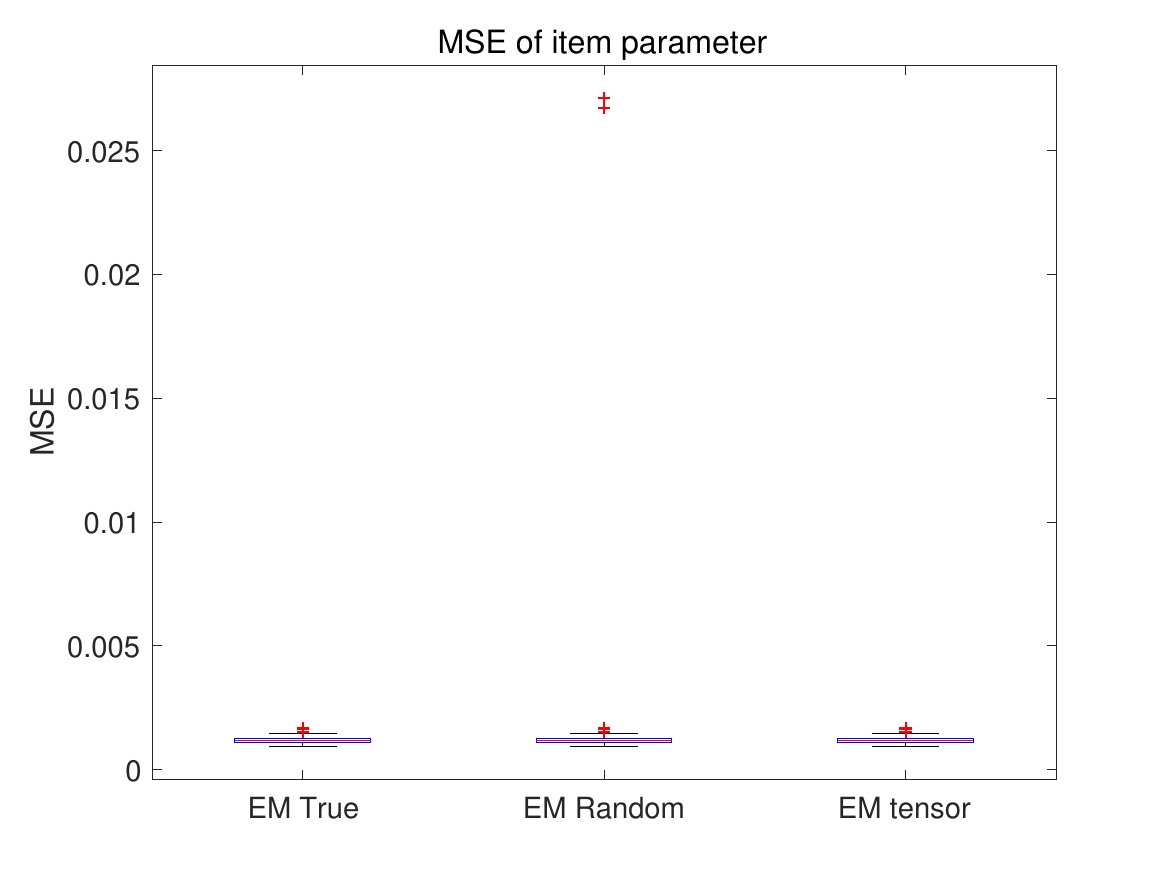}
		\end{minipage}%
	}%
	\centering
	\caption{Random-effect LCM, $N = 1000, J= 200, L=5, \theta_{j,a}\in \{0.2,0.4,0.6,0.8\}$}
\end{figure}

\begin{figure}[H]
	\centering
	\subfigure[MSE of item parameters $\rho=0.3$]{
		\begin{minipage}[t]{0.45\linewidth}
			\centering
			\includegraphics[width=2in]{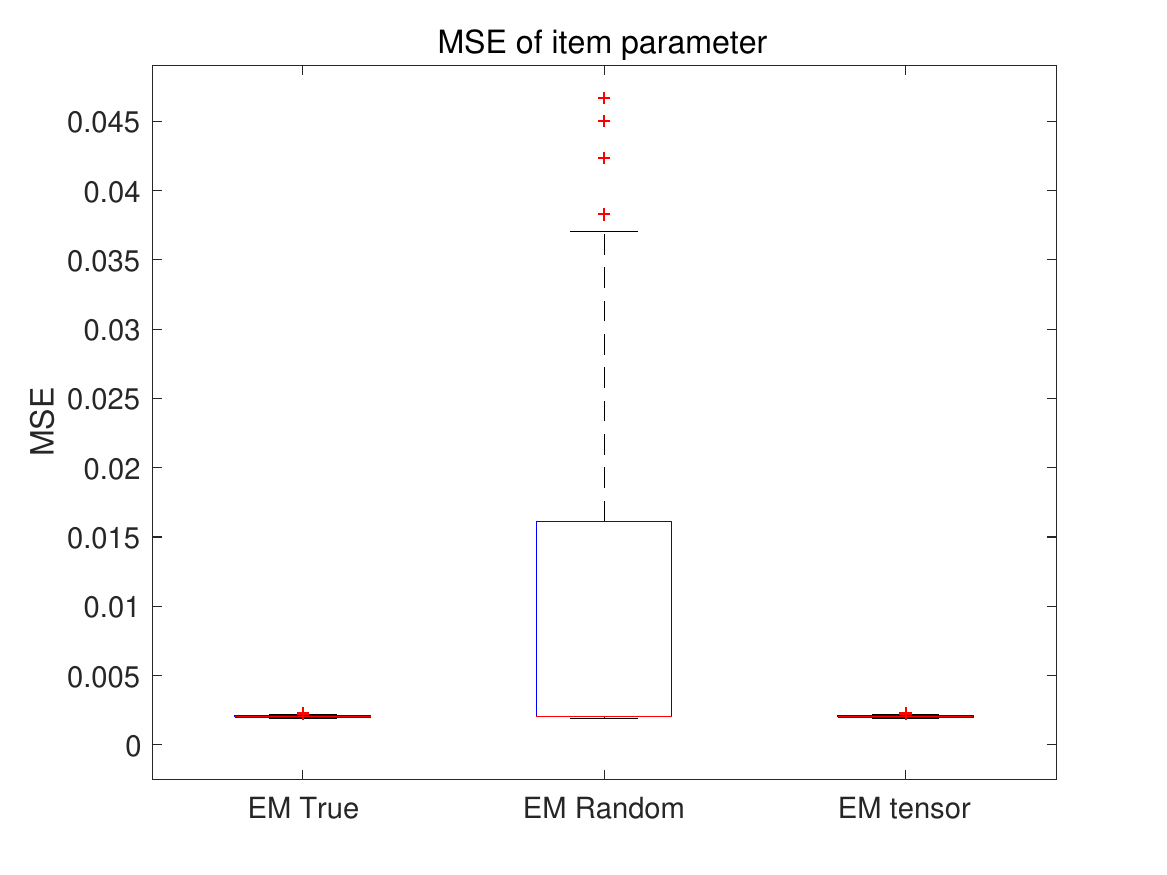}
		\end{minipage}%
	}%
	\subfigure[MSE of item parameters $\rho=0.7$]{
		\begin{minipage}[t]{0.45\linewidth}
			\centering
			\includegraphics[width=2in]{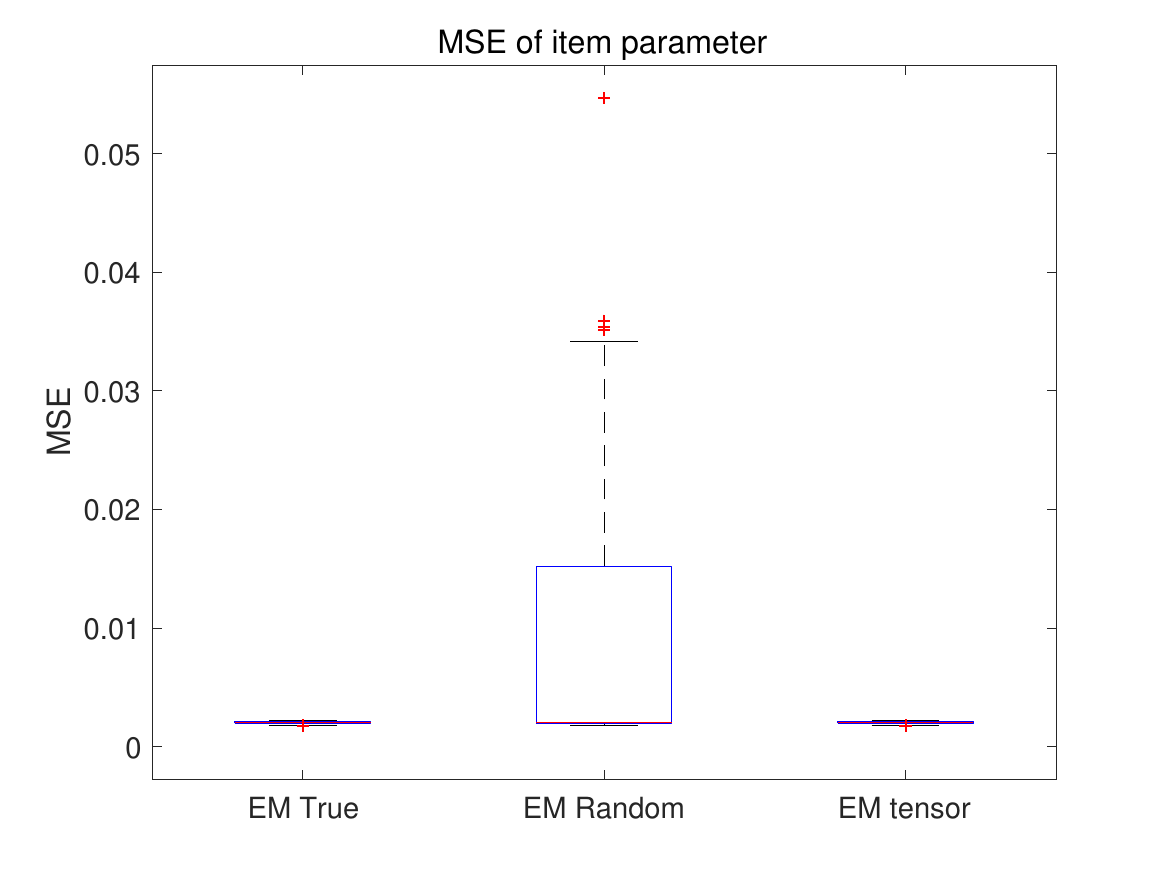}
		\end{minipage}%
	}%
	\centering
	\caption{Random-effect LCM, $N = 1000, J= 200, L=10, \theta_{j,a}\in \{0.2,0.4,0.6,0.8\}$}
\end{figure}
\end{document}